\documentclass [pdftex,b5paper,11pt,titlepage,oneside]{book}

\usepackage[b5paper,verbose]{geometry}
\usepackage[Lenny]{fncychap}
\usepackage{xpatch}
\makeatletter
\ChNameVar{\Large\rm\color{main_color}}

  \renewcommand{\DOTI}[1]{%
    \CTV\FmTi{#1}\par\nobreak 
    \vskip 20\p@}
  \renewcommand{\DOTIS}[1]{%
    \CTV\FmTi{#1}\par\nobreak 
    \vskip 20\p@}

\xpatchcmd\DOCH
  {\mghrulefill}{\color{black}\mghrulefill}
\xpatchcmd\DOTI
  {\mghrulefill}{\color{black}\mghrulefill}
  
\makeatother

\usepackage{titlesec}
\titlespacing*{\paragraph} {0pt}{1\baselineskip}{1\baselineskip}
  
\usepackage{graphicx}
\usepackage{caption}
\usepackage{subcaption}
\DeclareGraphicsExtensions{.pdf,.png,.jpeg,.mps,.bmp,.eps, .gif}
\usepackage{pdfpages}
\usepackage{epsfig}
\usepackage{tikz}
\usepackage{tikz-qtree}
\usepackage{picinpar}
\usepackage{bmpsize}
\usepackage{textcomp} 
\usetikzlibrary{trees} 
\usepackage[most]{tcolorbox}
\usepackage{xcolor}
\usepackage{scalerel}

\usepackage[utf8]{inputenc}
\usepackage[T1]{fontenc}
\usepackage{newlfont}
\usepackage{verbatim}
\usepackage{indentfirst}
\usepackage{ragged2e}
\usepackage{setspace}
\usepackage{chngpage}
\usepackage{ifoddpage}
\usepackage{textcomp}
\usepackage{calligra}
\usepackage{scrextend}
\usepackage[normalem]{ulem} 
\usepackage[backref=page]{hyperref} 
\DeclareUrlCommand\uri{}

\usepackage{threeparttable}
\usepackage[margin=6pt, font=small, labelfont=bf]{caption}
\hypersetup{
 pdftitle={},%
 pdfauthor={},%
 pdfsubject={},%
 pdfkeywords={},%
 colorlinks=true,%
 linkcolor=main_color,%
 urlcolor=main_color,%
 backref=page,%
 linktocpage=true,%
 pageanchor=true,%
 citecolor=main_color
 }
\renewcommand*{\backref}[1]{}
\renewcommand*{\backrefalt}[4]{{%
    \color{darkgray}
    \ifcase #1 Not cited.%
          \or Cited on page~#2.%
          \else Cited on pages #2.%
    \fi%
    }}
    
\usepackage{lipsum}  
\usepackage{blindtext}

\newlength{\defbaselineskip}
\newcommand{\setlinespacing}[1]%
           {\setlength{\baselineskip}{#1 \defbaselineskip}}
\usepackage [autostyle]{csquotes}
\MakeOuterQuote{"}
           
\usepackage{mathtools}
\usepackage{amsmath}
\usepackage{amsfonts}
\usepackage{amssymb}
\usepackage{amsthm}     
\usepackage{bm}
\usepackage{calc} 
\usepackage{calrsfs}
\usepackage{siunitx}
\usepackage{bigints} 
\usepackage{framed}
\usepackage{relsize} 
\usepackage{cases}   

\usepackage{tabularx}
\usepackage{tabulary}
\usepackage{colortbl}
\usepackage{array}
\usepackage{booktabs} 
\usepackage{multirow}
\usepackage{rotating}
\usepackage{supertabular}
\usepackage{longtable}

\usepackage[chapter]{algorithm}
\usepackage{algorithmic}
\newcommand{\algcomment}[1]{{\color{gray}\# #1}}

\usepackage[numbers,sort]{natbib}
\usepackage{acronym}
\usepackage{epigraph} 
\usepackage{nomencl}

\renewcommand{\nompreamble}{The following symbols are used within the thesis}
\makenomenclature
\usepackage{makeidx}

\definecolor{black}{rgb}{0, 0, 0}
\definecolor{cobalt}{rgb}{0.0, 0.28, 0.67}
\definecolor{darkblue}{rgb}{0.0, 0.0, 0.55}

\colorlet{base_color}{black}
\colorlet{main_color}{cobalt}
\colorlet{auxiliary_color}{cobalt}

\usepackage{fancyhdr}

\usepackage{etoolbox}
\makeatletter
\patchcmd{\footrule}
  {\if@fancyplain}
  {\color{main_color}\if@fancyplain}
  {}
\makeatother
\usepackage{enumitem}

\newcommand{\university}{Universit\`a degli Studi di Napoli Federico II }
\newcommand{\phdprogram}{Computational and Quantitative Biology }

\newcommand{\phdcycle}{XXXVI }
\newcommand{\phdchairman}{Prof. Michele Ceccarelli}
\newcommand{\thesistitle}{Modelling and Control of Spatial Behaviours in Multi-Agent Systems}
\newcommand{\thesissubtitle}{with Applications to Biology and Robotics}
\newcommand{\thesisauthor}{Andrea Giusti}
\newcommand{\advisor}{Prof. Mario di Bernardo}
\newcommand{\coadvisor}{Prof. Diego di Bernardo and prof. Marco Faella}

\newcommand{\addref}{{\color{red}[REF] }}
\newcommand{\add}[1]{{\color{red} #1}}
\newcommand{\note}[1]{{\color{red} [#1]}}
\newcommand{\new}[1]{{\color{blue} #1}}

\renewcommand{\note}[1]{}
\renewcommand{\add}[1]{}
\renewcommand{\addref}{\!\!}
\renewcommand{\new}[1]{#1}

\renewcommand{\vec}[1]{\mathbf{#1}}
\newcommand{\unitvec}[1]{ \vec{\hat{#1}} }
\newcommand{\func}[2]{\ \mathrm{#1}\!\left(#2\right)}
\newcommand{\norm}[1]{\left\lVert #1 \right\rVert}
\newcommand{\abs}[1]{\left\lvert #1 \right\rvert}
\renewcommand{\d}[0]{\mathrm{d}}
\newcommand{\x}[0]{$\times$}
\newcommand{\T}{^{\mathsf{T}}}

\newcommand{\B}[1]{\if#1\relax\bm{#1}\else\mathbf{#1}\fi} 
\newcommand{\R}[1]{\mathrm{#1}}						      
\newcommand{\C}[1]{\mathcal{#1}}
\newcommand{\BB}[1]{\mathbb{#1}}
\newcommand{\bigcdot}[0]{\boldsymbol{\cdot}}

\newtheorem{remark}{Remark}[chapter]
\newtheorem{theorem}{Theorem}[chapter]

\newtheorem{lemma}{Lemma}[chapter]
\newtheorem{proposition}{Proposition}[chapter]

\newtheorem{definition}{Definition}[chapter]
\newtheorem{assumption}{Assumption}[chapter]

\DeclareMathAlphabet{\mathcal}{OMS}{cmsy}{m}{n} 

\usepackage[pages=some]{background}
\backgroundsetup{
 scale=0.8,
 opacity=0.05,
 angle=0,
 contents={\includegraphics[width=0.5\paperwidth]{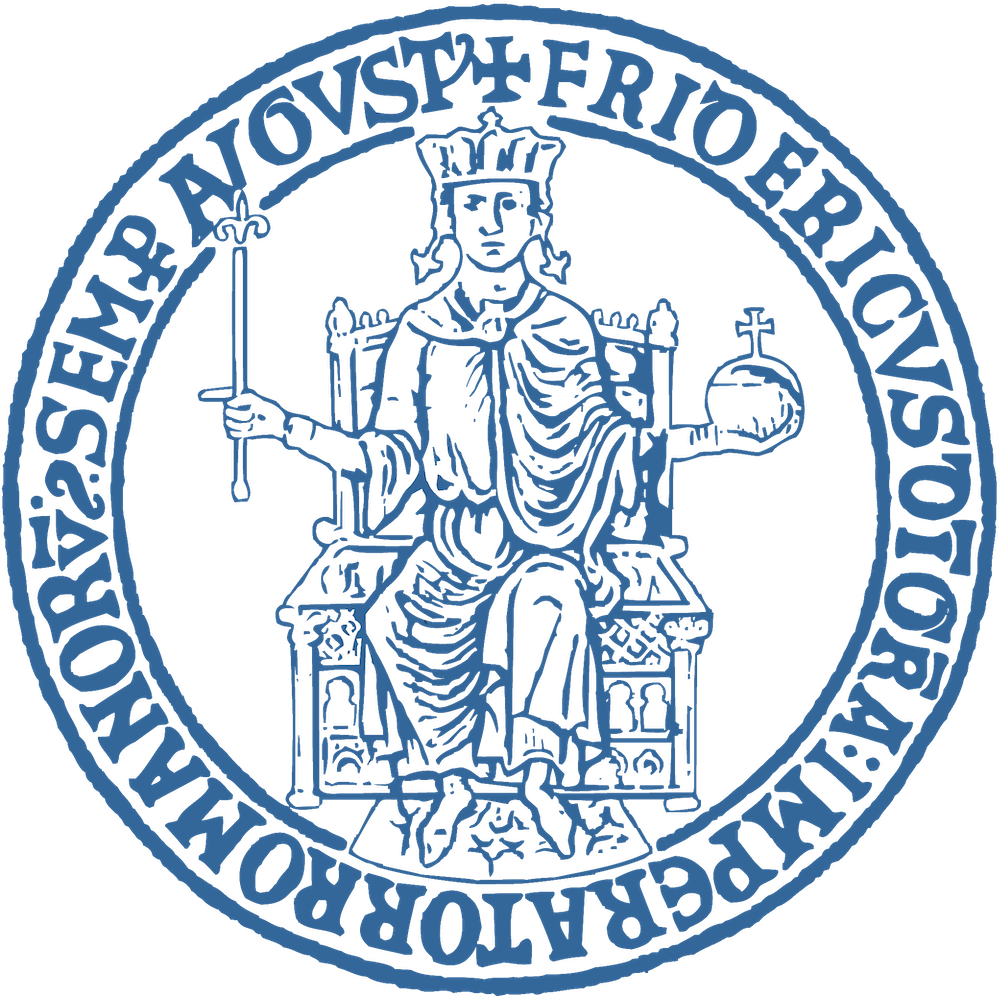}}
}
\makeatletter
    \def\cleardoublepage{\clearpage%
        \if@twoside
            \ifodd\c@page\else
                \vspace*{\fill}
                \hfill
	        \BgThispage
                \begin{center}
                \end{center}
                \vspace{\fill}
                \thispagestyle{empty}
                \newpage
                \if@twocolumn\hbox{}\newpage\fi
            \fi
        \fi
    }
\makeatother

\begin{document}
\frontmatter

	\begin{titlepage}
		\centering{\includegraphics[width=\columnwidth]{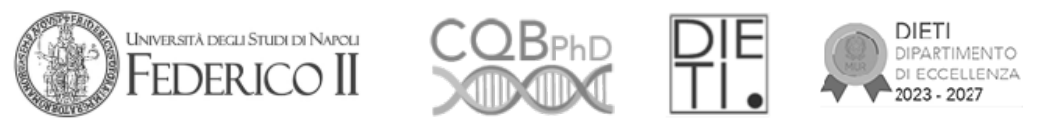}}\\

		\begin{center}
         	{\color{base_color}\Large{\university \par}}
		\vspace{-0.4cm}
		{\color{base_color}\singlespacing \large {Ph.D. Program in\\
		\color{auxiliary_color}C\color{base_color}omputational and \color{auxiliary_color}Q\color{base_color}uantitative  \color{auxiliary_color}B\color{base_color}iology \par\footnotesize{\vspace{3pt}\phdcycle Cycle\par}}}
		\vspace{1cm}
		{\color{base_color} \bfseries\large{\textsc{Thesis for the Degree of Doctor of Philosophy }}\par}
		\vspace{0.5 cm}
		{\color{main_color} \LARGE \bfseries \textsc{\thesistitle} \newline
        \large \textsc{\thesissubtitle} \par}
		{\color{main_color}\center{by\\\bfseries{\textsc\thesisauthor \par}
        \note{Annotated version}}}
		\end{center}
		\vspace{0.4 cm}
		\begin{flushleft}
			\color{base_color}{\normalsize{Advisor: \advisor\par}}
			\vspace{6pt}
			\color{base_color}{\small{Co-advisors: \coadvisor\par}}
		\end{flushleft}
		\vspace*{\fill}
		\begin{center}
			\centering{\includegraphics[width=0.2\columnwidth]{figs/logoUNINAblu.png}\par}
			\vspace{4pt}
			\color{base_color}{\bfseries\footnotesize\textsc{Scuola Politecnica e delle Scienze di Base}\par}
			\vspace{3pt}
			\bfseries\footnotesize\textsc{\color{auxiliary_color}D\color{base_color}{ipartimento di} \color{auxiliary_color}{I}\color{base_color}{ngegneria} \color{auxiliary_color}{E}\color{base_color}{lettrica e delle }\color{auxiliary_color}{T}\color{base_color}{ecnologie dell'}\color{auxiliary_color}{I}\color{base_color}{nformazione }}
		\end{center}
	\end{titlepage}

\newpage
 
\thispagestyle{empty} 
\centering
{\LARGE\textsc{\textcolor{main_color}{\textbf{\thesistitle\newline \large\thesissubtitle}}}}


\vspace{0.8cm}
\centering{\textcolor{base_color}{Ph.D. Thesis presented\\
for the fulfillment of the Degree of Doctor of Philosophy\\
in \phdprogram\\
by}}\\
 
\vspace{0.2cm}
\centering{\large\textsc{\textcolor{main_color}{\textbf{\thesisauthor}}}}

\vspace{0.7cm}
\centering{\large{December 2023}}\\

\vspace{0.3cm}
\centering{\includegraphics[width=0.2\columnwidth]{figs/logoUNINAblu.png}}\\
\vspace{1.cm}
\raggedright Approved as to style and content by \\

\vspace{0.5cm}
\raggedright \rule{4cm}{0.5pt} \\
\raggedright \advisor, Advisor


\vfill
\textcolor{main_color}{
\university \\
Ph.D. Program in \phdprogram\\
\scriptsize{\phdcycle cycle - Chairman: \phdchairman}
}
\clearpage

\normalsize

\thispagestyle{empty} 
\textbf{Candidate's declaration} \\

\vspace{0.5cm}
I hereby declare that this thesis submitted to obtain the academic degree of Philosophi\ae~Doctor (Ph.D.) in \phdprogram is my own unaided work, that I have not used other than the sources indicated, and that all direct and indirect sources are acknowledged as references. \\

Parts of this dissertation have been published in international journals and conference articles (see list of the author's publications at the end of the thesis).\\

\vspace{1cm}
Naples, \today

\vspace{1cm}
\rule{5cm}{0.5pt} 

\vspace{0.5cm}

\thesisauthor

\clearpage

\thispagestyle{empty} 

\vspace*{.25\textheight}

\begin{flushright}
\begin{minipage}{9cm}
\begin{flushright}
\textit{
Agli amici, lontani e vicini, perché mi ricordano che le cose importanti non sono nei libri. \\
Ai colleghi, con cui ho condiviso lunghe giornate ed innumerevoli pause caffè. \\
A chi mi ha insegnato a capire i dettagli, e chi invece mi ha insegnato a guardare oltre. \\
A chi ha condiviso la mia passione nel progettare e costruire qualcosa.
\paragraph{}
\noindent
A Franci, che c'era dall'inizio. \\
Ed alla mia famiglia, che in effetti, c'era da prima.
}
\end{flushright}
\end{minipage}
\end{flushright}

\thispagestyle{empty} 

\cleardoublepage

\pagenumbering{roman}

\phantomsection
\section*{Abstract}
\thispagestyle{empty} 
\setlength{\parindent}{1em}
\justifying
\addcontentsline{toc}{section}{Abstract}
Large-Scale Multi-Agent Systems (LS-MAS) consist of several autonomous components, interacting in a non-trivial way, so that the emerging behaviour of the ensemble depends on the individual dynamics of the components and their reciprocal interactions. 
These models can describe a rich variety of natural systems, as well as artificial ones, characterised by unparalleled scalability, robustness, and flexibility.
%
Indeed, a crucial objective is devising efficient strategies to model and control the spatial behaviours of LS-MAS to achieve specific goals.
However, the inherent complexity of these systems and the wide spectrum of their emerging behaviours pose significant challenges.
%
The overarching goal of this thesis is, therefore, to advance methods for modelling, analyzing and controlling the spatial behaviours of LS-MAS, with applications to cellular populations and swarm robotics. 
The thesis begins with an overview of the existing Literature, and is then organized into two distinct parts.
In the context of swarm robotics, Part \ref{part:geometric_pattern} deals with distributed control algorithms to spatially organize agents on geometric patterns.
The contribution is twofold, encompassing both the development of original control algorithms, and providing a novel formal analysis, which allows to guarantee the emergence of specific geometric patterns.
In Part \ref{part:spatial_microorganisms}, looking at the spatial behaviours of biological agents, experiments are  carried out  to study the movement of microorganisms and their response to light stimuli.
This allows the derivation and parametrization of mathematical models that capture these behaviours, and pave the way for the development of innovative approaches for the spatial control of microorganisms.
%
The results presented in the thesis were developed by leveraging formal analytical tools, simulations, and experiments, using innovative platforms and original computational frameworks.

\vfill
\noindent\textbf{\textit{Keywords}}: emerging behaviours, control theory, multi-agent systems, biology, swarm robotics.

\cleardoublepage

\phantomsection
\section*{Sintesi in lingua italiana}
\thispagestyle{empty} 
\setlength{\parindent}{1em}
\justifying
\addcontentsline{toc}{section}{Sintesi in lingua italiana}

I Sistemi Multi-Agente su Larga Scala (LS-MAS) sono costituiti da molteplici componenti autonome, che interagiscono tra di loro, così che il comportamento emergente dell'insieme dipenda dalla dinamica delle singole componenti e dalle loro reciproche interazioni. 
Questi modelli possono descrivere una ricca varietà di sistemi, sia naturali che artificiali, caratterizzati da scalabilità, robustezza e flessibilità. 
Infatti, lo sviluppo di strategie per modellare e controllare i comportamenti spaziali dei LS-MAS, rappresenta un problema cruciale. 
Tuttavia, l'intrinseca complessità  di questi sistemi e l’ampio spettro dei loro comportamenti, pongono grosse sfide. 
L'obiettivo di questa tesi è, quindi, quello di contribuire allo sviluppo di metodi per modellare, analizzare e controllare i comportamenti spaziali dei LS-MAS, con applicazioni alle popolazioni cellulari ed agli sciami robotici. 
La tesi inizia con una panoramica della Letteratura esistente, per poi essere organizzata in due parti distinte. 
Nel contesto della robotica degli sciami, la Parte \ref{part:geometric_pattern} tratta gli algoritmi di controllo distribuito per organizzare spazialmente gli agenti su pattern geometrici. 
Il duplice contributo, comprende sia lo sviluppo di nuovi algoritmi di controllo, sia l'introduzione di un'analisi formale, che consente di garantire l'emergere di specifici pattern. 
Nella Parte \ref{part:spatial_microorganisms}, al fine di esaminare i comportamenti spaziali di agenti biologici, vengono condotti esperimenti per studiare il movimento dei microrganismi e la loro risposta agli stimoli luminosi. 
Ciò consente la derivazione e la parametrizzazione di modelli matematici che catturano questi comportamenti, ed aprono la strada allo sviluppo di approcci per il controllo spaziale dei microrganismi. 
I risultati presentati nella tesi sono stati sviluppati sfruttando strumenti analitici, simulazioni ed esperimenti, utilizzando piattaforme innovative e strumenti computazionali originali.

\vfill
\noindent\textbf{\textit{Parole chiave}}: comportamenti emergenti, teoria del controllo, sistemi multi-agente, biologia, sciami robotici.

\cleardoublepage

\phantomsection
\pagestyle{plain}

\section*{Acknowledgements}





Part of the research presented in this dissertation was developed during a nine months period at the BioCompute Lab (University of Bristol), under the supervision of Dr. Thomas Gorochowski.

Moreover, it has been partially carried out in the framework of the projects PRIN 2017 “Advanced Network Control of Future Smart Grids” and PON INSIST 2019 "Sistema di monitoraggio intelligente per la sicurezza delle infrastrutture".

\addcontentsline{toc}{section}{Acknowledgements}
\cleardoublepage

\pagestyle{plain}
\tableofcontents
\clearpage

\phantomsection
\chapter*{List of Acronyms}
\noindent The following acronyms are used throughout the thesis.
\begin{acronym}[MAXACR]
\acro{DLP}{digital light projector}
\acro{DOME}{Dynamic Optical Micro-Environment}
\acro{LS-MAS}{Large-Scale Multi-Agent Systems}
\acro{OLS}{Ordinary Least Squares}
\acro{ODE}{Ordinary Differential Equation}
\acro{PDE}{Partial Differential Equation}
\acro{PTW}{Persistent Turing Walker}
\acro{SDE}{Stochastic Differential Equation}
\acro{UAV}{Unmanned Aerial Vehicle}
\end{acronym}


\addcontentsline{toc}{section}{List of Acronyms}
\clearpage





\pagestyle{fancy}
\pagenumbering{arabic}
\mainmatter
\setlength{\parindent}{1em}
\justifying
\singlespacing

\chapter{Introduction}
\thispagestyle{empty} 
\label{ch:intro}


\section{Motivation}

The behaviour of a rich variety of systems, in Nature and Technology, depends on the dynamics of multiple individual components and their reciprocal interactions. This is the case for both natural systems, such as gene networks in cells, flocks of birds, or global climate; and artificial ones, among others robotic swarms and internet \cite{Gorochowski2020,Schatten2014}. 
When the number of agents becomes extremely large, the role of their interconnections more intricate, and the interactions with the environment predominant, we talk of \acfi{LS-MAS}, or \emph{swarms} \cite{Brambilla2013}. 

The distinctive trait of such systems lies in the emergence of complex behaviours, stemming from the interactions among the individual agents.
Examples of such \emph{emerging behaviours} are the synchronized flashing of fireflies, the flocking of birds \cite{Ballerini2008}, or the formation of patterns by swarming bacteria \cite{Doshi2022}.
Furthermore, artificial multi-robot systems can be engineered to replicate these behaviours, and to address complex tasks, such as  manipulation or patrolling, in a cooperative manner. 
The coexistence and the interplay, between the \emph{microscopic} dynamics of the individual agents, and the emerging \emph{macroscopic} behaviour of the swarm, define the characteristic \emph{multi-scale} nature of these systems.
Crucially, the distributed nature of \ac{LS-MAS} guarantees \emph{scalability}, i.e. the ability to cope with large and often fluctuating numbers of agents, \emph{robustness} to agents' failure and various disturbances, and \emph{flexibility} to adapt to different environments and tasks \cite{Brambilla2013}.
These features catalysed interests around such systems, and fostered the development of new technologies and methodologies, with applications spanning from synthetic biology \cite{Gorochowski2020} to swarm robotics \cite{Brambilla2013}, from opinion dynamics \cite{Jusup2022} to infrastructures’ resilience \cite{Guerrero-Bonilla2019}.

\note{add figure esempi LSMAS}

\section{Open challenges}
A pressing problem for the deployment of real world applications is devising efficient strategies to model \ac{LS-MAS} and control their spatial behaviour to achieve specific goals, e.g. microbial pattern formation or swarming robots coordination.
However, the inherent complexity of such systems, and the wide spectrum of their possible emerging behaviours pose significant challenges to the systematic study of their dynamics and the design of general methodologies to, either, model or control their behaviour.
%
A first difficulty lies in the formal description and classification of the different behaviours that can emerge, which, in many cases, have not been clearly classified. 

The lack of general tools for the study of emerging behaviours in \ac{LS-MAS} constitutes a second challenge.
Conventional approaches from control theory can fail in guaranteeing stability, scalability and robustness, primarily due to the high dimensionality and the often non-linear nature of \ac{LS-MAS}; these difficulties worsening in the presence of a variable number of agents and uncertainties.
In particular, the study of stability offers unparalleled insights into the nature of dynamical systems and has been the pillar of this field, nevertheless the complexity of \ac{LS-MAS} often caused this formal analysis trudging behind \addref. Indeed, many emerging behaviours have been studied and a variety of control algorithms has been validated by simulations or experiments, while the analytical study of their stability remains unaddressed.

Moreover, we still lack effective methodologies to build mathematical models, able of describing the intricate behaviours of real-world \ac{LS-MAS}, such as the coordinated motion of fishes and birds, or the cooperation in insects colonies.
The nature of such systems poses great difficulties to the bottom-up development of models form first principles.
Therefore, we have to rely on top-down approaches, such as data-driven modelling \cite{Brunton2016, Stephany2022}.

Lastly, the design of distributed control algorithms to achieve desired emerging behaviours has been addressed primarily in specific cases, with a more comprehensive design framework still missing.
This is certainly due to intrinsic difficulties in developing control strategies able to close a loop between the microscopic and the macroscopic scales, but also the lack of models for real-world \ac{LS-MAS} represents a significant obstacle.

\section{Main contributions}
In this context, this thesis aims to contribute to the development of novel and effective methods for modelling, analyzing and controlling \ac{LS-MAS} of mobile agents, with a focus on studying their spatial behaviours for applications to cellular populations and swarm robotics. 
To enhance the understanding of these emerging behaviours, and of the existing modelling and control approaches, we provide a concise review of the existing Literature.
From this, we move to tackle the design and the formal analysis of distributed control algorithms for pattern formation.
Specifically, we develop an interaction law that allows a swarm to self organize onto geometric patterns, such as triangular and square lattices.
While this approach is validated via simulations and experiments, we apply formal tools to a similar, and very popular control algorithm, proving its local stability. 
Furthermore, we address the problem of constructing mathematical models able to describe the spatial behaviour of real world \ac{LS-MAS}.
In particular, starting from experimental data we acquired, we characterize and then model, the movement and the light response of some microorganisms, in order to provide the first step towards the development of new methodologies to control their spatial distribution.

To develop our results we employed formal tools, such as graph theory and Lyapunov stability; and developed a novel simulation platform, called  SwarmSim, to carry out agent-based simulations of  \ac{LS-MAS}.
Also, we performed experiments, using two innovative platforms, the Robotarium for swarm robotics applications and the \acs{DOME} for biological agents.
Overall, this thesis addresses general open problems in the field of \ac{LS-MAS} and their application to robotic and biological agents, providing both novel methodological results and new numerical tools. 


\section{Thesis outline}
The thesis is organized as follows; Chapter \ref{ch:background} provides a brief overview of the literature on \ac{LS-MAS} of mobile agents, presenting the most relevant aspects of modelling, classifying, analyzing, and controlling their emerging behaviours.
In particular we focus on those behaviours that influence the spatial organization of these systems. 
The rest of the thesis is then organized in two parts. Part \ref{part:geometric_pattern} concerns geometric pattern formation, and specifically the development of novel swarming strategies and their proof of stability, in the context of robotics applications.
Specifically, Chapter \ref{ch:patt_form_background}  introduces the problem of geometric pattern formation, reviews the existing Literature and the necessary preliminary mathematical formalism.
Chapter \ref{ch:dist_cont} presents our own solution to achieve this behaviour with a novel distributed approach, while the   formal study on the stability of specific geometric configurations is discussed in Chapter \ref{ch:convergence}.

Part \ref{part:spatial_microorganisms}, deals with a biological application, specifically the movement of microorganisms (protozoa and microscopic algae), and their response to light stimuli. 
The problem is first introduced and described in Chapter \ref{ch:miroorganisms_background}, while Chapter \ref{ch:dome} presents our experimental methodology and results.
Finally, Chapter \ref{ch:modelling} discusses how to mathematically model such behaviours, and how these can be leveraged to control the spatial distribution of the microorganisms.
Conclusions are drawn in Chapter \ref{ch:conclusions}.

Two software tools we developed in the process of our research are presented in the Appendices.
Specifically, Appendix \ref{ch:swarmsim} presents SwarmSim, an agent-based simulator we developed for the study of multi-agent systems, while  Appendix \ref{ch:tracker} presents our computer vision software for the automatic detection and tracking of moving objects in the \acs{DOME}.

\paragraph{}
The results reported in this thesis appeared in the articles listed in the
\hyperref[ch:author_publications]{Author's publications} section.




\chapter{Background}
\thispagestyle{empty} 
\label{ch:background}
   
     

As introduced in Chapter \ref{ch:intro}, \acfi{LS-MAS} consist of several autonomous components, or \emph{agents}, interacting in a non-trivial way \cite{Simon1962}, so that the emerging behaviour of the ensemble depends on the individual dynamics of the components and their reciprocal interactions, and vice versa \cite{Estrada2023}.
A strictly related concept is that of \emph{swarm}, which is mostly used in robotics to describe a \ac{LS-MAS} of mobile robots moving together in some coordinated manner and emphasizes the idea of physically embodied agents \cite{Brambilla2013}.
In the rest of the thesis we will use these two terms interchangeably.

\ac{LS-MAS} can describe a rich variety of natural systems, as well as artificial ones. 
Moreover, such systems show desirable properties, namely scalability, robustness and flexibility.
These aspects together motivate the interest in understanding, modelling and controlling their behaviour.
In this Chapter, we will discuss what characterizes the emergent behaviours of \ac{LS-MAS} and how these can be classified, focusing on those behaviours that influence the spatial distribution of the swarm.
Then, we will introduce the possible modelling approaches for \ac{LS-MAS}, and how these can be used to understand the emerging properties of the system.
Finally, an overview of the existing control approaches is presented.

\section{Emergent behaviours}
\label{sec:background_behaviours}
A distinctive characteristic of \ac{LS-MAS} is the wide range of complex emergent behaviors they can show. Unlike ordinary systems, whose steady-state behavior can generally be described in terms of equilibrium points, limit cycles, quasi-periodic or chaotic attractors \cite{Strogatz1994}, in \ac{LS-MAS} the relative behaviors of agents can determine the origin of complex emerging properties. 
The simpler of emerging behaviours are consensus and synchronization \cite{Russo2009}. 
Other examples include aggregation, flocking, area coverage, morphogenesis and more \cite{Brambilla2013,Majid2022}.

Various classifications of these behaviours have been proposed, but a general consensus is still lacking.
Overall, one of the most influential was proposed by Brambilla et al. \cite{Brambilla2013}, where emergent behaviours are classified in \emph{spatial organization}, \emph{navigation}, \emph{collective decision making} and \emph{others}.

Here, we do not aim to propose a new general classification, but try to adjust the existing ones to define the class of \emph{spatial behaviours} of our interest.
Specifically, we will say spatial behaviours are those characterized by a steady state configuration that satisfies some geometric constraints on the states of the agents.
These will include consensus, aggregation, density regulation and more (see below).
Contrarily, we will not include behaviours explicitly involving persistent movement%
\footnote{This does not require the agents to necessarily become static, as long as their motion does not influence the resulting steady state behaviour.}
(e.g. flocking and synchronization), or more complex ones requiring interactions with the external environment, such as object assembling (that instead is included within \emph{spatial organizing behaviours} in \cite{Brambilla2013}), cooperative transportation, mapping, herding, etc.

Moreover, some of the spatial behaviors, despite being treated in a vast existing Literature, have been mostly defined heuristically and ambiguously, with different authors adopting different definitions, e.g. \cite{Brambilla2013, Majid2022}.
Therefore, in the following, we list the spatial behaviours we identified, together with the proposed definitions and some possible applications.
We will loosely use $\vec{x}_i \in  \mathbb{R}^d$ to represent the state of agent $i$, and $\rho(\vec{x},t):\mathbb{R}^d \times \mathbb{R}^+ \rightarrow \mathbb{R}^+$ for the density distribution of the agents in a state space, of dimension $d$.
Moreover, we will implicitly refer to the steady state behaviour, and therefore omit the dependence on time:

\begin{itemize}
\item \emph{Consensus} 
describes the convergence of all the agents to the same equilibrium state, formally  $\vec{x}_i=\vec{x}_j \ \forall i,j$  with $\dot{\vec{x}}_i=0$. 
Possible examples encompass rendezvous \cite{Ji2007}, the alignment of magnetosomes in magnetotactic bacteria or of schooling fishes \cite{Couzin2005}.
Moreover, if the state of the agents represents their opinion, consensus can represent unanimous decisions in social networks \cite{Valentini2017}. 

\item \emph{Aggregation} 
is a weaker form of consensus, requiring all the agents to converge towards a bounded set of their state space, formally guaranteeing that eventually $\Vert \vec{x}_i-\vec{x}_j\Vert<M \ \forall i, j$. 
A classic example is spatial aggregation of mobile agents \cite{Leverentz2009, Gazi2002}, but it can also describe the emergence of bounded consensus in opinion dynamics models \cite{LoIudice2023}.

\item \emph{Dispersion} 
is the opposite behaviour to aggregation, requiring agents to spread and diverge from one another, formally $\Vert \vec{x}_i-\vec{x}_j\Vert\rightarrow \infty \ \forall i,j$. 
This simple behaviour \cite{Bodnar2005} can have important applications in exploration \cite{Duncan2022, Elamvazhuthi2016}, or search and localization \cite{Zarzhitsky2005}.
    
\item \emph{Morphogenesis} 
is the formation of organic-like shapes, a behaviour that eludes a formal definition.
It was famously introduced by Alan Turing in \cite{Turing1952}, and recently inspired algorithms for the spontaneous organization of robots \cite{Oh2014,Carrillo-Zapata2019}.
A similar behaviour, the natural formation of patterns by swarming bacteria, has also been used to encode digital information \cite{Doshi2022}.
    
\item \emph{Pattern formation}, in a very general formulation, requires each agent to converge to a different point $p$ of a specified set $A\subset \mathbb{R}^d$, formally $\forall i \ \exists \vec{p}\in A : \vec{x}_i=\vec{p}$, with $\vec{x}_i\neq \vec{x}_j \ \forall i,j$.
Depending on the actual definition of the set $A$, it can describe a quite various range of behaviours, such as geometric pattern formation \cite{Giusti2023FRAI, Zhao2019}, 1D shape formation \cite{Hsieh2008, Meng2013} and 2D shape formation \cite{Rubenstein2014}.

\item \emph{Density regulation} 
requires the swarm  to distribute accordingly to some density distribution of interest $\rho_d$, formally $\rho(\vec{x},t) \propto \rho_d(\vec{x})$.
It has been implemented in swarming robots \cite{Elamvazhuthi2016, Eren2017}, with applications to surveillance \cite{Spears2006}, in mammalian cells for tissue engineering \cite{Gentile2014}, and in microorganisms \cite{Lam2017, Massana-Cid2022}.

\end{itemize}

In this thesis we will mainly focus on two of these behaviours, specifically Part \ref{part:geometric_pattern} will focus on \emph{pattern formation}, while Part \ref{part:spatial_microorganisms} will deal with the motion of microorganisms and how this can induce 
\emph{density regulation}.

\section{Modelling \& Analysis}
\label{sec:background_modelling}
Modelling the dynamics of \ac{LS-MAS} is a crucial step, as it allows to perform simulations, carry out mathematical analysis of their emergent behaviours and, finally, design model-based control strategies.
Mathematical models of dynamical systems typically consists of a set of first-order \acp{ODE}, each one describing the time evolution of one of the \emph{state variables} of the system.
Such equations can be either \emph{linear} or \emph{non-linear}. In the first case, they can be easily written in matrix form and analysed with the tools of linear algebra \cite[Section 4.3]{Khalil2002}. 
On the contrary, non-linear systems do not admit such simple analysis. While their local behaviour around a specific state can, often, be studied through linearization, a more general analysis requires different tools, such as Lyapunov analysis \cite[Section 4.1]{Khalil2002}.
Moreover, some systems, for example biological ones, have inherently stochastic behaviours. In such cases \acp{ODE} are replaced by \acp{SDE}, which can include randomness, usually in the form of additive white noise \cite{Elamvazhuthi2020}. 

Clearly, most \ac{LS-MAS} of interest involve non-linear dynamics and, of course, an extremely large and often variable, number of state variables, therefore ad-hoc modelling approaches are necessary.
Specifically, models of \ac{LS-MAS} can be divided into two main groups, which reflect the multi-scale nature of their dynamics:
\begin{itemize}
    \item \emph{Microscopic models}, also known as agent-based models, describe each agent (cell, robot or component) individually, and the topology of their interactions is represented by a (di-)graph \cite{Bullo2008}.    
    This formulation is, usually, intuitive and easy to build, moreover it allows the straightforward integration of a distributed control action into the agents’ dynamics. 
    Nevertheless, inferring the emerging properties of the ensemble and their robustness to changes in the network topology remains challenging. 
    
    \item \emph{Macroscopic models} describe the whole ensemble (colony, consortium, swarm) at aggregate level as a continuum. Therefore, individual agents are disregarded and instead the average behavior is studied, capturing at once the time and spatial evolutions by using \acp{PDE}.
    Specifically, the resulting equation usually takes the form of a Fokker-Planck equation, also known as forward Kolmogorov equation \cite{Roy2016}.
    These models are not always feasible, especially in the case of non trivial and heterogeneous interaction topologies, and require analysis tools different from those based on \acp{ODE}.
    On the other hand, they are particularly suited for stochastic systems and certain applications, such as density control. 
    For a deeper discussion on this approach see \cite{Elamvazhuthi2020}.
\end{itemize}

Given the different characteristics of these modeling approaches, tools to obtain one model from the other are necessary. 
Specifically, moving from a discrete to a macroscopic description, a process called \emph{countinuification} \cite{Maffettone2023LCSS}, requires specific tools such as the mean-field limit \cite{Elamvazhuthi2020}, \emph{graphons} \cite{Lovsz2012}, or the continuous conversion of Partial Difference Equations \cite{Nikitin2022}.
Moreover, this process, by approximating the individual states of the agents with a continuous distribution, implies some loss of information \cite{Lamarche-Perrin2013}.
Also the reverse process of \emph{discretization} comes with its own challenges, as it requires to infer the specific state of agents and their interaction topology. Indeed, this process is rarely considered, but has a crucial role in obtaining a distributed control law (\acp{ODE}), to be implemented on the agents, out of a continuous control law (\acp{PDE}) designed on a macroscopic model, as done, for example, in \cite{Maffettone2023LCSS}.

The development of the model itself represents another crucial aspect.
Indeed, bottom-up approaches, that consist in deriving the dynamics of the system starting form first principles, are hardly applicable to \ac{LS-MAS}, due to their complexity.
Instead, top-down approaches (e.g. data-driven modelling), leverage the observation of specific dynamics and data to capture the aspects of interest in the dynamics of the system \cite{Brunton2016}.
These second approaches, while still presenting numerous challenges, can be more suited to deal with \ac{LS-MAS}. An example is the modelling approach described in Chapter \ref{ch:modelling}.

Once a \ac{LS-MAS} model has been obtained, in one of the possible representations, it needs to be analyzed to infer the emerging properties of the system, such as the existence and the stability of equilibria, or more complex steady state behaviors as those discussed in Section \ref{sec:background_behaviours}.
We will also be interested in studying  robustness to variations of the number of agents, noise and disturbance. 
Due to the variety of emerging behaviors \ac{LS-MAS} can show, the analysis of such systems requires ad-hoc methods. 
Two fundamental tools in the study of such behaviors are Lyapunov’s and graph theory \cite{Mesbahi2010}. 
In particular, Lyapunov theory has been extended in multiple directions, for example vector
Lyapunov functions allow to study the robustness to agents removal \cite{Siljak1972}, while graph theory allows to model and analyze the network of interactions between the agents.
In Chapter \ref{ch:convergence} we use these tools to prove the stability of specific geometric patter formations.
Other useful tools are \emph{passivity} \cite[Chapter 6]{Khalil2002} and \emph{contraction} \cite{Bullo2022} theories, that have been extended to assess emerging properties of the ensemble from the study of the individual agents' dynamics \cite{Bai2011, Russo2009}.

\section{Control algorithms}
\label{sec:background_control}

\ac{LS-MAS} can be engineered to achieve desired goals and leverage cooperation to solve complex problems, such as surveillance  \cite{Lopes2021}, exploration \cite{Kegeleirs2021}, herding \cite{Auletta2022} or transportation \cite{Bayindir2016, Gardi2022}. 
Unfortunately the classic paradigm of control theory, does not apply as is to these systems.
Here, there is no longer a single process and a controller, that observes the output of the process and computes the appropriate inputs to steer the process towards the desired state.
Instead, to control \ac{LS-MAS}, one needs to close a feedback loop across multiple scales, understanding how to drive the individual agents (microscopic level), to obtain the desired effect at a macroscopic level (see Figure \ref{fig:classic_vs_LSMAS}).
For example, in density regulation problems, one might be able to measure the (macroscopic) spatial density of the swarm, while the control action would be exerted by controlling the movement of the single agents. 
Currently, despite the existence of many specific solutions and a vast body of research, no general method has been developed to translate macroscopic specifications into instructions for the individual agents to follow \cite{Heinrich2022}.

\begin{figure}[t]
    \centering
    \begin{subfigure}[t]{0.64\textwidth}
    \centering
    \includegraphics[trim=0 0 0 0, clip, width=1\columnwidth]{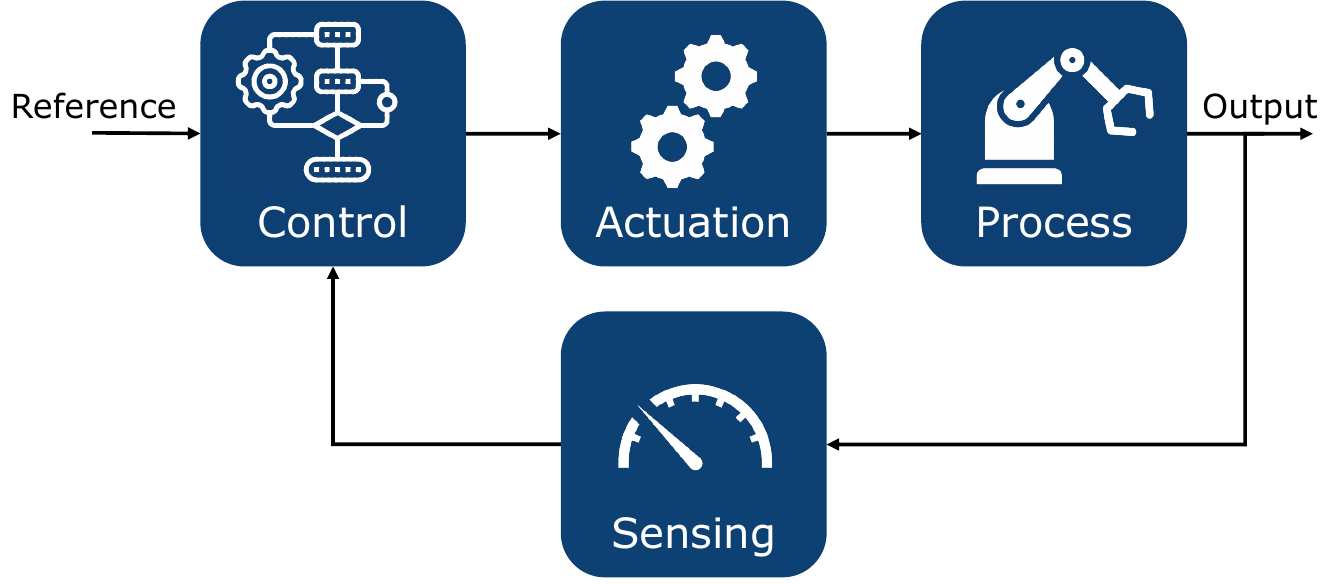}
    \caption{Classic control scheme}
    \end{subfigure}
    \hfill
    \begin{subfigure}[t]{0.33\textwidth}
    \includegraphics[trim=0 0 0 0, clip, width=1\columnwidth]{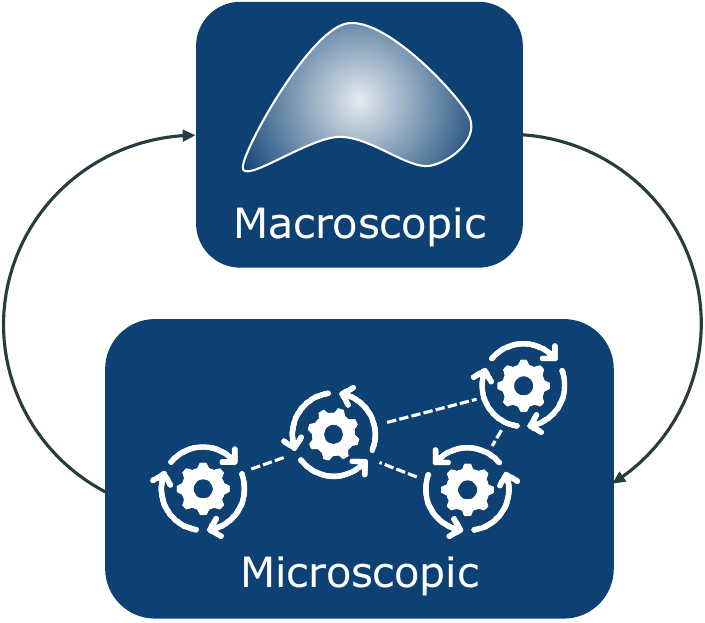}
    \caption{Multi-scale LS-MAS}
    \end{subfigure}
    \caption{
    Schematic comparison between classic systems and \ac{LS-MAS}.}
    \label{fig:classic_vs_LSMAS}
\end{figure} 

\subsection{Hierarchic classification of control algorithms}
To better frame the open challenges and to understand how to achieve control of \ac{LS-MAS}, it is useful to define a hierarchic structure of the control process \cite{Majid2022}. It is a standard approach in many engineering disciplines, according to the motto \textit{divide et impera}, but has not received the deserved attention in the context of \ac{LS-MAS}. 
Majid et al. in \cite{Majid2022}, divide the tasks commonly found in swarm robotics into two layers, \emph{low-level tasks}, such as aggregation or flocking, and \emph{high-level tasks}, including collective mapping or transportation, concerning the collective interaction of the swarm with the environment, but the connection between the two levels is fuzzy and not addressed directly.
A more refined hierarchy, focuses on the control of \acp{UAV} and encompasses five layers \cite{YZhou2020}.
The highest level, \emph{decision-making}, performs planning and task allocation, then the \emph{path planning} generates the path to reach the given goal.
The \emph{control} layer is in charge of the multi-agent coordination, including formation control and obstacle avoidance, and can rely on the services provided by the \emph{communication} layer.
Finally, the \emph{application} layer represents the lowest level and deals with the execution of the  application-specific task.
Here the analysis is narrowed down to the specific case of groups of \acp{UAV}, and decentralization is considered only for the lower levels.

Here, we propose a simple yet general hierarchical classification, distinguishing three control layers, each of which is in charge of tracking the reference generated by the upper level, as depicted in Figure \ref{fig:ControlHierarchy}. 
In general the higher the control level the greater the abstraction of the system and the slower the dynamics to be controlled, and therefore the execution frequency.

\begin{figure}[t]
    \centering
    \includegraphics[width=0.6\textwidth]{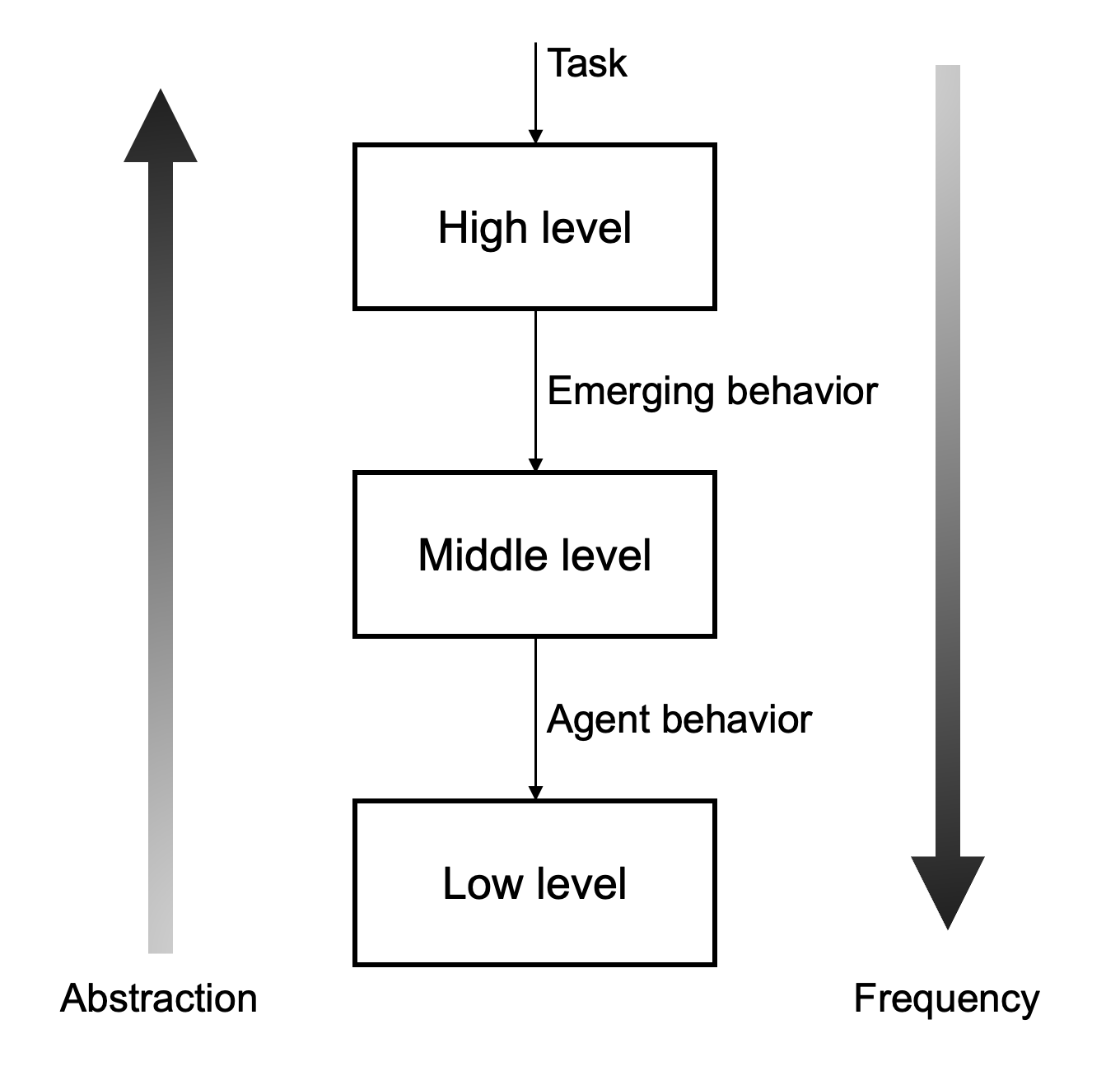}
    \caption{Hierarchy of control algorithms for \ac{LS-MAS}. The higher the layer the more abstract is the system representation and the lower the required execution frequency.}
    \label{fig:ControlHierarchy}
\end{figure}

\paragraph{High level control layer.}  
Starting from the high level (macroscopic) specification of the task to be solved (e.g. patrol a certain region), eventually provided by the human operator, high level control algorithms are in charge of defining the appropriate emerging behavior of the swarm that is required to solve the given task (e.g. flock towards north, then form a geometric pattern). 
The specific properties of the swarm, such as the number agents, and of the agents themselves are completely abstracted. 
The execution is most likely centralized and relatively slow, due to the slow evolution of the high level dynamics of interests.
Possibly, the high level control might also be executed off-line. 

\paragraph{Middle level control layer.}  
Given the desired emerging behavior of the swarm, computed by the high level control algorithm, the middle level algorithm computes the behaviour of each individual agent (e.g. agent $i$ moves north with a certain speed, while agent $j$ moves towards it).
Simple models of the agents' dynamics (e.g. first order integrator) and their interactions (e.g. a graph) are used, using any of the representations discussed in Section \ref{sec:background_modelling}. 
The execution might be either centralized or, preferably, distributed. 

\paragraph{Low level control layer.}  
Finally, the lower level implements the "servo" control of the single agent (e.g. turn the left wheel forward). 
The implementation is necessarily platform dependent and completely distributed, therefore an accurate model of the agent and its physical characteristics is necessary.
Most likely, an high execution frequency will be required, in order to provide the fast response necessary to the functioning of physical sensors and actuators.

\paragraph{}
The high and low levels are less relevant in this thesis, as they do not directly deal with the multi-agent aspect of the problem, and can, in most cases, be addressed with tools from other areas of engineering. 
Therefore, we will mostly focus on middle layer control strategies and particularly, on spatiotemporal control, aiming to regulate some quantity of interest, in both space and time, to obtain the desired emerging behaviour.

\subsection{Distributed and centralized implementations}
As said above, given the model of a \ac{LS-MAS} of interest and the desired emerging behavior, the middle level control is a feedback control algorithm that makes the target behavior emerge.
Such control law may be either centralized or distributed. In the first case a single entity observes the whole system and defines a control action that is imposed on the system. It can be implemented either by communicating the control action to each agent (\emph{centralized internal control}), or by applying some modification to the environment the system is evolving in (\emph{centralized external control}) \cite{Shannon2020}. 
On the other hand, a \emph{distributed control} law prescribes that each agent autonomously computes its own action.  Moreover, a distributed control algorithm should be local, that is, it only uses information that can be directly measured or estimated by the agents. 
Distributed and local control approaches intrinsically provide \emph{scalability} with the size of the system and \emph{robustness} to failure of agents \cite{Brambilla2013}. 
Indeed, most of the literature on the control of \ac{LS-MAS} focuses on these approaches, with applications to virtually any task.
This is also the approach used for pattern formation and discussed in Part \ref{part:geometric_pattern}.

Conversely, a centralized controller needs an infeasible increase of sensing and computation capabilities, as the number of agents increases, and represents a single point of failure.
One possible solution to achieve scalability in centralized architectures, is represented by macroscopic models (see Section \ref{sec:background_modelling}).
Indeed, such models provide a representation of the system that is independent of the number of agents and their individual state. 
A possible application of this approach is the spatial control of microorganisms, discussed in Part \ref{part:spatial_microorganisms}.

\section{Discussion}
This Chapter provided a brief overview of modelling and controlling emergent behaviours of \ac{LS-MAS}.
Firstly, the spatial behaviours were introduced as those influencing the steady-state distribution of the agents in the space. Among these, pattern formation and density regulation will be discussed in the rest of the thesis.
To study the properties of these behaviours, one needs to properly describe the dynamics of the system of interest. Hence, we discussed the microscopic and macroscopic modelling approaches, the respective features and the most relevant analysis tools.
Finally, we focused on the algorithms to control the emergent behaviours, and how these can deal with the multi-scale nature of \ac{LS-MAS}.
In particular, we discussed of a hierarchic control architecture, to address the different problems in the control of such systems, and the importance of decentralization, to improve scalability and robustness of the system.
The next Chapters will present our contributions on the control of geometric pattern formation (Chapters \ref{ch:patt_form_background} - \ref{ch:convergence} in Part \ref{part:geometric_pattern}), and the study of spatial behaviours in microorganisms (Chapters \ref{ch:miroorganisms_background} - \ref{ch:modelling} in Part \ref{part:spatial_microorganisms}).

\part{Geometric pattern formation}
\label{part:geometric_pattern}

\chapter{Introduction to geometric pattern formation}
\thispagestyle{empty} 
\label{ch:patt_form_background}

As introduced in Chapter \ref{ch:intro} spatial organisation is a crucial feature for both natural and artificial multi-agent systems. Here we will focus on \emph{geometric pattern formation} (see Section \ref{sec:background_behaviours}), where the agents self-organize their relative positions into some repeating geometric structure or {\em pattern}, e.g., arranging themselves on a lattice consisting of repeating adjacent triangles (see Figure \ref{fig:lattices_examples}).
\begin{figure}[t]
    \centering
    \begin{subfigure}[t]{0.25\textwidth}
    \centering
    \includegraphics[trim=4mm 4mm 4mm 4mm, clip, width=0.9\columnwidth]{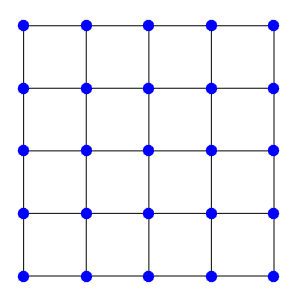}
    \caption{Square lattice}
    \end{subfigure}
    \begin{subfigure}[t]{0.25\textwidth}
    \includegraphics[trim=4mm 4mm 4mm 4mm, clip, width=1\columnwidth]{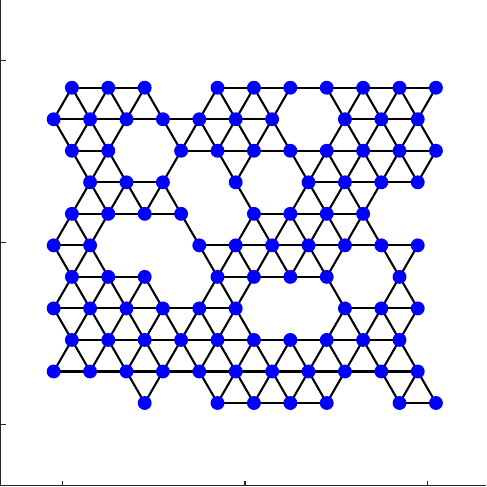}
    \caption{Triangular lattice}
    \label{subfig:lattice_example_triangular}
    \end{subfigure}
    \begin{subfigure}[t]{0.4\textwidth}
    \centering
    \includegraphics[trim=2mm 2mm 2mm 0mm, clip, width=0.9\columnwidth]{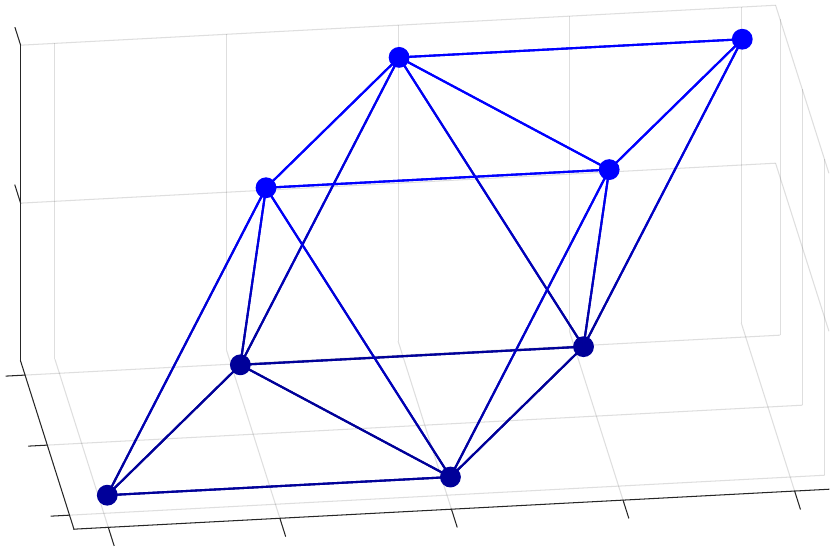}
    \caption{Tetrahedral-octahedral lattice}
    \end{subfigure}
    \caption{
    Examples of lattices in 2D (a-b) and 3D (c) spaces.}
    \label{fig:lattices_examples}
\end{figure} 
It is crucial in many tasks involving large-scale multi-agent systems and especially in \emph{swarm robotics} \cite{HOh2017}.
Examples include sensor networks deployment \cite{Zhao2019, Kim2014}, cooperative transportation and construction \cite{Rubenstein2013, Mooney2014, Gardi2022}, 2D or 3D exploration and mapping \cite{Kegeleirs2021} or area coverage \cite{HWang2020}.
Moreover, the formation of patterns is common in many biological systems where agents, such as cells or microorganisms, form organized geometric structures, e.g., \cite{Tan2022}.


To achieve pattern formation we have to overcome the two main difficulties commonly found in swarm robotics, as discussed in Section \ref{sec:background_control}.
Firstly, as there are no leader agents, the pattern must emerge by exploiting a control strategy that is the same for all agents, \emph{distributed} and \emph{local} (i.e., each agent can only use information about ``nearby'' agents).
Secondly, the number of agents is large and may change over time; therefore, the control strategy must also be  {\em scalable} to varying sizes of the swarm and \emph{robust} to uncertainties due to its possible variations.
This sets the problem of achieving pattern formation apart from the more classical \emph{formation control} problems \cite{Oh2015} where agents are typically fewer and have pre-assigned roles within the formation.
Moreover, geometric formations can also emerge as a by-product of \emph{flocking} algorithms as those described in \cite{Olfati-Saber2006, GWang2022}, but, in such cases the focus of the control strategy is to achieve coordinated motion, rather than desired regular formations to emerge.


The next pages will give an overview of the existing Literature on geometric pattern formation, with a particular attention to control strategies based on the use of virtual forces, and then introduce some useful mathematical tools.


\section{State of the art on geometric pattern formation}
\label{sec:literature_geom_patt_form}

Most of the existing distributed control algorithms for geometric pattern formation rely on the use of \emph{virtual forces} (or \emph{virtual potentials}) \cite{Spears2004, Casteigts2012, Zhao2019, Torquato2009, Mesbahi2010, Sakurama2021, Olfati-Saber2006}.
Within this framework, first introduced for obstacle avoidance \cite{Khatib1985}, agents move under the effect of virtual forces generated by the presence of their neighboring agents and the environment, causing attraction, repulsion, alignment, etc. 
In the most common case, where agents are homogeneous and exert radial forces depending on their relative distance, the resulting virtual force $\vec{u}_i$ acting on agent $i$ can be written as
\begin{equation}\label{eq:vf_control_law}
    \vec{u}_i(t) \coloneqq \sum_{j \in \C{I}_i(t)} f\left( \norm{\vec{x}_{i}(t)-\vec{x}_{j}(t)} \right)\, \frac{\vec{x}_{i}(t)-\vec{x}_{j}(t)}{\norm{\vec{x}_{i}(t)-\vec{x}_{j}(t)}},
\end{equation}
where $\C{I}_i(t)$ is the set of agents $i$ can interact with (usually depending on the relative distance), $\vec{x}_{i}$ represent the position of the agent and $f:\BB{R}_{\geq0}\to \BB{R}$ is the \emph{virtual interaction function}, describing the influence of interacting agents.
The most common interaction embeds short range repulsion, to avoid collisions, and long range attraction, to keep the swarm cohesive. This simple control law can then be modified or extended to obtain different behaviours (e.g. alignment, pattern formation, flocking, etc.).
The main advantage of this approach is that it is inherently local and distributed, two key features of swarm robotics algorithms.

To classify existing solutions to pattern formation, we employ the same taxonomy used in \cite{Oh2015}, and later extended in \cite{Sakurama2021}, which is based on the type of information available to the agents.
Namely, existing strategies can be classified as being (i) \emph{position-based} when it is assumed agents know their position and orientation and those of their neighbours, in a global reference frame;
(ii) \emph{displacement-based} when agents can only sense their own orientation with respect to a global reference direction (e.g., North) and the relative positions of their neighbours;
(iii) \emph{distance-based} when agents can measure the relative positions of their neighbours with respect to their local reference frame. 
In terms of sensor requirements, position-based solutions are the most demanding, requiring global positioning sensors, typically GPS, and communication devices, such as WiFi or LoRa. 
Differently, displacement-based methods require only a distance sensor (e.g., LiDAR) and a compass, although the latter can be replaced by a coordinated initialisation procedure of all local reference frames \cite{Cortes2009}.
Finally, distance-based algorithms are the least demanding, needing only the availability of some distance sensors.

\subsection{Position-based approaches}
In \cite{Pinciroli2008}, a position-based algorithm was proposed to achieve 2D triangular lattices in a constellation of satellites in a 3D space.
This strategy combines global attraction towards a reference point with local interaction among the agents to control both the global shape and the internal lattice structure of the swarm.
In \cite{Casteigts2012}, a position-based approach was presented that combines the common radial virtual force (also used in \cite{Spears2004,Hettiarachchi2005,Torquato2009}) with a normal force. 
In this way, a network of connections is built such that each agent has at least two neighbours. 
Importantly, this approach requires the acquisition of positions from two-hop neighbours.
In \cite{Zhao2019}, a position-based strategy is presented to achieve triangular and square patterns, as well as lines and circles, both in 2D and 3D; the control strategy features global attraction towards a reference point and re-scaling of distances between neighbours, with the virtual forces changing according to the goal pattern.
Therein, a qualitative comparison is also provided with the distance-based strategy from \cite{Spears2004}, showing more precise configurations and a shorter convergence time, due to the position-based nature of the solution. 
Finally, a simple position-based algorithm for triangular patterns, based on virtual forces and requiring communication between the agents, is proposed in \cite{Trotta2018} to have unmanned aerial vehicles perform area coverage. 

\subsection{Displacement-based approaches}

In \cite{Li2009}, a displacement-based approach is presented based on the use of a geometric control law similar to the one proposed in \cite{Lee2008}. The aim is to obtain triangular lattices but small persisting oscillations of the agents are present at steady state, as the robots are assumed to have a constant non-zero speed.
In \cite{Balch2000a, Balch2000}, an approach is discussed inspired by covalent bonds in crystals, where each agent has multiple attachment points for its neighbours. 
Only starting conditions close to the desired pattern are tested, as the focus is on navigation in environments with obstacles.
In \cite{Song2014} the desired lattice is encoded by a graph, where the vertices denote possible \emph{roles} the agents may play in the lattice and edges denote rigid transformations between the local frames or reference of pairs of neighbours.
All agents communicate with each other and are assigned a label (or identification number) through which they are organised hierarchically to form triangular, square, hexagonal or octagon-square patterns.
Formation control is similarly addressed in \cite{Coppola2019}. 
The algorithm proposed therein is made of a higher level
policy to assign positions in a square lattice to the agents,
and a lower level control, based on virtual forces, to have the agents reach these positions.
The algorithm can be readily applied to the formation of square geometric patterns, but not to triangular ones. Notably, the reported convergence time is relatively long and increases with the number of agents.
%
Finally, a solution to progressively deploy a swarm on a predetermined set of points is presented in \cite{Li2019}.
The algorithm can be used to perform both formation control and geometric lattice formation, even though the orientation of the formation cannot be controlled. 
Moreover, this strategy requires local communication between the agents and the knowledge of a common graph associated to the formation.

\subsection{Distance-based approaches}
A popular distance-based approach for the formation of triangular and square lattices, based on \emph{gravitational} virtual forces, was proposed in \cite{Spears1999} and later further investigated in \cite{Spears2004, Hettiarachchi2005}. 
In these studies, 
triangular lattices are achieved with long-range attraction and short-range repulsion virtual forces only, while square lattices are obtained through a selective rescaling of the distances between some of the agents (see Section \ref{subsec:comparison}
for further details).
The main drawback of the gravitational strategy \cite{Spears1999, Spears2004, Hettiarachchi2005, Sailesh2014} is that it can produce the formation of multiple aggregations of agents, each respecting the desired pattern, but with different orientations.
Another problem, described in \cite{Spears2004}, is that, for some values of the parameters, multiple agents can converge towards the same position and collide.

Similar approaches are also used to obtain triangular lattices in flocking algorithms \cite{Olfati-Saber2006, XWang2022, GWang2022}, but such solutions mostly focus on the coordination of the motion.
An extension of the gravitational strategy to achieve the formation of hexagonal lattices was proposed in \cite{Sailesh2014}, but with the requirement of an ad-hoc correction procedure to prevent agents from remaining stuck in the centre of a hexagon. 
%
In \cite{Torquato2009}, an approach exploiting Lennard-Jones-like virtual forces is numerically optimised to locally stabilise a hexagonal lattice.
This control law relies on stability conditions provided by \emph{harmonic approximation} \cite{Hinsen2005}, but such conditions are only necessary, and not sufficient to prove the stability of the lattice configuration.
Moreover, this approach requires time-varying control gains and synchronous clocks among the agents.
%
%
A different strategy, based on geometric arguments, was proposed in \cite{Lee2008}. It allows to build triangular lattices and, notably, its convergence was proved exploiting the Lyapunov method.
A 3D extension was later presented in \cite{Lee2010}.

\subsection{Open problems}
From this analysis emerges that many existing strategies are limited to 2D domains and can only achieve triangular patterns (see Figure \ref{subfig:lattice_example_triangular}) \cite{Pinciroli2008,Trotta2018}.
More flexible solutions either show good performance but require expensive sensors and communication devices \cite{Zhao2019,Song2014,Li2019}, or have lesser  requirements but result in more poor performance \cite{Spears2004,Balch2000,Coppola2019}.
Moreover, validation is mainly achieved numerically or experimentally \cite{Spears2004, Balch2000, Balch2000a, Li2009, Casteigts2012, Song2014, Zhao2019}, with only few strategies supported by formal proofs of convergence.
Therefore, two pressing open challenges in geometric pattern formation are (i) the design of local and distributed control strategies that can combine low sensor requirements with flexibility and consistently high performance, and (ii) the formal proof of convergence for such algorithms.

We tackled both this challenges. Chapter \ref{ch:dist_cont} will present the distributed control algorithm we developed to address the formation of both triangular and square lattices, while our formal study on the stability of specific lattice configurations is presented in Chapter \ref{ch:convergence}. 

\section{Mathematical preliminaries}
\label{sec:math_preliminaries}

Here we introduce some notation and mathematical tools that will be useful in Chapters \ref{ch:dist_cont} and \ref{ch:convergence}.

\paragraph*{Notation} 
Given a vector $\vec{v} \in \BB{R}^d$, $[\vec{v}]_i$ is its $i$-th element, $\norm{\vec{v}}$ its Euclidean norm, and $\unitvec{v}\coloneqq\frac{\vec{v}}{\norm{\vec{v}}}$ its direction. 
$\vec{0}$ denotes a column vector of appropriate dimension with all elements equal to 0.
Given a matrix $\vec{A}$, $[\vec{A}]_{ij}$ is its $(i, j)$-th element. 
Given a set $\C{B}$, its cardinality is denoted by $\abs{\C{B}}$. Finally, we refer to $\BB{R}^2$ as the \emph{plane}.

\subsection{Dynamical systems}
We start by recalling the concepts of equilibrium set and stability.
Given a continuous-time, autonomous dynamical system
\begin{equation}\label{eq:dynamical_system}
    \dot{\vec{x}}(t) = \vec{f}(\vec{x}(t)), \quad \vec{x}(0) = \vec{x}_0,
\end{equation}
 with state vector $\vec{x}(t) \in \BB{R}^d$, and $\vec{x}_0 \in \BB{R}^d$, we term as $\vec{\phi}(t, \vec{x}_0)$ its trajectory starting from $\vec{x}(0) = \vec{x}_0$.

\begin{definition}[Equilibrium set]
\label{def:equilibrium_set}
A set $\Xi  \subset \BB{R}^d$ is an \emph{equilibrium set} for system \eqref{eq:dynamical_system} if $ \vec{f}(\vec{x}) = \vec{0} \ \forall \vec{x} \in \Xi$.
\end{definition}

\begin{definition}[Local asymptotic stability {\cite[Definition~1.8]{Kuznetsov2004}}]
\label{def:LAS}
An equilibrium set $\Xi$ for system \eqref{eq:dynamical_system} is \emph{locally asymptotically stable} if $\forall \epsilon>0, \exists \delta > 0$ such that if $\min_{\vec{y} \in \Xi}\norm{\vec{x}_0 - \vec{y}} < \delta $, then
\begin{enumerate}
    \item 
    $\min_{\vec{y} \in \Xi}\norm{\vec{\phi}(t, \vec{x}_0)-\vec{y}} < \epsilon,  \ \forall t>0$, and
    \item
    $ \lim_{t \rightarrow +\infty} \vec{\phi}(t, \vec{x}_0) \in \Xi$.
\end{enumerate}
\end{definition}

\subsection{Frameworks and rigidity}
Next let us introduce the concept of \emph{framework}, useful to add a spatial dimension to the mathematical concept of \emph{graph}.
\begin{definition}[Incidence matrix]
\label{def:incidence_mat}
Given a digraph with $N$ vertices and $m$ edges, its \emph{incidence matrix} $\vec{B}\in \mathbb{R}^{N\times m}$ has elements defined as
\begin{equation*}
    [\vec{B}]_{ij} \coloneqq \begin{dcases}
    + 1, &\text{if edge $j$ starts from vertex $i$},\\
    - 1, &\text{if edge $j$ ends in vertex $i$},\\
    0,   &\text{otherwise}.
    \end{dcases}
\end{equation*}
\end{definition}

\begin{definition}[Framework {\cite[p.~120]{Mesbahi2010}}]\label{def:framework}
Consider a (di\mbox{-)}graph $\C{G}=(\C{V}, \C{E})$ with $N$ vertices, and a set of positions $\vec{p}_1, \dots, \vec{p}_n \in \mathbb{R}^d$ associated to its vertices, with $\vec{p}_i \neq \vec{p}_j \ \forall i,j \in \{1, \dots, N\}$.
A \emph{$d$-dimensional framework} is the pair $(\C{G},\bar{\vec{p}})$, where $\bar{\vec{p}} \coloneqq [\vec{p}_1\T \ \cdots \ \vec{p}_n\T]\T \in \mathbb{R}^{dN}$.
Moreover, the \emph{length} of an edge, say $(i,j)\in \C{E}$, is $\norm{\vec{p}_i - \vec{p}_j}$.
\end{definition}

\begin{definition}[Congruent frameworks {\cite[p. 3]{Jackson2007}}]\label{def:congurent_framework}
Given a graph $\C{G}=(\C{V},\C{E})$ and two frameworks $(\C{G},\bar{\vec{p}})$ and $(\C{G},\bar{\vec{q}})$, these are \emph{congruent} if $\norm{\vec{p}_i - \vec{p}_j}=\norm{\vec{q}_i - \vec{q}_j}\ \forall i,j \in \C{V}$.
\end{definition}
%


\begin{definition}[Rigidity matrix {\cite[p.\,5]{Jackson2007}}]
\label{def:rigidity_matrix}
Given a $d$-dimensional framework with $N\geq2$ vertices and $m$ edges, its \emph{rigidity matrix} $\vec{M}\in \mathbb{R}^{m\times dN}$ has elements defined as
\begin{equation*}
    [\vec{M}]_{e,(jd-d+k)}\coloneqq\begin{cases}
    [\vec{p}_j-\vec{p}_i]_k, & \parbox[t]{4cm}{ if edge $e$ goes from vertex $i$ to vertex $j$,}\\
    [\vec{p}_i-\vec{p}_j]_k, & \parbox[t]{4cm}{ if edge $e$ goes from vertex $j$ to vertex $i$,}\\
    0, &\mbox{otherwise.}
    \end{cases}
\end{equation*}
with $k \in \{1, \dots, d\}$.
\end{definition}

\begin{definition}[Infinitesimal rigidity {\cite[p.~122]{Mesbahi2010}}]
\label{def:inf_rigidity}
A framework with rigidity matrix $\vec{M}$ is \emph{infinitesimally rigid} if, for any infinitesimal motion, say $\vec{u}$,%
\footnote{$\vec{u}$ can be interpreted as either a velocity or a small displacement.}
of its vertices, such that the length of the edges is preserved, it holds that $\vec{M u} = \vec{0}$.
\end{definition}

To give a geometrical intuition of the concept of infinitesimal rigidity, we note that an infinitesimally rigid framework is also rigid \cite[p.~122]{Mesbahi2010}, according to the definition below.%
\footnote{Rarely, a rigid framework is not infinitesimally rigid; e.g. \cite[p.~7]{Jackson2007}.}

\begin{definition}[Rigidity {\cite[p. 3]{Jackson2007}}]
\label{def:rigidity}
A framework is \emph{rigid} if every continuous motion of the vertices, that preserves the length of the edges, also preserves the distances between all pairs of vertices.
\end{definition}
Consequently, in a rigid framework, a continuous motion that does \emph{not} preserve the distance between any two vertices also does \emph{not} preserve the length of at least one edge.

To more easily assess the infinitesimal rigidity of a framework, it is possible to use the following result.

\begin{theorem}[{\cite[Theorem 2.2]{Hendrickson1992}}]
\label{th:rigidity}
A $d$-dimensional framework with $N\geq d$ vertices and rigidity matrix $\vec{M}$ is infinitesimally rigid if and only if $\R{rank}(\vec{M})=dN-d(d+1)/2$.
\end{theorem}


\subsection{Swarms of mobile agents}
\label{subsec:swarm_definition}
Now, let us formally define a \emph{swarm} as a set of $N \in \mathbb{N}_{>0}$ identical agents, say $\C{S} \coloneqq \{1,2,\dots,N\}$, that can move in $\mathbb{R}^d$ and interact with their neighbors to generate emergent behavior \cite{Brambilla2013}.
For each agent $i \in \C{S}$, $\vec{x}_i(t)\in \mathbb{R}^d$ denotes its position at time $t \in \mathbb{R}_{\geq0}$.
Moreover, we call $\bar{\vec{x}}(t) \coloneqq [\vec{x}_1\T(t) \ \cdots \ \vec{x}_n\T(t)]\T \in \mathbb{R}^{dN}$ the \emph{configuration} of the swarm, define $\vec{x}_{\R{c}}(t) \coloneqq \frac{1}{N} \sum_{i = 1}^N \vec{x}_i(t) \, \in \mathbb{R}^d$ as its \emph{center}, and denote by $\vec{r}_{ij}(t) \coloneqq \vec{x}_{i}(t)-\vec{x}_{j}(t) \in \mathbb{R}^d$ the relative position of agent $i$ with respect to agent $j$.
Finally, for planar swarms ($d=2$) we define $\theta_{ij}(t) \in [0, 2\pi]$ as the angle between $\vec{r}_{ij}$ and the horizontal axis, and given any two pairs of agents, $(i,j)$ and $(h,k)$, we denote with $\theta_{ij}^{hk}(t) \in [0, 2\pi]$ the absolute value of the angle between the vectors $\vec{r}_{ij}$ and $\vec{r}_{hk}$.



\begin{definition}[Interaction set]
\label{def:interaction_set}
Given a swarm $\C{S}$ and a \emph{sensing radius} $R_{\text{s}} \in \mathbb{R}_{>0}$, the \emph{interaction set} of agent $i$ at time $t$ is 
$
    \mathcal{I}_i(t) \coloneqq \{ j \in \mathcal{S} \setminus \{i\} : \Vert \vec{r}_{ij}(t)\Vert \leq R_{\text{s}} \}.
$
\end{definition}

\begin{definition}[Adjacency set] 
\label{def:adjacency_set}
Given a swarm $\C{S}$ and some finite $R_{\min},\, R_{\max} \in \mathbb{R}_{>0}$, with $R_{\min}\leq R_{\max}$, the \emph{adjacency set} of agent $i$ at time $t$ is (see Figure~\ref{fig:lattices})
\begin{align}
    \mathcal{A}_i(t) &\coloneqq \{j \in \mathcal{S} \setminus \{i\} : R_{\min} \leq \Vert \vec{r}_{ij}(t)\Vert \leq R_{\max} \}.
\end{align}
\end{definition}

\begin{definition}[Links]
\label{def:links}
A \emph{link} is a pair $(i,j) \in \C{S} \times \C{S}$ such that $j \in \C{A}_i(t)$, and $\norm{\vec{r}_{ij}(t)}$ is its \emph{length}.
The set of all links existing in a certain configuration $\bar{\vec{x}}$ is denoted by $\C{E}(\bar{\vec{x}})$.
\end{definition}
We will say that two agents are connected if and only if $(i,j) \in \C{E}(\bar{\vec{x}})$, that also implies $(j,i) \in \C{E}(\bar{\vec{x}})$.
Moreover, notice that if $R_{\max} \leq R_{\text{s}}$ then $\mathcal{A}_i \subseteq \mathcal{I}_i$.


\begin{definition}[Swarm graph and framework]
\label{def:swram_graph}
The \emph{swarm graph} is the digraph $\C{G}(\bar{\vec{x}})\coloneqq(\C{S},\C{E}(\bar{\vec{x}}))$, whose vertices correspond to the agents in the swarm and whose edges correspond to the links %
\footnote{Formally, $\mathcal{G}(\bar{\vec{x}})$ is a directed graph, even though $\mathcal{E}(\bar{\vec{x}})$ is such that the existence of link $(i, j)$ implies the existence of link $(j, i)$.}.
The \emph{swarm framework} is $\C{F}(\bar{\vec{x}})\coloneqq(\C{G}(\bar{\vec{x}}), \bar{\vec{x}})$.
\end{definition}

\begin{definition}[Congruent configurations]\label{def:congruent_conf}
Given a configuration $\bar{\vec{x}}^\diamond$, we define \emph{the set of its congruent configurations} $\Gamma(\bar{\vec{x}}^\diamond)$ as the set of configurations with congruent associated frameworks (see Definition~\ref{def:congurent_framework}), that is
    $\Gamma(\bar{\vec{x}}^\diamond) \coloneqq \{ \bar{\vec{x}} \in \BB{R}^{dN} : \norm{\vec{x}_{i}-\vec{x}_{j}} = \norm{\vec{x}_{i}^\diamond-\vec{x}_{j}^\diamond}, \forall i,j \in \C{S} \}$.
\end{definition}
These configurations are obtained by translations and rotations of the framework $\C{F}(\bar{\vec{x}}^\diamond)$; thus, it is immediate to verify that $\Gamma(\bar{\vec{x}}^\diamond)$ is connected and unbounded for any $\bar{\vec{x}}^\diamond$.

\section{Discussion}
In this Chapter we introduced the problem of geometric pattern formation, discussing the possible applications and the main challenges.
We reviewed the existing solutions, classified according to the sensory requirements, and the open challenges, mainly the design of flexible control algorithms and their analytical validation.
Finally, we introduced some notation and the mathematical tools, used in Chapters \ref{ch:dist_cont} and \ref{ch:convergence} to address, respectively, the problems of designing and analysing solutions for the formation of geometric patterns.

\chapter{Distributed control for geometric pattern formation}
\thispagestyle{empty} 
\label{ch:dist_cont}

This Chapter follows the contents and the structure of our work \cite{Giusti2023FRAI}, where we introduced a distributed displacement-based control law that allows large groups of agents to achieve triangular and square lattices in 2D, with low sensor requirements and without needing communication between the agents. 
Also, a simple, yet powerful, adaptation law is proposed to automatically tune the control gains in order to reduce the design effort, while improving robustness and flexibility.
We show the validity and robustness of our approach via numerical simulations and experiments, comparing it, where possible, with other approaches from the existing literature.

\section{Control goal and performance metrics}
\label{sec:performance_metrics}

To describe the planar lattices studied in this Chapter let us introduce the following definition.

\begin{definition}[$(L, R)$-Lattice]
\label{def:lattice}
Given some $L \in \{4, 6\}$ and $R \in \mathbb{R}_{>0}$, a \emph{$(L, R)$-lattice} is a set of points in the plane that coincide with the vertices of an associated \emph{regular tiling} \cite{Engel2004};
$R$ is the distance between adjacent vertices and $L$ is the number of adjacent vertices each point has.
\end{definition}
In Definition \ref{def:lattice}, $L=4$, and $L=6$ correspond to square and triangular lattices,%
\footnote{\label{foot:tilings}
Regular tilings exist only for $L \in \{3, 4, 6\}$.
The case $L=2$ corresponds more trivially to a line, rather than a planar structure.
The case $L = 3$, corresponding to hexagonal tilings, where vertices appear in two different spatial configurations (one with edges at angles $\pi/2$, $7/6\,\pi$, $11/6\,\pi$, and one with edges at angles $\pi/6$, $5/6\,\pi$, $3/2\,\pi$---plus an optional offset).
Hence, the control strategy we propose here would need to be extended to select one or the other of the possible hexagonal configurations.
}
respectively, as portrayed in Figure~\ref{fig:lattices}. 
\begin{figure}
    \centering
    \begin{subfigure}[t]{0.3\textwidth}
        \includegraphics[width=1\textwidth]{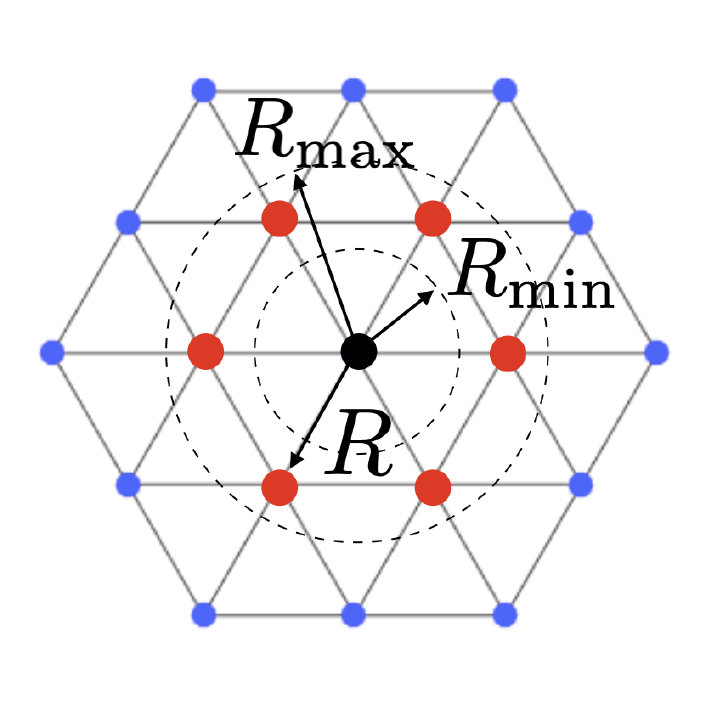}
        \caption{$L=6$}
    \end{subfigure}
    \begin{subfigure}[t]{0.28\textwidth}
        \includegraphics[width=1\textwidth]{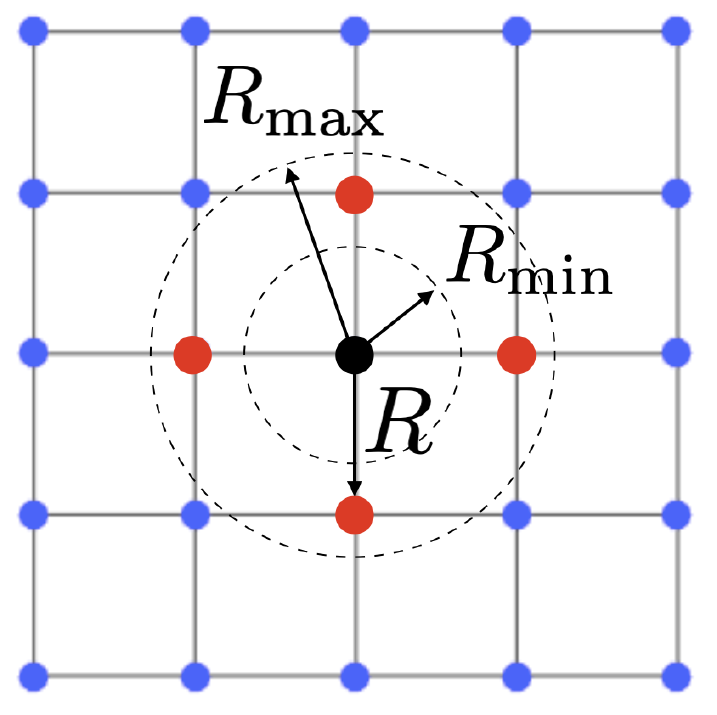}
        \caption{$L=4$}
    \end{subfigure}
    \caption{$(L,R)$-lattice formations: triangular ($L=6$) (a) and square ($L=4$) (b).
    Red dots are agents in the adjacency set ($\mathcal{A}_i$) of the generic agent $i$ depicted as a black dot.
    }
    \label{fig:lattices}
\end{figure}
Let us recall the concept of \emph{swarm} as introduced in Section \ref{subsec:swarm_definition}, we will say it \emph{self-organises into a $(L,R)$-lattice} if (i) each agent has at most $L$ links (see Definition \ref{def:links}), and (ii) given any two links $(i,j)$ and $(h,k)$ in $\mathcal{E}$ it holds that $\theta_{ij}^{hk}$ (i.e. the angle between the two as defined in Section \ref{subsec:swarm_definition}) is some multiple of $2\pi/L$.
Therefore, to assess whether a swarm self-organises into some desired $(L,R)$-lattice, we introduce the following two metrics.

\begin{definition}[Regularity metric]
\label{def:regularity}
Given a swarm and a desired $(L,R)$-lattice, the \emph{regularity metric} $e_{\theta}(t) \in [0,1]$ is
\begin{align}
    e_{\theta}(t) \coloneqq \frac{L}{\pi}\cdot \theta_{\mathrm{err}}(t),
    \label{eq:regMet}
\end{align}
where, omitting the dependence on time,
\begin{equation}
    \theta_{\mathrm{err}} \coloneqq\frac{1}{\vert \mathcal{E}\vert^2-2\vert \mathcal{E}\vert} \mathlarger{\sum}_{(i,j)\in \mathcal{E}} \ \mathlarger{\sum}_{(h,k) \in \mathcal{E}} \min_{q\in \mathbb{Z}} \left\vert\theta_{ij}^{hk} -q\frac{2\pi}{L}\right\vert .
    \label{eq::SpearsMetrics}
\end{equation}
\end{definition}

The regularity metric $e_{\theta}$, derived from \cite{Spears2004}, quantifies the incoherence in the orientation of the links in the swarm.
In particular, $e_{\theta}=0$ when all the pairs of links form angles that are multiples of $2\pi/L$ (which is desirable to achieve the $(L, R)$-lattice), while $e_{\theta}=1$ when all pairs of links have the maximum possible orientation error, equal to $\pi/L$. ($e_{\theta} \approx 0.5$ generally corresponds to the agents being arranged randomly.)

\begin{definition}[Compactness metric]
\label{def:compactness}
Given a swarm of $N$ agents and a desired $(L,R)$-lattice, the \emph{compactness metric} $e_L(t) \in [0,(N-1-L)/L]$ 
is
\begin{align}
    e_L(t)&\coloneqq\frac{1}{N} \sum_{i=1}^N \frac{\big\lvert \vert \mathcal{A}_i (t)\vert - L \big\rvert}{L},
    \label{eq::compactness_metric}
\end{align}
where $\mathcal{A}_i (t)$ is the adjacency set of agent $i$, as given in Definition \ref{def:adjacency_set}.
\end{definition}
The compactness metric $e_L$ measures the average difference between the number of neighbours each agent has and the one they are ought to have if they were arranged in a $(L,R)$-lattice.
According to this definition, $e_L$ reaches its maximum value, $e_{L,\max} = (N-1-L)/L$, when all agents are concentrated in a small region, and links exist between all pairs of agents, while
$e_L=1$ when all the agents are scattered loosely in the plane, and no links exist between them, finally,
 $e_L = 0$ when all the agents have $L$ links (typically we will require that $e_L$ is below some acceptable threshold, see Section \ref{subsubsec:performance_eval}).
It is important to remark that, if the number $N$ of agents is finite, $e_L$ can never be equal to zero, because the agents on the boundary of the group will always have less than $L$ links (see Figure~\ref{fig:lattices}), but this effect gets less relevant as $N$ increases.
Note that a similar metric was also independently defined in \cite{Song2014}. 
We remark that the compactness metric inherently penalizes the presence of holes in the configuration and the emergence of detached swarms, as those scenarios are characterized by larger boundaries.

For the sake of brevity, in what follows we will omit dependence on time when that is clear from the context.

\section{Design of the control law}

\subsection{Problem formulation}
\label{sec:ProblemStatement}
Consider a planar swarm $\mathcal{S}$  whose agents' dynamics is described by the first order model
\begin{align}
    \dot{\vec{x}}_i(t) = \vec{u}_i(t), \ \ \forall \ i\in\mathcal{S},
    \label{eq:firstOrdDynamics}
\end{align}
where $\vec{x}_i(t)\in \mathbb{R}^2$ represent the position of agent $i$, as described in Section \ref{subsec:swarm_definition} and depicted in Figure \ref{fig:two_agents_positions}, and $\vec{u}_i(t) \in \mathbb{R}^2$ is the input signal determining its velocity.

\begin{figure}
    \begin{center}
    \includegraphics[width=0.4\columnwidth]{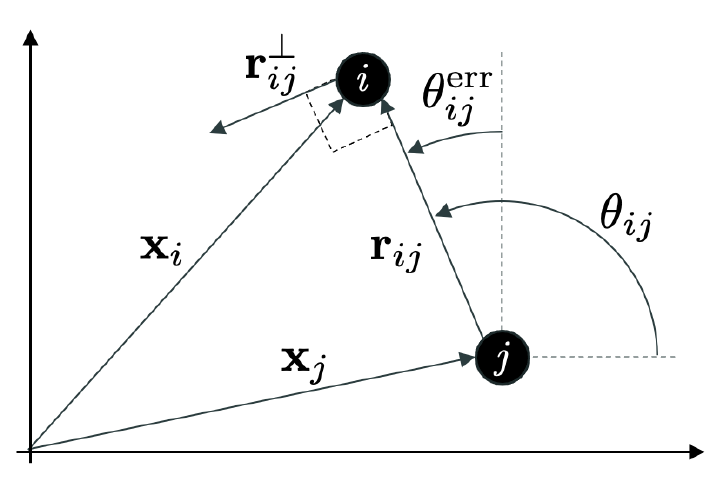}
    \end{center}
    \caption{Schematic diagram of two agents, $i$ and $j$, showing the key variables describing the agents' position and their geometrical relationship.}
    \label{fig:two_agents_positions}
\end{figure}

\begin{remark}
\label{rem:second_order_to_first}
First order models like \eqref{eq:firstOrdDynamics} are often used in the literature \cite{Lee2008, Lee2010, Casteigts2012, Zhao2019}.
In some other works \cite{Spears2004, Sailesh2014} a second order model is used, given by $m \ddot{\vec{x}}_i + \mu \dot{\vec{x}}_i = \vec{u}_i$, where $\vec{u}_i$ is a force, $m$ is a mass and $\mu$ is a viscous friction coefficient.
Under the simplifying assumptions of small inertia ($m\Vert \dot{\vec{v}}_i \Vert \ll \mu \Vert \vec{v}_i \Vert $) and $\mu = 1$, the two models coincide.
\end{remark}

We want to design a \emph{distributed} feedback control law  $\vec{u}_i = \vec{g}(\{\vec{r}_{ij}\}_{j\in \C{I}_i}, L, R)$
to let the swarm self-organise into a desired triangular or square lattice, starting from any set of initial positions in some disk of radius $r$, while guaranteeing the control strategy to be: 
\begin{enumerate}
    \item \emph{robust} to failures of agents and to noise;
    \item \emph{flexible}, allowing dynamic reorganisation of the agents into different patterns; 
    \item \emph{scalable}, allowing the number of agents $N$ to change dynamically.
\end{enumerate}

We will assess the effectiveness of the proposed strategy by using the performance metrics $e_{\theta}$ and $e_L$ introduced above (see Definitions \ref{def:regularity} and \ref{def:compactness}).

\subsection{Distributed control law}
\label{sec:controlDesign}
To solve this problem we propose a distributed displacement-based control law of the form
\begin{align}\label{eq:controlLaw}
    \vec{u}_i(t) &= \vec{u}_{\text{r},i}(t)+ \vec{u}_{\text{n},i}(t),
\end{align}
where $\vec{u}_{\text{r},i}$ and $\vec{u}_{\text{n},i}$ are the \emph{radial} and \emph{normal} control inputs, respectively.
The two inputs have different purposes and each comprises  several \emph{virtual forces}.
The radial input $\vec{u}_{\text{r},i}$ is the sum of attracting/repelling actions between the agents, with the purpose of aggregating them into a compact swarm, while avoiding collisions and keeping the desired inter-agent distance.
The normal input $\vec{u}_{\text{n},i}$ is also the sum of multiple actions, used to adjust the angles of the relative positions of the agents.

Note that the control strategy in \eqref{eq:controlLaw} is \emph{displacement-based} because it only requires  each agent $i$ (i) to be able to measure the relative positions of the agents close to it (in the sets $\C{I}_i$ and $\mathcal{A}_i$, given in Definitions \ref{def:interaction_set} and \ref{def:adjacency_set}), and (ii) to possess knowledge of a common reference direction.
Next, we describe in detail each of the two control actions in \eqref{eq:controlLaw}.

\subsubsection{Radial Interaction}
\label{subsec:RIF}
The radial control input $\vec{u}_{\text{r},i}$ in \eqref{eq:controlLaw} is defined as the sum of several virtual forces, one for each agent in $\C{I}_i$ (neighbours of $i$), each force being attractive (if the neighbour is far) or repulsive (if the neighbour is close).
Specifically, we set
\begin{align}\label{eq:radInput}
    \vec{u}_{\text{r},i} &= G_{\text{r},i}\sum_{j\in \C{I}_i} f_{\text{r}}(\Vert \vec{r}_{ij}\Vert) \frac{\vec{r}_{ij}}{\Vert \vec{r}_{ij} \Vert},
\end{align}
where $G_{\text{r},i}\in \mathbb{R}_{\geq0}$ is the radial control gain.
Note that $\vec{u}_{\text{r},i}$ is termed as \emph{radial} input because in \eqref{eq:radInput} the attraction/repulsion forces are parallel to the vectors $\mathbf{r}_{ij}$ (see Figure \ref{fig:two_agents_positions}).
The magnitude and sign of each of these forces depend on the distance, $\Vert \vec{r}_{ij}\Vert$, between the agents, according to the \emph{radial interaction function} $f_{\text{r}} : \mathbb{R}_{\geq 0} \rightarrow \mathbb{R}$.
Here, we select $f_{\text{r}}$  as the Physics-inspired Lennard-Jones function \cite{Brambilla2013}, given by
\begin{equation}
    f_{\text{r}}(\Vert \vec{r}_{ij} \Vert) = \min \left\lbrace \left( \frac{a}{\Vert \vec{r}_{ij} \Vert^{2c}}-\frac{b}{\Vert \vec{r}_{ij} \Vert^c}\right), \ 1 \right\rbrace,
    \label{eq:radial_Lennard-Jones}
\end{equation}
where $a, b \in \mathbb{R}_{> 0}$ and $c \in \mathbb{N}$ are design parameters.
In \eqref{eq:radial_Lennard-Jones}, $f_{\text{r}}$ is saturated to $1$ to avoid divergence for $\lVert \vec{r}_{ij} \rVert \to 0$.
$f_\mathrm{r}$ is portrayed in Figure \ref{subfig:interaction_func_radial}

\begin{figure}[t]
    \centering
    \begin{subfigure}[t]{0.45\textwidth}
        \includegraphics[width=1\textwidth]{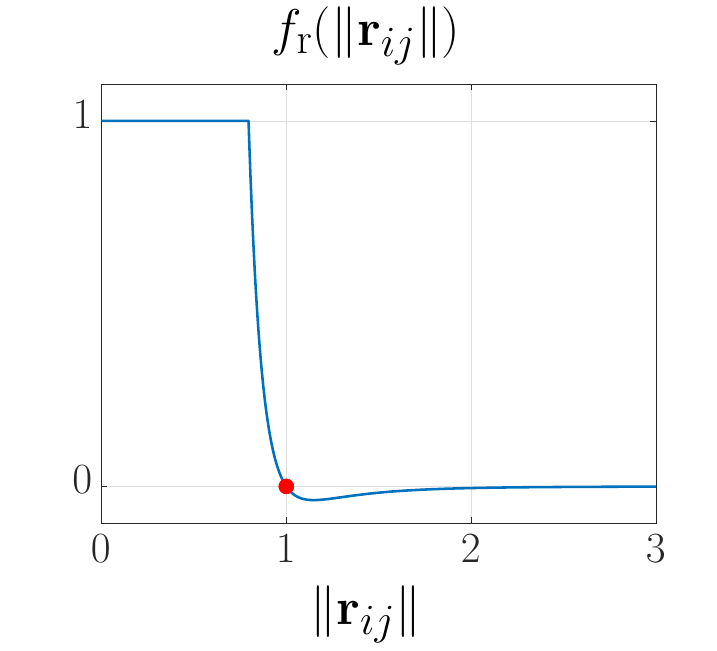}
        \caption{Radial interaction function}
        \label{subfig:interaction_func_radial}
    \end{subfigure}
    \begin{subfigure}[t]{0.45\textwidth}
        \includegraphics[width=1\textwidth]{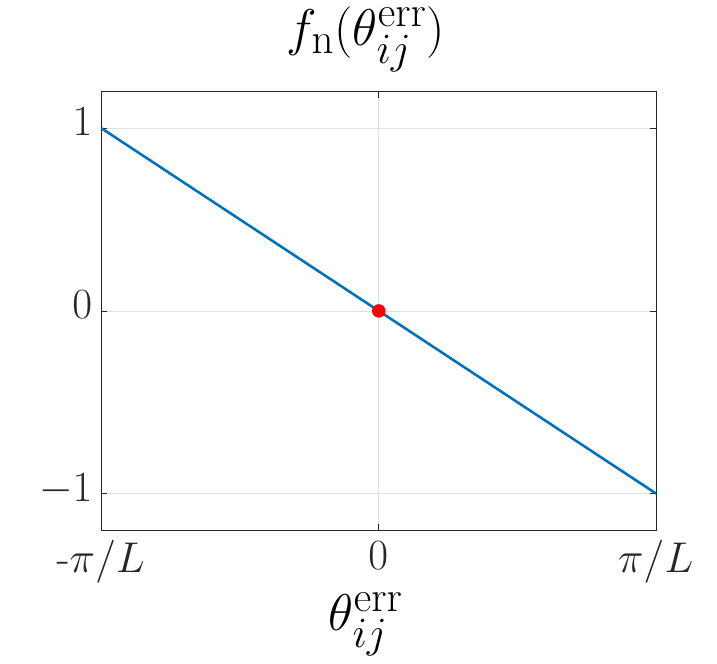}
        \caption{Normal interaction function}
        \label{subfig:interaction_func_normal}
    \end{subfigure}
    \caption{Interaction functions.
    Red dots highlight zeros of the functions.
    Parameters are taken from Tab \ref{tab:parameters}.    }
    \label{fig:interaction_functions}
\end{figure}

\subsubsection{Normal Interaction}
\label{subsec:NIF}

For any link $(i,j)$, we define the \emph{angular error} $\theta_{ij}^{\mathrm{err}} \in \left]- \frac{\pi}{L}, \frac{\pi}{L}\right]$ as the difference between $\theta_{ij}$, introduced in Section \ref{subsec:swarm_definition}, and the closest multiple of $2\pi/L$ (see Figure \ref{fig:two_agents_positions}), that is,
\begin{align}\label{eq:theta_ij^err}
    \theta_{ij}^{\mathrm{err}} &\coloneqq \theta_{ij} - 
    \frac{2 \pi}{L}
    \arg \min_{q \in \mathbb{Z}} \left\{ \left\vert \theta_{ij} - q\frac{2 \pi}{L} \right\vert \right\},
\end{align}

Then, the normal control input $\vec{u}_{\text{n},i}$ in \eqref{eq:controlLaw} is chosen as 
\begin{align}\label{eq:normalInput}
    \vec{u}_{\text{n},i} = G_{\text{n},i}\sum_{j \in \mathcal{A}_i} f_{\text{n}}(\theta_{ij}^{\mathrm{err}}) \frac{\vec{r}_{ij}^\perp}{\norm{\vec{r}_{ij}}},
\end{align}
where $G_{\text{n,i}} \in \mathbb{R}_{\geq 0}$ is the normal control gain.
Note that each of the normal virtual forces is applied in the direction of $\vec{r}_{ij}^\perp$, that is the vector normal to $\vec{r}_{ij}$, obtained by applying a $\pi/2$ counterclockwise rotation (see Figure \ref{fig:two_agents_positions}). The magnitude and sign of these forces are determined by the \emph{normal interaction function} $f_{\text{n}}: \,\left]-\frac{\pi}{L},\frac{\pi}{L}\right] \rightarrow \left[-1, 1\right[$, given by
\begin{align}\label{eq:LinearFn}
    f_{\text{n}}(\theta_{ij}^{\mathrm{err}}) & =  -\frac{L}{\pi} \,  \theta_{ij}^{\mathrm{err}}.
\end{align}
$f_\mathrm{n}$ is portrayed in Figure \ref{subfig:interaction_func_normal}. 

\section{Numerical validation}
\label{sec:results}

In this section, we assess the performance and the robustness of our proposed control algorithm \eqref{eq:controlLaw} through an extensive simulation campaign. The experimental validation of the strategy is later reported in Section \ref{sec:robotarium}.
First in Section \ref{subsec:tuning}, using a numerical optimisation procedure, we tune the control gains $G_{\text{r},i}$ and $G_{\text{n},i}$ in \eqref{eq:radInput} and \eqref{eq:normalInput}, as the performance of the controlled swarm strongly depends on these values. 
Then in Section \ref{sec:robustness_analysis}, we assess the robustness of the control law with respect to (i) agents' failure, (ii)   noise, (iii) flexibility to pattern changes, and (iv) scalability.
Finally in Section \ref{subsec:comparison}, we present a comparative analysis of our distributed control strategy and other approaches previously presented in the literature.
The simulations and experiments performed in this and the next Sections are summarised in Table \ref{tab:simulations}
\begin{table}[t]
\begin{center}
    \begin{tabular}{@{}lll@{}}
    \toprule
    Scenario & Section & Figure  \\
    \midrule
    \emph{Control law \eqref{eq:controlLaw},\eqref{eq:radInput},\eqref{eq:normalInput}}\\
    \ \ Tuning & \ref{subsec:tuning} & \ref{fig:tuning_cost_map}\\
    \ \ Validation & \ref{subsec:tuning} & \ref{fig:simulation_triangular_and_sqaure}\\
    \ \ Robustness to faults & \ref{subsec:agents_removal} & \ref{fig:AgentsRemoval}\\
    \ \ Robustness to noise & \ref{subsec:noise} & \ref{fig:NoiseTest}\\
    \ \ Flexibility & \ref{subsec:DynLatt} & \ref{fig:dynlattices}\\
    \ \ Scalability & \ref{subsec:Scalability} & \ref{fig:Scalability}\\
    \ \ Comparison with established algorithm & \ref{subsec:comparison} & \ref{fig:spears}\\
    \hline
    \emph{Adaptive gain tuning \eqref{eq:adaptationLaw} }\\
    \ \ Validation & \ref{sec:adaptive} & \ref{fig:AdaptiveSquares}\\
    \ \ Robustness to faults & \ref{subsec:adaptiveRobustnessTests} & \ref{fig:adaptive_agents_removal}\\
    \ \ Flexibility & \ref{subsec:adaptiveRobustnessTests} & \ref{fig:adaptive_flexibility}\\
    \ \ Scalability & \ref{subsec:adaptiveRobustnessTests} & \ref{fig:adaptive_scalability}\\
    \hline
    Robotarium experiment & \ref{sec:robotarium} & \ref{fig:Robotarium}\\
    \hline
\end{tabular}
\caption{List of simulations and experiments with indication of the corresponding section and figures.}
\label{tab:simulations}
\end{center}
\end{table}
\subsection{Simulation setup}
\label{sec:simulation_setup}
We consider a swarm consisting of $N = 100$ agents (unless specified differently).
To represent the fact that the agents are deployed from a unique source (as typically done in the literature, see e.g., \cite{Spears2004}), their initial positions are drawn randomly with uniform distribution from a disk of radius $r=2$ centred at the origin.%
\footnote{That is, denoting with $U([a,b])$ the uniform distribution on the interval $[a,b]$, the initial position of each agent in polar coordinates $\vec{x}_i(0) \coloneqq (d_i,\phi_i)$ is obtained by independently sampling $\phi_i \sim U\left([0,2\pi[\right)$ and $d_i$ is chosen according to the  probability density function $p_\mathrm{l}(\xi): [0,r]\mapsto \mathbb{R}_{\geq 0}$ defined as $p_l(\xi) = 2 \xi/r^2$.}\textsuperscript{, }%
\footnote{\label{foot:deployment} We also considered different deployment strategies (e.g., agents starting uniformly distributed from a larger disk or several disjoint disks) and verified that the results are qualitatively similar.}

Initially, for the sake of simplicity and to avoid the possibility of some agents becoming disconnected from the group, we assume that the sensing radius $R_{\text{s}}$ in Definition \ref{def:interaction_set} is large enough so that
\begin{equation}\label{eq:assumption_sensing_radius}
    \forall i \in \mathcal{S}, \forall t \in \mathbb{R}_{\ge 0}, \quad \mathcal{I}_i(t) = \mathcal{S} \setminus i;
\end{equation}
i.e., any agent can sense the relative position of all others.
Later, in Section \ref{sec:robustness_analysis}, we will drop this assumption and show the validity of our control strategy also for smaller values of $R_{\text{s}}$.
All simulation trials are conducted in {\sc Matlab} using the agent based simulator we developed SwarmSim V1 (for more information see Appendix \ref{ch:swarmsim} or visit \url{https://github.com/diBernardoGroup/SwarmSimPublic/tree/SwarmSimV1}).
The agents' speed is limited to $V_{\mathrm{max}} > 0$ and their dynamics is integrated using the forward Euler method with a fixed time step $\Delta t > 0$,
The values of the parameters used in the simulations are reported in Table \ref{tab:parameters}.

\begin{table}[t]
\begin{center}
    \begin{tabular}{@{}lll@{}}
    \toprule
    Parameter & Description & Value \\
    \midrule
    $R$ & Desired link length &  $1 \, \text{m}$ \\
    $R_{\min}$ & Minimum link length & $0.6 \, \text{m}$\\
    $R_{\max}$ & Maximum link length &  $1.1 \, \text{m}$ \\
    $V_{\max}$ & Maximum speed & $5 \, \text{m/s}$ \\
    $t_{\max}$ & Maximum simulation time & $200$ s\\
    $\Delta t$ & Integration step & $0.01 \, \text{s}$\\
    $T_{\text{w}}$ & Time window & $10 \, \text{s}$ \\
    $a$ & Radial interaction function $f_{\text{r}}(\cdot)$ & 0.15 \\
    $b$ & '' & 0.15\\
    $c$ & '' & 5\\    
    \hline
    \end{tabular}
\caption{Simulation parameters}
\label{tab:parameters}
\end{center}
\end{table}

\subsubsection{Performance evaluation}
\label{subsubsec:performance_eval}
To assess the performance of the controlled swarm, we exploit the metrics $e_{\theta}$ and $e_L$ given in Definitions \ref{def:regularity} and \ref{def:compactness}.
Namely, we select empirically the thresholds $e_{\theta}^* = 0.2$ and $e_L^* = 0.3$, which are associated to satisfactory compactness and regularity of the swarm.
Then, letting $T_{\text{w}}>0$ be the length of a time window, we say that $e_{\theta}$ is at \emph{steady-state} from time $t'=k \Delta t$ (for $k \in \mathbb{Z}$) if 
\begin{equation}
  \vert e_{\theta}(t') - e_{\theta}(t' - j\Delta t) \vert \leq 0.1 \, e_{\theta}^*, \
  \forall j \in \left\{ 1, 2, \dots, \left\lfloor \frac{T_{\text{w}}}{\Delta t} \right\rfloor \right\}.
\end{equation}
We give an analogous definition for the steady state of $e_L$ (using $e_L^*$ rather than $e_\theta^*$).
Then, we say that in a trial the swarm \emph{achieved steady-state} at time $t_{\mathrm{ss}}$ if there exists a time instant such that both $e_{\theta}$ and $e_L$ are at steady state, and $t_{\mathrm{ss}}$ is the smallest of such time instants.
Moreover, we deem the trial \emph{successful} if $e_{\theta}(t_{\mathrm{ss}}) < e_{\theta}^*$ and $e_L(t_{\mathrm{ss}}) < e_L^*$.
If in a trial steady-state is not reached in the time interval $[0, t_{\max}]$, the trial is stopped (and deemed unsuccessful).
We define
\begin{equation}
    e_{\theta}^{\mathrm{ss}} \coloneqq
    \begin{dcases}
         e_{\theta}(t_{\mathrm{ss}}), & \text{if steady state is achieved}, \\
         e_{\theta}(t_{\mathrm{max}}), & \text{otherwise}. \\
    \end{dcases}
\end{equation}
\begin{equation}
    e_L^{\mathrm{ss}} \coloneqq
    \begin{dcases}
         e_L(t_{\mathrm{ss}}), & \text{if steady state is achieved}, \\
         e_L(t_{\mathrm{max}}), & \text{otherwise}. \\
    \end{dcases}
\end{equation}

Finally, to asses how quickly the pattern is formed, we define
\begin{align}
    T_{\theta} &\coloneqq\min \{t' \in \mathbb{R}_{\geq 0} :\ e_{\theta}(t') \leq e_{\theta}^*, \;\; \forall t \geq t'\},\\
    T_L &\coloneqq\min \{t'' \in \mathbb{R}_{\geq 0} :\ e_L(t'') \leq e_L^*, \;\; \forall t \geq t'' \},\\
    T &\coloneqq \max\{T_\theta, T_L\}
    \label{eq:convergence_time}.
\end{align} 

\subsection{Tuning of the control gains}
\label{subsec:tuning}
For the sake of simplicity, in this section we assume that 
$G_{\text{r},i} = G_{\text{r}}$ and $G_{\text{n},i} = G_{\text{n}}$, for all $i \in \mathcal{S}$; later, in Section \ref{sec:adaptive}, we will present an adaptive control strategy allowing each agent to independently vary online its own control gains.
To select the values of $G_{\text{r}}$ and $G_{\text{n}}$ giving the best performance in terms of regularity and compactness, we conducted an extensive simulation campaign and evaluated.
Specifically, for each pair $(G_{\text{r}}, G_{\text{n}}) \in \{0, 1, \dots ,30\} \times \{0, 1, \dots ,30\}$, we executed $30$ trials, starting from random initial conditions, and averaged the following cost function:
\begin{equation}
    C(e_{\theta}^{\mathrm{ss}}, e_{L}^{\mathrm{ss}}) \coloneqq 
    \left(\frac{e_{\theta}^{\mathrm{ss}}}{e_{\theta}^*}\right)^2 + \left(\frac{e_{L}^{\mathrm{ss}}}{e_{L}^*}\right)^2.
    \label{eq:costFunction}
\end{equation}
The results are reported in Figure \ref{fig:tuning_cost_map} for the triangular ($L=6$) and the square ($L=4$) lattices; in the former case, the pair $(G_{\text{r}}^*, G_{\text{n}}^*)_{L=4}$ minimising $C$ is $(22, 1)$, whereas in the latter case it is $(G_{\text{r}}^*, G_{\text{n}}^*)_{L=6} = (15, 8)$.
Both pairs achieve $C\le 1$, implying $e_\theta^{\mathrm{ss}} \le e^*_\theta$ and $e_L^{\mathrm{ss}} \le e_L^*$.

\begin{figure}[t!]
    \centering
    \begin{subfigure}[t]{0.48\textwidth}
        \includegraphics[width=1\textwidth]{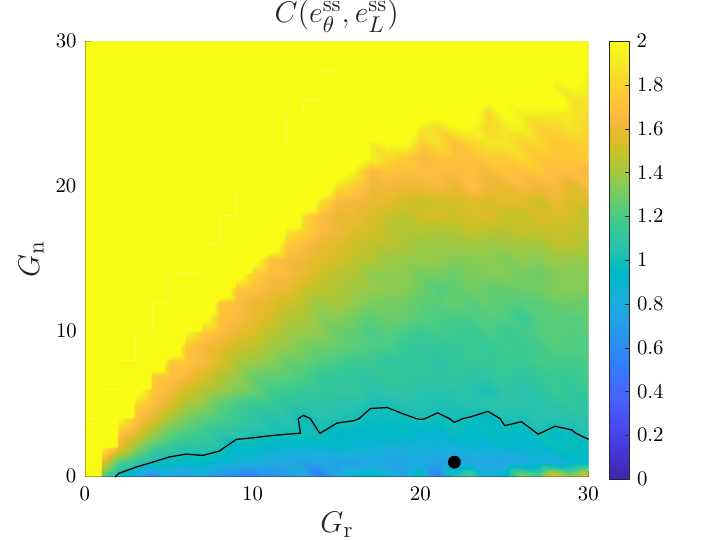}
        \caption{Triangular lattice ($L=6$)}
    \end{subfigure}
    \begin{subfigure}[t]{0.48\textwidth}
        \includegraphics[width=1\textwidth]{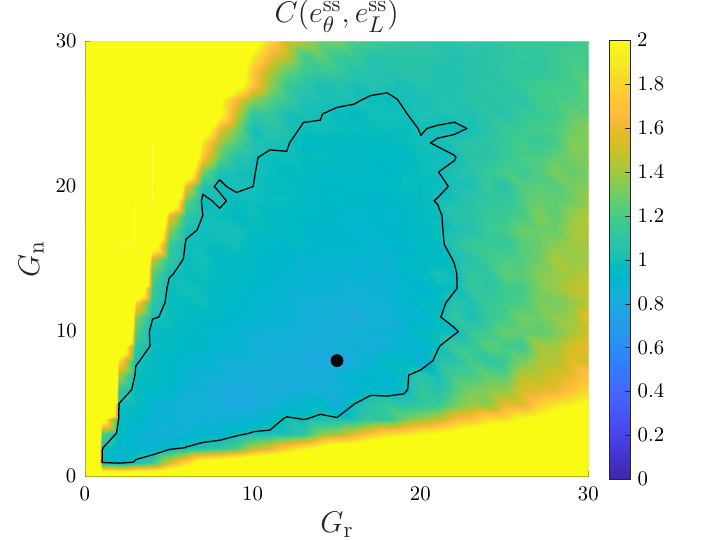}
        \caption{Square lattice ($L=4$)}
    \end{subfigure}
    \caption{Tuning of the control gains $G_{\text{r}}$ and $G_{\text{n}}$.
    The black dots correspond to $(G_{\text{r}}^*,G_{\text{n}}^*)_{L=6}$ and $(G_{\text{r}}^*,G_{\text{n}}^*)_{L=4}$, minimising the metric cost $C$ defined in \eqref{eq:costFunction}. 
    The black curves delimit the regions where $C \leq 1$.
    }
    \label{fig:tuning_cost_map}
\end{figure}

In Figure \ref{fig:simulation_triangular_and_sqaure}, we report four snapshots, at different time instants, of two representative simulations, together with the metrics $e_{\theta}(t)$ and $e_L(t)$, for the cases of a triangular and a square lattice, respectively.
The control gains were set to the optimal values $(G_{\text{r}}^*,G_{\text{n}}^*)_{L=6}$ and $(G_{\text{r}}^*,G_{\text{n}}^*)_{L=4}$.
In both cases, the metrics quickly converge below their prescribed thresholds, as $T$ $< 2.75 \, \text{s}$.
Moreover, note that $e_L(t)$ decreases faster than $e_{\theta}(t)$, meaning that the swarm tends to first reach the desired level of compactness and then agents' positions are rearranged to achieve the desired pattern.
Finally, we note, and it is immediate to verify, that it is possible to control the orientation of the resulting lattice simply by applying a uniform offset to the agents' compasses.

\begin{figure}[t!]
    \includegraphics[width=1.05\textwidth]{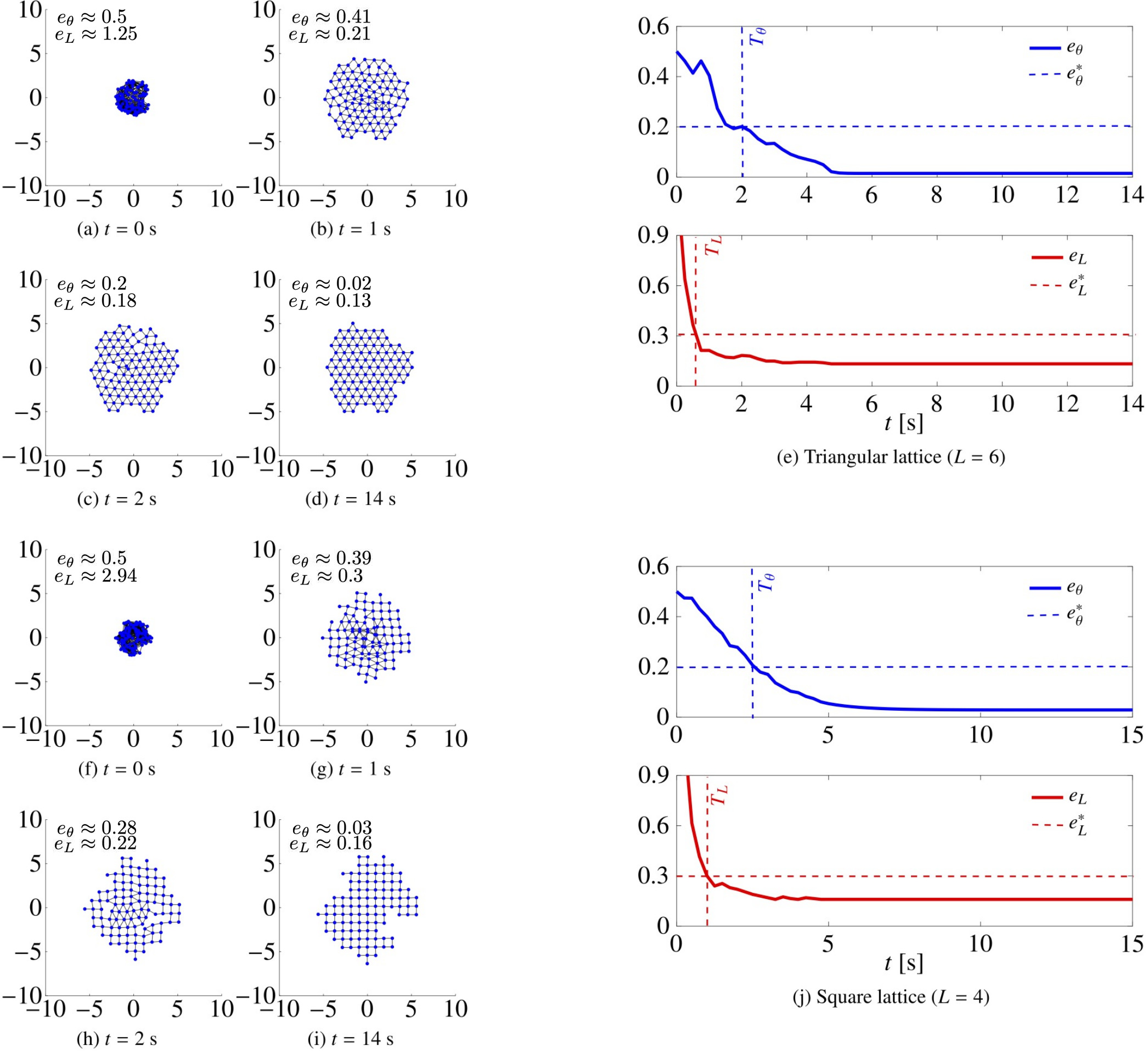}
    \caption{
    Snapshots at different time instants of a swarm of $N=100$ agents being controlled to form a triangular lattice (a-d) and a square lattice (f-i). For each snapshot, we also report the values of $e_\theta$ and $e_L$.
    (e) and (j) show the time evolution of the metrics $e_{\theta}$ and $e_L$ for $L=6$ and $L=4$, respectively.
    When $L=6$, we set $(G_{\text{r}},G_{\text{n}}) = (G_{\text{r}}^*,G_{\text{n}}^*)_{L=6}$; when $L=4$, we set $(G_{\text{r}},G_{\text{n}}) = (G_{\text{r}}^*,G_{\text{n}}^*)_{L=4}$. (See \ref{subsec:tuning} for details on how the gains were tuned.)
    }
    \label{fig:simulation_triangular_and_sqaure}
\end{figure}

\subsection{Robustness analysis}
\label{sec:robustness_analysis}

In this section, we investigate numerically the properties that we required in Section \ref{sec:ProblemStatement}, that is robustness to faults and noise, flexibility, and scalability.

\subsubsection{Robustness to faults}
\label{subsec:agents_removal}
To analyse the robustness of the controlled swarm to agents' faults, we ran a series of simulations in which we removed a percentage of the agents at a certain time instant, and assessed the capability of the swarm to recover the desired pattern.
For the sake of brevity, we report only one of them as a representative example in Figure \ref{fig:AgentsRemoval}, where, with $L=4$, 30\% of the agents were removed at random at time $t = 30 \, \mathrm{s}$.
We notice that, as the agents are removed, $e_L(t)$ and $e_{\theta}(t)$ suddenly increase, but, after a short time, they converge again to values below the thresholds, recovering the desired pattern, despite the formation of small holes in the pattern at steady-state that increase $e_L^{\mathrm{ss}}$. 
Finally, we also considered the case where the faulty agents stay still in their positions after the fault, with other agents having to form the lattice around them.
We observed that when the fault takes place after a satisfying structure is formed, the metrics are not affected by the event (the numerical results are omitted here as redundant).

\begin{figure}[t!]
    \begin{minipage}{0.45\textwidth}
    \centering
    \begin{subfigure}[t]{0.48\textwidth}
        \includegraphics[width=1\textwidth]{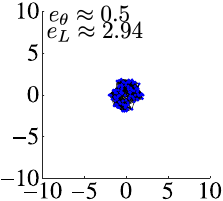}
        \caption{$t=0\,\text{s}$}
    \end{subfigure}
    \begin{subfigure}[t]{0.48\textwidth}
        \includegraphics[width=1\textwidth]{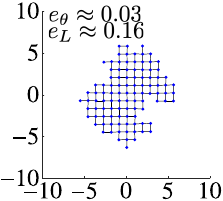}
        \caption{$t=29.99\,\text{s}$}
    \end{subfigure}
    
    \begin{subfigure}[t]{0.48\textwidth}
        \includegraphics[width=1\textwidth]{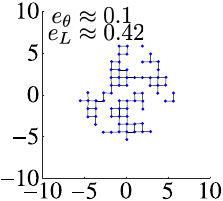}
        \caption{$t=30\,\text{s}$}
    \end{subfigure}
    \begin{subfigure}[t]{0.48\textwidth}
        \includegraphics[width=1\textwidth]{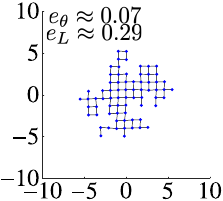}
        \caption{$t=60\,\text{s}$}
    \end{subfigure}
    \end{minipage}
    \begin{minipage}{0.55\textwidth}
    \begin{subfigure}[t]{1\textwidth}
        \includegraphics[width=1\textwidth]{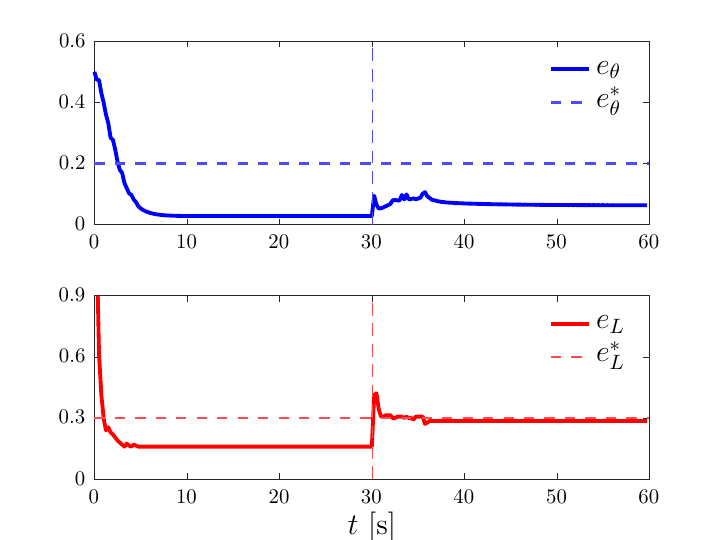}
        \caption{}
    \end{subfigure}
    \end{minipage}
    \caption{Robustness to agents' removal.
    (a-d) Snapshots at different time instants of a swarm achieving square lattice. Initially, there are $N=100$ agents with $30$ agents being removed at $t=30 \, \text{s}$.
    (e) Time evolution of the metrics; dashed vertical lines denote the time instant when agents are removed. 
    Here $L=4$, and $(G_{\text{r}},G_{\text{n}}) = (G_{\text{r}}^*,G_{\text{n}}^*)_{L=4}$.
    }
    \label{fig:AgentsRemoval}
\end{figure}

\subsubsection{Robustness to noise}

\label{subsec:noise}
%
We assessed the robustness to noise both on actuation and on sensing, in two separate tests.
In the first case, we assumed that the dynamics \eqref{eq:firstOrdDynamics} of each agent is affected by additive white Gaussian noise with standard deviation $\sigma_{\mathrm{a}}$.
In the second case, we assumed that, for each agent, both the distance measurements $\Vert \vec{r}_{ij} \Vert$ in \eqref{eq:radInput} and the angular measurements $\theta_{ij}^{\mathrm{err}}$ in \eqref{eq:normalInput} are affected by additive white Gaussian noise (i.i.d. for each $i$ and $j$) with standard deviation $\sigma_{\mathrm{m}}$ and $\sigma_{\mathrm{m}} \frac{\pi}{L}$, respectively.

In particular, we set $L=4$ and varied either $\sigma_{\mathrm{a}}$ or $\sigma_{\mathrm{m}}$ in intervals of interest with small increments%
. 
For each value of $\sigma_{\mathrm{a}}$  and $\sigma_{\mathrm{m}}$, we ran $M = 30$ trials, starting from random initial conditions, and reported the average values of $e_{\theta}^{\mathrm{ss}}$ and $e_{L}^{\mathrm{ss}}$ in Figure \ref{fig:NoiseTest}.
%
%
We observe that, while in the ranges $\sigma_{\mathrm{a}} \in [0,0.45]$ and $\sigma_{\mathrm{m}} \in [0,0.125]$ the strategy guarantees robustness, for large enough noise ($\sigma_{\mathrm{a}} \geq 0.45$ or $\sigma_{\mathrm{m}} \geq 0.125$) performance is increasingly worsened with trials eventually becoming unsuccessful (the swarm never achieving the desired lattice configuration).
Interestingly, we find that for smaller noise ($0 < \sigma_{\mathrm{a}} \leq 0.2$ or $0 < \sigma_{\mathrm{m}} \leq 0.1$) performance is actually improved, as small random inputs can prevent the agents from getting stuck in undesired configurations, including those containing holes.

We obtained qualitatively similar results when we assumed the presence of noise on the compass measurements of the agents (obtained by adding Gaussian noise on the variables $\theta_{ij}^{\mathrm{err}}$, with the noise value being the same for $\theta_{ij}^{\mathrm{err}}$ and $\theta_{kl}^{\mathrm{err}}$ when $i=k$).

%

\begin{figure}[t!]
    \centering
    \begin{subfigure}[t]{0.48\textwidth}
        \includegraphics[width=1\textwidth]{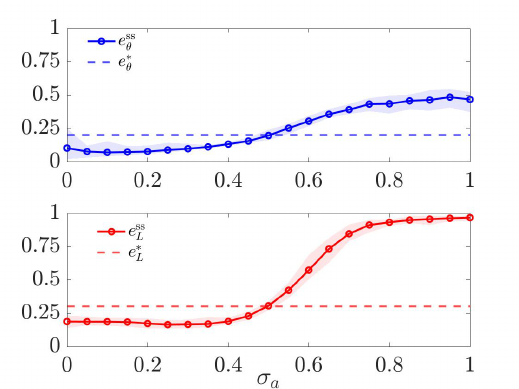}
        \caption{Actuation noise}
    \end{subfigure}
    \begin{subfigure}[t]{0.48\textwidth}
        \includegraphics[width=1\textwidth]{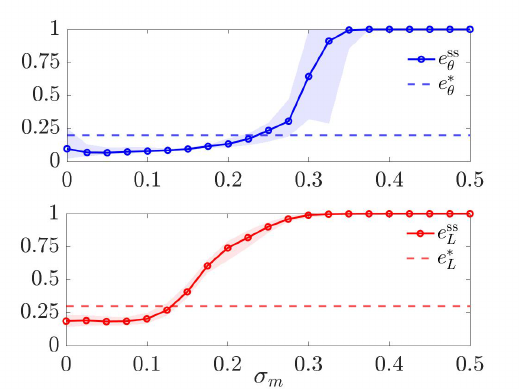}
        \caption{Measurement noise}
    \end{subfigure}
    \caption{Robustness to noise.
    Value of the metrics $e_\theta^\mathrm{ss}$ and $e_L^\mathrm{ss}$, averaged over $M = 30$ trials, when (a) the intensity $\sigma_{\mathrm{a}}$ of the actuation noise is varied and (b) the intensity $\sigma_{\mathrm{m}}$ of the measurement noise is varied.
    The shaded areas represent the maximum and minimum values obtained over the $M$ trials.
    Here $L=4$, and $(G_{\text{r}},G_{\text{n}}) = (G_{\text{r}}^*, G_{\text{n}}^*)_{L=4}$.
    }
    \label{fig:NoiseTest}
\end{figure}

\subsubsection{Flexibility}
\label{subsec:DynLatt}
%
%
In Figure \ref{fig:dynlattices}, we report a simulation where $L$ was initially set equal to $4$ (square lattice), changed to $6$ (triangular lattice) at time $t=30  \ \mathrm{s}$, and finally changed back to $4$ at $t=60 \ \mathrm{s}$.
The control gains are set to $(G_{\text{r}}^*,G_{\text{n}}^*)_{L=4}$ and kept constant during the entire the simulation.
Clearly, as $L$ is changed, both $e_L$ and $e_{\theta}$ suddenly increase, but the swarm is quickly able to reorganise and reduce them below their prescribed thresholds in less than $5\, \mathrm{s}$, thus achieving the desired patterns.
%
\begin{figure}[t!]
    \centering
    \includegraphics[width=0.5\columnwidth]{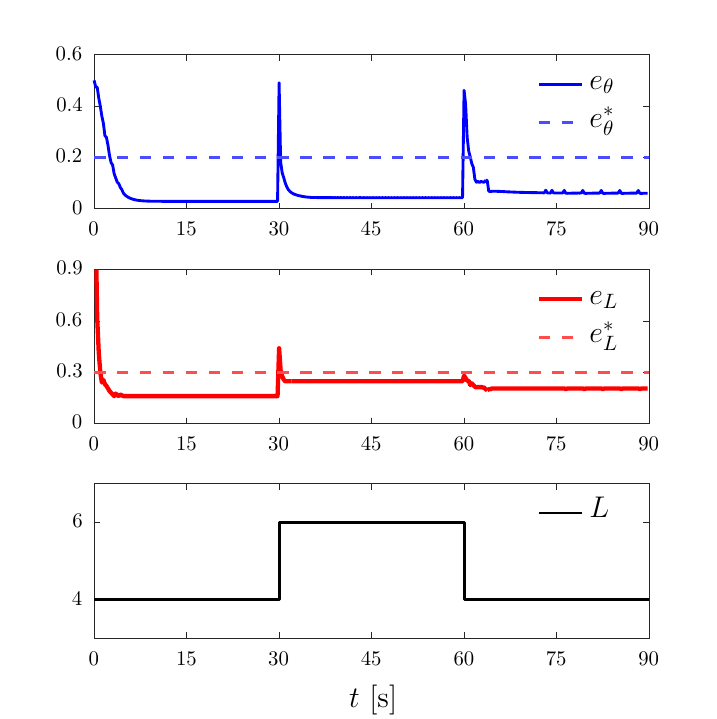}
    \caption{Flexibility to spatial reorganisation. 
    Time evolution of the metrics $e_\theta^\mathrm{ss}$ and $e_L^\mathrm{ss}$ as $L$ changes as shown in the bottom panel. The gains are set as
    $(G_{\text{r}},G_{\text{n}}) = (G_{\text{r}}^*,G_{\text{n}}^*)_{L=4}$.
    }
    \label{fig:dynlattices}
\end{figure}
%

\subsubsection{Scalability}
\label{subsec:Scalability}
%
We relaxed the assumption \eqref{eq:assumption_sensing_radius}  and characterised $e_L^{\mathrm{ss}}$ as a function of the sensing radius $R_{\text{s}}$.
The results are portrayed in Figure \ref{subfig:scalability_sensing_radius}, showing that the performance starts deteriorating for approximately $R_{\text{s}} < 6 \, R$, until it becomes unacceptable for about $R_{\text{s}} < 1.1 \, R$.
%
%
Therefore, as a good trade-off between performance and feasibility, we set $R_{\text{s}} = 3 \, R$.
Then, to test for scalability, we varied the number $N$ of agents (initially, $N = 100$), reporting the results in Figure \ref{subfig:scalability_number_of_agents}.
We see that (i) the controlled swarm correctly achieves the desired pattern for at least four-fold changes in the size of the swarm, (ii) compactness ($e_L^{\mathrm{ss}}$) improves as $N$ increases, and (iii) the average convergence time $T$ increases as $N$ increases.

%
%
\begin{figure}[t!]
    \centering
    \begin{subfigure}[t]{0.48\textwidth}
        \includegraphics[width=1\textwidth]{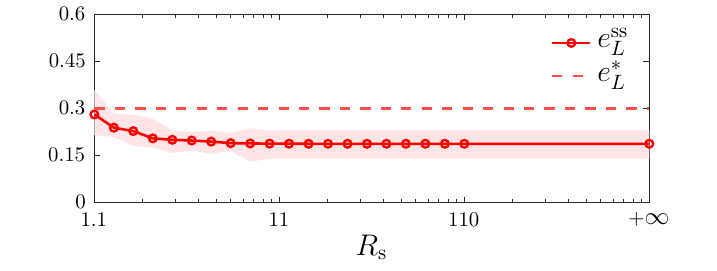}
        \caption{}
        \label{subfig:scalability_sensing_radius}
    \end{subfigure}
    \begin{subfigure}[t]{0.48\textwidth}
        \includegraphics[width=1\textwidth]{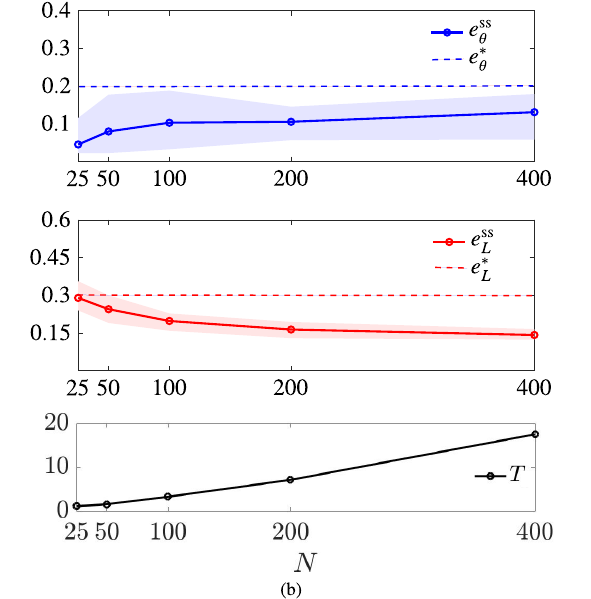}
        \caption{}
        \label{subfig:scalability_number_of_agents}
    \end{subfigure}
    \caption{Scalability.
    (a) $e_L^\mathrm{ss}$ averaged over $M=30$ trials for different values of the sensing radius $R_{\text{s}}$.
    (b) Metrics $e_\theta^\mathrm{ss}$ and $e_L^\mathrm{ss}$ and the convergence time $T$ averaged over the trials, with varying $N$, while 
    $R_{\text{s}} = 3\, \text{m}$, and agents' initial positions are drawn with uniform distribution from a disk with radius $r=\sqrt{N/25}$.
    The shaded areas represent the maximum and minimum values over the $M$ trials. Here
    $L=4$ and the gains are set as $(G_{\text{r}},G_{\text{n}}) = (G_{\text{r}}^*,G_{\text{n}}^*)_{L=4}$.
    }    
    \label{fig:Scalability}
\end{figure}

\subsection{Comparison with other established algorithm}
\label{subsec:comparison}

\begin{figure}
    \centering
    \begin{subfigure}[t]{0.48\textwidth}
        \includegraphics[width=1\textwidth]{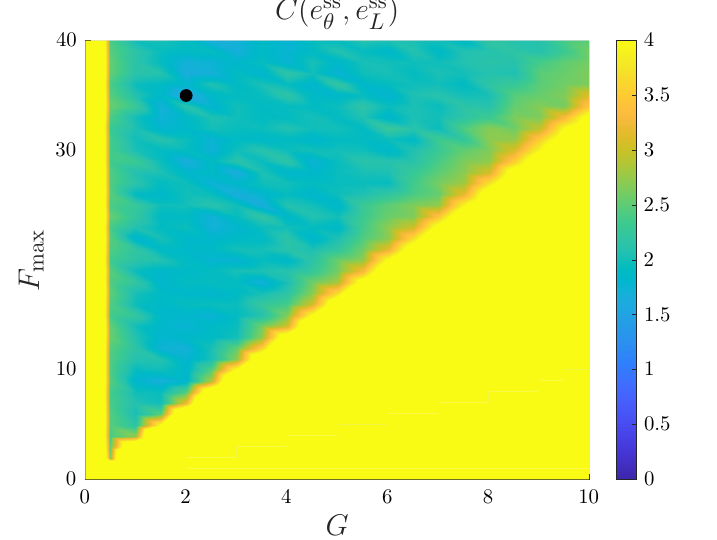}
        \caption{}
        \label{subfig:spears_tuning_cost}
    \end{subfigure}
    \begin{subfigure}[t]{0.48\textwidth}
        \includegraphics[width=1\textwidth]{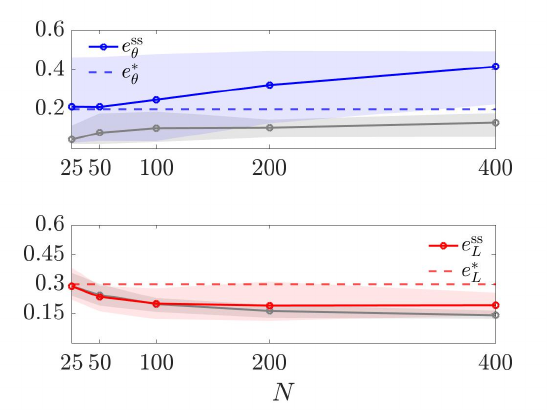}
        \caption{}
        \label{subfig:spears_scalability}
    \end{subfigure}    
    \caption{Testing the algorithm from \cite{Spears2004}. 
    (a) Tuning of parameters $G$ and $F_{\max}$ for the square lattice ($L=4$). The black dot denotes the optimal pair $(G^*,F_{\max}^*)$. 
    (b) Scalability test.
    The metrics $e_L^\mathrm{ss}$ and  $e_\theta^\mathrm{ss}$ are averaged over $M=30$ trials, as $N$ varies, and plotted against our results (in gray) in the same scenario (see Figure \ref{fig:Scalability} for the sake of comparison).
    Agents' initial positions are drawn with uniform distribution from a disk of radius $r=\sqrt{N/25}$.
    The shaded area represents the maximum and minimum values over the trials. Here
    $L=4$, and $(G, F_{\max}) = (G^*, F_{\max}^*)$.}
    \label{fig:spears}
\end{figure}
As done in related literature \cite{Zhao2019} (yet for a position-based solution), we compared our control law \eqref{eq:controlLaw} to the so-called ``gravitational virtual forces'' strategy first introduced in \cite{Spears1999}  and later extended in \cite{Spears2004}, that represents an established solution to geometric pattern formation problems. 

In these works the agents' dynamics is described by 
\begin{align}
\begin{cases}
\dot{\vec{x}}_i =\vec{v}_i, \\
\dot{\vec{v}}_i =\frac{1}{m}(\vec{u}_i-\mu \vec{v}_i),
\end{cases}
\quad \forall i \in \mathcal{S},
\label{eq::SecondOrderIntegrDamped}
\end{align}
where $\vec{u}_i\in \mathbb{R}^2$ is the control input, $m\in \mathbb{R}_{>0}$ is the mass of the agent and $\mu\in \mathbb{R}_{>0}$ is the friction damping factor. 
The control input $\vec{u}_i$ is given by the classic virtual forces control law, that is
\begin{equation}
    \vec{u}_i = \sum_{j=1}^N f(\Vert \vec{r}_{ij}\Vert) \frac{\vec{r}_{ij}}{\Vert\vec{r}_{ij}  \Vert}
    \label{eq::spearsContrInp},
\end{equation}
where $f$ is a gravitational-like virtual force, namely
\begin{equation}
    f(\Vert \vec{r}_{ij}\Vert )= \begin{cases}
   \left[ G \frac{m^2}{\Vert \vec{r}_{ij}\Vert ^2} \right]_0^{F_{\max}}, 
   &\mbox{if } \ 0\leq \Vert \vec{r}_{ij}\Vert \leq R,\\
    -\left[G \frac{m^2}{\Vert \vec{r}_{ij}\Vert ^2}\right]_0^{F_{\max}},
    &\mbox{if } \ R<\Vert \vec{r}_{ij}\Vert \leq1.5 R,\\
    0,
    &\mbox{otherwise}.
    \end{cases}
    \label{eq::spearsForce}
\end{equation}
Here $G$ and $F_{\max} \in \mathbb{R}_{\geq0}$ are tunable control gains, and $R \in \mathbb{R}_{>0}$ is the desired link length, while the \emph{saturation} function $\left[x(t) \right]_a^b:\BB{R}\to [a;b]$ is defined as
\begin{equation*}
    \left[x(t) \right]_a^b \coloneqq \begin{cases}
    a, &\mbox{if} \;\; x(t)<a, \\
    x(t), &\mbox{if} \;\; a\leq x(t) \leq b,\\
    b, &\mbox{if} \;\; x(t)>b.
    \end{cases}
\end{equation*}
The control law given by \eqref{eq::spearsContrInp} and  \eqref{eq::spearsForce} was showed to work for triangular lattices. 
To make it suitable for square patterns, a binary variable called \emph{spin} is introduced for each agent, and the swarm is divided in two subsets, depending on the value of their spin.
Then, agents with different spin aggregate at distance $R$, while agents with the same spin do so at distance $\sqrt{2}R$. 
The extension to the case of hexagonal lattice is discussed in \cite{Sailesh2014} and requires communication among the agents.
Notice that in \eqref{eq::SecondOrderIntegrDamped}, a second order damped dynamics is considered for the agents.
Nevertheless, for the sake of comparison and by recalling Remark \ref{rem:second_order_to_first} we reduced it to the first order model in \eqref{eq:firstOrdDynamics}, by assuming that the viscous friction force is significantly larger than the inertial one.
Then, to select the gravitational gain $G$ and the saturation value $F_{\max}$ in the control law from \cite{Spears2004}, we applied the same tuning procedure described in Section \ref{subsec:tuning}.
In particular, we considered $(G, F_{\max}) \in \{0, 0.5, \dots, 10\} \times \{0, 1, \dots, 40\}$, and performed $30$ trials for each pair of parameters, obtaining as optimal pair for the square lattice $(G^*, F_{\max}^*) = (35, 2)$ (see Figure \ref{subfig:spears_tuning_cost}).
All other parameters where kept to the default values in Table \ref{tab:parameters}.

Then, we performed the same scalability test in Section \ref{subsec:Scalability} and report the results in Figure \ref{subfig:spears_scalability}.
Remarkably, by comparing these results with ours, we see that our proposed control strategy performs better, obtaining much smaller values of $e_{\theta}^{\mathrm{ss}}$, regardless of the size $N$ of the swarm.
In particular, the control law from \cite{Spears2004} only rarely achieves $e_{\theta}^{\mathrm{ss}} \le e_{\theta}^*$, implying a low success rate. 

\section{Adaptive tuning of control gains}
\label{sec:adaptive}

Tuning the control gains (here $G_{\text{r}, i}$ and $G_{\text{n}, i}$) can in general be a tedious and time-consuming procedure.
Therefore, to avoid it, we propose the use of a simple, yet effective adaptive control law, that might also improve the robustness and flexibility of the swarm.
Specifically, for the sake of simplicity, $G_{\text{r},i}$ is set to a constant value $G_{\text{r}}$ for all the swarm, while each agent computes its gain $G_{\text{n},i}$ independently, using only local information.
Letting $e_{\theta,i} \in [0,1]$ be the \emph{average angular error} for agent $i$, given by
\begin{equation}
    e_{\theta,i}\coloneqq\frac{L}{\pi} \frac{1}{\vert \mathcal{A}_i \vert} \sum_{j\in \mathcal{A}_i} \vert \theta_{ij}^{\mathrm{err}} \vert,
    \label{eq:avgAngError}    
\end{equation}
$G_{\text{n},i}$ is varied according to the law
\begin{subequations}
\begin{align}
    \frac{\mathrm{d}}{\mathrm{d}t}{G}_{n,i}(t)&=\begin{cases}
    \alpha \, (e_{\theta,i}(t)-e_{\theta}^*), &\mbox{if } e_{\theta,i}(t) > e_{\theta}^*, \\
    0, &\text{otherwise}.
    \end{cases} \label{eq:adaptationLaw:1}\\
    G_{\text{n},i}(0)&=0,
    \label{eq:adaptationLaw:2}
\end{align}
\label{eq:adaptationLaw}
\end{subequations}
\noindent where $\alpha > 0$ is an adaptation gain and $e_{\theta}^*$ (introduced in Section \ref{sec:simulation_setup}) is used to determine the amplitude of the dead-zone. 
Here, we empirically choose $\alpha=3$. 
To evaluate the effect of the adaptation law, we also define the average normal gain of the swarm $\bar{G}_{\text{n}}(t) \coloneqq \frac{1}{N} \sum_{i=1}^N G_{\text{n},i}(t)$.

In Figure \ref{fig:AdaptiveSquares}, we report the time evolution of $e_L$, $e_{\theta}$, and of $\bar{G}_{\text{n}}$ for a representative simulation. 
First, we notice that the average normal gain $\bar{G}_{\text{n}}$ eventually settles to a constant value.
Moreover, comparing the results with the case in which the gains $G_{\mathrm{n},i}$ are static (see Section \ref{subsec:tuning} and Figure \ref{fig:simulation_triangular_and_sqaure}j), here $T_\theta$, $T_L$ and $t_\mathrm{ss}$ are larger (meaning longer convergence time), but $e_\theta^\mathrm{ss}$ and $e_L^\mathrm{ss}$ are smaller (meaning better regularity and compactness performance).

\begin{figure}[t!]
    \centering
    \includegraphics[width=0.6\columnwidth]{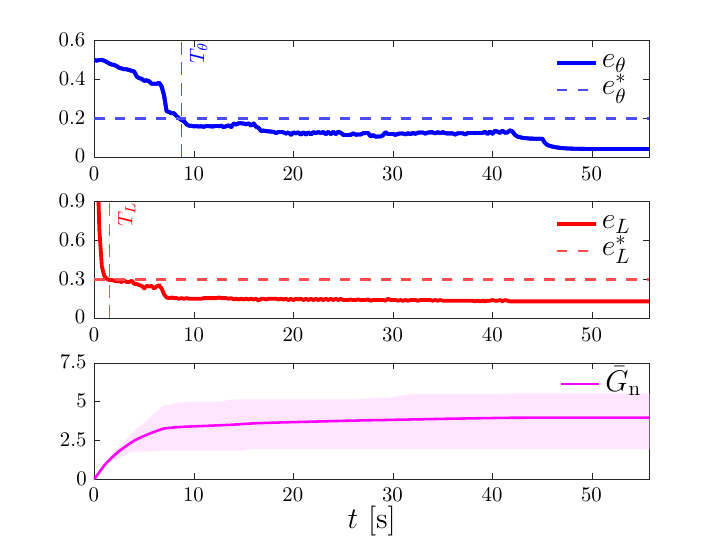}
    \caption{Pattern formation using the adaptive tuning law \eqref{eq:adaptationLaw}.
    Initial conditions are the same as those of the simulation in Figure \ref{fig:simulation_triangular_and_sqaure}.
    The shaded magenta area is delimited by $\max_{i \in \mathcal{S}} G_{\mathrm{n}, i}$ and $\min_{i \in \mathcal{S}} G_{\mathrm{n}, i}$, while the average across all agents is depicted by a solid magenta line. Here, $L=4$, and $G_{\text{r}}=15$.
    }
    \label{fig:AdaptiveSquares}
\end{figure}

\subsection{Robustness analysis}
\label{subsec:adaptiveRobustnessTests}
Next, we test robustness to faults, flexibility, and scalability for the adaptive law \eqref{eq:adaptationLaw}, similarly to what we did in Section \ref{sec:robustness_analysis}.

We ran a series of agents removal tests. 
For the sake of brevity, we report the results of one of such tests with $L=4$ in Figure \ref{fig:adaptive_agents_removal}.
At $t=30$ s, 30\%  of  the  agents are  removed; yet, after  a short time the swarm reaggregates to recover  the  desired  lattice. 
As observed in Section \ref{subsec:agents_removal}, at $t=30$ s both metrics increase, but, after a short transient, they converge again  below  the  respective  thresholds.

\begin{figure}
    \begin{minipage}{0.45\textwidth}
    \centering
    \begin{subfigure}[t]{0.48\textwidth}
        \includegraphics[width=1\textwidth]{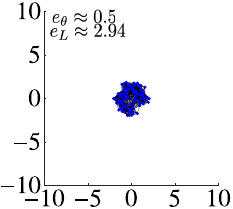}
        \caption{$t=0\,\text{s}$}
    \end{subfigure}
    \begin{subfigure}[t]{0.48\textwidth}
        \includegraphics[width=1\textwidth]{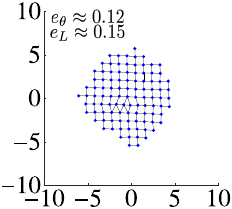}
        \caption{$t=29.99\,\text{s}$}
    \end{subfigure}
    
    \begin{subfigure}[t]{0.48\textwidth}
        \includegraphics[width=1\textwidth]{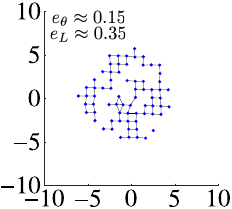}
        \caption{$t=30\,\text{s}$}
    \end{subfigure}
    \begin{subfigure}[t]{0.48\textwidth}
        \includegraphics[width=1\textwidth]{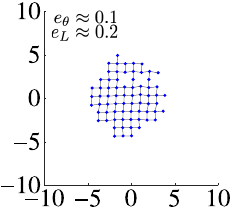}
        \caption{$t=60\,\text{s}$}
    \end{subfigure}
    \end{minipage}
    \begin{minipage}{0.54\textwidth}
    \begin{subfigure}[t]{1\textwidth}
        \includegraphics[width=1\textwidth]{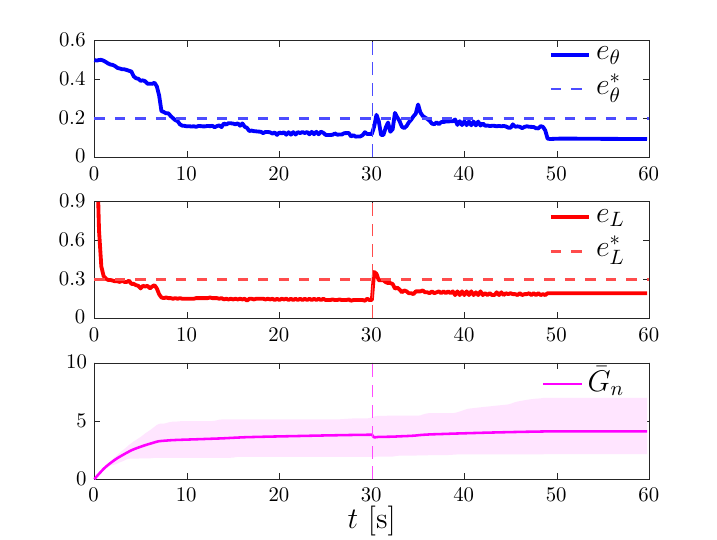}
        \caption{}
    \end{subfigure}
    \end{minipage}
    \caption{Robustness to agents removal using the adaptive tuning law \eqref{eq:adaptationLaw}. The simulation starts with $N=100$ agents, $30$ agents are removed at $t=30 \, \text{s}$. Initial conditions are the same as those of the simulation in Figure \ref{fig:AgentsRemoval}.
    (a--d) snapshots of the agents' configurations at different time instants. 
    (e) time evolution of the metrics $e_\theta$  and $e_L$, and of the adaptive gain $G_{\mathrm{n}}$ (the shaded magenta area is delimited by $\max_{i \in \mathcal{S}} G_{\mathrm{n}, i}$ and $\min_{i \in \mathcal{S}} G_{\mathrm{n}, i}$ while the average value of the gain over all the agents is shown as a solid magenta line). Dashed vertical lines denote the time instant when agents are removed. Here $L=4$. 
    }
    \label{fig:adaptive_agents_removal}
\end{figure}

When performing the flexibility test (time-varying reference lattice), a non trivial issue was how to set the value of the control gains, as, in a single trial, the algorithm was tested over different patterns to form dynamically.
For the sake of simplicity, we tested our solution using the optimal gains found for the formation of square lattices, as this structure is intrinsically more complex than the triangular one. 
We re-propose here the same experiment, but considering adaptive tuning of the normal gain.
Specifically, we set $G_{\text{r}} = 18.5$ (that is the average between the optimal gain for square and triangular patterns), and set $G_{\text{n}, i}$  according to law \eqref{eq:adaptationLaw}, resetting all $G_{\mathrm{n}, i}$ to $0$ when $L$ is changed.
The results are shown in Figure \ref{fig:adaptive_flexibility}. 
When compared to the non-adaptive case (Figure \ref{fig:dynlattices}), here $e_\theta^\mathrm{ss}$ and $e_L^\mathrm{ss}$ are smaller (better pattern formation), but $T_\theta$ and $T_L$ are larger (slower), especially when forming square patterns. 
Interestingly, when $L = 4$, $\bar{G}_{\mathrm{n}}$ settles to about 5, while when $L = 6$ it settles to about 0.3, a much smaller value.


Finally, we repeated the test in Section \ref{subsec:Scalability}, setting again the sensing radius $R_{\text{s}}$ to 3 $R$ and assessing performance while varying the size $N$ of the swarm; results are shown in Figure \ref{fig:adaptive_scalability}. 
First, we notice that the larger the swarm is, the larger the steady state value of $\bar{G}_{\mathrm{n}}$ is.
Comparing the results with those obtained for static gains shown in Figure \ref{subfig:scalability_number_of_agents}, here we observe a slight improvement of performance, with a slightly smaller $e_\theta^{\mathrm{ss}}$, while we verified that the convergence time is comparable to the one observed for the static policy.

\begin{figure}
\centering
\begin{minipage}{0.5\textwidth}
    \includegraphics[width=\textwidth]{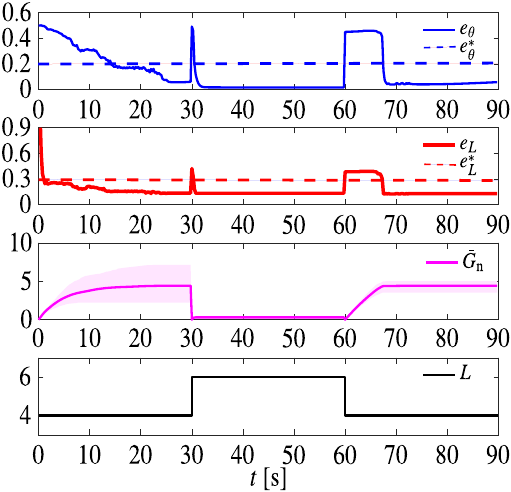}
    \caption{Flexibility tests using the adaptive tuning law \eqref{eq:adaptationLaw}.
    Initial conditions are the same as those of the simulation in Figure \ref{fig:dynlattices}.
    The shaded magenta area is delimited by $\max_{i \in \mathcal{S}} G_{\mathrm{n}, i}$ and $\min_{i \in \mathcal{S}} G_{\mathrm{n}, i}$. 
    Here $L=4$, and $G_{\text{r}}=15$, $R_{\text{s}} = 3\, \text{m}$.
    }
    \label{fig:adaptive_flexibility}
\end{minipage}
\begin{minipage}{0.49\textwidth}
    \centering
    \includegraphics[width=0.98\textwidth]{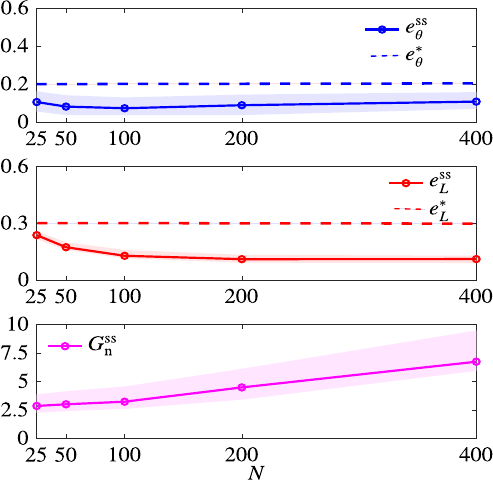}
    \caption{Scalability test using the adaptive tuning law \eqref{eq:adaptationLaw}. 
    $e_\theta^\mathrm{ss}$ and $e_L^\mathrm{ss}$ are averaged over $M=30$ trials with varying $N$.
    $R_{\text{s}} = 3\, \text{m}$; agents' initial positions are drawn with uniform distribution from a disk with radius $r=\sqrt{N/25}$.
    $G_{\text{n}}^{\text{ss}} \coloneqq \bar{G}_{\text{n}}(t_\mathrm{ss})$.
    The shaded areas represent the maximum and minimum values over the trials.
    $L=4$, $G_{\text{r}}=15$.
    }
    \label{fig:adaptive_scalability}
\end{minipage}
\end{figure}

\section{Experimental validation}
\label{sec:robotarium}

To further validate our control algorithm, we tested it in a real robotic scenario, using the open access \emph{Robotarium} platform; see \cite{Pickem2017, Mayya2020} for further details. 
The experimental setup features 20 differential drive robots (GRITSBot \cite{Pickem2015}), that can move in a 3.2 m $\times$ 2 m rectangular arena. The robots have a diameter of about 11 cm, a maximum (linear) speed of 20 cm/s, and a maximum rotational speed of about 3.6 rad/s.
To cope with the limited size of the arena, distances $\left\Vert \vec{r}_{ij} \right\Vert$ in \eqref{eq:radial_Lennard-Jones} are doubled, while control inputs $\vec{u}_i$ are halved.
The Robotarium implementation includes a collision avoidance safety protocol and transforms the velocity inputs \eqref{eq:controlLaw} into appropriate acceleration control inputs for the robots.
Moreover, we run an initial routine to yield an initial condition in which the agents are aggregated as much as possible at the centre, similarly to what considered in Section \ref{sec:results}.




As a paradigmatic example, we performed a flexibility test (similarly to what done in Sec \ref{subsec:DynLatt} and reported in Figure \ref{fig:dynlattices}).
During the first $33$ seconds, the agents reach an aggregated initial condition.
Then we set $L=4$ for $t \in [33, 165)$, $L=6$ for $t \in [165, 297)$, and $L=4$ for $t \in [297, 429]$, ending the simulation.
We used the static control law \eqref{eq:controlLaw}-\eqref{eq:radInput} and \eqref{eq:normalInput}, and to comply with the limited size of the arena, we scaled the control gains to the values $G_{\text{r}}=0.8$ and $G_{\text{n}}=0.4$, selected empirically.

The resulting movie is available online at {\small\url{https://github.com/diBernardoGroup/SwarmSimPublic/tree/SwarmSimV1/Media}}, while representative snapshots are reported in Figure \ref{fig:Robotarium}, with the time evolution of the metrics.
The metrics qualitatively reproduce the behaviour obtained in simulation (see Figure \ref{fig:dynlattices}).
In particular, we obtain $e_{\theta}^\mathrm{ss} < e_{\theta}^*$, with both triangular and square patterns. 
On the other hand, we obtain $e_{L}^\mathrm{ss} < e_{L}^*$
when forming square patterns, but $e_{L}^\mathrm{ss} > e_{L}^*$ with triangular patterns; this does not mean that the pattern is not achieved, as it can be seen in Figure \ref{subfig:robotarium_snap_3} showing the pattern is successfully achieved.
This minor  performance degradation is due to (i) the reduced number of agents, (ii) unmodelled dynamics of the differential-drive robots such as non-holonomic constraints and finite acceleration, and (iii) additional constraints such as the finite size of the arena and the size of the robots themselves.

The overall outcome is that, with a relative small effort, our solution can be implemented in real robots and is capable to overcome the simulation to reality gap.

\begin{figure}[t]
    \begin{minipage}{0.55\textwidth}
    \centering
    \begin{subfigure}[t]{0.49\textwidth}
        \includegraphics[width=1\textwidth]{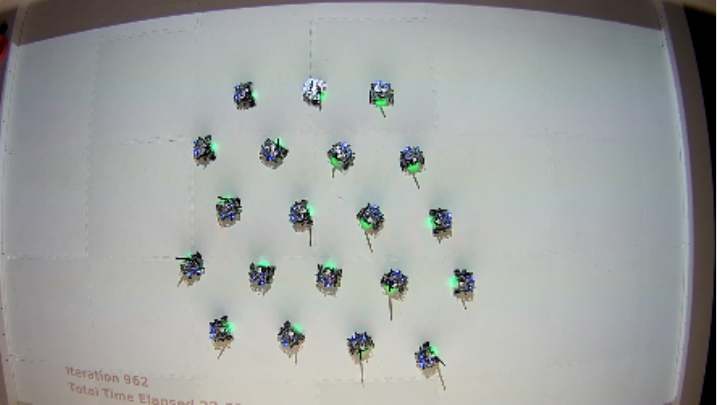}
        \caption{$t=33\,\text{s}$}
    \end{subfigure}
    \begin{subfigure}[t]{0.49\textwidth}
        \includegraphics[width=1\textwidth]{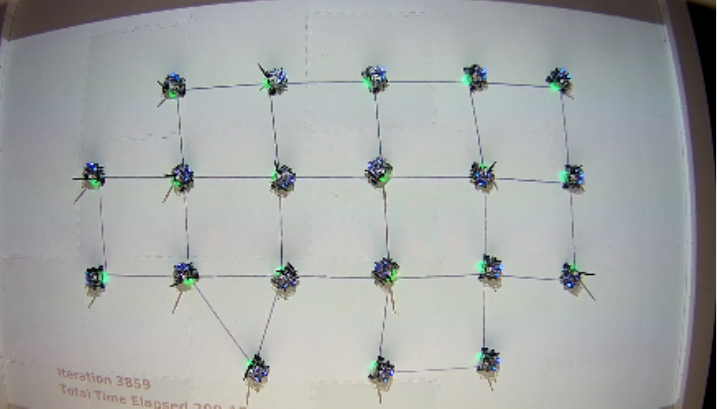}
        \caption{$t=127\,\text{s}$}
    \end{subfigure}
    
    \begin{subfigure}[t]{0.49\textwidth}
        \includegraphics[width=1\textwidth]{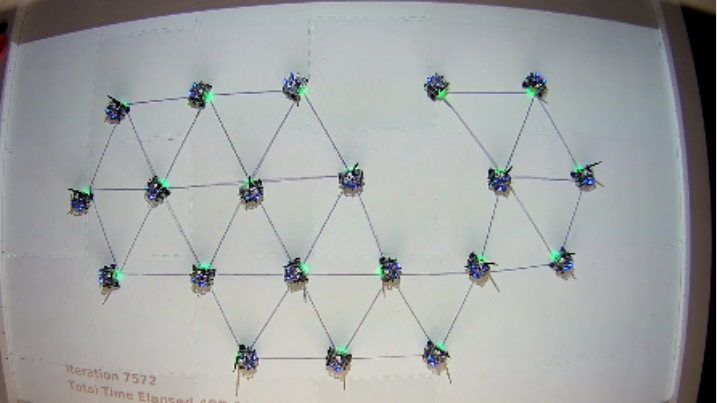}
        \caption{$t=250\,\text{s}$}
        \label{subfig:robotarium_snap_3}
    \end{subfigure}
    \begin{subfigure}[t]{0.49\textwidth}
        \includegraphics[width=1\textwidth]{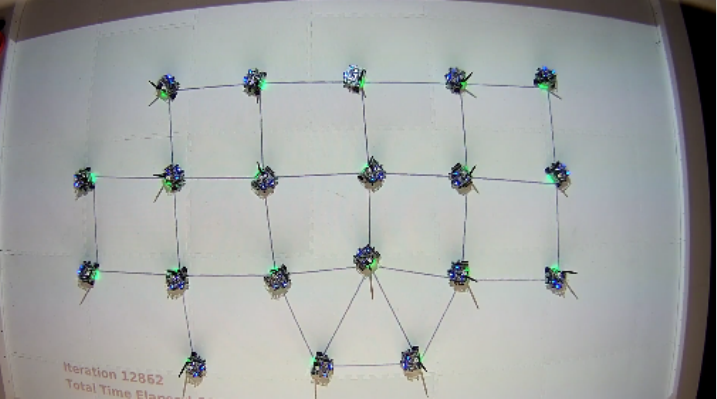}
        \caption{$t=425\,\text{s}$}
    \end{subfigure}
    \end{minipage}
    \begin{minipage}{0.45\textwidth}
    \begin{subfigure}[t]{1\textwidth}
        \includegraphics[width=1\textwidth]{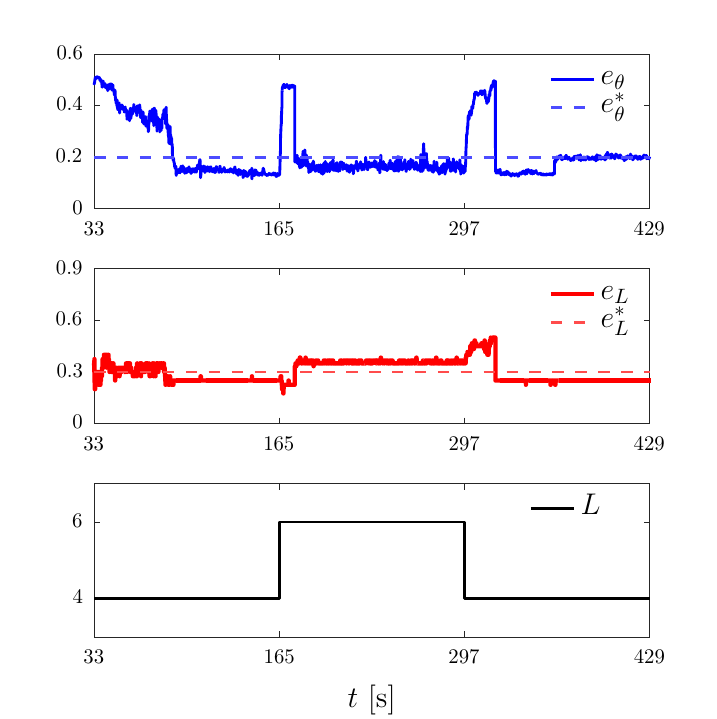}
        \caption{}
    \end{subfigure}
    \end{minipage}
    \caption{Flexibility test in the Robotarium. (a--d) swarm at different time instants.
    (e) time evolution of the metrics and the parameter $L$. The gains are set as
    $(G_{\text{r}},G_{\text{n}}) = (G_{\text{r}}^*,G_{\text{n}}^*)_{L=4}$.
    During the first 33 seconds, the agents reach an aggregated initial condition.
    }
    \label{fig:Robotarium}
\end{figure}

\section{Discussion}
\label{sec:dist_cont_discussion}

In this Chapter we presented a middle level control algorithm (see Section \ref{sec:background_control}) to solve problems of geometric pattern formation, specifically our solution uses virtual forces to form square and triangular lattices.
The proposed solution is distributed and displacement-based, i.e. only requires distance sensors and a compass, and does not need communication between the agents.
We showed via exhaustive simulations and experiments that the strategy is effective in achieving both triangular and square lattices. As a benchmark, we also compared it with the well established distance-based strategy in \cite{Spears2004}, observing better performance particularly when the goal is that of achieving square lattices.
Additionally, we showed that the control law is robust to failures of the agents and to noise, it is flexible to changes in the desired lattice and scalable with respect to the number of agents.
We also presented a simple yet effective gain adaptation law to automatically tune the gains so as to be able to switch the goal pattern in real-time.

On the other hand, a complete proof of the convergence for this control algorithm is still lacking, mostly because of the difficulties generated by the interplay between the two virtual forces (i.e. radial and normal).
Nevertheless, in the simplified case where only the radial force is present, such proof was indeed obtained and will be discussed in the next Chapter.

\chapter{Convergence toward geometric patterns}
\thispagestyle{empty} 
\label{ch:convergence}

As introduced in Chapter \ref{ch:patt_form_background}, proving formally the stability of distributed control algorithms for geometric pattern formation is a challenging yet crucial task.
In this Chapter, which follows the contents and the structure of our work \cite{Giusti2023LCSS}, we 
complement the numerical and experimental results presented in Chapter \ref{ch:dist_cont}, with a novel analytical result. 
Specifically, we revisit the problem of geometric pattern formation using \emph{attraction/repulsion} virtual forces with the aim of bridging a gap in the existing literature and deriving a general proof of convergence when considering the formation of \emph{rigid lattice}  configurations in multi-dimensional spaces.
When compared to previous work, our stability results (i) can be applied to most control laws based on virtual forces (or potentials), rather than only to a specific algorithm \cite{Lee2008}, (ii) are sufficient rather than necessary conditions, as, e.g., in \cite{Hinsen2005}, (iii) characterize the asymptotic configuration of the agents, rather than just proving its  boundedness \cite{Gazi2002}, and (iv) guarantee the emergence of rigid lattices rather than less regular ones, e.g., the $\alpha$-lattices studied in \cite{Olfati-Saber2006}, which allow for disconnected graphs and the coexistence of heterogeneous patterns (e.g., triangular and square).

\section{Problem formulation}
\label{sec:problem_statement}

Let us formulate the problem by recalling the concepts of \emph{swarm}, a set of $N$ mobile agents as introduced in Section \ref{subsec:swarm_definition}, and of \emph{adjacency set} (Definition \ref{def:adjacency_set}).
Then set $R_{\min}=0$ and the \emph{maximum link length} $R_{\max}=R_{\text{a}}$, so that two agents are connected by a link (see Definition \ref{def:links}) if and only if their distance is at most $R_\R{a}$; see Figure \ref{fig:R_a_R}.

Now let us recall the \emph{infinitesimal rigidity} of \emph{frameworks} (Definitions \ref{def:framework} and \ref{def:inf_rigidity}), to introduce the central topic of this Chapter, i.e. the \emph{rigid lattice}.

\begin{definition}[Rigid lattice]
\label{def:triangular_lattice}
Given a swarm with framework $\C{F}(\bar{\vec{x}}^*)$, we call $\bar{\vec{x}}^*$ a \emph{rigid lattice configuration} if
\begin{enumerate}[label=(\Alph*)]
    \item\label{condition:rigidity}
    $\C{F}(\bar{\vec{x}}^*)$
    is infinitesimally rigid, and
    \item\label{condition:link_length} 
    $\norm{\vec{r}_{ij}}=R,\ \forall (i,j)\in \C{E}(\bar{\vec{x}}^*)$,
\end{enumerate}
where $R \in \BB{R}_{>0}$ denotes the \emph{desired link length}.
\end{definition}

Figs.~\ref{fig:rigid_lattice_2d}, \ref{fig:rigid_lattice_3d} portray examples of rigid lattices, for $d=2$ and $d=3$: a tessellation of triangles, and one of tetrahedra and octahedra, respectively.
It is immediate to verify that rigid lattices  are characterized by connected graphs where each agent has at least $d$ links, yielding robustness to link failure. 
A similar structure is the $\alpha$-lattice from \cite{Olfati-Saber2006}, which requires \ref{condition:link_length} but not \ref{condition:rigidity} (hence, a rigid lattice is an $\alpha$-lattice, but the converse is false).
Thus, $\alpha$-lattices can display more heterogeneous structures, containing different polytopes (e.g., squares, cubes), or even be disconnected, which can be unsuited for applications such as region coverage or distributed sensing.
Note however that \emph{vacancies}, i.e. holes in the lattice, can be present in both rigid and $\alpha$-lattices.

In a rigid lattice, we denote by $R_{\R{next}}$ the minimum distance between two not directly connected agents (e.g., $R_{\R{next}}=R\sqrt{3}$ if $d=2$ and $R_{\R{next}}=R\sqrt{2}$ if $d=3$).
Here, we assume that
    $R_\R{a}\in \, ] R; R_{\R{next}} [$,
so that, when the swarm is in a rigid lattice configuration, the adjacency set (Definition~\ref{def:adjacency_set}) of any agent includes only the agents in its immediate surroundings, and all the links (Definition~\ref{def:links}) have length $R$ (see Figure \ref{fig:rigid_lattices_examples}).
Moreover, defining $\C{T} \subset \BB{R}^{dN}$ as the set of all rigid lattice configurations; it is immediate to verify that $\C{T}$ is unbounded and disconnected.  
\begin{figure}[t]
    \centering
    \subfloat[]{
    \includegraphics[trim=4mm 4mm 4mm 4mm, clip, width=0.28\columnwidth]{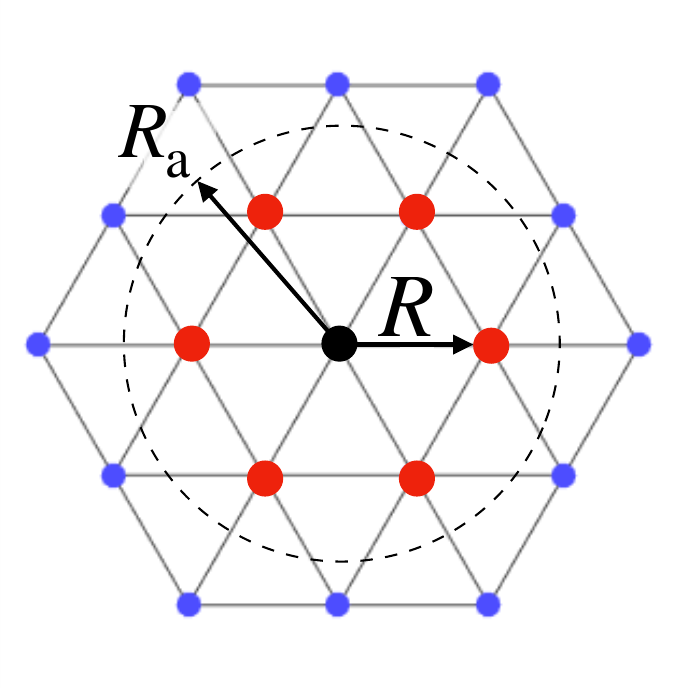}\label{fig:R_a_R}}
    \subfloat[]{
    \includegraphics[trim=4mm 4mm 4mm 4mm, clip, width=0.28\columnwidth]{figs/N=100_sparse_rep2_small.pdf}\label{fig:rigid_lattice_2d}}
    \subfloat[]{
    \includegraphics[trim=2mm 2mm 2mm 0mm, clip, width=0.4\columnwidth]{figs/lattice_3D_N=8.pdf}\label{fig:rigid_lattice_3d}}
    \caption{
    (a) Adjacency set (red) of an agent (black).
    (b) A rigid lattice with $d=2$ and $N=100$.
    (c) A rigid lattice with $d=3$ and $N=8$.}
    \label{fig:rigid_lattices_examples}
\end{figure} 
%
Also, note that any configuration congruent to a rigid lattice (according to Definition \ref{def:congruent_conf}) is a rigid lattice itself, formally $\bar{\vec{x}}^* \in \C{T} \Leftrightarrow \Gamma(\bar{\vec{x}}^*) \subset \C{T}$ (see Figure \ref{fig:rigid_lattices_sets}), and
\begin{equation}
    \label{eq:TasUnion}
    \C{T}=\bigcup_{\bar{\vec{x}}^* \in \C{T}} \Gamma(\bar{\vec{x}}^*).
\end{equation}

\begin{figure}[t]
    \centering
    \subfloat[]{
    \includegraphics[width=0.3\columnwidth]{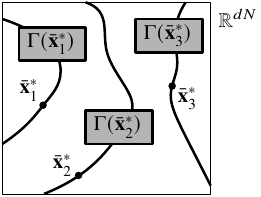}
    \label{fig:rigid_lattices_sets}}
    \subfloat[]{
    \includegraphics[width=0.3\columnwidth]{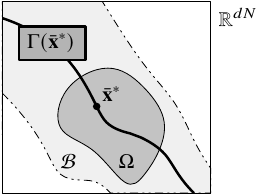}
    \label{fig:all_sets}}
    \caption{(a) Sets of rigid lattices configurations. (b) Sets used in the proof of Theorem \ref{th:local_stability_lyap}.}
\end{figure}



Now, consider a swarm $\C{S}$ and assume the agents' dynamics is described by
\begin{equation}\label{eq:model}
    \dot{\vec{x}}_i(t) = \vec{u}_i(t), \ \ \forall i \in \C{S},
\end{equation}
where $\vec{x}_i(t) \in \mathbb{R}^d$ is the position of agent $i$, and $\vec{u}_i(t)\in \mathbb{R}^d$ its velocity input.
We aim to select and validate a distributed control law to compute this input and let the swarm achieve a rigid lattice configuration.
%
Therefore, let us recall the interaction set $\C{I}_i(t)$ from Definition \ref{def:interaction_set} and
assume that $R_{\R{s}}\geq R_\R{a}$, so that
\begin{equation}\label{eq:A_subset_I}
    \C{A}_i \subseteq \C{I}_i, \quad \forall i\in \C{S}.
\end{equation}
We can now select $\vec{u}_i(t)$ in \eqref{eq:model} as the distributed \emph{virtual forces} control law
\begin{equation}\label{eq:control_law}
    \vec{u}_i(t) \coloneqq \sum_{j \in \C{I}_i(t)} f\left( \norm{\vec{r}_{ij}(t)} \right)\, \unitvec{r}_{ij}(t),
\end{equation}
where $f : \mathbb{R}_{>0} \rightarrow \mathbb{R}$ is the \emph{interaction function}.

Let us introduce our first result, which slightly extends {\cite[Lemma 1]{Gazi2002}}, and will be used in the next Section.
\begin{lemma}\label{th:invariant_center}
The position of the center of the swarm, 
say $\vec{x}_{\R{c}}$, under the control law \eqref{eq:control_law} is invariant, that is 
$
    \dot{\vec{x}}_{\R{c}} = \vec{0}\ \forall \bar{\vec{x}} \in \BB{R}^{dN}.
$
\end{lemma}

\begin{proof}
Exploiting \eqref{eq:model} and \eqref{eq:control_law}, the dynamics of the center of the swarm is given by
    $\dot{\vec{x}}_{\R{c}} \coloneqq \frac{1}{N} \sum_{i = 1}^N \dot{\vec{x}}_i = \frac{1}{N} \sum_{i = 1}^N {\vec{u}}_i =
    \frac{1}{N} \sum_{i = 1}^N \sum_{j \in \C{I}_i} f(\norm{\vec{r}_{ij}})\, \unitvec{r}_{ij}$.
Since the existence of any link $(i,j)$ implies the existence of link $(j,i)$ (see Definition~\ref{def:adjacency_set}), for any term $f(\norm{\vec{r}_{ij}}) \, \unitvec{r}_{ij}$ there exists a term $f(\norm{\vec{r}_{ji}}) \, \unitvec{r}_{ji}=-f(\norm{\vec{r}_{ij}}) \, \unitvec{r}_{ij}$ (because $\norm{\vec{r}_{ij}} = \norm{\vec{r}_{ji}}$ and $\unitvec{r}_{ij} = - \unitvec{r}_{ji}$).
Therefore, the sum of the two is zero, yielding the thesis.
\end{proof}

\section{Convergence to a rigid lattice configuration}
\label{sec:main_results}
We can now state our main result, i.e., that, given an interaction function $f$ (in \eqref{eq:control_law}) generating short range repulsion and long range attraction, the set of rigid lattice configurations is locally asymptotically stable (see Definition \ref{def:LAS}).
For this result to hold we need the following assumption on the interaction function.

\begin{assumption}\label{ass:interaction_function}
    $f$ (in \eqref{eq:control_law}) is such that:
\begin{enumerate}[label=(a\arabic*)]
    \item \label{hp:null_point} $f(R)=0$,
    \item \label{hp:attraction_repulsion} $f(z) > 0$ for $z \in ]0;R [$ and $f(z) < 0$ for $z\in ]R;R_{\R{a}} [$, 
    \item \label{hp:integrable} $f(z)$ is continuous in $]0; R_\R{a}]$,
    \item \label{hp:vanishing} $f(z) = 0$ for any $z > R_\R{a}$.
    \end{enumerate}
\end{assumption}

An exemplary interaction function fulfilling the assumption above is portrayed in Figure \ref{fig:interaction_function_and_potential}.

Without loss of generality, we further assume that, under Assumption \ref{ass:interaction_function}, in a sufficiently small neighborhood of a rigid lattice configuration, all other equilibria are also rigid lattice configurations (supporting evidence showing that this assumption is not restrictive is reported in Section \ref{sec:rigid_latt_complementary_results}).

\begin{figure}[t]
    \centering
    \includegraphics[width=0.35\columnwidth]{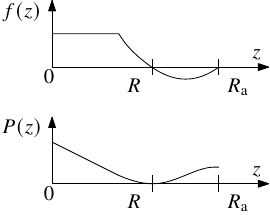}
    \caption{
    An interaction function $f$ satisfying Assumption \ref{ass:interaction_function} and its potential $P$.
    }
    \label{fig:interaction_function_and_potential}
\end{figure}


\begin{theorem}\label{th:local_stability_lyap}[Stability of rigid lattices]
Let Assumption \ref{ass:interaction_function} hold.
Then, for any rigid lattice configuration $\bar{\vec{x}}^*$, 
$\Gamma(\bar{\vec{x}}^*)$ is a locally asymptotically stable equilibrium set.
Consequently, $\C{T}$ is also a locally asymptotically stable equilibrium set.
\end{theorem}
\begin{proof}
Let us consider \emph{any} rigid lattice configuration $\bar{\vec{x}}^* \in \C{T}$, 
with center $\vec{x}^*_{\R{c}} \coloneqq \frac{1}{N} \sum_{i = 1}^N \vec{x}_i^*$ and relative positions $\B{r}_{ij}^*$,
and the set $\Gamma(\bar{\vec{x}}^*)$ of its congruent configurations.
Recalling Definition~\ref{def:triangular_lattice}.\ref{condition:link_length} and \ref{hp:null_point}, we have that $\bar{\vec{x}}^*$ is an equilibrium point of \eqref{eq:model}--\eqref{eq:control_law}; thus, $\Gamma(\bar{\vec{x}}^*)$ and $\C{T}$ are equilibrium sets, according to Definition \ref{def:equilibrium_set}.
Next, we will prove local asymptotic stability of $\Gamma(\bar{\vec{x}}^*)\subset \C{T}$, which implies local asymptotic stability of $\C{T}$ through \eqref{eq:TasUnion}.

\paragraph*{Step 1 (Lyapunov function)}

Given a configuration $\bar{\vec{x}} \in \BB{R}^{dN}$ with center $\vec{x}_{\R{c}}$ and inducing the links in $\C{E}(\bar{\vec{x}})$ according to Definition~\ref{def:links}, let $m \coloneqq \lvert \C{E}(\bar{\vec{x}}) \rvert$ and order the links in $\C{E}(\bar{\vec{x}})$ arbitrarily, so that $\vec{r}_1, \dots, \vec{r}_m$ refer to the relative positions $\vec{r}_{ij}$ for $(i,j)\in \C{E}(\bar{\vec{x}})$.
Recalling \ref{hp:integrable}, we can define the potential function
$P:\, ]0, R_\R{a}]\to \BB{R}$ given by $P(z)=- \int_{R}^{z} f(y) \, \R{d}y$ (see Figure \ref{fig:interaction_function_and_potential}). 
Note that $P(R)=0$, $\frac{\R{d}P}{\R{d} z}(z)= -f(z)$, and, from \ref{hp:attraction_repulsion},
\begin{equation}\label{eq:P_positive}
    P(z) > 0 \quad \forall z \in \BB{R}_{>0} \setminus \{R\} .    
\end{equation}

Then, let us consider the candidate Lyapunov function
\begin{equation}\label{eq:V+cm}
    V(\bar{\vec{x}}) \coloneqq \norm{\vec{x}^*_\R{c} - \vec{x}_\R{c} }^2 + \sum_{k\in\C{E}(\bar{\vec{x}})} P(\norm{\vec{r}_{k}}).
\end{equation}
By \eqref{eq:P_positive}, it holds that
    $V(\bar{\vec{x}})\geq0 \ \forall \bar{\vec{x}} \in \BB{R}^{dN}$,
and $V = 0$ if and only if both $\vec{x}_{\R{c}} =\vec{x}_{\R{c}}^*$ and Definition~\ref{def:triangular_lattice}.\ref{condition:link_length} holds%
. 

\paragraph*{Step 2 (Properties of $V$)}

$V(\bar{\vec{x}})$ is discontinuous over $\BB{R}^{dN}$ (because $\C{E}(\bar{\vec{x}})$ changes when links (dis-)appear).
However, $V(\bar{\vec{x}})$ is continuous and differentiable in any subset of $\BB{R}^{dN}$ where the set $\C{E}(\bar{\vec{x}})$ of links is constant.
To find such a set, we seek conditions on $\vec{\bar{x}}$ such that $\C{E}(\vec{\bar{x}}) = \C{E}(\vec{\bar{x}}^*)$ (see Definitions \ref{def:adjacency_set} and \ref{def:links}), i.e.,
\begin{subequations}
\begin{align}
    \norm{\vec{r}_{ij}} &< R_\R{a}, \quad \forall (i, j) \in \C{E}(\vec{\bar{x}}^*),\label{eq:links_preserved}\\
    \norm{\vec{r}_{ij}} &> R_\R{a}, \quad \forall (i, j) \not\in \C{E}(\vec{\bar{x}}^*).\label{eq:no_links_created}
\end{align}
\end{subequations}
\eqref{eq:links_preserved} means that all links in $\C{E}(\vec{\bar{x}}^*)$ are preserved in $\C{E}(\vec{\bar{x}})$, while \eqref{eq:no_links_created} means that no new links are created in $\C{E}(\vec{\bar{x}})$ with respect to $\C{E}(\vec{\bar{x}}^*)$.
With simple algebraic manipulations it is possible to show that  \eqref{eq:links_preserved} and \eqref{eq:no_links_created} hold if $\bar{\vec{x}} \in \C{B}$, where
\begin{equation}
    \label{eq:setB}
    \C{B} \coloneqq \{ \bar{\vec{x}} \in \mathbb{R}^{dN} : \abs{ \norm{\vec{r}_{ij}} - \norm{\vec{r}_{ij}^*} } < \beta, \ \forall i,j \in \C{S} \},
\end{equation}
and $\beta < \min_{i,j\in \C{S}}  \abs{ R_\R{a} - \norm{\vec{r}_{ij}^*} }$; $\C{B}$ can be intended as a ``neighborhood'' of $\Gamma(\bar{\vec{x}}^*)$ with ``width'' $\beta$ (see Figure \ref{fig:all_sets}). 
Thus, $\C{E}(\vec{\bar{x}}) = \C{E}(\vec{\bar{x}}^*)$ in $\C{B}$ and, hence, $V$ is continuously differentiable in $\C{B}$.

\paragraph*{Step 3 (Analysis of $\dot{V}$)}
 
Now, we restrict our analysis to the set $\C{B}$ to study the attractivity of $\Gamma(\bar{\vec{x}}^*)$.
We start by studying the dynamics of the agents.
From \eqref{eq:model}--\eqref{eq:control_law}, we have
    $\dot{\vec{x}}_{i}= \sum_{j\in \C{I}_i} f(\Vert \vec{r}_{ij}\Vert)  \unitvec{r}_{ij}$.
Hypothesis \ref{hp:vanishing} and \eqref{eq:A_subset_I} imply that
    $\sum_{\substack{j\in \C{I}_i}} f(\norm{\vec{r}_{ij}})  \unitvec{r}_{ij} =
    \sum_{\substack{j\in \C{A}_i}} f(\norm{\vec{r}_{ij}})  \unitvec{r}_{ij}$.
Hence, using
the incidence matrix $\vec{B}$ (Definition\,\ref{def:incidence_mat}) of the swarm graph,
we get
\begin{equation} \label{eq:x_i_dot_Incidence}
\begin{aligned} 
    \dot{\vec{x}}_{i}= \sum_{j\in \C{A}_i} f(\Vert \vec{r}_{ij}\Vert)  \unitvec{r}_{ij} = \sum_{k=1}^m [\vec{B}]_{ik} f(\Vert \vec{r}_{k}\Vert) \unitvec{r}_{k}.
\end{aligned}
\end{equation}
Moreover we can write the dynamics of the relative positions along a link $k$ as
$
    \dot{\vec{r}}_{k}= \sum_{i=1}^N [\vec{B}]_{ik} \dot{\vec{x}}_{i}.
$ 
Therefore, exploiting \eqref{eq:V+cm}, Lemma \ref{th:invariant_center}, and \eqref{eq:x_i_dot_Incidence},
we get
\begin{multline}\label{eq:Vdot}
    \dot{V}(\bar{\vec{x}}) = \sum_{k=1}^m \frac{\partial V}{\partial \norm{\vec{r}_k}}\ \frac{\partial \norm{\vec{r}_k}}{\partial \vec{r}_k} \ \dot{\vec{r}}_k 
    = \sum_{k=1}^m P'(\norm{\vec{r}_k}) \ \unitvec{r}_k\T  \, \sum_{i=1}^N [\vec{B}]_{ik} \dot{\vec{x}}_{i} \\
    =-\sum_{i=1}^N  \sum_{k=1}^m f(\norm{\vec{r}_k}) \ [\vec{B}\T]_{ki} \ \unitvec{r}_k\T \dot{\vec{x}}_{i} 
    =-\sum_{i=1}^N \dot{\vec{x}}_i\T \dot{\vec{x}}_i = - \dot{\bar{\vec{x}}}\T \dot{\bar{\vec{x}}} \le 0,
\end{multline}
where we also used that $P' = -f$ and that $\frac{\partial \norm{\vec{r}_k}}{\partial \vec{r}_k} = \unitvec{r}_k\T$.
We can hence conclude that $\dot{V}(\bar{\vec{x}})=0$ if and only if $\dot{\bar{\vec{x}}}=\vec{0}$, i.e., in correspondence of equilibrium configurations.

Choosing $\beta$ in \eqref{eq:setB} small enough, we exclude the presence of equilibrium configurations not belonging to $\Gamma(\bar{\vec{x}}^*)$, and hence 
\begin{equation}
\begin{cases}\label{eq:characterization_V_dot}
    \dot{V}(\bar{\vec{x}})=0, 
    &\text{if }\bar{\vec{x}} \in \Gamma(\bar{\vec{x}}^*),\\
    \dot{V}(\bar{\vec{x}}) < 0,
    &\text{if }\bar{\vec{x}} \in \C{B} \setminus \Gamma(\bar{\vec{x}}^*).
\end{cases}
\end{equation}

\paragraph*{Step 4 (Applying LaSalle's invariance principle)}
To complete the proof, we define a forward invariant neighborhood of $\bar{\vec{x}}^*$ and then apply LaSalle's invariance principle. 
Given some $\omega \in \BB{R}_{>0}$, let $\Omega$ be the largest connected set containing $\bar{\vec{x}}^*$ such that
$ V(\bar{\vec{x}}) \leq \omega \ \forall \bar{\vec{x}} \in \Omega$ (see Figure \ref{fig:all_sets}).
In particular, we select $\omega$ small enough that $\Omega \subseteq \C{B}$.%
\footnote{\label{fn:existence_omega}
Such $\omega$ exists because $\C{B}$ is a ``neighborhood'' of $\Gamma(\bar{\vec{x}}^*)$ (in the sense of \eqref{eq:setB}) and,  by the rigidity of framework $\C{F}(\vec{\bar{x}}^*)$ (Definition~\ref{def:rigidity}), any continuous motion of the vertices that changes the distance between any two vertices also changes the length of at least one link, causing $V$ to increase.}
Since $V(\bar{\vec{x}}) \leq \omega$ and $\dot{V}(\bar{\vec{x}}) \leq 0$ for all $ \bar{\vec{x}} \in \Omega$, then $\Omega$ is forward invariant.
Moreover, $\Omega$ is closed, because $V$ is continuous in $\Omega$, and $\Omega$ is the inverse image of the closed set $[0,\omega]$.
$\Omega$ is also bounded because (i) translations too far from $\bar{\vec{x}}^*$ cause $V$ to increase beyond $\omega$ (see \eqref{eq:V+cm}), and (ii) $\Omega \subseteq \C{B}$ implies that the deformations of the framework are bounded (see \eqref{eq:setB}).

As $\Omega$ is closed, bounded (thus compact) and forward invariant, we can apply LaSalle's invariance principle \cite[Theorem~4.4]{Khalil2002}, and noting that, in $\Omega$, $\dot{V}(\bar{\vec{x}})=0$ if and only if $\bar{\vec{x}} \in \Gamma(\bar{\vec{x}}^*)$ (see \eqref{eq:characterization_V_dot}), we get that all the trajectories starting in $\Omega$ converge to $\Gamma(\bar{\vec{x}}^*) \cap \Omega$.
This and the forward invariance of $\Omega$ imply that $\Gamma(\bar{\vec{x}}^*)$ is locally asymptotically stable (see Definition \ref{def:LAS}), and so is $\C{T}$ because of \eqref{eq:TasUnion}.
\end{proof}

Moreover, this proof opens the door to the following ancillary results.
\begin{proposition}\label{pro:collision_avoidance}[Collision avoidance]
Let $P^0 \coloneqq \lim_{z \searrow 0}P(z)$.
(i) No collisions between agents occur if $P^0 = \infty$.
(ii) In a sufficiently small neighborhood of a rigid lattice configuration, no collisions occur if $\bar{\B{x}}(0)$ is such that $\sum_{k\in\C{E}(\bar{\vec{x}})} P(\norm{\vec{r}_{k}}) < P^0$.
\end{proposition}
\begin{proof}
    When a collision occurs, at least one $\B{r}_k$ becomes zero and thus, from \eqref{eq:P_positive}, $\sum_{k\in\C{E}(\bar{\vec{x}})} P(\norm{\vec{r}_{k}}) \ge P^0$.
    Equations \eqref{eq:V+cm} and \eqref{eq:Vdot} yield the first statement.
    The second statement is obtained by recalling that $\Omega \subseteq \C{B}$ and that $\Omega$ is forward invariant (see Step 4 of the proof of Theorem \ref{th:local_stability_lyap}).
\end{proof}

\begin{remark}\label{rem:tracking}[Path tracking]
    Path tracking can be obtained by adding a velocity term $\B{w}(t)$ on the right hand side of \eqref{eq:model}.
    Theorem \ref{th:local_stability_lyap} still holds, as the analysis can be carried out on new states $\B{y}_i$, with $\B{y}_i(t) = \B{x}_i(t) - \int_{0}^t \B{w}(\tau) \R{d} \tau$ and $\dot{\B{y}}_i = \B{u}_i$.
\end{remark}

Remark \ref{rem:tracking} only aims to show feasibility of path tracking; clearly, more sophisticated strategies can be designed.

\begin{remark}\label{rem:second_order_nonlinear}[Second order dynamics]
    It is possible to show that the results in Theorem \ref{th:local_stability_lyap} also hold in the case of second order nonlinear dynamics, that is $\dot{\B{x}}_i = \B{v}_i$, $\dot{\B{v}}_i = g(\norm{\B{v}_i}) \unitvec{\B{v}}_i + \B{u}_i$, where $\B{x}_i$ and $\B{v}_i$ are the position and velocity of agent $i$, and $g : \BB{R}_{\ge0} \rightarrow \BB{R}_{\le 0}$ is a friction term with $g(z)=0 \Leftrightarrow z=0$ and such that $\B{v}_\R{c}(t) \coloneqq \sum_{i = 1}^N \B{v}_i(t) \rightarrow \B{0}$.
    Namely, the proof of Theorem \ref{th:local_stability_lyap} can be adapted by using the function 
    $
    V = \sum_{k\in\C{E}(\bar{\vec{x}})} P(\norm{\vec{r}_{k}}) + \frac{1}{2} \sum_{i=1}^N \B{v}_{i}\T \B{v}_i$
    in \eqref{eq:V+cm} and exploiting that $\B{x}_\R{c}(t) \coloneqq \sum_{i = 1}^N \B{x}_i(t)$ remains bounded, to apply LaSalle's invariance principle.
\end{remark}


\section{Numerical validation}
\label{sec:validation}

In this section, we validate numerically the result presented in Section \ref{sec:main_results} and estimate the basin of attraction of $\C{T}$.

\subsection{Simulation setup}

We set the number of agents to $N=100$, the desired link length to $R=1$, the sensing radius $R_\R{s}=3$, and the maximum link length to 
$R_\R{a}= (1+R_{\R{next}})/2$ (i.e. $R_\R{a}\approx 1.37$ if $d=2$; $R_\R{a}\approx 1.21$ if $d=3$).
We validate our strategy 
using two interaction functions, depicted in Figure \ref{fig:interaction_function}.
\begin{figure}[t]
    \centering
    \includegraphics[width=0.5\columnwidth]{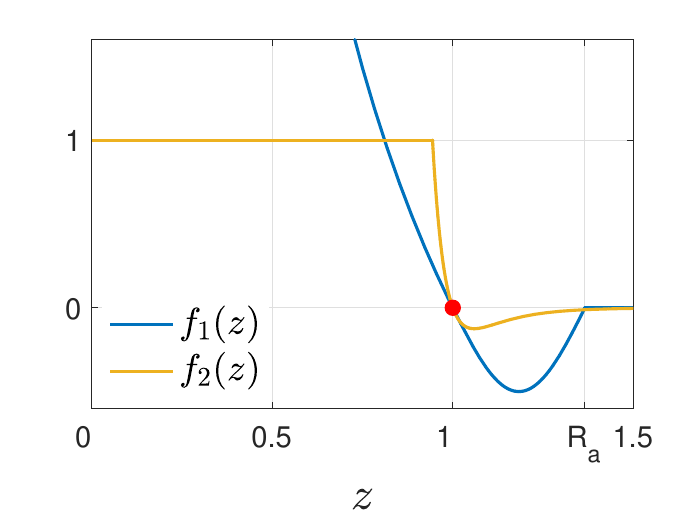}
    \caption{Interaction functions \eqref{eq:power_law} and \eqref{eq:Lennard-Jones} in the case $d=2$.    
    The red dot highlights the zero of the functions in $z=R$.}
    \label{fig:interaction_function}
\end{figure}
The first one is
\begin{equation}\label{eq:power_law}
    f_1(z) = \begin{cases}
        g \left(\frac{1}{z}-\frac{1}{R}\right)\frac{\pi R^2}{R_{\R{a}}-R} &\mbox{if } z\in \ ]0; R],\\
        -g \sin{\left((z-R)\frac{\pi}{R_{\R{a}}-R}\right)} &\mbox{if } z\in \ ]R; R_{\R{a}}],\\
        0 &\mbox{if } z>R_{\R{a}};
    \end{cases}
\end{equation}
with $g = 0.5$.
$f_1$ satisfies Assumption \ref{ass:interaction_function}, is smooth in $]0; R_{\R{a}}[$ and $\lim_{z\searrow 0} f_1(z) = \infty$, which, through Proposition \ref{pro:collision_avoidance}, guarantees the absence of collisions.
The second interaction function $f_2$ is the Physics-inspired Lennard-Jones function \cite{Brambilla2013}, i.e.,
\begin{equation}\label{eq:Lennard-Jones}
    f_2(z) = \min \left\lbrace \left( \frac{a}{z^{2c}}-\frac{b}{z^c}\right), \ 1 \right\rbrace,
\end{equation}
where we select $a = b = 0.5$ and $c=12$ when $d=2$ and $c=24$ when $d=3$; see Figure \ref{fig:interaction_function}.
$f_2$ saturates to $1$ as $z \searrow 0$ to comply with possible actuator saturation.
Moreover, $f_2$ satisfies \ref{hp:null_point}, \ref{hp:attraction_repulsion} and \ref{hp:integrable} in Assumption \ref{ass:interaction_function} exactly, but \ref{hp:vanishing} only approximately.
This is intentional as it allows to account for long range attraction between the agents, which is frequently required in swarm robotics applications \cite{Gazi2002}.
Notice that $f_2$ is equivalent to the radial interaction function \eqref{eq:radial_Lennard-Jones} used in previous Chapter for the formation of triangular and square lattices.

To assess if the swarm is in a rigid lattice configuration, we check conditions \ref{condition:rigidity}, \ref{condition:link_length} in Definition~\ref{def:triangular_lattice}.
%
To evaluate \ref{condition:rigidity} we use Theorem \ref{th:rigidity}.
To evaluate \ref{condition:link_length}, we define the \emph{error} $e(t) \coloneqq \max_{k\in\mathcal{E}(t)} \abs{\norm{\vec{r}_{k}(t)} - R}$, which is zero when \ref{condition:link_length} holds.
Also, as long as $e(t)$ stays strictly lower than $R_\R{a}-R$, links in the configuration of interest are neither created nor destroyed.

For each simulation, the initial positions of the agents are obtained by picking a random rigid lattice configuration and then applying, to each agent, a different random displacement drawn from a uniform distribution over a disk (when $d=2$) or a sphere (when $d=3$), having radius $\delta \in \BB{R}_{\ge 0}$.

All simulation are run in {\sc Matlab} using the agent-based simulator SwarmSim V2 we developed (for more information see Appendix \ref{ch:swarmsim} or visit \url{https://github.com/diBernardoGroup/SwarmSimPublic}).
In particular, agents' dynamics \eqref{eq:model}--\eqref{eq:control_law} are integrated using the forward Euler method with a fixed time step equal to $0.01\, \text{s}$.

\subsection{Numerical results}
To validate Theorem \ref{th:local_stability_lyap}, in Figure \ref{fig:simulation} we report the time evolution of the error $e(t)$ for 10 simulations where the swarm starts from a perturbed rigid lattice configuration.
Simulations are presented for $d \in \{2, 3\}$ and for both interaction functions $f_1$ and $f_2$.
In all cases, infinitesimal rigidity is preserved and $e(t)$ converges to zero, denoting local stability of the lattice.

\begin{figure}[t]
    \centering
    \subfloat[$f_1, \ d=2, \ \delta=0.2$]{\includegraphics[trim={3mm 4mm 4mm 4mm},clip,width=0.49\columnwidth]{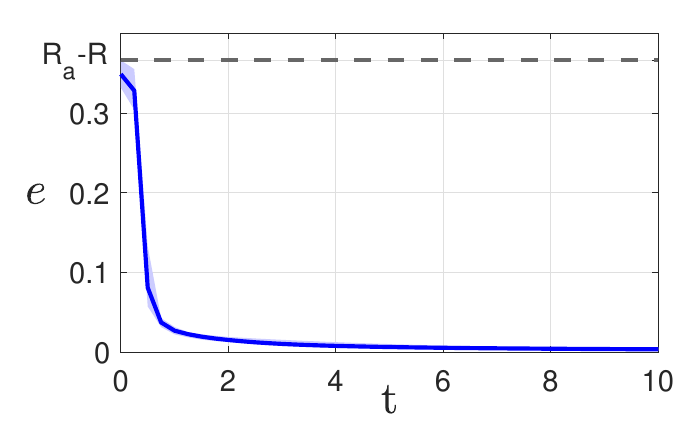}\label{fig:simulation_f_1_d_2}}
    \subfloat[$f_2, \ d=2, \ \delta=0.2$]{\includegraphics[trim={3mm 4mm 4mm 4mm},clip,width=0.49\columnwidth]{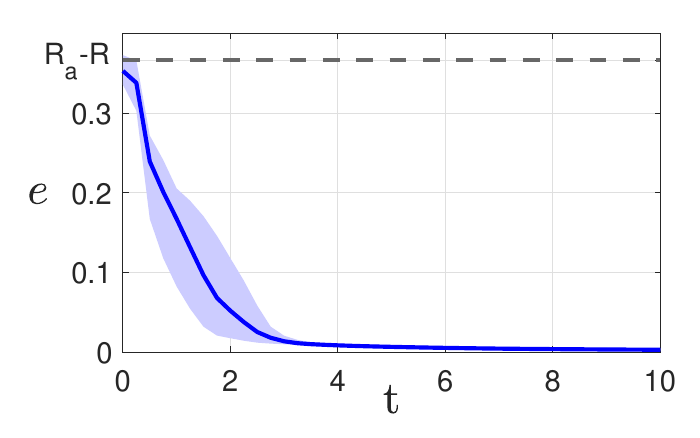}\label{fig:simulation_f_2_d_2}}\\
    \subfloat[$f_1, \ d=3, \ \delta=0.1$]{\includegraphics[trim={3mm 4mm 4mm 4mm},clip,width=0.49\columnwidth]{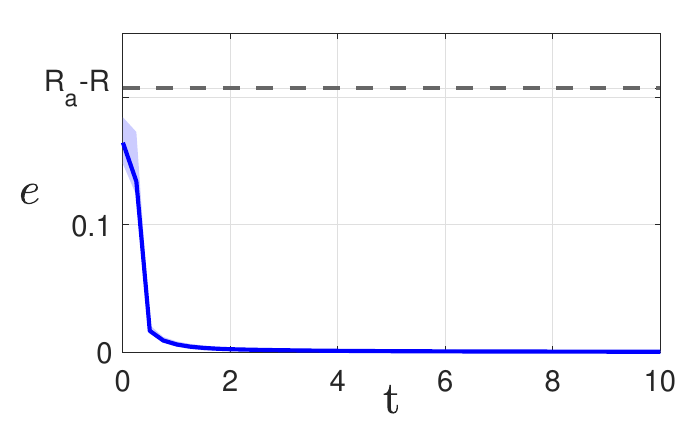}\label{fig:simulation_f_1_d_3}}
    \subfloat[$f_2, \ d=3, \ \delta=0.1$]{\includegraphics[trim={3mm 4mm 4mm 4mm},clip,width=0.49\columnwidth]{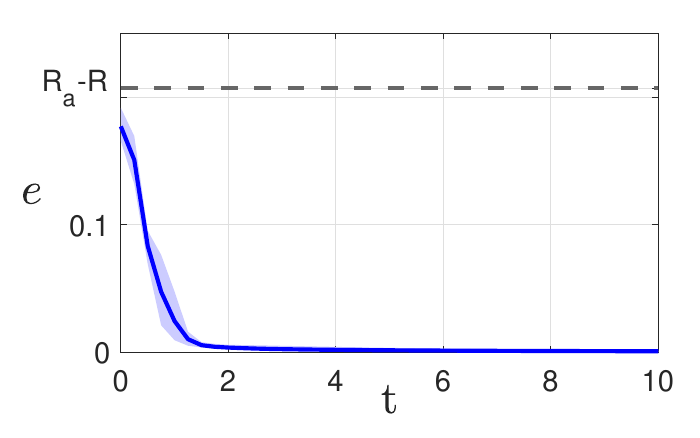}\label{fig:simulation_f_2_d_3}}

    \caption{
    Time evolution of $e$ for various interaction functions and values of $d$. 
    In each panel, 10 simulations with random initial conditions are showed; the solid line is the mean; the shaded area is the minimum and maximum.
    }
    \label{fig:simulation}
\end{figure}

To estimate the basin of attraction of the set of rigid lattice configurations, we performed extensive simulations for various values of $\delta$, and characterize the steady state configurations in Figure \ref{fig:varyingdelta}.
Namely, for $\delta$ smaller than $0.25$ for $d=2$ and $0.2$ for $d=3$ all simulations converge to a rigid lattice configuration. 
Then, as $\delta$ increases,
fewer simulations converge to rigid lattices, until none does.
Note that $e(0)\leq 2\delta$, therefore $\delta = 0.25$ (respectively $\delta = 0.2$) corresponds to a perturbation of up to 50\% (respectively 40\%) of the initial link length, giving an estimation of the basin of attraction 
of $\C{T}$.
%
%

\begin{figure}[t]
    \centering
    \subfloat[Terminal values for $d=2$]{\includegraphics[trim={2mm 2mm 2mm 8mm}, clip, width=0.4\columnwidth]{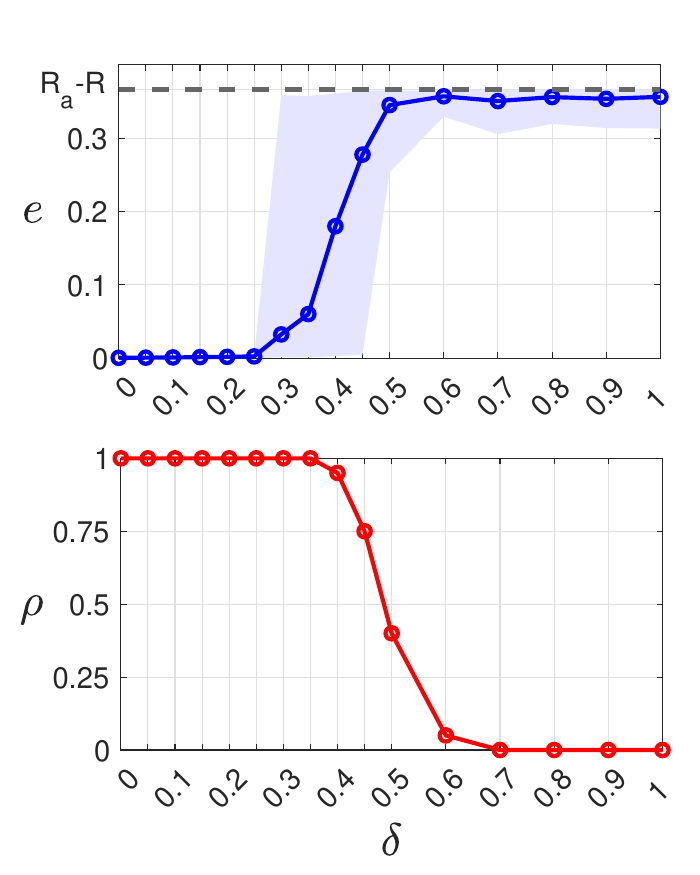} \label{fig:metrics_campaign_2d}}
    \subfloat[Terminal values for $d=3$]{\includegraphics[trim={2mm 2mm 2mm 8mm}, clip, width=0.4\columnwidth]{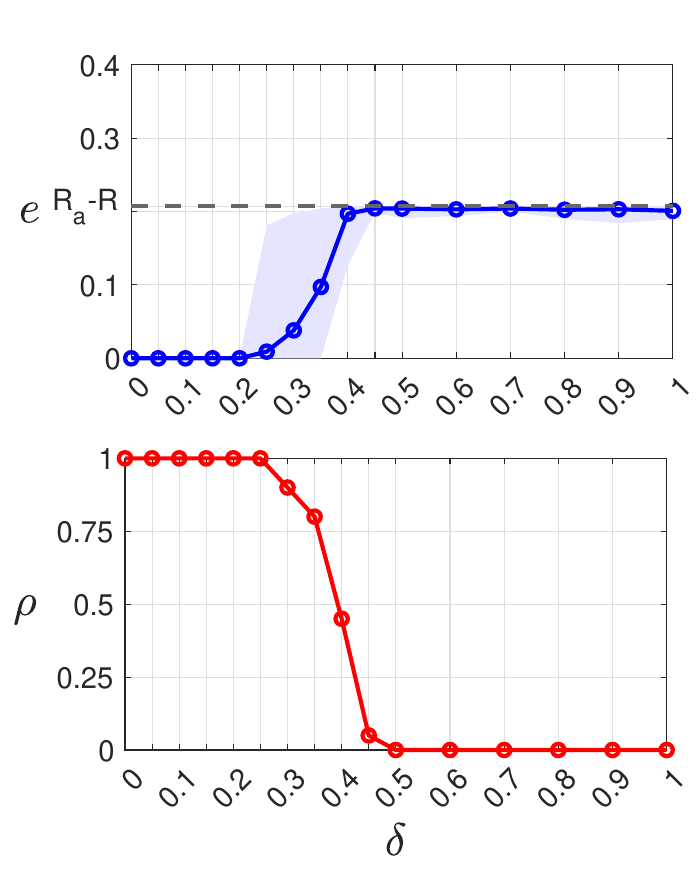} \label{fig:metrics_campaign_3d}}
    \\
    %
    \subfloat[Initial configurations ($d=2$)]{
    \captionsetup[subfigure]{labelformat=empty}
    \subfloat[$\delta=0.2$]{\includegraphics[trim={2mm 10mm 2mm 12mm},clip,width=0.25\columnwidth]{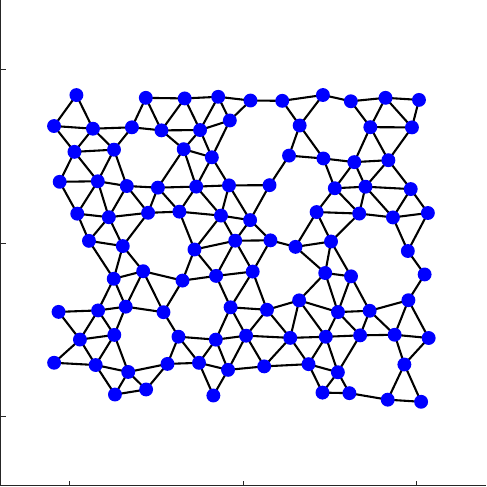}}
    \subfloat[$\delta=0.4$]{\includegraphics[trim={2mm 10mm 2mm 12mm},clip,width=0.25\columnwidth]{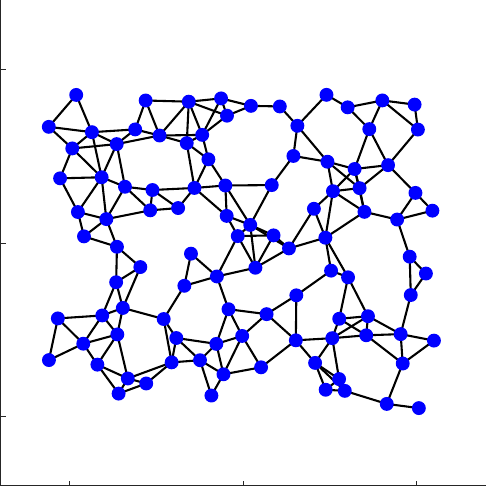}}
    \subfloat[$\delta=0.6$]{\includegraphics[trim={2mm 10mm 2mm 12mm},clip,width=0.25\columnwidth]{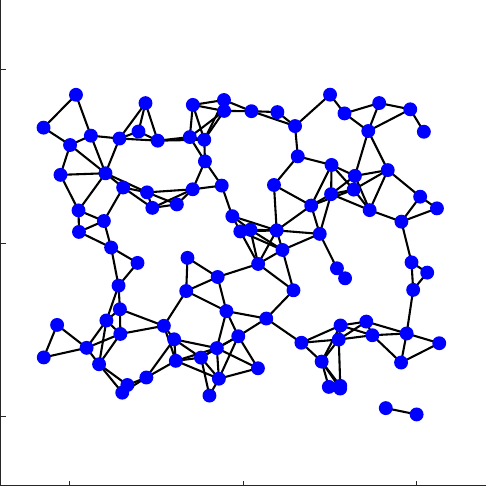}}
    \setcounter{subfigure}{3}}
    \\
    \subfloat[Final configurations ($d=2$)]{
    \captionsetup[subfigure]{labelformat=empty}
    \subfloat[$\delta=0.2$]{\includegraphics[trim={2mm 10mm 2mm 12mm},clip,width=0.25\columnwidth]{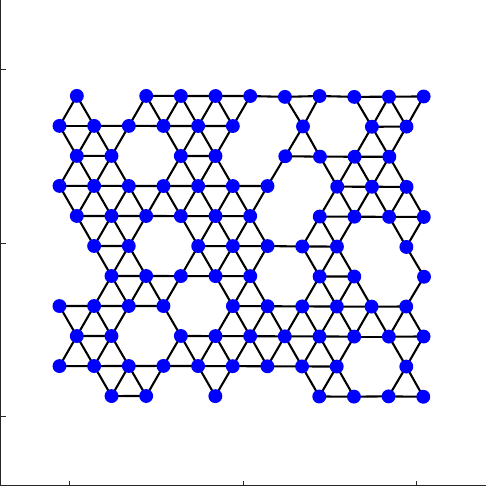}}
    \subfloat[$\delta=0.4$]{\includegraphics[trim={2mm 10mm 2mm 12mm},clip,width=0.25\columnwidth]{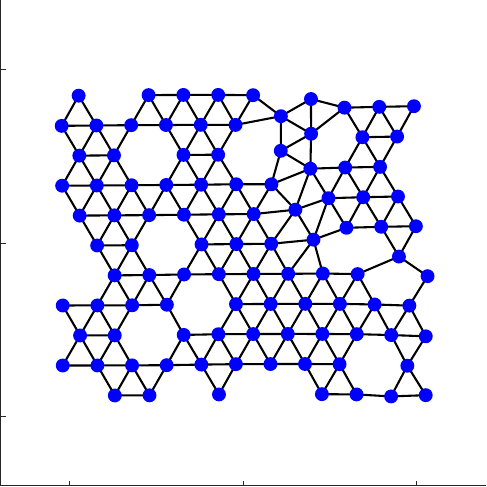}}
    \subfloat[$\delta=0.6$]{\includegraphics[trim={2mm 10mm 2mm 12mm},clip,width=0.25\columnwidth]{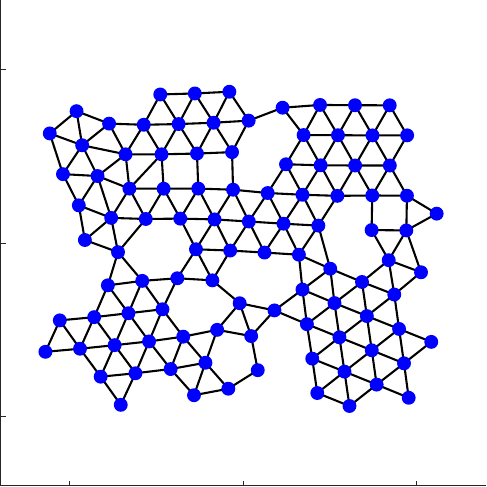}}
    \setcounter{subfigure}{4}}
    %
    \caption{Simulations for different values of $\delta$ and interaction function $f_2$.
    (a), (b): Terminal values of $e$ and $\rho$, respectively for $d=2$ and $d=3$.
    $\rho$ is the fraction of simulations converging to an infinitesimally rigid configuration.
    For $e$, the solid line is the mean; the shaded area is the minimum and maximum. 
    20 simulations with random initial conditions are performed for each value of $\delta$, and last 20\,s.
    (c), (d): Initial and final configurations of representative simulations for specific values of $\delta$ in the case that $d=2$.
    }
    \label{fig:varyingdelta}
\end{figure}


\clearpage
\section{Further results}
\label{sec:rigid_latt_complementary_results}
To further confirm the effectiveness of our theoretical results, we provide below complementary semi-analytical evidence that the set of rigid lattice configurations $\C{T}$ is locally asymptotically stable, which also excludes the presence of other equilibria in an arbitrarily small neighborhood of it.
To do so, we linearize system \eqref{eq:model}--\eqref{eq:control_law} around a rigid lattice configuration, say $\bar{\vec{x}}^*$, obtaining
$
    \label{eq:linearization}
    \dot{\bar{\vec{x}}} \approx \vec{J}(\bar{\vec{x}}^*) \, (\bar{\vec{x}}-\bar{\vec{x}}^*)
$, 
with $\vec{J}(\bar{\vec{x}}^*) \in \BB{R}^{dN \times dN}$ derived as follows.

\paragraph*{Jacobian of \eqref{eq:model}--\eqref{eq:control_law}}
System \eqref{eq:model}--\eqref{eq:control_law} can be recast as
\begin{equation}\label{eq:model_stack}
    \dot{\bar{\vec{x}}} = ((\B{B} \vec{F} \B{G}^{-1} \B{B}\T) \otimes \B{I}_{d}) \bar{\vec{x}}= ((\B{B} \vec{H} \B{B}\T) \otimes \B{I}_{d}) \bar{\vec{x}},
\end{equation}
where $\B{F}, \B{G}, \B{H} \in \BB{R}^{m \times m}$ are diagonal matrices; $[\B{F}]_{ii} \coloneqq f(\norm{\vec{r}_i})$, $[\B{G}]_{ii} \coloneqq \norm{\vec{r}_i}$, and $\vec{H} \coloneqq \B{F} \B{G}^{-1}$. 
The Jacobian of \eqref{eq:model_stack} is 
\begin{equation}
\label{eq:Jacobian}
\begin{aligned}    
    \vec{J}&=\left(
    \vec{B} \frac{\partial \vec{H}}{\partial \bar{\vec{x}}}  \vec{B}\T \otimes \vec{I}_{d} \right) \vec{\bar{x}} +
    (\B{B} \vec{H} \B{B}\T) \otimes \B{I}_{d} \eqqcolon \vec{J}_1 +\vec{J}_2,
\end{aligned}
\end{equation}
where  $\frac{\partial \vec{H}}{\partial \bar{\vec{x}}} \in \BB{R}^{m\times m \times dN}$ is a tensor, and 
$
    \left[\frac{\partial \vec{H}}{\partial \bar{\vec{x}}}  \B{B}\T \right]_{:,:,k} =  \left[ \frac{\partial \vec{H}}{\partial \bar{\vec{x}}} \right]_{:,:,k} \B{B}\T \quad \in \BB{R}^{m\times N}
$, 
with notation $[\ \cdot\ ]_{:,:,k}$ denoting the matrix obtained by fixing the third index of the tensor.
From \ref{hp:null_point}, for all rigid lattice configurations we have $\vec{J}_2=(\B{B} \vec{H} \B{B}\T) \otimes \B{I}_{d} =\vec{0}$.
Then,
$
    \left[ \vec{J}_1 \right]_{:,k} 
    = \left(
    \vec{B} \left[ \frac{\partial \vec{H}}{\partial \bar{\vec{x}}} \right]_{:,:,k}  \vec{B}\T \otimes \vec{I}_{d} \right) \vec{\bar{x}}
$. 
From {\cite[p. 20]{Jackson2007} 
}%
we have 
$\frac{\partial \norm{\vec{r}_i}^2}{\partial [\bar{\vec{x}}]_k}  = 2 [\vec{M}]_{i,k}$ (see Definition \ref{def:rigidity_matrix}),
that is
$\frac{\partial \norm{\vec{r}_i}}{\partial [\bar{\vec{x}}]_k}  = \frac{1}{\norm{\vec{r}_i}} [\vec{M}]_{i,k}$, and thus
\begin{subequations}\label{eq:matrix_H}
\begin{align}
    \left[ \frac{\partial \vec{H}}{\partial \bar{\vec{x}}} \right]_{i,i,k} 
    &= \frac{\partial [f(\norm{\vec{r}_i})/\norm{\vec{r}_i}]}{\partial \norm{\vec{r}_i}} \frac{\partial \norm{\vec{r}_i}}{\partial [\bar{\vec{x}}]_k}\nonumber\\
    &= [f'(\norm{\vec{r}_i}) \norm{\vec{r}_i}  - f(\norm{\vec{r}_i})] \norm{\vec{r}_i}^{-3} [\vec{M}]_{i,k},\\
    \left[ \frac{\partial \vec{H}}{\partial \bar{\vec{x}}} \right]_{i,j,k} &= 0, \quad \text{if} \ i \ne j.
\end{align}
\end{subequations}

\paragraph*{Numerical analysis}
We set $R = 1$ and generated $1520$ random rigid lattice configurations ($10$ per each $N \in \{25, 26, \dots, 100\}$, and each $d\in\{2, 3\}$).
For each of these configurations, assuming $f$ (in \eqref{eq:control_law}) is in the form \eqref{eq:Lennard-Jones}, we computed $\B{J}$ using \eqref{eq:Jacobian}--\eqref{eq:matrix_H} and found that in all cases $\vec{J}$ has $d(d+1)/2$ zero eigenvalues with eigenvectors $\{\vec{w}_i^{0}\}_{i}$, and $dN-d(d+1)/2$ negative eigenvalues with eigenvectors $\{\vec{w}_j^{\pm}\}_j$.
Moreover, $\B{M} \B{w}_i^{0} = \B{0}$ and $\B{M} \B{w}_j^{\pm} \ne \B{0}$; thus, from Definition~\ref{def:inf_rigidity}, the span of $\{\B{w}_i^{0}\}$ corresponds to roto-translations and is a hyperplane locally tangent to $\Gamma(\bar{\vec{x}}^*)$ (see Definition \ref{def:congruent_conf}), while $\{\B{w}_j^{\pm}\}$ correspond to other motions.
Therefore, the \emph{center manifold theorem} \cite[Theorem 5.1]{Kuznetsov2004} yields that $\Gamma(\bar{\vec{x}}^*)$ is a \emph{center manifold} of system \eqref{eq:model}--\eqref{eq:control_law}.
Moreover, as expected from Theorem \ref{th:local_stability_lyap}, the \emph{reduction principle} \cite[Theorem 5.2]{Kuznetsov2004} confirms that the dynamics locally converge onto the 
equilibrium set $\Gamma(\bar{\vec{x}}^*)$, and excludes the presence of other equilibria in an  arbitrarily small neighborhood of it.

\section{Discussion}

We proved analytically local asymptotic stability of rigid lattice configurations, in multi-dimensional spaces, for swarms under the action of a distributed control law based on virtual attraction/repulsion forces.
The theoretical results were supported by exhaustive numerical simulations, providing also an estimate of the basin of attraction, and further semi-analytical derivations.
The mild hypotheses required on the interaction function allow for wide applicability of the theoretical results.
This result addresses an important gap in the Literature on geometric pattern formation, providing a theoretical support for numerous existing solutions.
The analytical characterization of the basin of attraction of $\C{T}$, and the extension of the results to other geometric lattices, such as square, remain open.

In conclusion, Part \ref{part:geometric_pattern} presented our contribution to the control of geometric pattern formation in \ac{LS-MAS}.
In particular, the results presented in Chapters \ref{ch:dist_cont} and \ref{ch:convergence} represent an advancement in both the design and the validation of distributed control algorithms for the formation of geometric patterns. 
The next Part will deal with the study of spatial behaviours of microorganisms, and how these emerge from the movement of the agents and their response to external stimuli.

\part{Spatial behaviours of microorganisms}
\label{part:spatial_microorganisms}

\chapter{Introduction to spatiotemporal control of microorganisms}
\thispagestyle{empty} 
\label{ch:miroorganisms_background}

Understanding how natural \ac{LS-MAS} navigate the environment and organize their spatial distribution is as relevant as implementing these behaviours in artificial systems, as discussed in Part \ref{part:geometric_pattern}.
Controlling the movement and the spatial distribution of microorganisms, cells or artificial micro-agents is a crucial goal to improve our ability to control the microscopic world, with applications ranging from tumor treatment \cite{Hauert2014} to wound healing \cite{Wijewardhane2022}, from soil and water remediation \cite{VanDerMeer2010,Zarzhitsky2005} to population control \cite{Grandel2021}.
All these applications require the ability to control the spatial distribution of the agents (i.e. microorganisms or micro-robots) or the expression of a specific behaviour (e.g.  expression of a certain protein).
Specifically, \emph{spatiotemporal} control aims to regulate some quantity of interest in both space and time, so as, for example, to implement density regulation by prescribing the agents to uniformly distribute in the environment and then, at a specific time, converge towards a certain region.
To influence the behaviour of the agents of interest some actuation input is needed.
Chemical inputs (e.g. antibiotics, sugars or molecules such as IPTG) have been traditionally, and are still today, used as the main control inputs to steer the cell behaviour. 
Given the ubiquity of chemical reactions, the right choice of the input molecules allows to influence virtually any biological organism.
Moreover, chemicals are easily stored and administered by human operators.
On the other hand, chemical inputs, because of diffusion and degradation dynamics, provide poor resolution, both in space and time.
In many applications such tight control is not required, for example the spatial patterns generated by swarming bacteria (e.g. \textit{Proteus mirabilis}) have been used to encode information \cite{Doshi2022}.
Combinations of molecules in different concentrations were used as inputs, encoded in the features of the swarming pattern and then decoded using computer vision.

Instead, when precise actuation is required chemical inputs can be replaced by physical ones, such as light \cite{Fielding2023,Denniss2022,Lam2017} or magnetic fields \cite{Gardi2022}.
Both allow unparalleled temporal resolution, with actuation times in the order of milliseconds. 
Especially light, thanks to the use of \acf{DLP}, provides high resolution in both space and time, together with fine quantization (usually 1/256) and multiplexing capabilities by combining different colors (usually red, green and blue).
Moreover, light can transmit energy and many microorganisms naturally react to it \cite{Nultsch1988}, changing, for example, their movement, while others can be genetically engineered to show light induced behaviours \cite{Lugagne2022}. 
We therefore focused our research on spatiotemporal control of microorganisms using light inputs.

In the following, we will discuss how microorganisms move, and, crucially, how their motion can be modelled.
Then, we will deal with the influence of light on their movement and how this can be exploited to achieve spatiotemporal control of cellular populations.

\section{Movement of microorganisms}
\label{sec:microorganisms_movement}
Microorganisms can propel themselves using two main solutions, either \emph{cilia}, or one or multiple \emph{flagella} \cite{Fawcett1961, Bondoc-Naumovitz2023} (see Figure \ref{fig:ciliate_flagellate}).
Flagella are long filaments that protrude from the posterior part of the body. Despite some relevant differences in structure and composition this solution is widespread across Bacteria (e.g. \textit{E. coli}), Archaea and microscopic Eukaryotes (e.g. \textit{Euglena}).
Cilia, instead, are exclusive of Eukaryotes. These are small filaments that cover the cell membrane and can address various functions. In particular \emph{ciliate} microorganisms (e.g. \textit{Paramecium}) use  cilia for locomotion \cite{Lynn2008}.

\begin{figure}
    \centering
    \begin{subfigure}[t]{0.4\textwidth}
        \centering
        \includegraphics[width=1\textwidth]{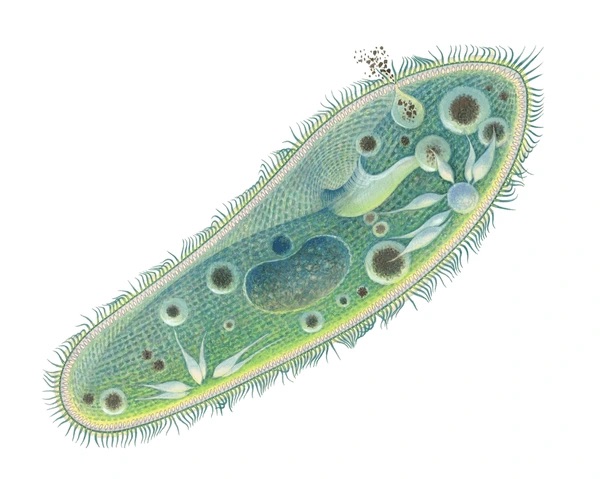}
        \caption{Ciliate}
        \label{subfig:ciliate}
    \end{subfigure}
        \begin{subfigure}[t]{0.3\textwidth}
        \centering
        \includegraphics[width=1\textwidth]{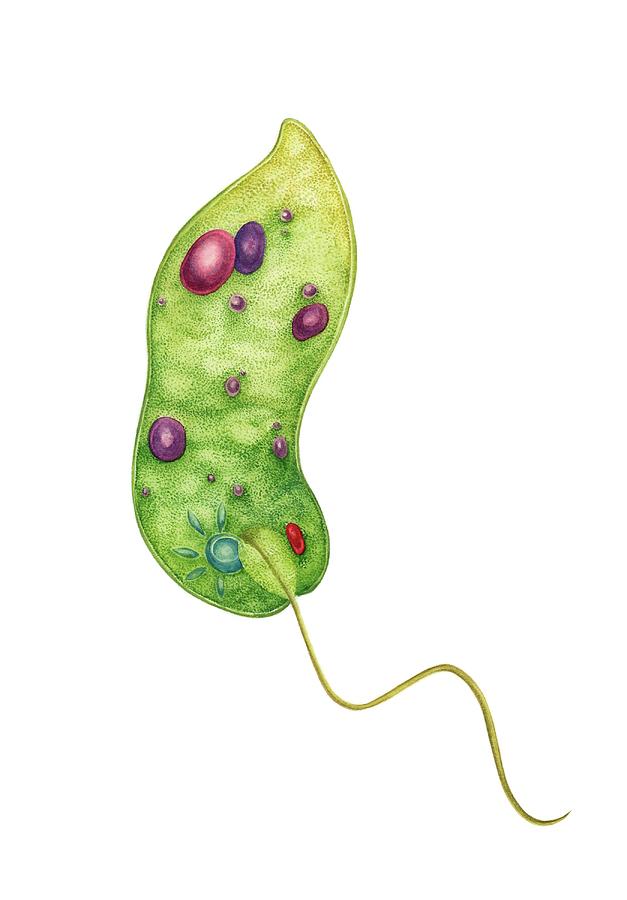}
        \caption{Flagellate}
        \label{subfig:flagellate}
    \end{subfigure}
    \caption{Graphical representation of ciliate and flagellate microorganisms.
    }
    \label{fig:ciliate_flagellate}
\end{figure}

We are now interested in abstracting the details of  locomotion, such as the propulsion mechanism or the biochemistry behind it, to focus on a higher level description, so as to uncover and mathematically describe the kinematic features of the movement and the resulting trajectories.
Such stochastic dynamics cannot be described by classical dynamical models using \acfp{ODE}, therefore stochastic models are introduced to include randomness \cite{Codling2008} (see Chapter \ref{ch:background} for more about modelling multi-agent systems).
%
%

\subsection{Run and tumble dynamics}
The motion of many microorganisms is often described in terms of "\emph{run and tumble}" dynamics \cite{Berg1993}.
This is mostly the case of bacteria, such as \textit{E. coli} \cite{Berg1993}, but also unicellular Eukaryotes, such as the microscopic alga \textit{Chlamydomonas reinhardtii} \cite{Polin2009}.
This behaviour is characterized by almost straight movements (runs) alternated by sudden changes of direction (tumbles).
Runs are usually longer and executed at a relatively high constant speed.
On the contrary, tumbles occur over shorter time and produce a negligible displacement.
A similar behaviour is the run-reverse-flick (adopted for example by some marine bacteria such as \emph{Vibrio alginolyticus}) \cite{Son2016}. In this case only 180 degrees (reverse) and $\sim90$ degrees (flick) rotations are present.

Both run and tumble, and run-reverse-flick behaviours can be mathematically described by Lévy walks \cite{Zaburdaev2015}.
This random motion resembles the biological run and tumble. Indeed, the agent moves with a constant velocity $v$ for a randomly distributed time span $\tau$, then randomly selects a new direction $\theta$ and start moving again with the same speed. 
This can be represented as the hybrid dynamical system \cite{Henzinger1996} depicted in Figure \ref{fig:levy_walk_model}.
In this representation, $(x,y)$ is the position of the agent in a 2D space, and $f_\tau:\mathbb{R}_{\geq0} \to \mathbb{R}_{\geq0}$ and $f_\theta:[-\pi;+\pi] \to \mathbb{R}_{\geq0}$ are the probability density distributions of $\tau$ and $\theta$ respectively.
\begin{figure}
    \centering
    \includegraphics[width=0.7\linewidth]{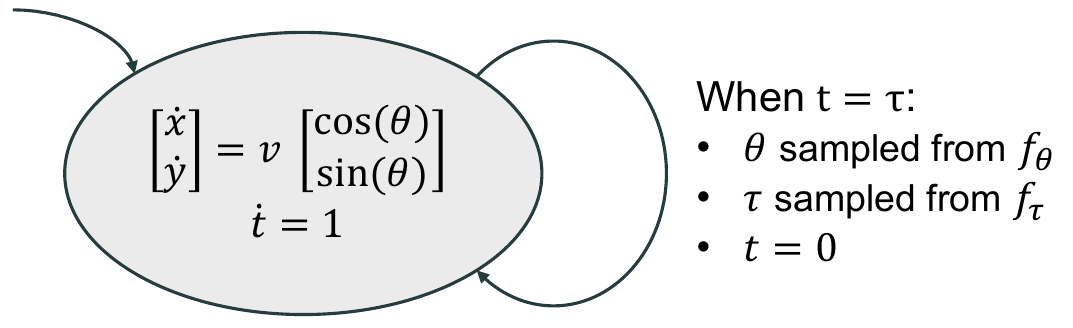}
    \caption{Hybrid dynamical model of the Lévy walker.}
    \label{fig:levy_walk_model}
\end{figure}

In the most classic case the duration of each \emph{run} is sampled from either an exponential or a power-law distribution \cite{Duncan2022}, while the new direction after a \emph{tumble} is assumed to be uniformly distributed (i.e. $f_\theta=U_{[-\pi;+\pi]}$).
Naturally, many variations of this simple model have been proposed to describe richer behaviours. For example the selection of a new direction might depend on the the current one \cite{Berg2004}, this allows to model the run-reverse-flick behaviour or the fact that smaller rotations might be more likely (e.g. the average angle between two consecutive movements of \textit{E. coli} is 68° \cite{Berg1993}).
Other variations have been proposed to model tumbles with non-zero duration \cite{Berg2004}, or to introduce randomness in the run phase, for example considering some additive random vibration \cite{Zaburdaev2011}.
\new{Moreover, these models can be \emph{continuified} to obtain a macroscopic description (see Section \ref{sec:background_modelling}), that captures the evolution of the ensemble of agents \cite{Codling2008,Conte2023}.}

\subsection{Persistent Turning Walkers}
Many other (micro)organisms do not show such abrupt turns, but smoother trajectories.
This is the case of microscopic algae such as \textit{Volvox} \cite{Pedley2016} and \textit{Chlamydomonas nivalis} \cite{Hill1993}, but also animals, such as fishes \cite{Gautrais2009,Zienkiewicz2015}.
To model such smooth and continuous turns it is more suitable to consider the heading direction and the angular velocity as state variables.
This is the case of \acf{PTW} models \cite{Gautrais2009}, in which the evolution of the angular velocity is described by a \ac{SDE}. Such equations extend the framework of \acp{ODE} to include randomness. 
In the most simple case the longitudinal speed $v$ is assumed to be constant, while the evolution of the angular velocity $\omega$ is described by a Ornstein-Uhlenbeck process, which is a type of linear \ac{SDE}, also known as Vasicek models and often used to model the evolution of interest rates or dynamical systems with noise \cite{Bishwal1923}. 
The resulting model then is
\begin{equation}
\begin{cases}
    \d v &= 0 \\
    \d \omega &= -\theta \omega\d t + \sigma\d W.
\end{cases}
\end{equation}
where the parameters $\theta$ and $\sigma$ represent respectively the \emph{rate} and the \emph{volatility} of the Ornstein-Uhlenbeck process; while \emph{W} is a Wiener process.
Extensions of this model have been proposed to include a dynamic evolution of the speed in the form of a second Ornstein-Uhlenbeck process, allowing also for dependencies between the speed and the angular velocity \cite{Zienkiewicz2015}.
\new{A similar model, specifically developed to describe the 3D motion of \textit{Euglena gracilis}, assumes the agent moves with constant speed  while rotating around an axis fixed to its own body frame \cite{Tsang2018}.}

The main difference between Lévy walk and \ac{PTW} models, besides the different smoothness of the trajectories, is in the resulting distributions of the  angular velocity. Indeed measuring the speed and the angular velocity of Lévy walkers at some time instants we expect to observe two clearly distinct behaviours, either high longitudinal speed and approximately zero angular velocity or low speed and significant angular velocity.
Contrary, for \ac{PTW} agents we expect the angular velocity to be continuously distributed over a range of values.
A second difference is given by the structure of the models, indeed in the \ac{PTW} model, the evolution of the state variables is explicitly described by \acp{SDE}, that can be easily modified to accommodate non-linear terms or exogenous inputs.
Last, while Lévy walks allow a straightforward extension to the 3D case, this requires non-trivial adjustments for \ac{PTW} models, as in \cite{Tsang2018}.
Figure \ref{fig:Levy_PTW_trajectories} shows the comparison between trajectories generated by these two models.

\begin{figure}
    \centering
    \begin{subfigure}[t]{0.4\textwidth}
        \centering
        \includegraphics[width=1\textwidth]{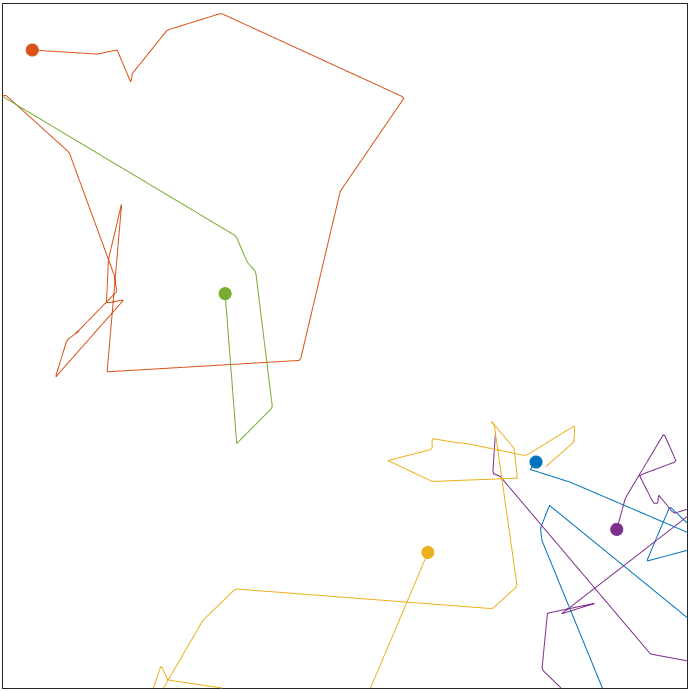}
        \caption{Lévy walk trajectories}
        \label{subfig:LevyWalk_trajectories}
    \end{subfigure}
    \quad
    \begin{subfigure}[t]{0.4\textwidth}
        \centering
        \includegraphics[width=1\textwidth]{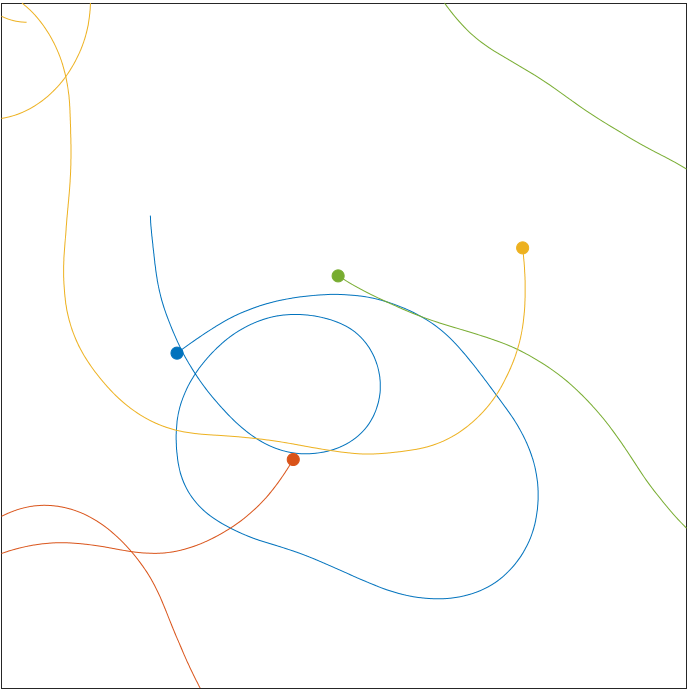}
        \caption{PTW trajectories}
        \label{subfig:PTW_trajectories}
    \end{subfigure}
    \caption{Example trajectories generated by Lévy and persistent turning walkers.
    }
    \label{fig:Levy_PTW_trajectories}
\end{figure}

Finally, it is worth to notice that some  microorganisms might even show a combination of these two behaviours, with both smooth and abrupt turns.
Besides the movement of individual agents, the interactions between them can induce more complex behaviours at the population level.
For example, inter-cellular chemical communication can trigger bacterial swarming \cite{Beer2019}, that is a rapid and highly coordinated movement of colonies of flagellate bacteria; while hydrodynamic interactions between swimming starfish embryos can make them assembly into active chiral crystals \cite{Tan2022}.

\section{Light induced behaviours}
\label{sec:microorganisms_light_response}
As previously discussed, light represents an ideal control input to influence the behaviour of micro-agents, either biological or artificial, mainly because high temporal and spatial resolution can be easily achieved \cite{Fielding2023}. 
Therefore, we are interested in understanding how light can influence the movement of such motile agents, and, once more, we will disregard biochemical and mechanical details, to focus on the resulting kinematic effects of such light inputs.

Some microorganisms can perform \emph{photomovements}, i.e. react to light and modify their motion in a variety of ways (see \cite{Iwatsuki1988, Nultsch1988}, or \cite{Burr1984} for a discussion about the nomenclature), namely
\begin{itemize}
    \item \emph{Photokinesis}: change of the movement speed. Positive if the speed increases with light, negative otherwise.
    \item \emph{Photoklinokinesis}: change of the turning rate or in the frequency of direction changes. Positive if the turns increase with light, negative otherwise.
    \item \emph{Phototaxis} or \emph{Phototactic orientation}: alignment with unidirectional light source. Positive if the alignment points toward the light source, negative otherwise.
    \item \emph{Photophobic response}: sudden stop followed by a change of direction. Such response can be further classified as:
    \begin{itemize}
    \item \emph{Step-down}: triggered by light decrease.
    \item \emph{Step-up}: triggered by light increase.
    \end{itemize}
\end{itemize}
Multiple of these responses are often present at the same time, in different proportions and possibly induced by different light intensities and wavelengths \cite{Hader2017}, or depending on light polarization \cite{Yang2021}.
The combination of the these individual responses can lead to the emergence of macroscopic density regulation behaviours at the population level \cite{Sgarbossa2002},
\begin{itemize}
    \item \emph{Photoaccumulation}: accumulation of microorganisms in illuminated areas.
    It is promoted by step-down photophobic, negative photokinesis and positive photoklinokinesis.
    \item \emph{Photodispersion}: accumulation of microorganisms in darker areas.
    It is promoted by step-up photophobic and positive photokinesis and negative photoklinokinesis.
\end{itemize}
This landscape of behaviours is further enriched by the presence of adaptation, indeed many microorganisms, such as \textit{Paramecium bursaria} \cite{Iwatsuki1988} or \textit{Volvox} \cite{Drescher2010}, show some adaptation to constant light inputs.
This implies that not only the intensity and the color, but also the duration, of light inputs can influence the organisms' response.

\section{Control applications \& experimental platforms}
Many applications of spatiotemporal control at the micro-scale have been proposed, and various types of micro-agents have been used.
Examples include both biological agents, such as microscopic algae, bacteria or mammalian cells, and artificial micro-robots.
Here we discuss some of these applications based on light inputs, with a focus on the proposed experimental platforms.

The movement of microscopic algae \textit{Euglena} (see Section \ref{sec:microorganisms} for details about this microorganism) can be influenced exploiting its light response. 
In \cite{Hossain2016} a prototype experimental platform is proposed, which embeds a microscopic camera for vision, while light actuation is provided by four LEDs placed around the sample. These LEDs can be used to induce phototaxis in the \textit{Euglena} and steer their movement direction toward one of the sides. 
The main features of this platform are the possibility to be used remotely and to automatically handle the loading of the sample, making it more accessible; nevertheless no automatic control application is proposed, either in the original or in subsequent works, such as \cite{Washington2019}.
A similar solution is proposed in \cite{Lam2017}. This platform lacks the remote access but embeds, besides the four LEDs, a \acf{DLP}, which provides more flexible actuation. Indeed here the phototactic response of \textit{Euglena} can be coupled with the photophobic one, induced by highly bright areas generated by the projector. 
A simple strategy is proposed to remove the microorganisms from a certain area of the sample, but the feedback is minimal, as it does not take into account any information on the movement of the specimens, and the control strategy consists only in the projection of a constant light pattern and the lighting of one of the LEDs.
Therefore, more advanced control strategies may easily improve the performance and solve more general tasks.
Another experimental platform, called the \acf{DOME}, was introduced in \cite{Denniss2022} (see Section \ref{sec:dome} for details about this platform). It lacks the lateral actuation LEDs present in \cite{Hossain2016,Lam2017}, but the embedded projector provides flexible actuation. Indeed, the \ac{DOME} has been used to influence the movement of microscopic algae \textit{Volvox}, to promote migration of human cells \cite{Wijewardhane2022}, and its use has also been proposed to control micro-robots \cite{Uppington2022}.
Specifically, it was shown in \cite{Denniss2022} that the movement of \textit{Volvox} (see Section \ref{sec:microorganisms} for details about this microorganism) can be influenced using repeating light inputs to induce negative photokinesis and, thus, reducing their speed. This was first demonstrated in \cite{Denniss2022} using an open loop strategy, and then improved in \cite{Denniss2022b} with the use of Q-learning.
Also motile micro-robots, made of contractile and light-reactive units, are currently being developed and tested in the \ac{DOME} \cite{Uppington2022}.

Spatiotemporal control has also been applied to bacteria. 
The speed and the spatial distribution of genetically engineered \textit{E. coli} were controlled in \cite{Massana-Cid2022} using light inputs.
Specifically the bacteria were engineered to show positive photokinesis. A computer vision algorithm was then used to detect the microorganisms and a bright spot was projected next to each agent, towards the desired direction of movement, so that only the agents moving in the desired direction exhibited augmented speed. This was shown to produce an overall flow of bacteria towards the desired area.
Light sensitive bacteria were also used to propel micro-robots. In \cite{Steager2015} multiple specimens of \textit{Serratia marcescens} were attached on a star-shaped plate. The propulsion generated by the bacteria induces a rotation of the micro-robot, whose velocity can be influenced via UV light acting on the swimming of the bacteria. 
Acting directly on the movement of bacteria is not the only way of controlling their spatial distribution. For example the relative distribution of two bacterial populations is spatially controlled in \cite{Velema2015} by the use of selective light-induced antibiotics.
In some application the physical distribution of the bacteria might not be the only quantity of interest, for example spatiotemporal control of the expression of green fluorescence protein, is implemented in \cite{Lugagne2022}, using genetically engineered \textit{E. coli} and light inputs of different colors.

Light can also be used to influence the migration of mammalian cells.
The \ac{DOME}, modified to project UV light, was used in \cite{Wijewardhane2022} to induce DNA damage in the cells at the edges of a wound, accelerating their migration and promoting healing of the wound. 
While sub-cellular lighting was used in \cite{Town2023} to influence the steering of migrating human cells.


\section{Discussion}
In this Chapter we described the concept of spatiotemporal control of microorganisms and how to achieve it.
Specifically, we discussed how microorganisms move in their environment, and two different approaches to mathematically model their stochastic motion.
We then discussed how light influences both, the movement of the individual organisms, and the macroscopic behaviour of the population.
Finally, we reviewed the existing solutions that exploit light to achieve spatiotemporal control of micro-agents, including microscopic algae, bacteria and mammalian cells. 
The main outcome of this analysis is that, despite various approaches having been proposed to control the movement of such micro-agents, a comprehensive framework is still missing. Indeed, most control strategies are designed heuristically and specific to a single species. This may be explained be the lack of mathematical models describing the motion of such microorganisms, and how it is influenced by light.

In the following we propose a general methodology to build such models, without a priori knowledge on the behaviour of a specific species, but using experimental data.
The next Chapter will present our experimental setup, describing the experimental platform and the species of microorganisms we studied.
Moreover, it will discuss the experiments we executed, and how the resulting data allow to characterize the movement and the light response of the microorganisms.
Chapter \ref{ch:modelling} will describe our modelling approach, and how experimental data can be used to estimate the values of the parameters and validate the resulting model.
Finally, we will discuss how these models allow to improve our spatiotemporal control over the microorganisms, and ultimately achieve density regulation.

\chapter{Experimental setup and data processing}
\thispagestyle{empty} 
\label{ch:dome}

As discussed in the previous Chapter characterizing the behaviour of microorganisms is a crucial step towards their spatiotemporal control. 
Specifically, we are interested in studying the movement of different species of motile microorganisms, together with their response to light.
Such study starts from experimental data to characterize and model the most relevant features of these behaviours, and then allow the design of feedback control laws.

In this chapter we will describe the experimental setup used to study the movement of microorganisms and their response to light, the algorithmic pipeline to analyse the data, and the results of such analysis. 

\section{DOME: the experimental platform}
\label{sec:dome}
Our research was carried out using the experimental platform called the \acf{DOME}.
This low-cost and open-access equipment was designed by Denniss et al. \cite{Denniss2022} at the BioComputeLab (\url{www.biocomputelab.github.io}) to study the behaviour of microscopic agents (e.g. protozoa, micro-robots \cite{Uppington2022}, tissue samples \cite{Wijewardhane2022}, etc.) under light inputs.
Such platform extends the functionalities of a classic digital microscope by implementing, besides the imaging capabilities, the possibility of acting on the sample with light inputs, and embedded computing.
These features, thanks to both the hardware design and the software we developed, allow high-throughput data acquisition and the execution of feedback control strategies.

\subsection{Hardware}
The \ac{DOME} is made of two main blocks arranged around the sample, the sensing and the actuation blocks (see Figure \ref{fig:dome}).
The first unit is responsible for the acquisition of the microorganisms images. It is made up of an HD camera, a light filter, a magnifying tube lens and a magnifying objective lens.
The objective lens are optional and replaceable, allowing to adjust the magnification level depending on the size of the microscopic agents to be studied. 
The tube lens provides a first magnification and can be moved vertically to adjust the back-focus distance from the camera sensor.
The light filter is an optional and easily replaceable component that can be chosen to image only specific wavelengths. For example, a long-passing filter, with the right cutting wave length, will absorb green and blue light, allowing only red light to reach the camera.
The second unit allows for the spatiotemporal control of the light inputs delivered to the organisms. It consists of a condenser lens, a projector and a Raspberry Pi board.
The Raspberry Pi ZERO board drives the \acf{DLP}, while the condenser lens focuses the projected light on the sample.
The two blocks are connected by a Raspberry Pi 4, which is connected to the camera, runs the software to acquire the images, process them and compute the light pattern to be projected. This is sent via an ad-hoc WiFi connection to the Raspberry Pi ZERO driving the projector.
The sample stage is placed between the sensing and actuation blocks.
Its vertical position can be adjusted to focus the camera on the sample, while a caliper controls the horizontal position.
The liquid sample is loaded on either a Petri dish or a microscopy slide.
Specifically, we used microscopy slides fitted with a 3D printed plastic frame, allowing the use of larger liquid samples (up to 300 $\mu$L) \cite{Denniss2022}.
For further details about the \ac{DOME} refer to \cite{Denniss2022} or visit \url{www.theopendome.org}.

\begin{figure}
    \centering
    \begin{subfigure}[t]{0.44\textwidth}
        \centering
        \includegraphics[width=1\textwidth]{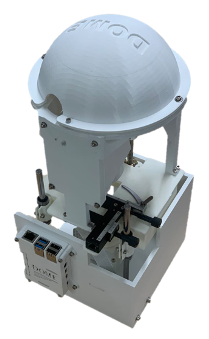}
        \caption{}
        \label{fig:dome_pic}
    \end{subfigure}
    \begin{subfigure}[t]{0.55\textwidth}
        \centering
        \includegraphics[trim={0 -70 10 0},clip,width=1\textwidth]{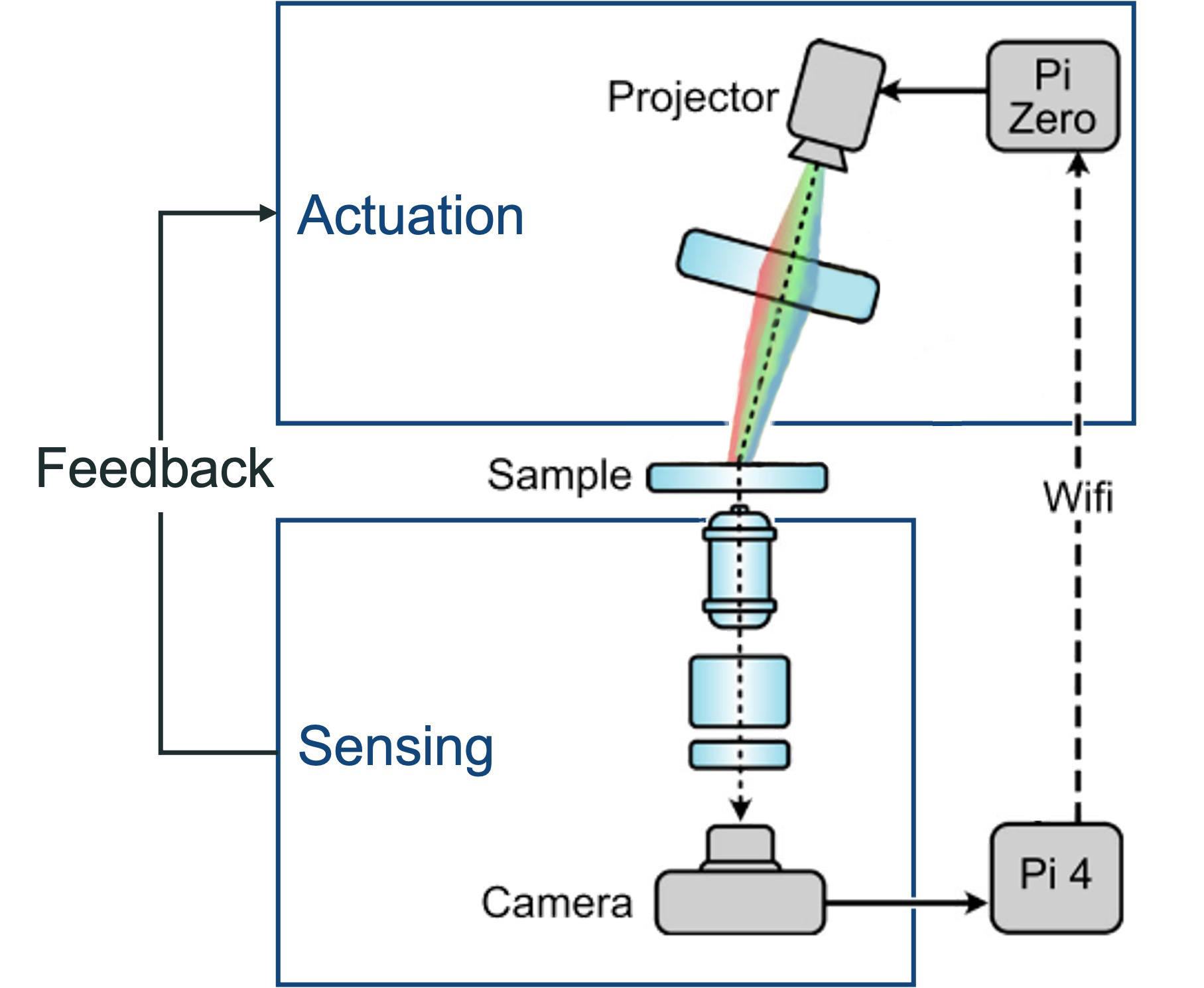}
        \caption{}
        \label{fig:dome_scheme}
    \end{subfigure}
    \caption{\ac{DOME} experimental platform: (a) picture and (b) schematic adapted from \cite{Denniss2022}.}
    \label{fig:dome}
\end{figure}

The projection and imaging specifications of the \ac{DOME} are reported in Table \ref{tab:DOME_spechs}, while the components used in our setup are listed in Table \ref{tab:DOME_components}. Additionally, to interact with the \ac{DOME}, we used a USB keyboard, a mouse, and an HDMI monitor connected to the Raspberry Pi 4.
Note that we used the \ac{DOME} in two different configurations, either with a 10X objective lens (90X total magnification) or without (9X total magnification), this allowed us to work with microorganisms of very different sizes. 
Moreover, we used a long wavelength passing filter to screen the camera sensor from blue light. 

\begin{table}
    \centering
    \begin{tabular}{p{2cm}p{2cm}p{2cm}p{2cm}p{2cm}}
    \hline
    Total magnification & Imaging pixel size [$\mu$m] & Imaging resolution [px]  & Projection pixel size [$\mu$m] & Projection resolution [px]  \\
    \hline
    9X & 4.44\x4.44 & \multirow{2}{*}{1080\x1920} & \multirow{2}{*}{30\x30} & \multirow{2}{*}{480\x854}  \\
    90X & 1.25\x1.25 &   &  & \\
    \hline
    \end{tabular}
    \caption{Projection and imaging specifications of the \ac{DOME} in our setup.}
    \label{tab:DOME_spechs}
\end{table}

\begin{table}
    \centering
    \begin{tabular}{p{0.3cm}p{2.5cm}p{6cm}p{2.8cm}}
    \hline
        & Component &	Description &	Manufacturer \\
    \hline
    \multirow{13}{*}{\rotatebox[origin=c]{90}{Optics}}
    \arrayrulecolor{lightgray}
        & Projector &	DLP evaluation module  (DLPM2000EVM) &	Texas Instruments \\ \cline{2-4}
		& Condenser Lens&	Ø50mm\x44mm FL, PCX condenser lens	& Edmund Optics \\ \cline{2-4}
        & Tube Lens	& 9X tube lens &	Edmund Optics \\ \cline{2-4}
        & Objective Lens	& 10X DIN semi-plan standard objective	& Edmund Optics \\ \cline{2-4}
        & Long Passing Filter&	Red light passing filter, Ø25.4mm to cut out blue light from the camera	& Edmund Optics \\ \cline{2-4}
        & Camera Lens	& Ø25mm uncoated glass window to seal the camera. &	Edmund Optics \\ \cline{2-4}
        & O-ring	&23.5-25.5mm O-ring to seal the camera	& \\
    \arrayrulecolor{black} \hline
    \multirow{12}{*}{\rotatebox[origin=c]{90}{Electronics}}
    \arrayrulecolor{lightgray}
        & Raspberry Pi 4 &	Model B, 4GB &	Raspberry Pi \\ \cline{2-4}
		& Raspberry Pi ZERO	& Raspberry Pi ZERO W	& Raspberry Pi \\ \cline{2-4}
        & Camera	& Raspberry Pi High Quality Camera Module	&Raspberry Pi \\ \cline{2-4}
        & Camera Connector	& Flexible cable for Raspberry Pi Camera - 300mm	& Raspberry Pi \\ \cline{2-4}
        & MicroSD Card	& MicroSD Card (Class 10 A1) 32GB &	SanDisk \\ \cline{2-4}
        & Raspberry Power Supply &	Official Raspberry Pi 4 power supply (5.1V, 3A)	& Raspberry Pi \\ \cline{2-4}
        & Power Supply &  Plug-in power supply (20W, 5V, 4A)	& RS \\ \cline{2-4}
        & Interface PCB	& Custom board to connect the Pi ZERO adapter	with the projector& \\ 
    \arrayrulecolor{black} \hline
    \multirow{6}{*}{\rotatebox[origin=c]{90}{Mechanical}}
    \arrayrulecolor{lightgray}
        & 3D printed parts & Custom 3D printed parts & \\ \cline{2-4}			
        & Linear Rail Set &	Set of lead screw, linear bearing, rod rail support & Glvanc	\\ \cline{2-4}
        & X-Y Caliper &	Caliper for the stage	& \\ \cline{2-4}
        & Linear Ball Bearing & Long linear motion ball bearings & \\
    \arrayrulecolor{black} \hline
    \end{tabular}
    \caption{Components of the \ac{DOME} in our setup. For up to date information visit \url{www.theopendome.org}.}
    \label{tab:DOME_components}
\end{table}

\subsection{Software}
\label{sec:dome_software}
To readily use the \ac{DOME} we developed a software package. The package is written in Python and made of two parts. 
The first includes scripts working on Raspberry Pi OS to control both the Raspberry Pi boards embedded in the \ac{DOME}.
It manages the WiFi communication between the two boards, and allows to acquire images and video from the camera, project light patterns, run experiments and save the data in a structured way. 
The second part collects scripts, that can be executed on any PC, to read the data acquired during the experiments, perform the automatic tracking (see Appendix \ref{ch:tracker}) of the microorganisms from the images and analyze the resulting trajectories.
This package is currently being refined and will soon be made available.

\section{Experiments with microorganisms}
\label{sec:experiments}

\subsection{Microorganisms}
\label{sec:microorganisms}
For our experiments we selected four species of microorganisms (see Figure \ref{fig:microorganisms}), two \textit{Paramecia} (\textit{P. caudatum} and \textit{P. bursaria}) and two microscopic algae (\textit{Volvox} and \textit{Euglena gracilis}).
Together these species represent a significant sample of motile and light sensitive eukaryotic microorganisms, that will allow to test the wide adaptability of both the experimental platform and the analytical pipeline.

\begin{figure}
    \centering
    \begin{subfigure}[t]{0.3\textwidth}
        \centering
        \includegraphics[trim={100 20 100 100},clip,width=1\textwidth]{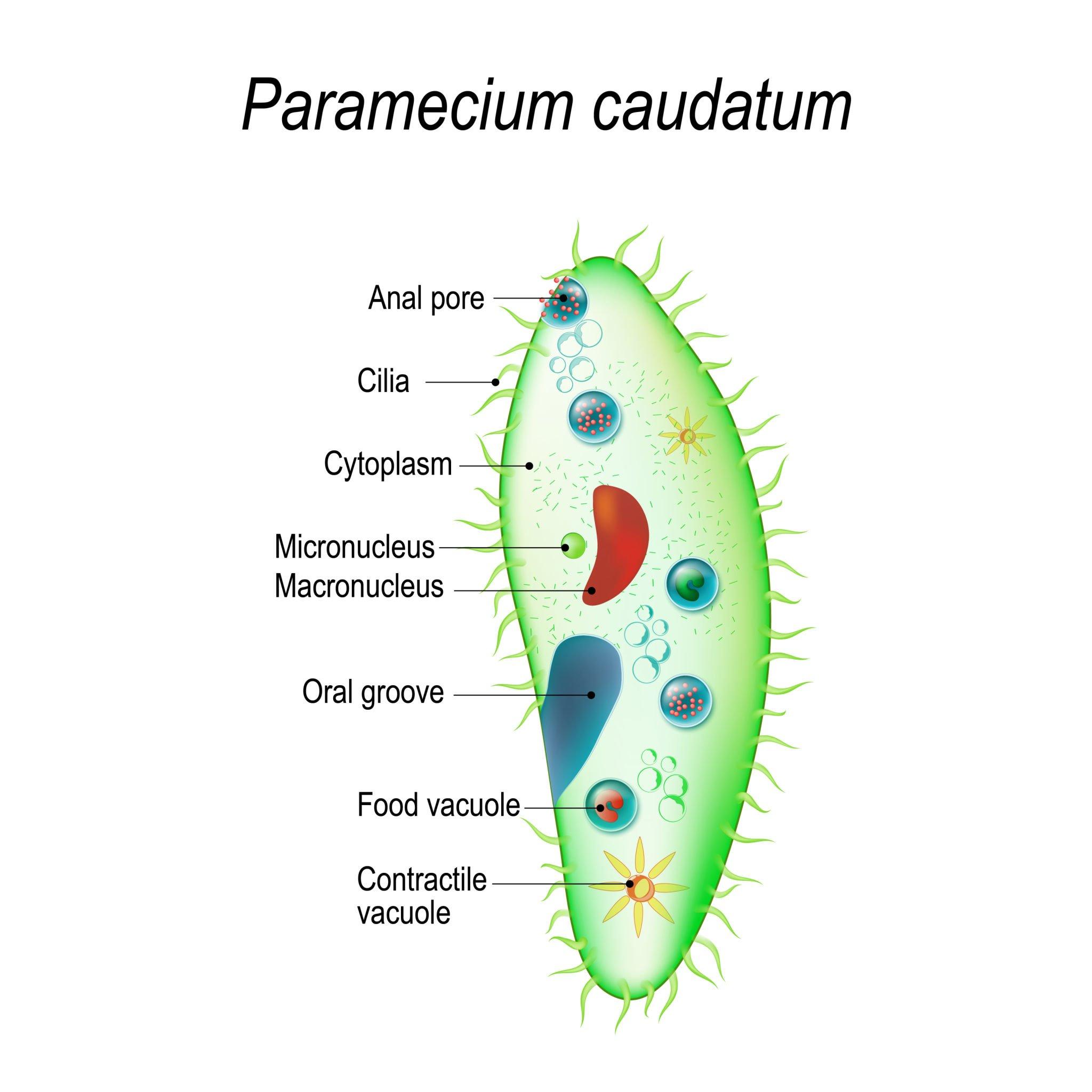}
        \caption{\textit{P. caudatum}}
        \label{fig:paramecium_caudatum}
    \end{subfigure}
        \begin{subfigure}[t]{0.45\textwidth}
        \centering
        \includegraphics[trim={40 20 20 60},clip,width=1\textwidth]{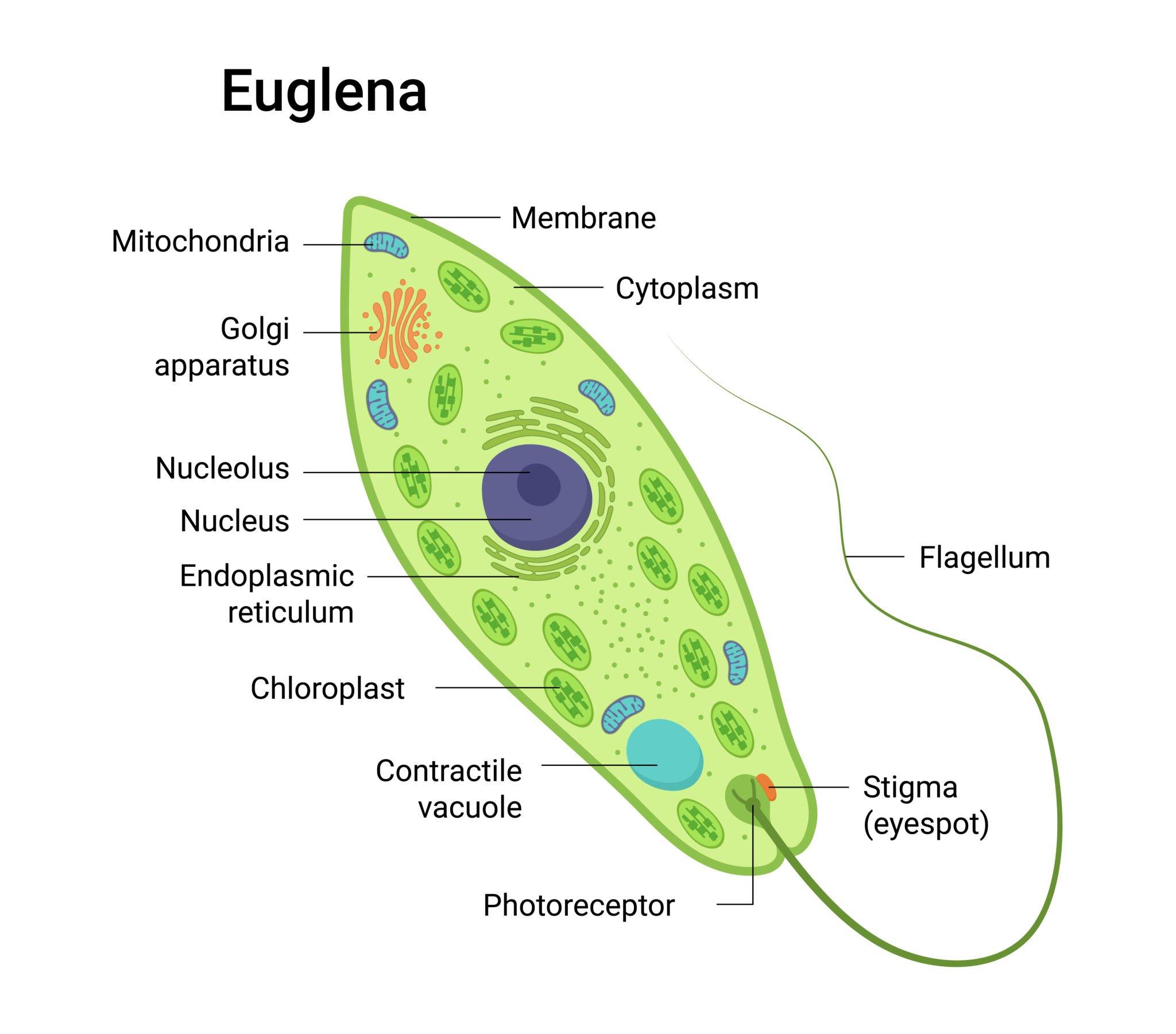}
        \caption{\textit{Euglena gracilis}}
        \label{fig:euglena}
    \end{subfigure}
    
    \begin{subfigure}[t]{0.24\textwidth}
        \centering
        \includegraphics[width=1\textwidth]{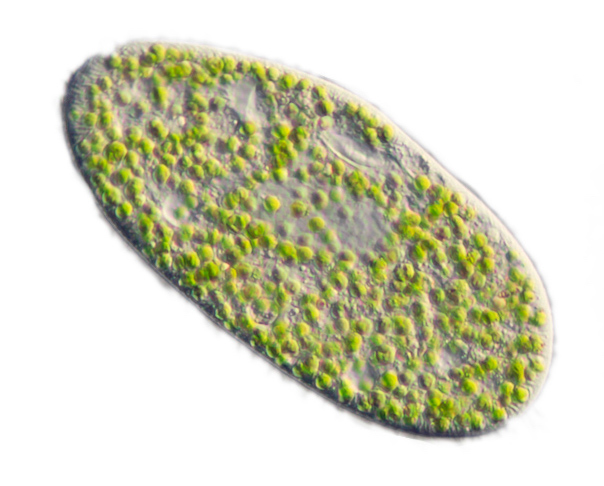}
        \caption{\textit{P. bursaria}}
        \label{fig:paramecium_bursaria}
    \end{subfigure}
    \quad \quad \quad
    \begin{subfigure}[t]{0.24\textwidth}
        \centering
        \includegraphics[width=1\textwidth]{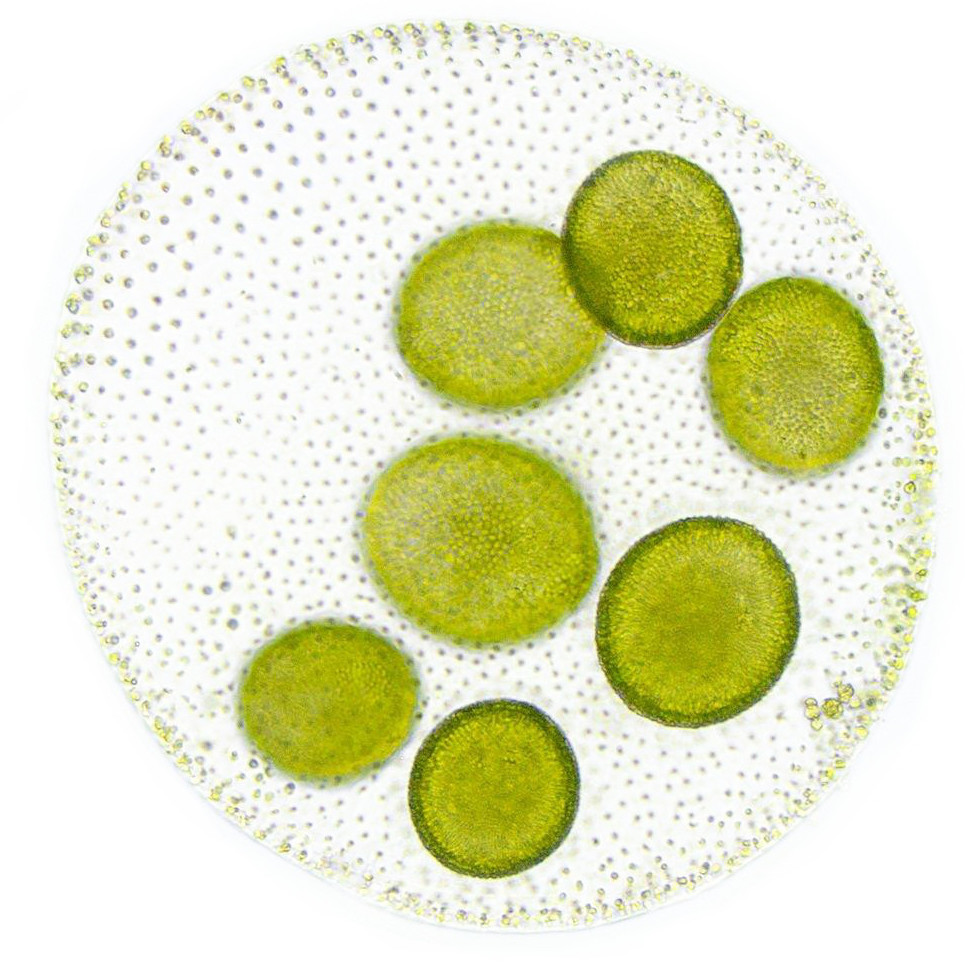}
        \caption{\textit{Volvox}}
        \label{fig:volvox}
    \end{subfigure}
    \caption{Species of microorganisms used in the experiments. 
    (a) Diagram of \textit{P. caudatum}.
    (b) Diagram of \textit{Euglena gracilis}.
    (c) Picture of \textit{P. bursaria}, where symbiotic \textit{Chlorella} are easily visible.
    (d) Picture of a \textit{Volvox} colony.
    }
    \label{fig:microorganisms}
\end{figure}

\textit{Paramecium caudatum} is a common species of paramecium widespread in stagnant and  fresh water environments \cite{Wichterman1986}. 
The locomotion along the main axis is characterized by a spiral path and powered by the cilia.
It shows photophobic step-down response \cite{Okumura1963} (see Section \ref{sec:microorganisms_light_response} for an introduction to light responses of microorganisms), with peak response to UV light (300-330nm) and a linear correlation between the intensity of light and the percentage of responding agents.
Compared to the other three species \textit{P. caudatum} is large \new{(200-300 $\mu$m)} and fast moving \new{(1100-1500 $\mu$m/s \cite{Okumura1963})}, but its light response is weaker and less documented.

\textit{Paramecium bursaria} is another ciliate fresh water species \cite{Wichterman1986}. It is characterized by the symbiotic relationship with the \textit{Chlorella} microscopic algae, which can live inside the these microorganisms and influences their response to light stimuli. 
Specifically, the light-adapted (i.e., cultured in a lit environment for a few hours before the experiment) \textit{Chlorella-}containing paramecium show a combination of step-down photophobic response, negative photokinesis and positive klinokinesis \cite{Cronkite1981}, resulting in photoaccumulation. 
Conversely, \textit{Chlorella-}containing dark-adapted (i.e., cultured in a dark environment for a few hours before the experiment) and \textit{Chlorella-}free specimens perform step-up photophobic response, resulting in photodispersion \cite{Iwatsuki1988}. Both behaviors are predominantly induced by blue and green light (440-680nm).
The complex light response of \textit{P. bursaria} results in a large literature, but often unclear and even contrasting results \cite{Iwatsuki1988,Saji1974,Cronkite1981}.
Moreover, the small size \new{(85-150 $\mu$m)} combined with a low population density (0.4-1.0 specimens/$\mu$L) and a fast movement \new{(1500-1900 $\mu$m/s \cite{Iwatsuki1988})} make the simultaneous imaging of multiple organisms problematic. 

\textit{Euglena gracilis} is a fresh water unicellular alga. It swims rolling around its main axis and is propelled by a flagellum. It shows strong light responses and, therefore, is commonly used in studies on light-responsive microorganisms \cite{Lam2017, Hossain2016, Washington2019}.
Specifically it shows phototaxis (positive when light intensity is low and negative otherwise \cite{Hader1987}), and a strong step-up photophobic response \cite{Tsang2018}, resulting in a clear photodispersion behaviour, with peak response  induced by blue light (420-480nm).
It is the smallest within the species we considered \new{(50-90 $\mu$m)}, but the relatively slow movement \new{(30-120 $\mu$m/s \cite{Muku2023,Tsang2018})} combined with a high population density (50-70 specimens/$\mu$L) implies that a significant amount of organisms can be contained in a relatively small volume and be observed with higher magnification.

The last species we worked with is \textit{Volvox}. These unicellular algae form spherical colonies of up to 50,000 cells.
The colonies, despite being almost spherical, have a main axis dictating the internal structure and the direction of movement. 
The propulsion is generated by the coordinated motion of the cilia of the cells on the surface of the colony.
They show negative photokinetic response to green light (470-530nm) \cite{Denniss2022}, resulting in photoaccumulation \cite{Ueki2010} when the light intensity is not too high, otherwise a photodispersion behaviour emerges \cite{Drescher2010}.
\textit{Volvox} is the largest species we worked with \new{(350–500 $\mu$m)}, moreover their spherical shape and the relatively slow movement with respect to the body size \new{(300-600 $\mu$m/s \cite{Pedley2016, Solari2008})} facilitate the automatic detection and tracking.


\begin{table}
    \centering
    \small
    \begin{tabular}{p{2cm}p{1.2cm}p{1.2cm}p{1.6cm}p{5cm}}
        \hline
        Species     & Size [$\mu$m] & Speed [$\mu$m/s] & Peak light sens. [nm]& Ligth response \\
        \hline
        \arrayrulecolor{lightgray}
        \textit{P. caudatum} & 200-300 & 1100-1500 & 300-330 & Step-down photophobic.  \\ \hline
        \textit{P. bursaria} & 85-150 & 1500-1900 & 440-680 & Light-adapted: step-down photophobic, negative photokinesis and positive klinokinesis, photoaccumulatio. Dark-adapted or \textit{Chlorella-}free: step-up photophobic response, photodispersion \\ \hline
        \textit{E. gracilis} & 50-90 & 30-120 & 420-480  & Phototaxis (negative or positive depending on the intensity), step-up photophobic, photodispersion. \\ \hline
        \textit{Volvox}     & 350-500 & 300-600 & 470-530 & Negative photokinesis, photoaccumulation. \\
        \arrayrulecolor{black} \hline
    \end{tabular}
    \caption{Species of microorganisms used in the experiments and their main features.}
    \label{tab:microorganisms}
\end{table}

These species show remarkably different features (i.e. size, movement speed, light response, etc), and together represent a significant sample of motile eukaryotic microorganisms. See Table \ref{tab:microorganisms} for a summary.

\textit{P. caudatum}, \textit{Volvox} and  \textit{Euglena} were sourced from BladesBio UK (\url{www.blades-bio.co.uk}), while \textit{P. bursaria} were sourced from Carolina Biological Supply (\url{www.carolina.com}).
\textit{Volvox} were cultured in Alga Grow medium,  \textit{Euglena} in its ad-hoc medium and both \textit{Paramecium} species in Protozoan Pellet medium. 
Euglena, \textit{Volvox} and \textit{P. bursaria} were kept under an artificial light bank with a daily activation time of 12h%
\footnote{Culturing and media preparation protocols were taken from \url{www.carolina.com/teacher-resources/Document/protozoa-invert-care-handling-instructions/tr10466.tr}.}.

\subsection{Experimental protocol}
To collect a rich data-set and to be able to characterize both the movement in a dark environment and the light response of the sample microorganisms we performed a set of open loop experiments, projecting predefined spatial and temporal light inputs.
Each experiment lasts 3 minutes with the acquisition of an image every 0.5s.
The sampling time was chosen due to constraints on the execution frequency given by the platform, while the duration of the experiments was selected large enough to collect abundant data and be able to observe slow dynamics, such as adaptation, while keeping a reasonable burden in terms of storage memory and time.

For each experiment a given volume (see Table \ref{tab:exp_param}) is taken from the culture jar and placed on the microscopy slide, eventually in the well made by the plastic frame. The sample is then placed on the sample stage and the \ac{DOME}  covered with a dark hood to screen external, not controlled, light.
The projector is then set to shine constant and low intensity red light (640nm, 5\% of the projector's maximum red brightness), providing the illumination for dark-field imaging.
Blue light (460nm), being within the sensibility range of most microorganisms, is then used as input. The presence of a red light filter prevents blue light from reaching the camera sensor, guaranteeing a constant illumination for better image acquisition%
\footnote{The light filter was not present during preliminary experiments, resulting in uneven quality images due to sudden changes in luminosity.}.
Due to the different sizes of the species involved in the experiments we used 90X magnification, when working with \textit{Euglena}, while the other species were imaged at 9X magnification, see Section \ref{sec:dome} for details about the components.

\begin{table}
    \centering
    \begin{tabular}{p{2.1cm}p{2.1cm}p{2cm}p{2.2cm}p{2.5cm}}
        \hline
        Species   & Magnification & Density [$\mu $L$^{-1}$]  & Sample \newline volume [$\mu$L]& Plastic frame  \\
        \hline
        \textit{P. caudatum} & 9X & 0.19-0.95 & 100-150 & Yes  \\
        \textit{P. bursaria} & 9X & 0.4-1.0 & 20-100 &  Only for larger volumes\\
        \textit{E. gracilis} & 90X & 50-70 & 15-20  & No\\
        \textit{Volvox}     & 9X & 0.2-0.8 &120-150 &  Yes\\
        \hline
    \end{tabular}
    \caption{Experimental parameters.}
    \label{tab:exp_param}
\end{table}

We run different experiments to characterize (i) the movement of the microorganisms in absence of light inputs, and (ii) their response to light inputs of different intensities and duration.
Specifically we run experiments with light inputs of different duration, from 1 to 60 seconds, and different intensity, ranging from 0\% to 100\% of the projector's maximum blue brightness.
Table \ref{tab:experiments} and Figure \ref{fig:experiments_inputs} show a short description and a graphical representation of the different types of experiments, while Figure \ref{fig:images_example} shows examples of the images acquired during the experiments.
Each experiment was repeated, for each species, for a minimum of 3 biological and 2 technical replicates, resulting in more than 300 experiments.

 \begin{table}
    \centering
    \begin{tabular}{p{0.5cm}p{2.5cm}p{8cm}}
    \hline
     & Experiment & Description \\
    \hline
    \multirow{3}{*}{\rotatebox[origin=c]{90}{Default}}
         & \multirow{3}{*}{No input}         & \multirow{3}{*}{Constant OFF, only red background illumination.}\newline \newline \\
    \hline
    \multirow{5}{*}{\rotatebox[origin=c]{90}{Intensity}}
    \arrayrulecolor{lightgray}
         & 30\% Intensity   & 1min OFF, 1min with 30\% input, 1min OFF. \\ \cline{2-3}  
         & 60\% Intensity   & 1min OFF, 1min with 60\% input, 1min OFF. \\ \cline{2-3} 
         & 100\% Intensity  & 1min OFF, 1min with 100\% input, 1min OFF.\\ \cline{2-3} 
         & Ramp             & 10s OFF, linearly increasing input from 0\% to 100\%, 10s with constant 100\% input.\\
    \arrayrulecolor{black} \hline
    \multirow{6}{*}{\rotatebox[origin=c]{90}{Frequency}}
    \arrayrulecolor{lightgray}
         & Switch 10s       & 10s OFF, repeating ON-OFF input with 20s period and 50\% duty cycle.\\ \cline{2-3}
         & Switch 5s        & 10s OFF, repeating ON-OFF input with 10s period and 50\% duty cycle.\\ \cline{2-3}
         & Switch 1s        & 10s OFF, repeating ON-OFF input with 2s period and 50\% duty cycle.\\
    \arrayrulecolor{black} \hline
    \end{tabular}
    \caption{Types of experiments executed to characterize the movement and light response of microorganisms.}
    \label{tab:experiments}
\end{table}

\begin{figure}[t]
    \centering
    \includegraphics[trim={0, 152, 0, 0},clip,width=0.9\linewidth]{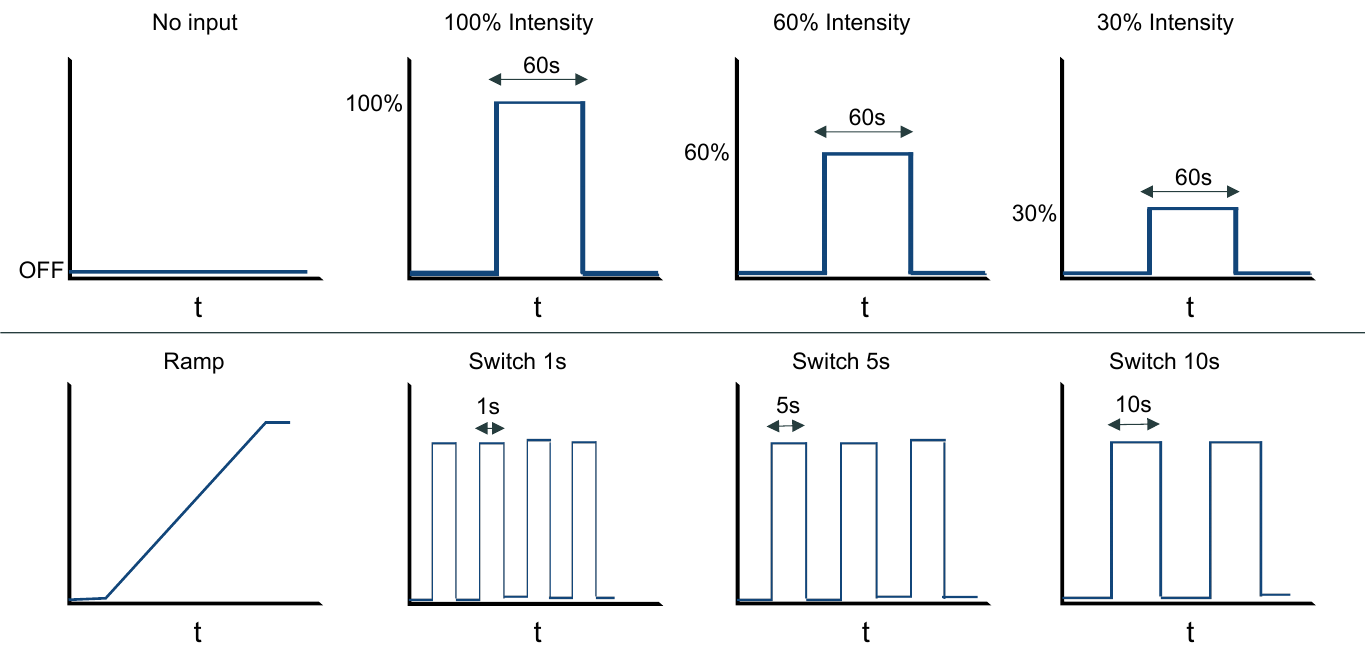}
    \caption{Graphical representation of the types of experiments executed to characterize the movement and light response of microorganisms.
    }
    \label{fig:experiments_inputs}
\end{figure}

\begin{figure}[t]
    \centering
    \begin{subfigure}[t]{0.48\textwidth}
        \includegraphics[width=\linewidth]{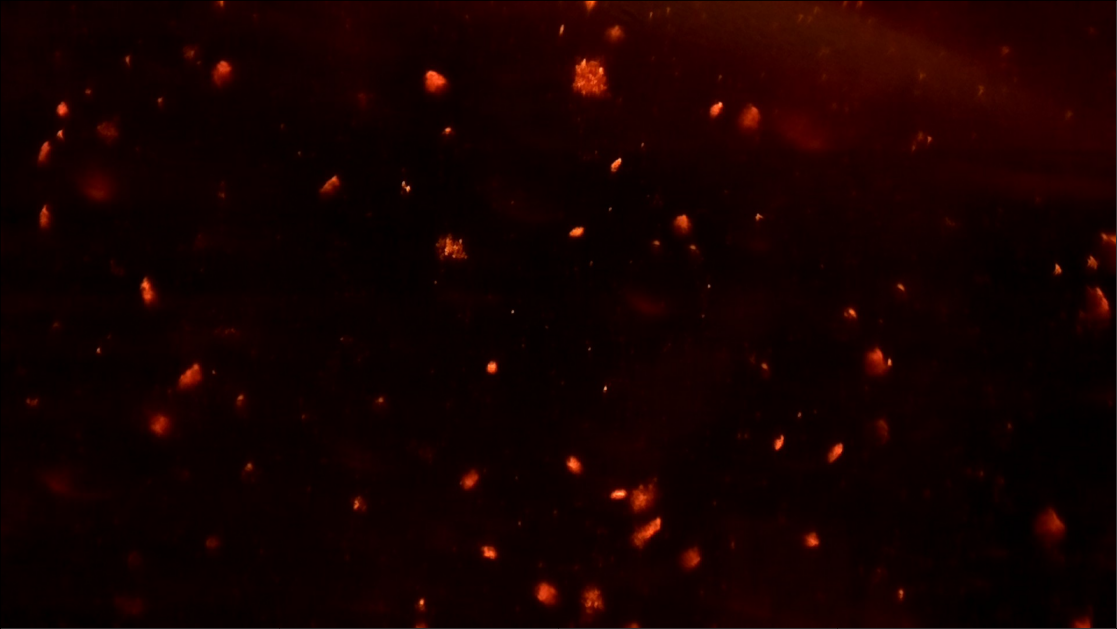}
        \caption{\textit{Euglena}}
    \end{subfigure}
    \begin{subfigure}[t]{0.48\textwidth}
        \includegraphics[width=\linewidth]{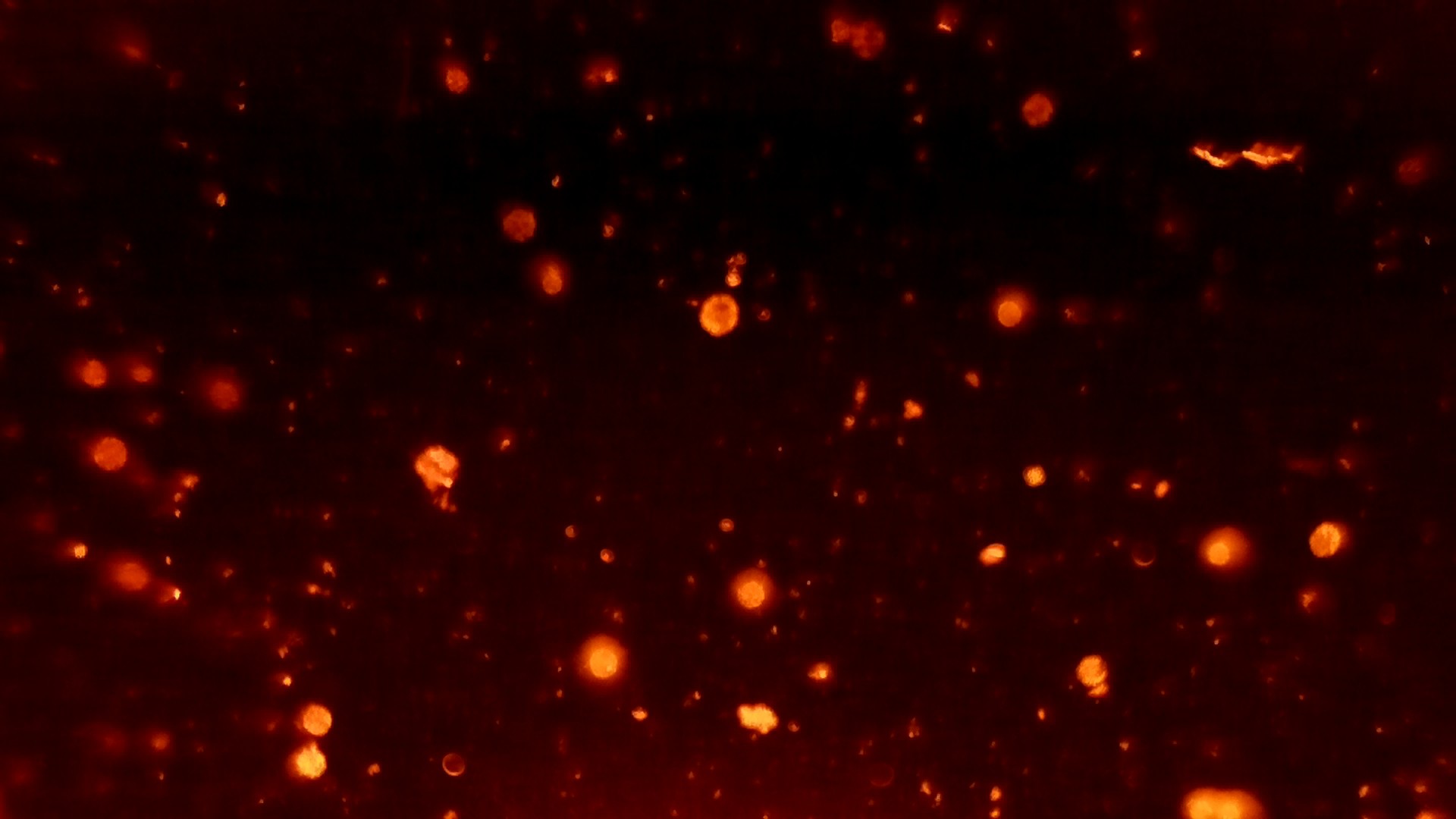}
        \caption{\textit{Volvox}}
    \end{subfigure}

    \begin{subfigure}[t]{0.48\textwidth}
        \includegraphics[width=\linewidth]{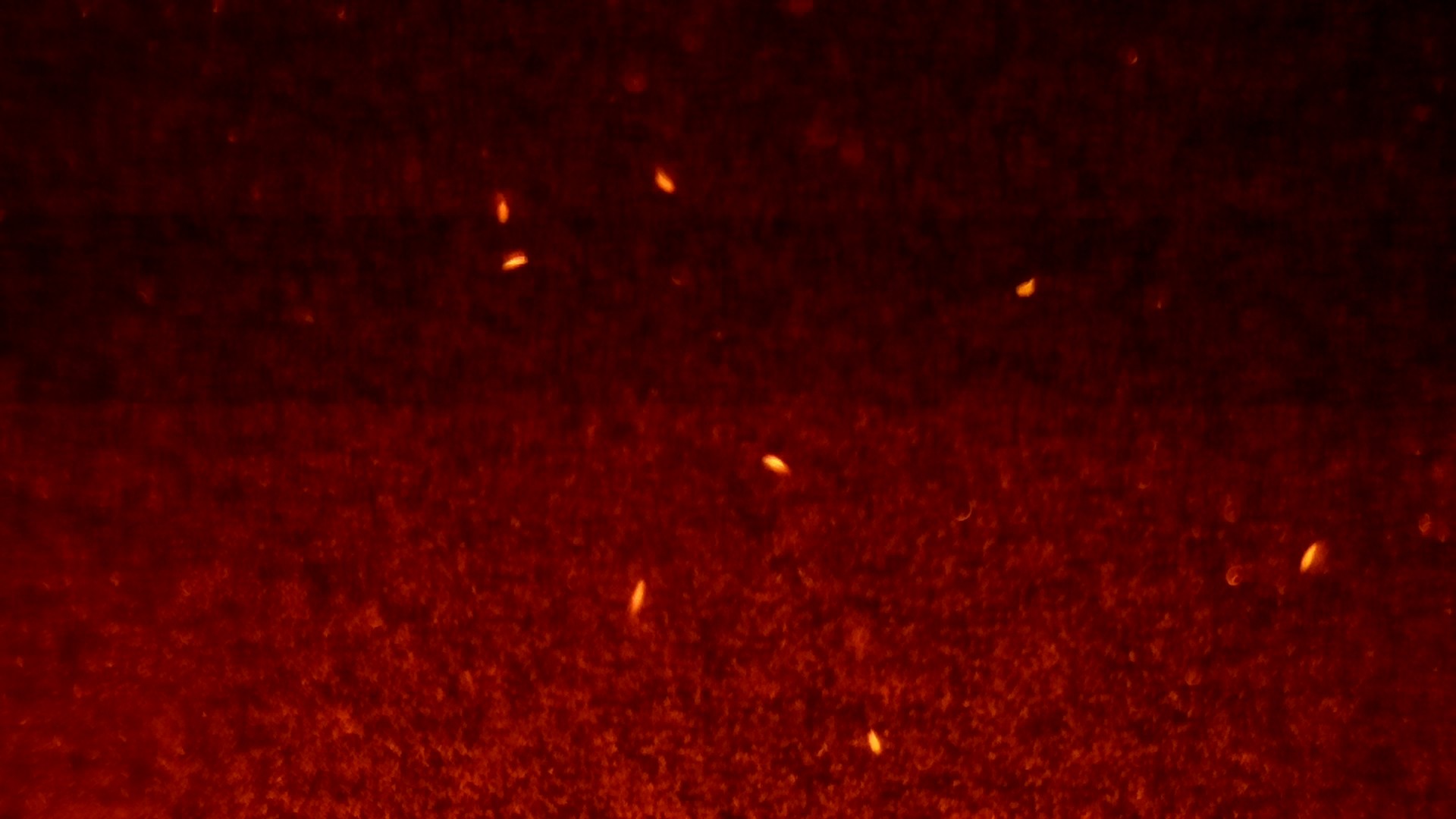}
        \caption{\textit{P. bursaria}}
    \end{subfigure}
    \begin{subfigure}[t]{0.48\textwidth}
        \includegraphics[width=\linewidth]{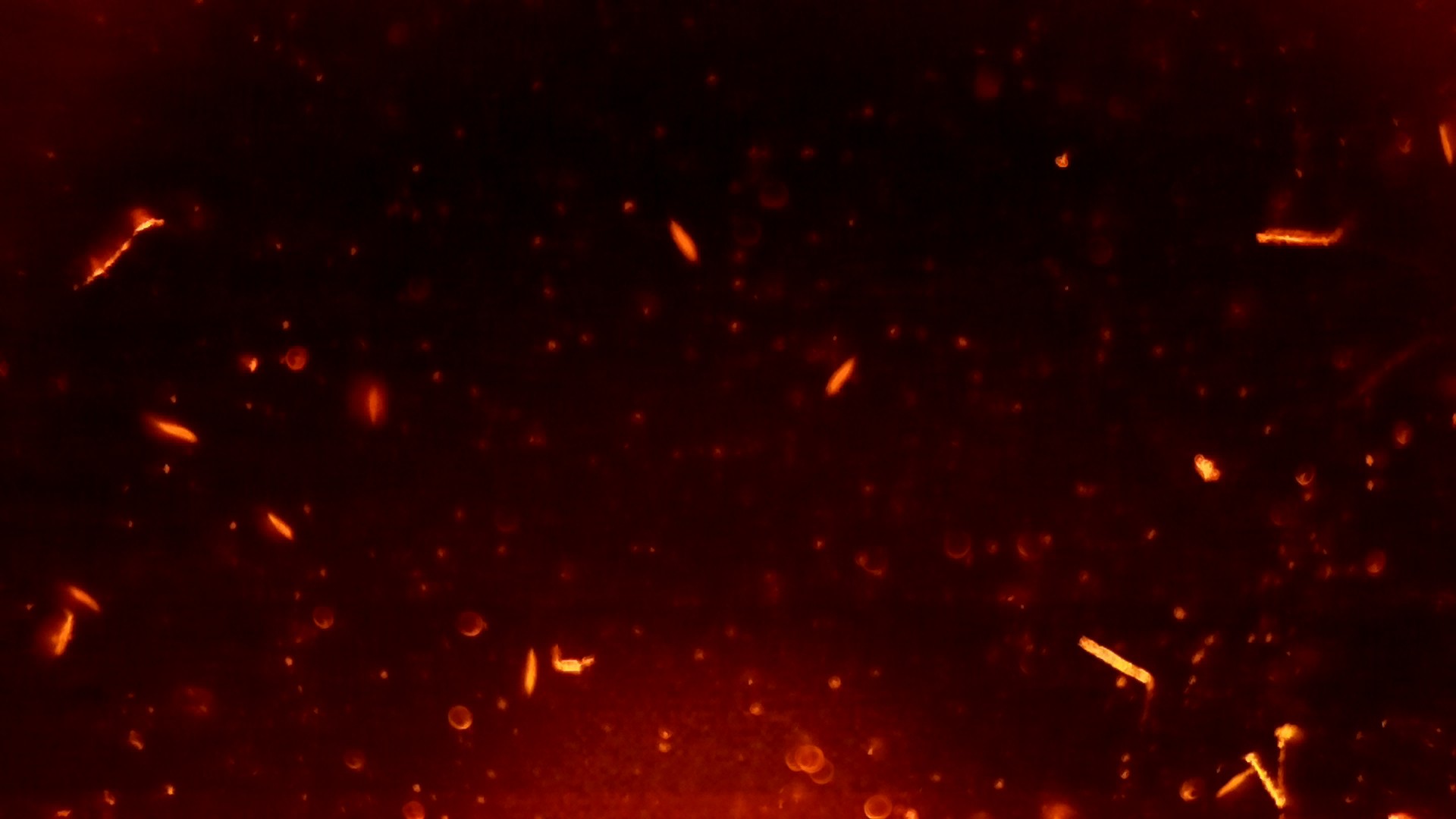}
        \caption{\textit{P. caudatum}}
    \end{subfigure}
    \caption{Example images acquired during the experiments. 
    The experimental conditions are those reported in Table \ref{tab:exp_param}.
    }
    \label{fig:images_example}
\end{figure}

\section{Data extraction and analysis}
\label{sec:analysis}

To study the movement of microorganisms we need to extract the relevant kinematic variables (position, velocity, orientation, etc.) from the images acquired during the experiments.
Therefore, we developed a software that elaborates these images, automatically detects the microorganisms and tracks their position over time (see Appendix \ref{ch:tracker} for more details about the tracking algorithm).

A preliminary processing of the recorded trajectories was performed prior of the data analysis.
Specifically, the trajectories are selected according to their duration, and those shorter than 5s are discarded.
This helps in removing erroneously detected trajectories, and ensures remaining ones have a sufficient number of points.
The remaining trajectories are then smoothed with a moving average of window size 3 to reject noise or high frequencies movements.
The velocity vectors are computed using second order accurate central differences, to then extract the longitudinal speed (i.e. the magnitude of the velocity vector) and the angular velocity (i.e. the rate of change of the movement direction).
Specifically, given a single trajectory of points $\{\vec{x}_1, \vec{x}_2,...,\vec{x}_n\}$ in the camera frame $[0;1080]\times[0;1920]$, sampled every $\Delta T$ seconds, the velocity  is computed as
\begin{equation}
    \vec{v}_k =
        \begin{cases}
        \frac{\vec{x}_{k+1}-\vec{x}_{k}}{\Delta T} &\mbox{if  } k=1; \\
        \frac{\vec{x}_{k}-\vec{x}_{k-1}}{\Delta T} &\mbox{if  } k=n; \\
        \frac{\vec{x}_{k+1}-\vec{x}_{k-1}}{2\Delta T} &\mbox{otherwise} \\
        \end{cases}
    \label{eq:numerical_differentiation}
\end{equation}
and measured in px/s.
Then the speed $v_k$ is simply given by the norm of $\vec{v}_k$, while the angular velocity $\omega_k$, measured in rad/s, is computed as

\begin{equation}
    \omega_k = \frac{1}{\Delta T} \func{atan2}{\frac{\vec{v}_{k} \times \vec{v}_{k+1}}{\vec{v}_{k} \cdot \vec{v}_{k+1}}}
\end{equation}

Each of these time series is then further smoothed by applying again the moving average of size 3.

We then check possible tracking errors, by applying Algorithm \ref{alg:detect_outliers}, with $m=2.5$, to detect the outliers in the resulting speed data.
For each outlier the corresponding agent and time instant are reported to the user that can visually check the tracking video, and, in case of errors, correct it by removing erroneous trajectories or merging multiple trajectories belonging to the same microorganism.


\begin{algorithm}[t]
\caption{Given a data-set (vector or multidimensional array) $d$ and a positive threshold $m$, the function  \texttt{detect\_outliers($d,m$)} identifies the outliers in the $d$.}
\label{alg:detect_outliers}
\begin{algorithmic}
\tt
\STATE function detect\_outliers($d,m$) \begin{addmargin}[1em]{0em}
require $m>0$ \qquad\quad\algcomment{Check threshold}
\STATE $outliers$ = empty list \algcomment{Initialize empty list}
\FOR{$k$ {\bf in} indices($d$)}
    \STATE \algcomment{Compute normalised variation}
    \STATE $s$ = |$d$[$k$]-median($d$)| / median(|$d$-median($d$)|)

\IF{$s>m$}
    \STATE  $outliers$.append($k$) \algcomment{Insert $k$ in the list of outliers}
\ENDIF
\ENDFOR
\end{addmargin}
\RETURN $outliers$

\end{algorithmic}
\end{algorithm}

Note that methods described in this and the following pages hold for any of the species we considered, but, for the sake of time, we have, up to this point, only analysed the data from the experiments with \textit{Euglena}.


\subsection{The behaviour in a dark environment}
\label{sec:analysis_dark}

First we characterized the behaviour of microorganisms without light inputs (i.e., in a dark environment), with only low intensity red illumination.
Figure \ref{fig:dark_time_evolution} shows the time evolution of speed and angular velocity from an experiment with \textit{Euglena} (see Section \ref{sec:experiments} for details about the experiments).
It can be easily observed that the data are very variable, but their macroscopic features (i.e., mean and variance) are consistent over time.

\begin{figure}
    \centering
    \includegraphics[width=1\linewidth]{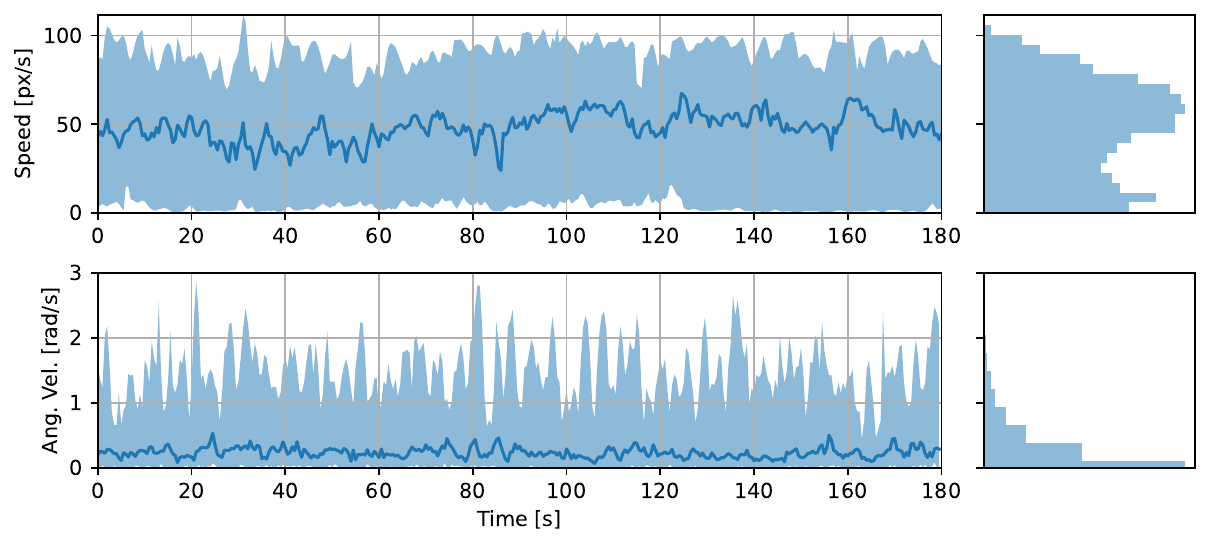}
    \caption{Longitudinal speed and absolute angular velocity of \textit{Euglena} specimens in a dark environment.
    The solid lines show the sample's median, while the shaded areas show the minimum and the maximum.
    }
    \label{fig:dark_time_evolution}
\end{figure}

We then statistically characterized these data. 
Specifically, we computed the mean speed and mean absolute angular velocity of each agent, and observed a clear negative correlation between these quantities (see Figure \ref{fig:dark_scatter_mean}), with faster agents having a lower angular velocity.
Nevertheless, we notice this might be a byproduct of discrete time sampling and numerical differentiation. Indeed, if during a certain time step an agent turns, while keeping a constant longitudinal velocity, at the next sampling time the displacement from the initial position will be smaller, and numerical differentiation will result in a smaller velocity.

Moreover, we studied the deviation of agents from their average behaviour. Figure \ref{fig:dark_scatter_speed} shows a non monotonic dependence between the standard deviation and the mean of the speed of the agents, while Figure \ref{fig:dark_scatter_angv} shows the expected linear correlation between standard deviation and mean value of the angular velocity.
Therefore, we can conclude that faster agents have more consistent motion, with lower angular velocity and proportionally less variable speed.

Moreover, none of the plots in Figure \ref{fig:dark_scatter} shows clearly separated clusters, this suggests the presence of a single homogeneous, yet variable, population; rather that the coexistence of more populations with clearly different behaviours.


\begin{figure}
    \centering
    \begin{subfigure}[t]{0.49\textwidth}
        \centering
        \includegraphics[width=1\textwidth]{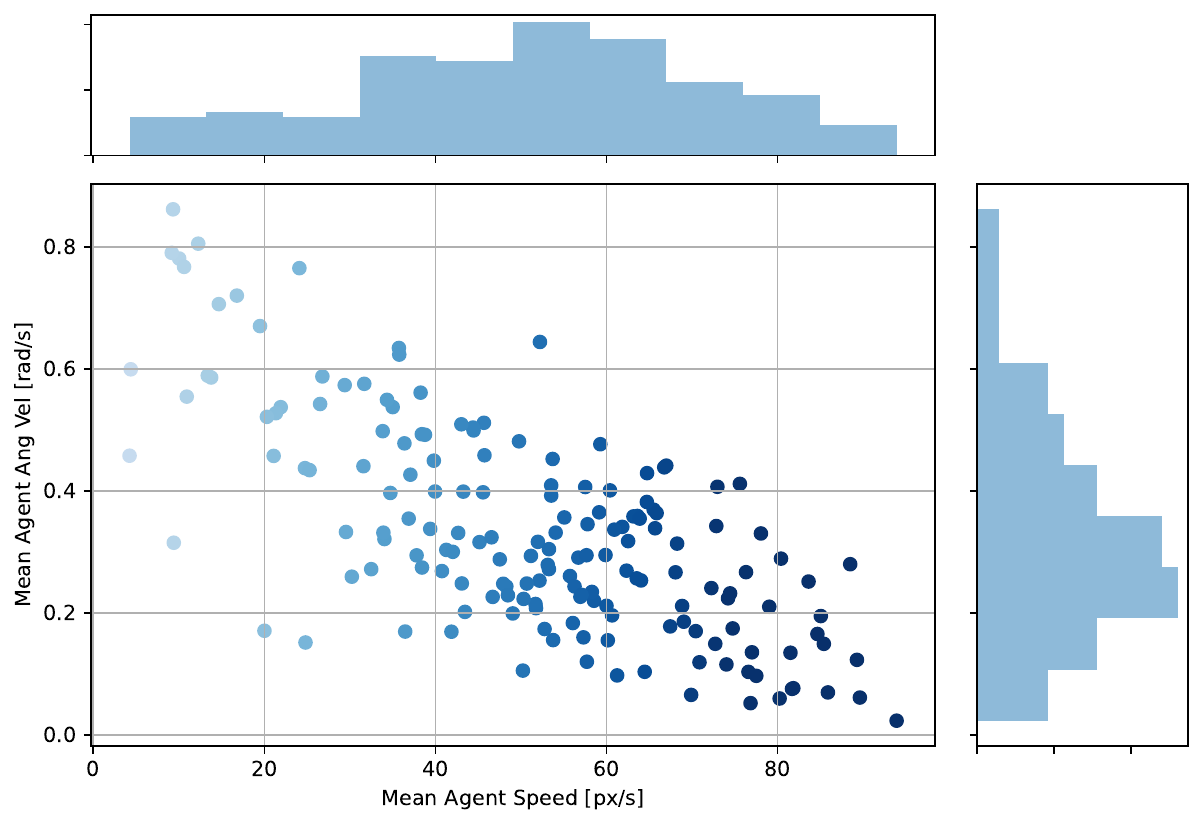}
        \caption{}
        \label{fig:dark_scatter_mean}
    \end{subfigure}
    
    \begin{subfigure}[t]{0.49\textwidth}
        \centering
        \includegraphics[width=1\textwidth]{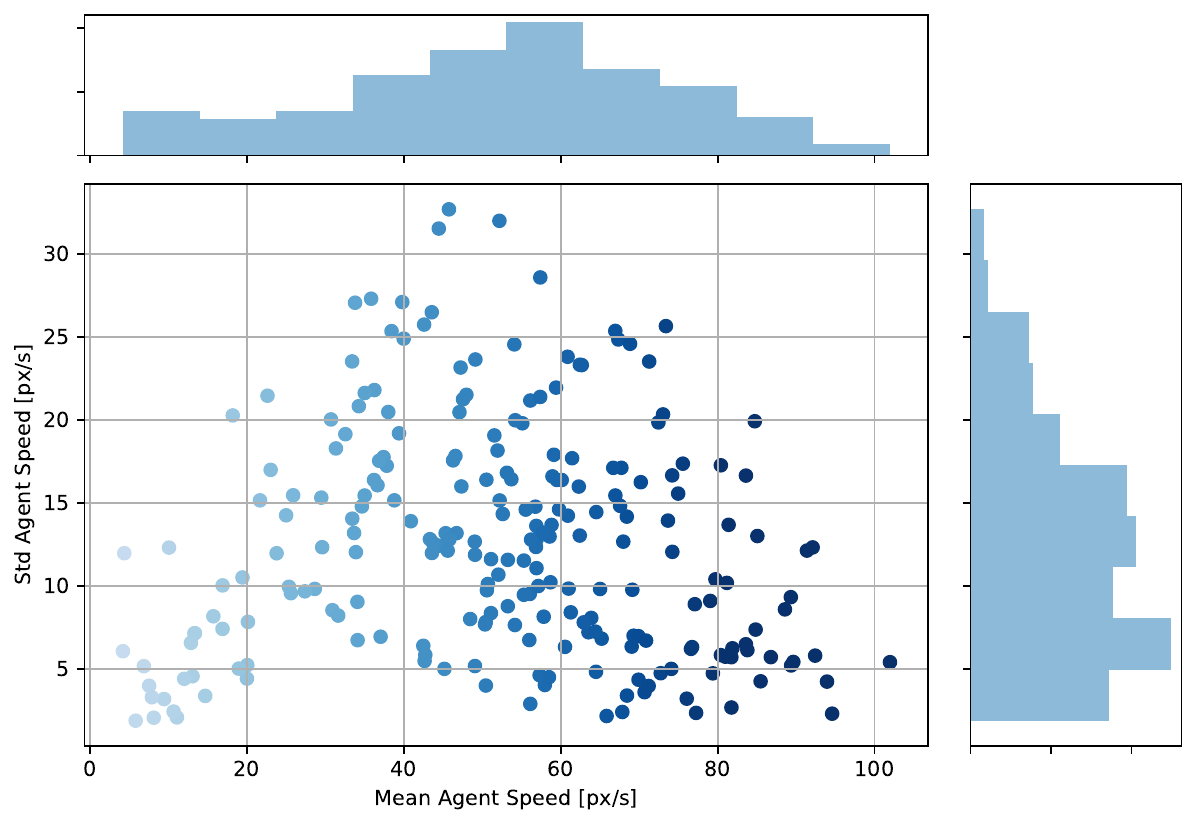}
        \caption{}
        \label{fig:dark_scatter_speed}
    \end{subfigure}
    \begin{subfigure}[t]{0.49\textwidth}
        \centering
        \includegraphics[width=1\textwidth]{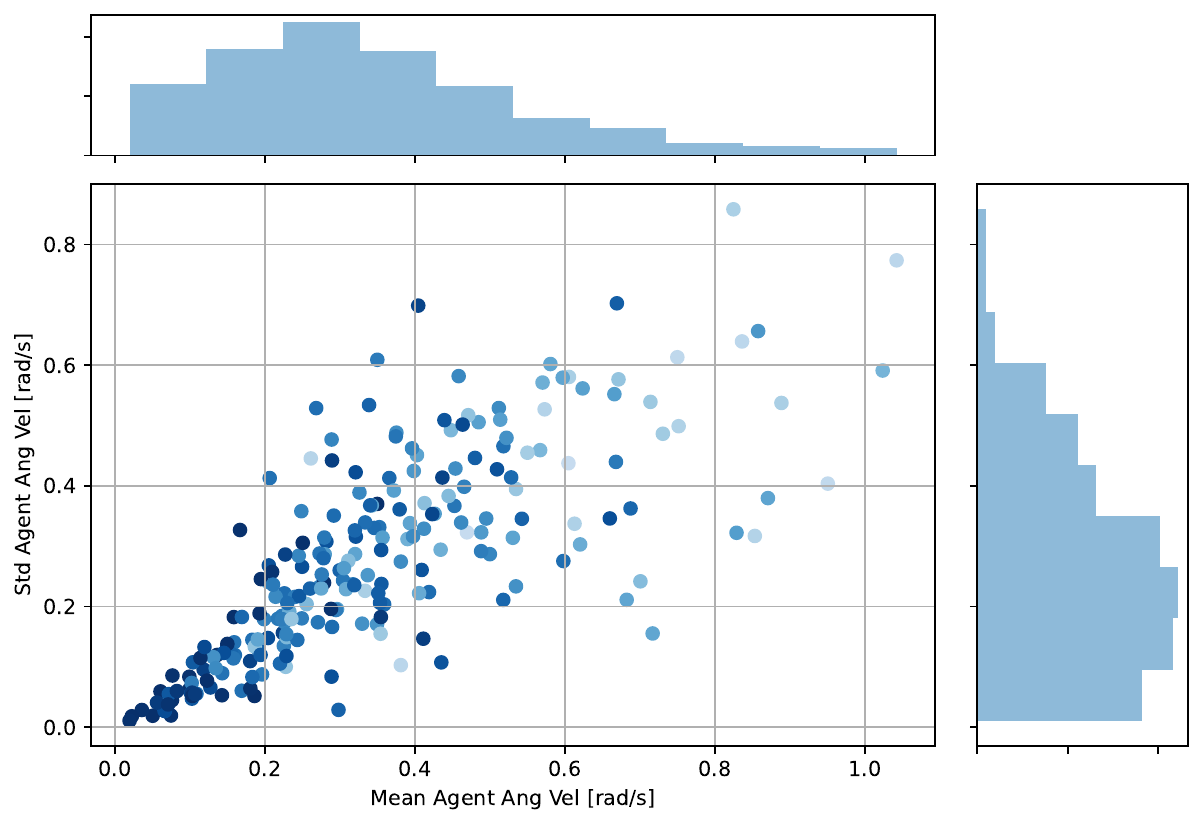}
        \caption{}
        \label{fig:dark_scatter_angv}
    \end{subfigure}

    \caption{Characterization of the motion \textit{Euglena} specimens in a dark environment. 
    (a) Average speed and average absolute angular velocity of the agents.
    (b) Average and standard deviation of the speed of the agents.
    (c) Average and standard deviation of the absolute angular velocity of the agents.
    Each dot represents one agent and is colored according to the mean speed of the agent, so that darker points correspond to faster agents.
    }
    \label{fig:dark_scatter}
\end{figure}

\subsection{The influence of light}
We then characterized the response of the microorganisms to light inputs.
Figure \ref{fig:light_inputs_frequency} shows the behaviour of \textit{Euglena} with light inputs of different duration, respectively 60s, 10s and 5s, and spatially uniform (see Section \ref{sec:experiments} for details about the experiments).
The time evolution of the quantities of interest clearly shows a rapid (i.e. in approximately one second) decrease of the speed and a simultaneous increase of the absolute angular velocity when the blue light is switched on.
In the following $\sim$10\,s both quantities return back, close to the original values, showing a strong adaptation capability of this species.
This implies that a constant illumination or very fast bursts (i.e. during a second or less) do not cause large effects in the long term, while inputs with intermediate duration (i.e. during between 5 and 10 seconds)  do (see Figure \ref{fig:varying_input_duration}).
These effects are compatible with the well known step-up photophobic response of \textit{Euglena} to blue light (see Table \ref{tab:microorganisms}), and, potentially, weaker negative photokinesis and positive photoklinokinesis.

\begin{figure}[t]
    \centering
    \vspace{-0.6cm}
    \begin{subfigure}[t]{1\textwidth}
        \centering
        \includegraphics[trim={0 8 0 0},clip,width=0.97\textwidth]{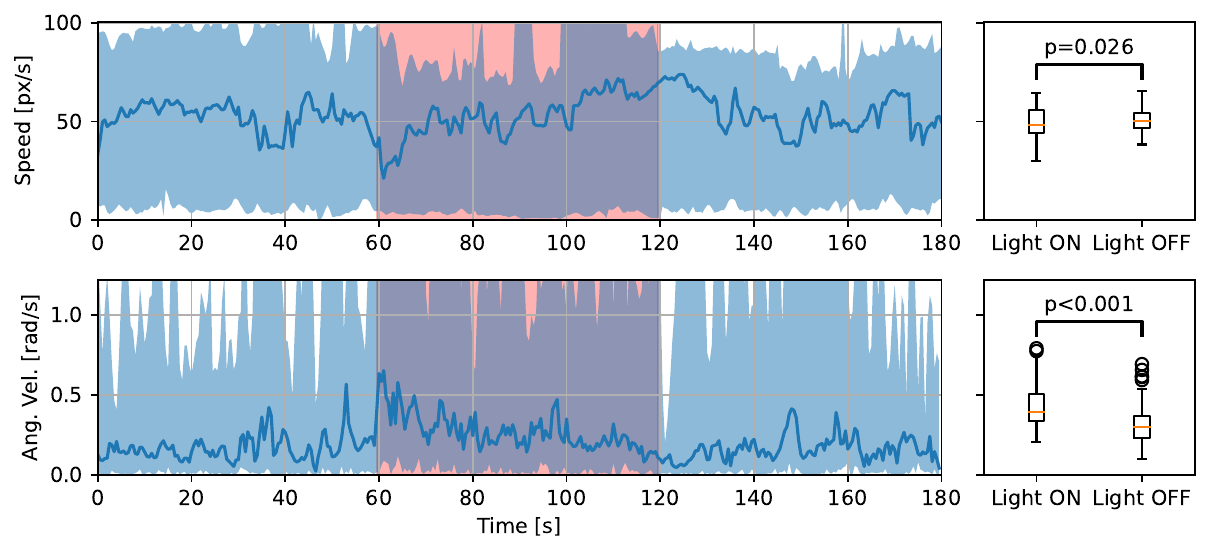}
        \vspace{-0.2cm}
        \caption{Intensity 100\% experiment with 60s activation.}
        \label{fig:light_inputs_255ON}
    \end{subfigure}

    \begin{subfigure}[t]{1\textwidth}
        \centering
        \includegraphics[trim={0 8 0 0},clip,width=0.97\textwidth]{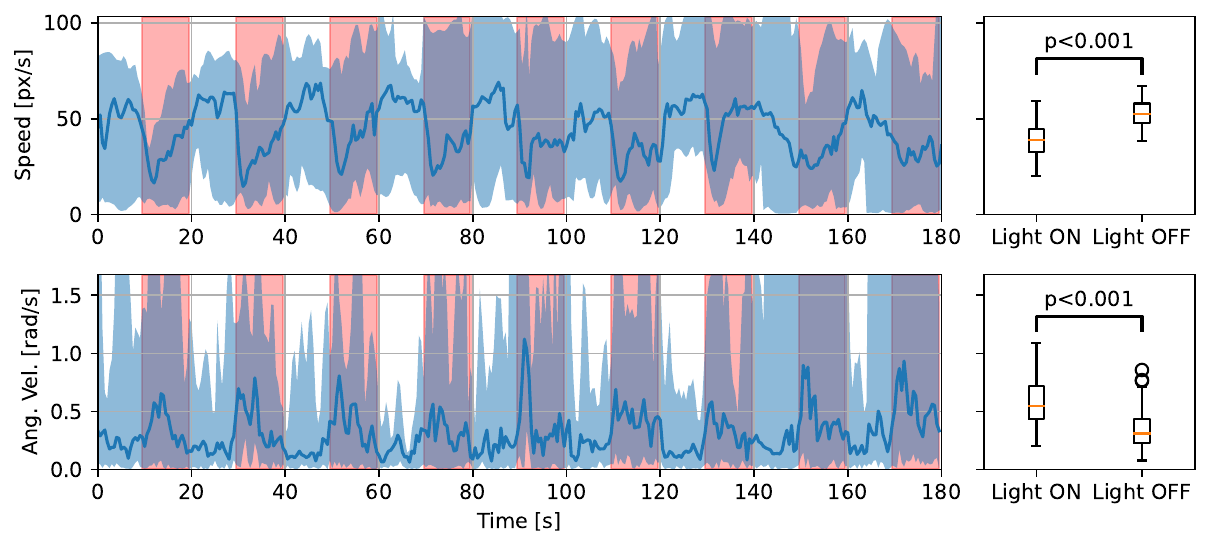}
        \vspace{-0.2cm}
        \caption{Switch 10s experiment.}
        \label{fig:light_inputs_switch10}
    \end{subfigure}
    
    \begin{subfigure}[t]{1\textwidth}
        \centering
        \includegraphics[trim={0 8 0 0},clip,width=0.97\textwidth]{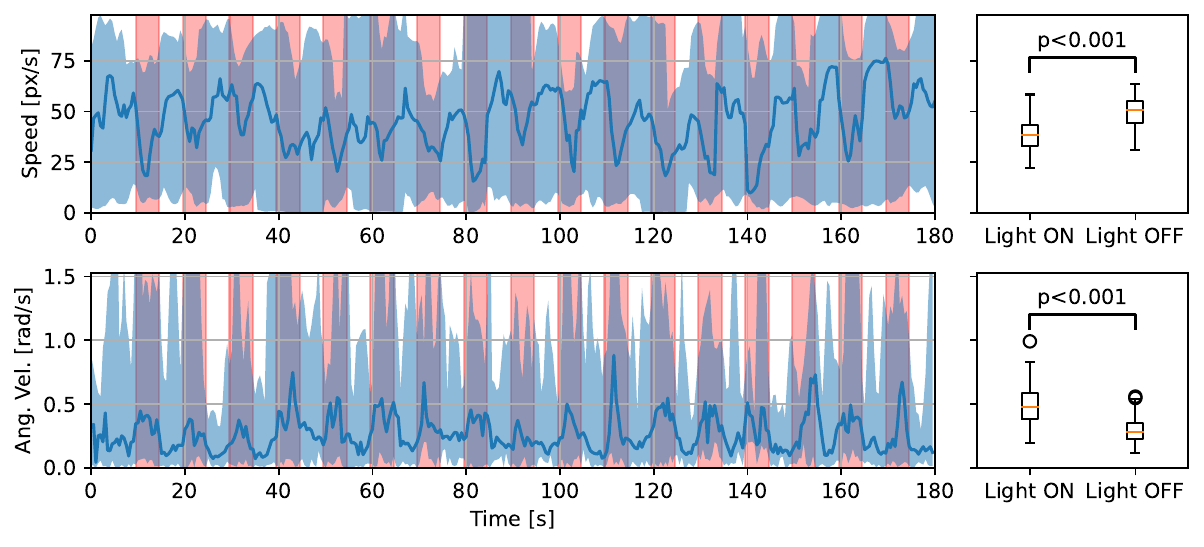}
        \vspace{-0.2cm}
        \caption{Switch 5s experiment.}
        \label{fig:light_inputs_switch5}
    \end{subfigure}
    \vspace{-0.3cm}
    \caption{Response of \textit{Euglena} to light inputs of different duration. The red bands indicate the periods in which the light inputs were active. Boxplots show the distribution of the mean speeds (or angular velocity) of the agents. Statistical significance is computed by Mann-Whitney U test \cite{Mann1947}.}
    \label{fig:light_inputs_frequency}
\end{figure}

We then studied how the light intensity affects the response of \textit{Euglena}.
Figure \ref{fig:light_inputs_intensity} shows the effect of light inputs with different intensities, respectively 100\%, 60\% and 30\%, and spatially uniform.
The results suggest that the intensity of light do not affect significantly the response, as even the minimum tested intensity is sufficient to trigger the photophobic response.
Figure \ref{fig:varying_input_intensity} further supports this finding, and shows that, regardless of input intensity, long exposure does not cause significant effects in the long term.
\begin{figure}[t]
    \centering
    \vspace{-0.6cm}
    \begin{subfigure}[t]{1\textwidth}
        \centering
        \includegraphics[trim={0 8 0 0},clip,width=0.97\textwidth]{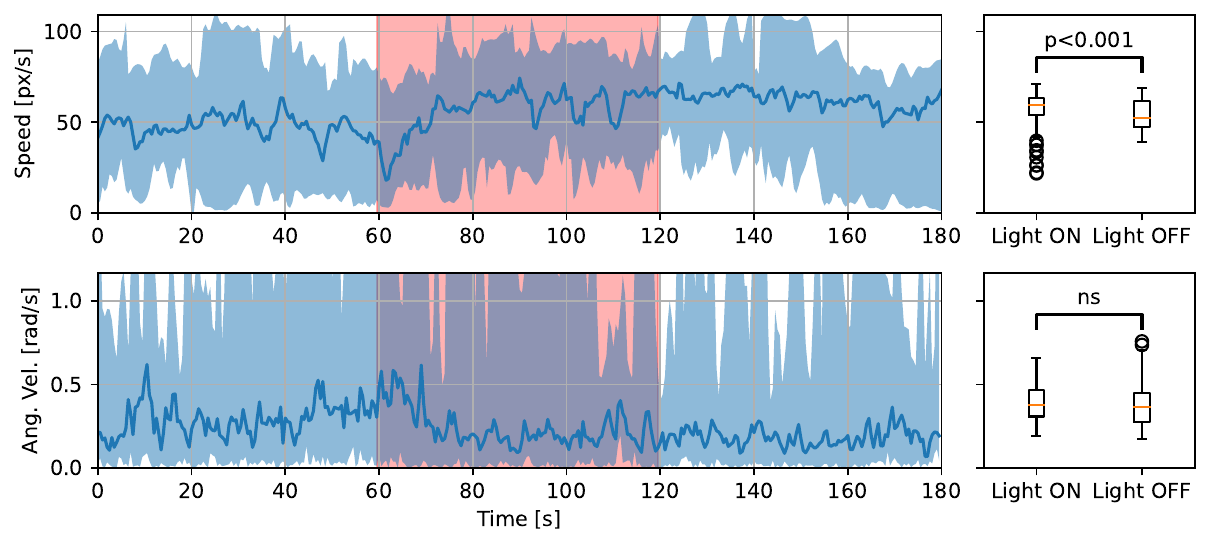}
        \vspace{-0.2cm}
        \caption{Intensity 30\% experiment.}
        \label{fig:light_inputs_75ON}
    \end{subfigure}

    \begin{subfigure}[t]{1\textwidth}
        \centering
        \includegraphics[trim={0 8 0 0},clip,width=0.97\textwidth]{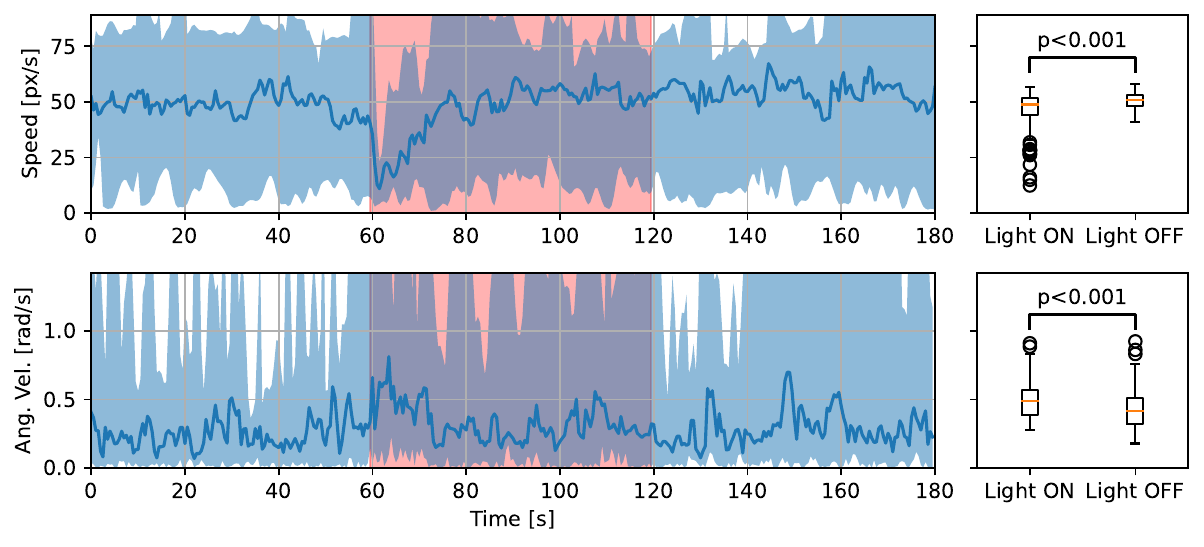}
        \vspace{-0.2cm}
        \caption{Intensity 60\% experiment.}
        \label{fig:light_inputs_150ON}
    \end{subfigure}
    
    \begin{subfigure}[t]{1\textwidth}
        \centering
        \includegraphics[trim={0 8 0 0},clip,width=0.97\textwidth]{figs/Euglena_255ON_time_evolution_boxplots.pdf}
        \vspace{-0.2cm}
        \caption{Intensity 100\% experiment.}
        \label{fig:light_inputs_255ON_int}
    \end{subfigure}
    \vspace{-0.3cm}
    \caption{Response of \textit{Euglena} to light inputs of different intensities. The red bands indicate the periods in which the light inputs were active. Boxplots show the distribution of the mean speeds (or angular velocity) of the agents. Statistical significance is computed by Mann-Whitney U test \cite{Mann1947}.
    }
    \label{fig:light_inputs_intensity}
\end{figure}
\begin{figure}[t]
    \centering
    \vspace{-0.5cm}
    \begin{subfigure}[t]{0.45\textwidth}
        \centering
        \includegraphics[trim={0 0 0 20},clip,width=1\textwidth]{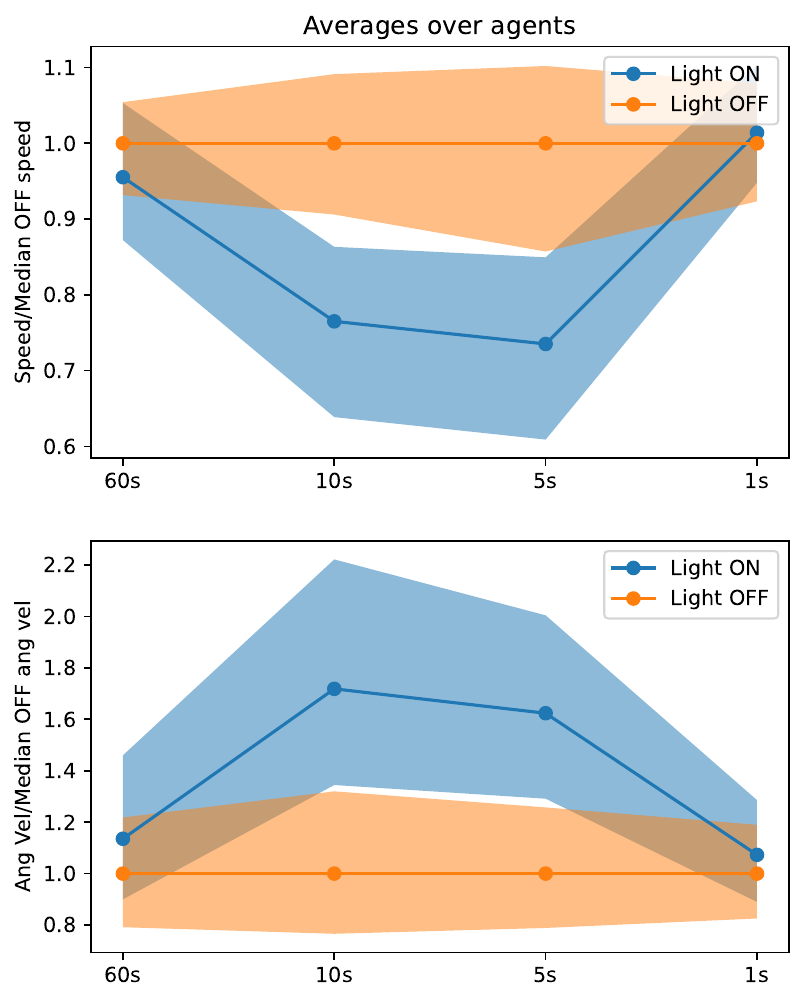}
        \caption{Varying input duration.}
        \label{fig:varying_input_duration}
    \end{subfigure}
    \begin{subfigure}[t]{0.46\textwidth}
        \centering
        \includegraphics[trim={0 0 0 20},clip,width=1\textwidth]{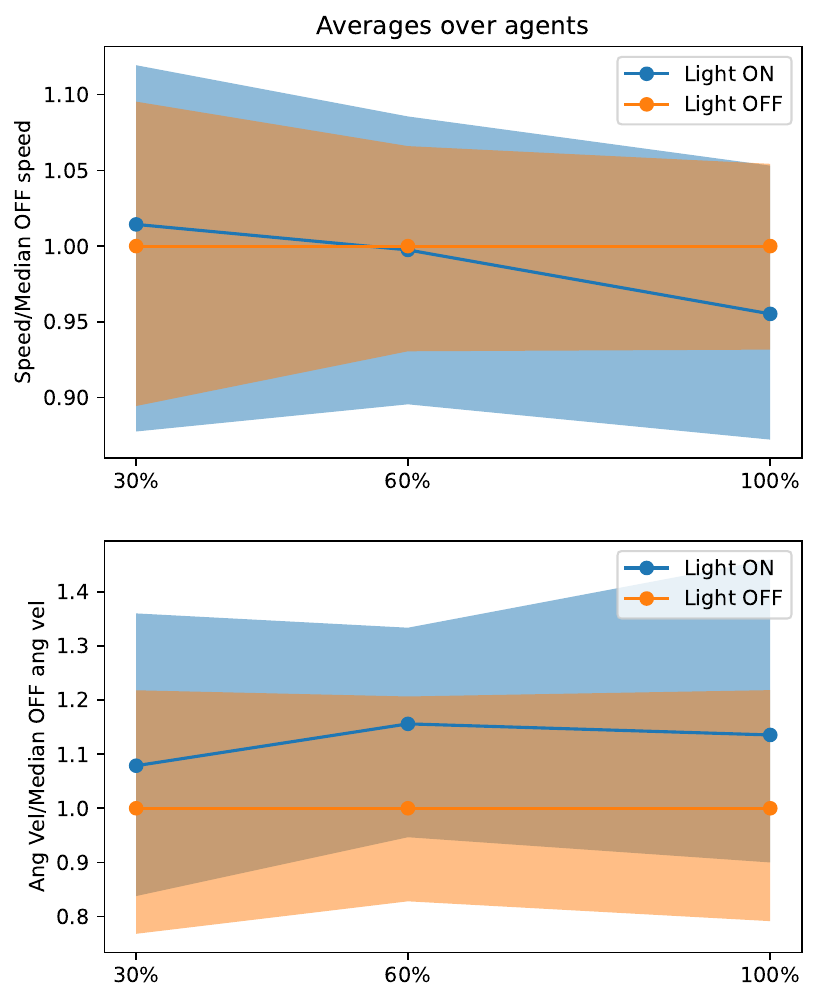}
        \caption{Varying input intensity.}
        \label{fig:varying_input_intensity}
    \end{subfigure}
    \caption{Response of \textit{Euglena} to light inputs of different (a) duration and (b) intensities. 
    The plots are obtained by aggregating the results from three replicates for each case. For each experiment the average over the agents is computed, and normalized on the median value recorder when the light inputs are absent.
    The solid lines show the median, while the shaded areas show the range between first and third quartiles.
    }
    \label{fig:varying_input}
\end{figure}
These results are summarised in Table \ref{tab:results}, where can be easily observed that switching inputs generate the largest effects, with up to 25\% speed decrease and 67\% increase in the angular velocity.

\begin{table}
    \centering
    \begin{tabular}{p{2.3cm}p{1.5cm}p{1.5cm}p{0.9cm}p{1.5cm}p{1.5cm}p{0.8cm}}
        \hline
         \multirow{2}{*}{Scenario} & \multicolumn{3}{l}{Speed [px/s]} & \multicolumn{3}{l}{Angular velocity [rad/s]} \\
           & OFF & ON & & OFF& ON &\\
        \hline
        No input &48.1±25.7	&---------	&	&0.35±0.40	&---------	&  \\
        30\% Intensity  &49.6±27.3	&49.6±27.0	&+0\%	&0.40±0.47	&0.43±0.47	&+6\%	\\
        60\% Intensity &46.6±27.3	&46.2±26.9	&-1\%	&0.40±0.47	&0.46±0.49	&+14\%\\
        100\% Intensity &51.4±27.9	&49.5±26.9	&-4\%	&0.34±0.43	&0.39±0.45	&+14\% \\
        Switch 1s &45.8±26.0	&46.8±25.9	&+2\%	&0.41±0.44	&0.44±0.46	&+7\%	 \\
        Switch 5s &46.9±29.0	&35.0±24.9	&-25\%	&0.38±0.50	&0.59±0.56	&+57\%  \\
        Switch 10s &51.4±28.2	&39.2±26.7	&-24\%	&0.32±0.43	&0.53±0.55	&+67\%  \\
        \hline
    \end{tabular}
    \caption{Experimental results on the movement of \textit{Euglena}. Data are obtained aggregating the results of three biological replicates per case. Shown values are mean ± standard deviation and, for light ON case, percentage variation of the mean value with respect to the light OFF case.}
    \label{tab:results}
\end{table}

\clearpage
\section{Discussion}
In this Chapter we introduced our experimental setup, and characterized the movement and the light response of \textit{Euglena}.
We started by describing the \ac{DOME}, an innovative experimental platform that was crucial for our research. In particular, the ability to stimulate the microorganisms with light inputs allowed to study their response.
Then, we introduced the four species of motile microorganisms used in our study, and the series of experiment designed to investigate their response to light stimuli.
Finally, we analysed the kinetics of the movement of \textit{Euglena}, studying their behaviour in a dark environment and the modifications induced by light.
In particular, we observed these microorganisms to strongly react to sudden luminosity changes by slowing down and increasing their turning rate.
These results will be used in the next Chapter to direct the modelling of the motion of these microorganisms.

\chapter{Data-driven modelling}
\thispagestyle{empty} 
\label{ch:modelling}

Mathematical models allow to achieve better understanding of the experimental results and new insights on the studied system, trough both simulation or analysis.
Moreover, dynamical models are a fundamental tool for the design of control strategies.
Nevertheless, as discussed in Chapter \ref{ch:miroorganisms_background}, general approaches to model the movement and the light response of microorganisms are currently lacking, hampering the development of effective control strategies.
Therefore, this Chapter will introduce our data-driven approach to model stochastic mobile agents, which also allows to include inputs and might be used to describe the movement of a variety of microorganisms, together with their response to light. 
Specifically we will use the data presented in Section \ref{sec:analysis} to parameterize and validate the model, using \emph{Euglena} as a representative example.
Finally, we will present some preliminary results in the direction of controlling the spatial distribution of microorganisms.

\section{Stochastic model of motion}
Models of dynamical systems usually consist of  \acp{ODE}, nevertheless, these descriptions cannot grasp the nature of highly stochastic behaviours, as often found in biology.
In Section \ref{sec:microorganisms_movement}, we discussed two of the most common approaches to model the stochastic movement of swimming organisms, namely the Lévy walk \cite{Berg1993} and the \acf{PTW} \cite{Gautrais2009}.
Given that we observed \textit{Euglena} to be continuously moving and turning for most of the time, that we are only interested in modelling 2D trajectories, and that the \ac{PTW} model can be easily adapted to include inputs, we adopted this modelling framework.
Specifically, we use two \acfp{SDE} to model, respectively, the time evolution of the speed $v_i$ and the angular velocity $\omega_i$ \cite{Zienkiewicz2015}.
So that, for each agent $i$, we have
%
\begin{subnumcases}{\label{eq:SDEmodel}}
    \d v_i= [\theta_{v,i} (\mu_{v,i}-v_i) + \alpha_{v,i} u_i + \beta_{v,i} \max(\dot{u}_i,0) ] \d t + \sigma_{v,i} \d W_{v,i} \label{eq:SDEmodel_v}\\
    \d \omega_i= [-\theta_{\omega,i} \omega_i +\!\func{sign}{\omega_i} (\alpha_{\omega,i} u_i + \beta_{\omega,i} \max(\dot{u}_i,0))] \d t + \sigma_{\omega,i} \d W_{\omega,i} \label{eq:SDEmodel_w}
\end{subnumcases}

where $\theta_{\bigcdot,i}$ represents the \emph{rate}, $\mu_{\bigcdot,i}$ the \emph{mean}, $\sigma_{\bigcdot,i}$  the \emph{volatility}, $\alpha_{\bigcdot,i}$ and $\beta_{\bigcdot,i}$ the input gains of the corresponding equation; while $u_i$ is the intensity of the light input measured by agent $i$ and $\dot{u}_i$ its time derivative; finally, $W_{\bigcdot,i}$ represent independent standard Wiener processes.
Notice that we assumed zero mean for the angular velocity (i.e. $\mu_{\omega,\bigcdot}=0$), giving the agents no preference between left or right turns, but, for all the other parameters, we are allowing different values for each agent.

Concerning the agents' reaction to light inputs, we included two terms in each equation, describing the effects of light intensity and its time derivative.
In particular, by using the term $\max(\dot{u}_i,0)$ we only considered the effect of positive variations of the inputs, that can indeed trigger the step-up photophobic response of these microorganisms%
\footnote{Here we decided to only model the step-up photophobic response, nevertheless the extension to also include the step-down response is straightforward. It can be implemented by inserting an additional term $+\gamma_{\bigcdot,i}\min(\dot{u}_i,0)$ in both equations.}.
Moreover, the term $\func{sign}{\omega_i}$ in \eqref{eq:SDEmodel_w} keeps the consistency between the parameters $\alpha_{\omega,i}$ and $\beta_{\omega,i}$, describing the variation of the absolute value of $\omega_i$, and $\omega_i$ itself, being either positive or negative.
All together this model allows to capture any combination of photokinetic (trough $\alpha_{v,i}$), photoklinokinetic (trough $\alpha_{\omega,i}$), and step-up photophobic responses (trough $\beta_{v,i}$ and $\beta_{\omega,i}$) (see Section \ref{sec:microorganisms_light_response} for more information about possible light responses of microorganisms).

Compared with the original \ac{PTW} model \cite{Gautrais2009}, which assumed constant speed, our model can capture richer dynamics, allowing for fluctuations of the longitudinal speed.
This improvement was already proposed in \cite{Zienkiewicz2015}, which also introduced an explicit dependency between the speed and the angular velocity.
While, for the sake of simplicity, we assumed the two variables to be independent, we introduced input terms in both equations \eqref{eq:SDEmodel}. This, allows our model to capture the influence of external factors, such as light, that can, indeed, induce a coupling between speed and angular velocity.
\new{A similar model, specifically designed to capture the step-up photophobic response of \textit{Euglena g.}, was proposed in \cite{Tsang2018}. This, while describing 3D movement and some light responses, assumes constant longitudinal speed, so that some crucial aspects, such as velocity fluctuations and photokinesis, cannot be captured.}

\section{Model parametrization}
Model \ref{eq:SDEmodel} is linear in the parameters, allowing for a relatively straightforward parametrization. 
Therefore, we extended the calibration technique described in \cite{vanDenBerg2011} to identify the values of the parameters.


We need to identify the parameters in both \eqref{eq:SDEmodel_v} and \eqref{eq:SDEmodel_w}, which have the same structure,
\begin{equation}
    \d x= [\theta (\mu-x) + \alpha u] \d t + \sigma \d W.
    \label{eq:SDE}
\end{equation}
Assume we can acquire data points $\{x_1,x_2,\dots,x_n\}$ and $\{u_1,u_2,\dots,u_n\}$ from $x$ and $u$, with a sampling time step $\Delta T$.
Then, \ac{OLS} regression \cite{freedman2009} can be used to fit the straight line $x_{k+1}=a x_k+b u_k +c$ to these data, obtaining estimates for the parameters $a$, $b$ and $c$, and the residuals $\epsilon_k=x_{k+1} - (a x_k+b u_k + c)$.

Also, discretizing \eqref{eq:SDE} we have
\begin{equation}
    x_{k+1}=e^{-\theta \Delta T} x_{k} + \mu(1-e^{-\theta \Delta T}) 
    + \int_0^{\Delta T} e^{-\theta \tau} b \, \d \tau \, u_k
    + \sigma \sqrt{\frac{1-e^{-2\theta \Delta T}}{2\theta}} N_{0,1},
    \label{eq:SDEdiscretization}
\end{equation}
where $N_{0,1}$ is a normally distributed random variable.
By comparing it to the equation of the fitting line one gets the parameters of the original continuous time \ac{SDE}, namely
\begin{align}
    \theta &= -\frac{\ln(a)}{\Delta T}, \\
    \mu &=\frac{c}{1-a}, \\
    \alpha &= \frac{\ln(a)}{\Delta T(a-1)} b\\
    \sigma &= \func{std}{\epsilon}\sqrt{\frac{-2\ln(a)}{(1-a^2)\Delta T}}.
\end{align}

Given the time series from an experiment we apply this procedure to calibrate both the equations \eqref{eq:SDEmodel}.
Specifically, for each agent $i$ and independently for the recorded speed $v_i$ and the absolute angular velocity $\abs{\omega_i}$, we identified the parameters $\theta_{\bigcdot,i}$, $\mu_{\bigcdot,i}$, $\alpha_{\bigcdot,i}$, $\beta_{\bigcdot,i}$, and $\sigma_{\bigcdot,i}$.
As inputs $u_i$ and $\dot{u}_i$ we used, for all the agents, the intensity of blue light given by the projector, divided by 255 to be in the range [0; 1], and its approximate time derivative (numerical differentiation is described in \eqref{eq:numerical_differentiation}).
We then discarded all the agents for which $\theta_{v,i}$ or $\theta_{\omega,i}$ was negative, because that would result in unstable dynamics.
Moreover, we discarded all the agents for whom any of the identified parameters represents an outlier in the set of values of that parameter (according to Algorithm \ref{alg:detect_outliers} and with a threshold $m=5$).
This ensures robustness of the identification procedure, which involves highly non-linear functions and would otherwise be susceptible to noise or uncertainties in the data.
Notice that this procedure can also identify $\mu_{\omega,i}$, nevertheless we assumed the angular velocity to have zero mean, hence this parameter will not be used.

Figure \ref{subfig:identification_parameters_boxplot} shows the resulting estimated values of the parameters for a representative experiment with light inputs switching every 10\,s.
We can observe a good qualitative agreement between these data and the results presented in Section \ref{sec:analysis}, in particular $\alpha_v$ being mostly negative describes the observed negative photokinesis, similarly $\alpha_\omega$ being mostly positive describes the observed positive photoklinokinesis, while $\beta_v$ and $\beta_\omega$ being respectively negative and positive capture the photophobic response.
Figure \ref{subfig:identification_parameters_corrplot} shows the pairwise correlation between the estimated values of the parameters.
Notice that, while most pairs of parameters show weak correlation, the relations observed in the experimental data (see Section \ref{sec:analysis}), such as negative correlation between speed and angular velocity (see $(\mu_v,\mu_\omega)$), or positive correlation between the mean and the variability of the angular velocity (see $(\mu_\omega,\sigma_\omega)$), are here retrieved.
Moreover we can observe new properties, for example \new{faster agents are less responsive to light inputs (see the negative correlations} between $(\mu_v,\alpha_v)$ and $(\mu_v,\beta_v)$).

\begin{figure}[t]
    \vspace{-0.5cm}
    \centering 
    \begin{subfigure}[t]{1\textwidth}
        \includegraphics[trim={80 440 80 50},clip,width=0.48\textwidth]{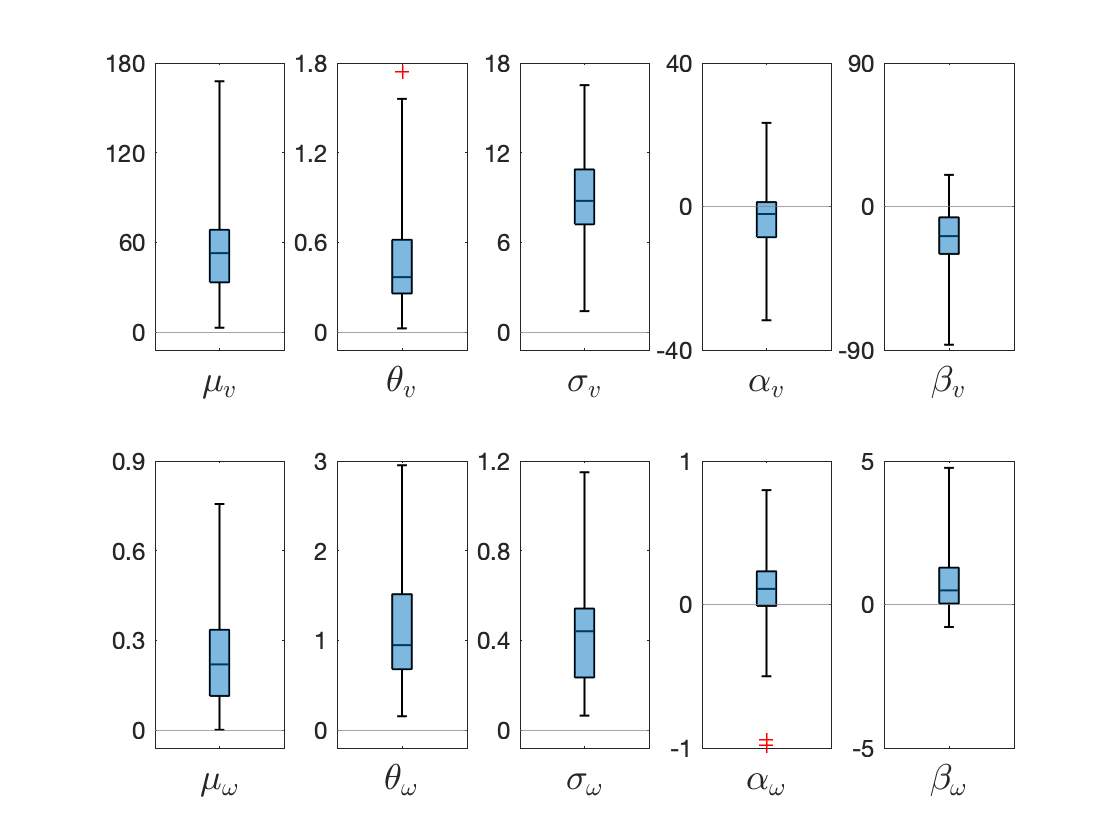} 
        \includegraphics[trim={80 40 80 400},clip,width=0.48\textwidth]{figs/Euglena_input_identified_parameters_boxplots.png} 
        \caption{Distribution of the estimated parameters.}
        \label{subfig:identification_parameters_boxplot}
        \vspace{0.5cm}
    \end{subfigure}
    \begin{subfigure}[t]{1\textwidth}
        \includegraphics[trim={150 100 150 140},clip,width=1\textwidth]{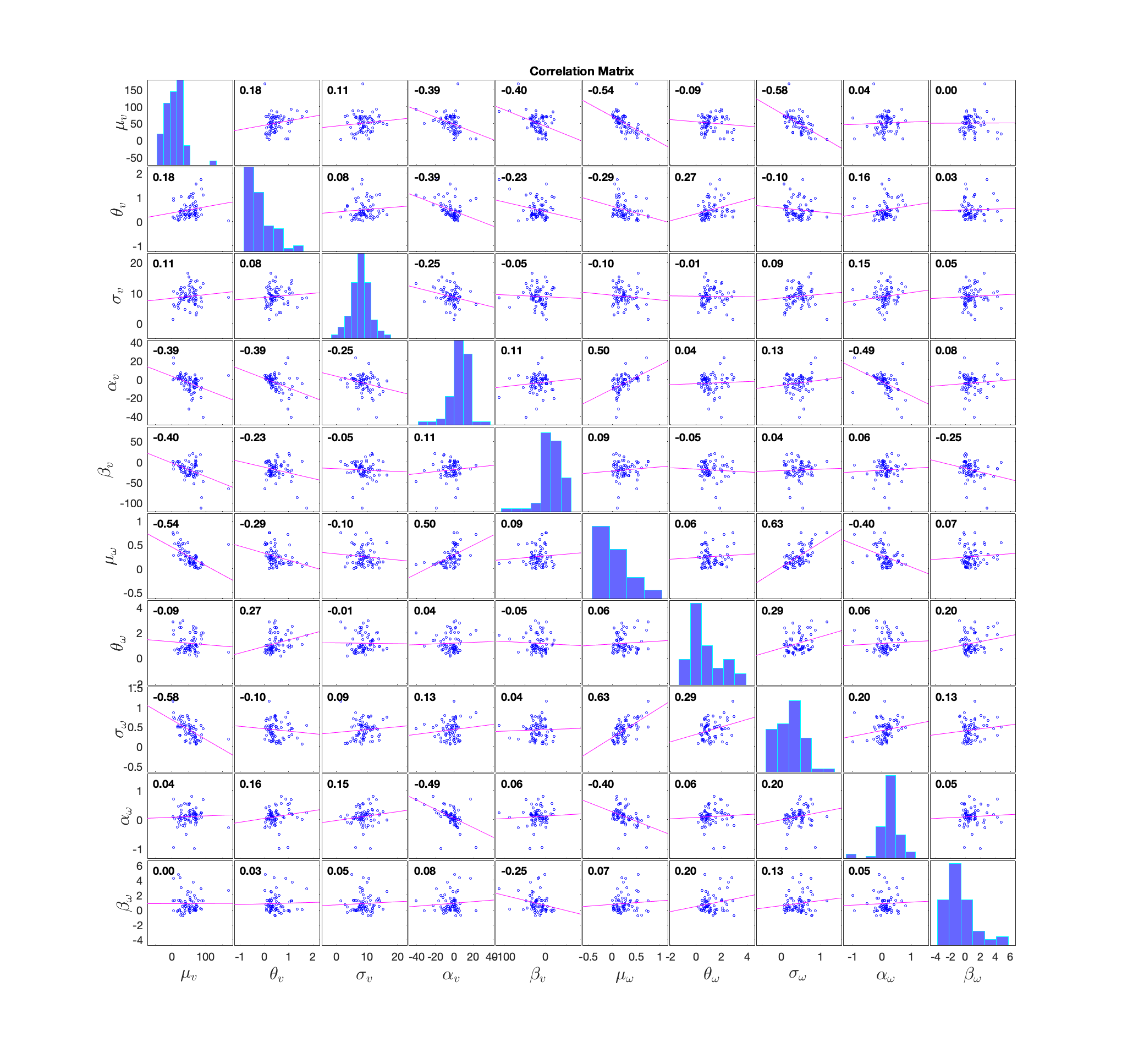} 
        \caption{Pairwise correlation between the estimated parameters.}
        \label{subfig:identification_parameters_corrplot}
    \end{subfigure}
    \caption{Estimated values of the parameters (after the selection procedure) (a), and their distribution and pairwise Pearson correlation coefficient (b). 
    The data comes from an experiment with light inputs switching every 10\,s.
    Input: 78 time series of speed and absolute angular velocity.
    Output: 71 valid sets of parameters.
    }
    \label{fig:identification_parameters}
\end{figure}

\section{Numerical validation}
To validate our model \eqref{eq:SDEmodel} and the parametrization procedure we implemented them in the agent based simulator we developed, called SwarmSim (for more information about SwarmSim see Appendix \ref{ch:swarmsim} or visit \url{www.github.com/diBernardoGroup/SwarmSimPublic}).
Specifically, starting from the trajectories recorded during an experiment, we applied the identification procedure to all $N'$ detected agents, obtaining a collection of $N''\leq N'$ valid sets of parameters.
By randomly sampling sets of parameters from this collection we can instantiate $N$ virtual agents%
\footnote{By allowing for repetitions, $N$ can be larger than $N''$.}
in our simulator, and then compare the synthetic data to the experimental ones.
Figure \ref{fig:identification_validation_trajectories} shows real and simulated trajectories, which, indeed appear to be qualitatively similar.
Moreover, Figure \ref{fig:identification_validation} shows the comparison between experimental and synthetic data, for two experiments, one with no inputs and the other in the presence of switching inputs.
Both simulations are obtained by simulating the agents identified from the experimental data with switching inputs. 
The remarkable similarity between experimental and simulated data, also validated by Mann-Whitney U statistical test \cite{Mann1947} (see Figures \ref{subfig:identification_validation_dark_box} and \ref{subfig:identification_validation_input_box}), confirms the effectiveness of both the model \eqref{eq:SDEmodel} and the identification procedure described in the previous pages.
In particular Figure \ref{subfig:identification_validation_input_time} shows how the simulated model reproduces the complex light response observed in the real organisms, made of a combination of step-up photophobic, photokinetic and photoklinokinetic responses, followed by an adaptation in the next $\sim$10\,s.

\begin{figure}[t]
    \centering   
    \begin{subfigure}[t]{0.49\textwidth}
        \centering
        \includegraphics[trim={0 0 0 0},clip,width=1\textwidth]{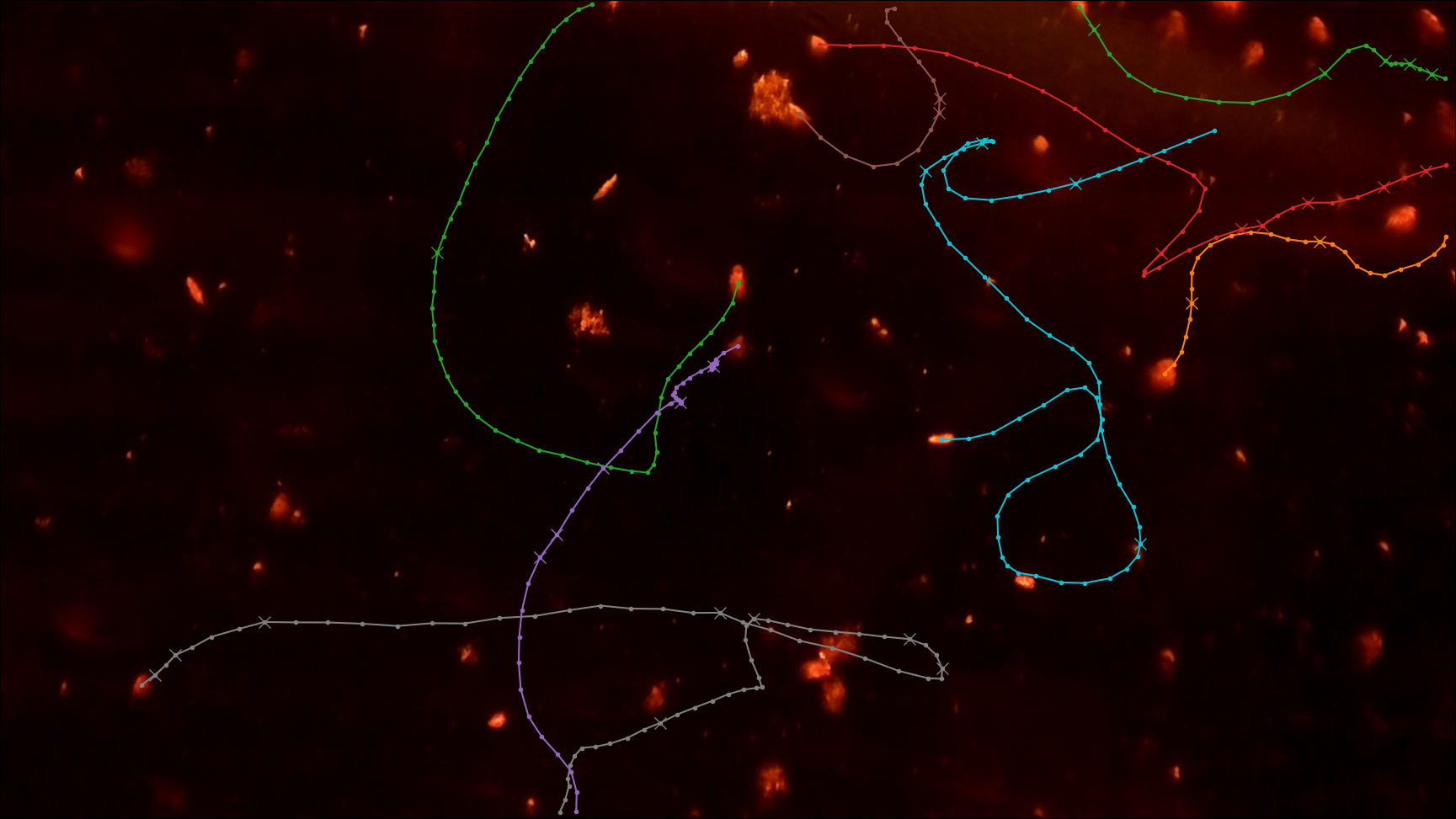} 
        \caption{}
    \end{subfigure} 
    \begin{subfigure}[t]{0.49\textwidth}
        \centering
        \includegraphics[trim={0 0 0 0},clip,width=1\textwidth]{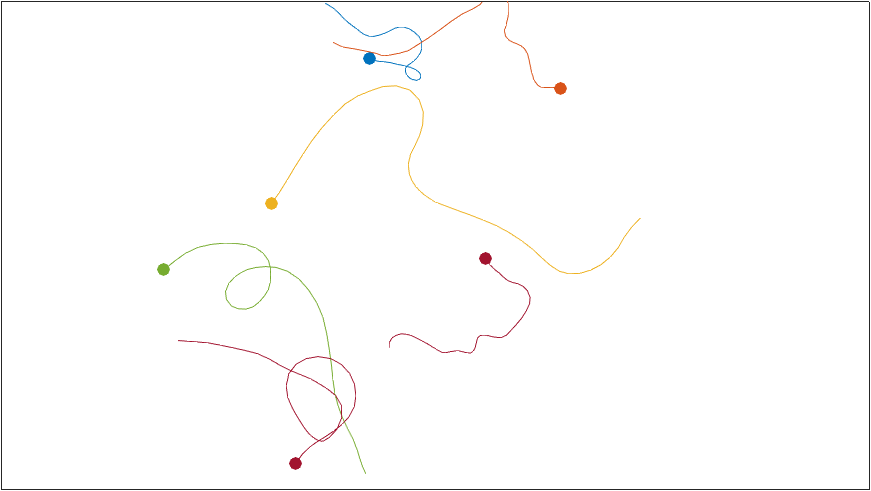} 
        \caption{}
    \end{subfigure} 
    
    \caption{Comparison between real (a) and simulated (b) trajectories of \textit{Euglena} in a dark environment.
    }
    \label{fig:identification_validation_trajectories}
\end{figure}

\begin{figure}[t]
\vspace{-0.6cm}
    \centering   
    \begin{subfigure}[t]{0.49\textwidth}
        \centering
        \includegraphics[trim={0 30 0 50},clip,width=1\textwidth]{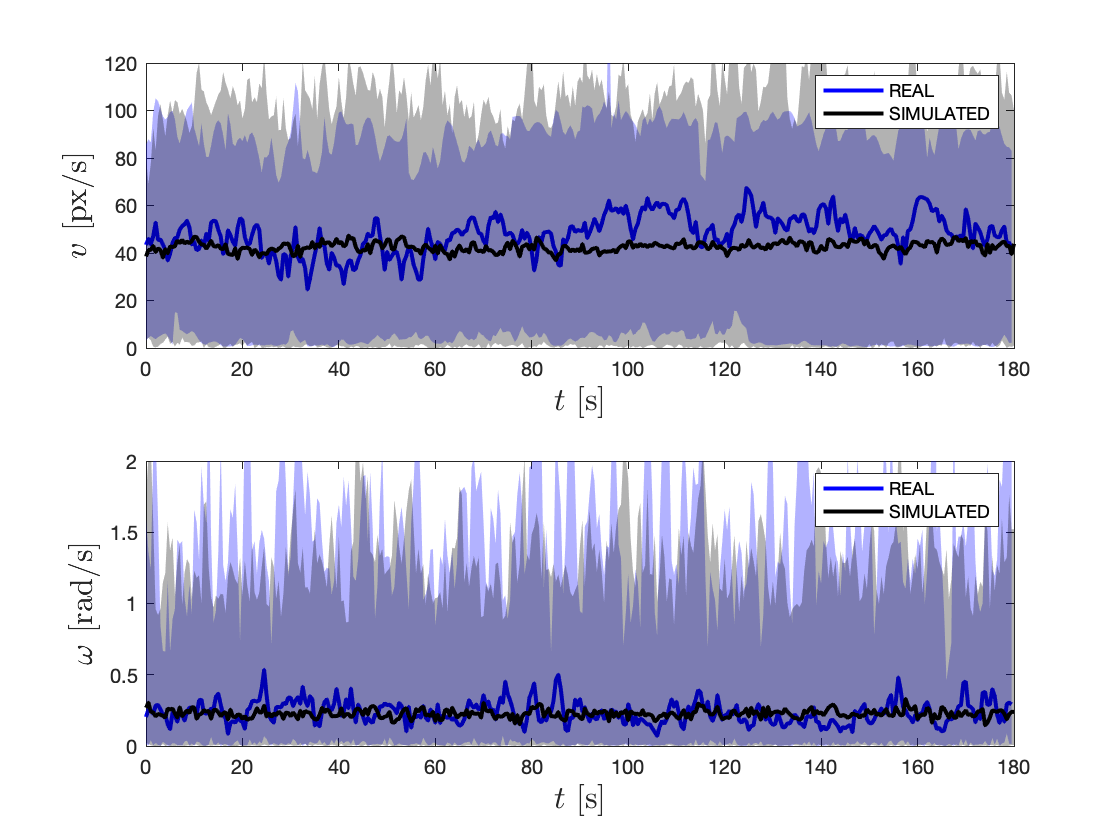} 
        \caption{}
        \label{subfig:identification_validation_dark_time}
    \end{subfigure} 
    \begin{subfigure}[t]{0.49\textwidth}
        \centering
        \includegraphics[trim={0 30 0 50},clip,width=1\textwidth]{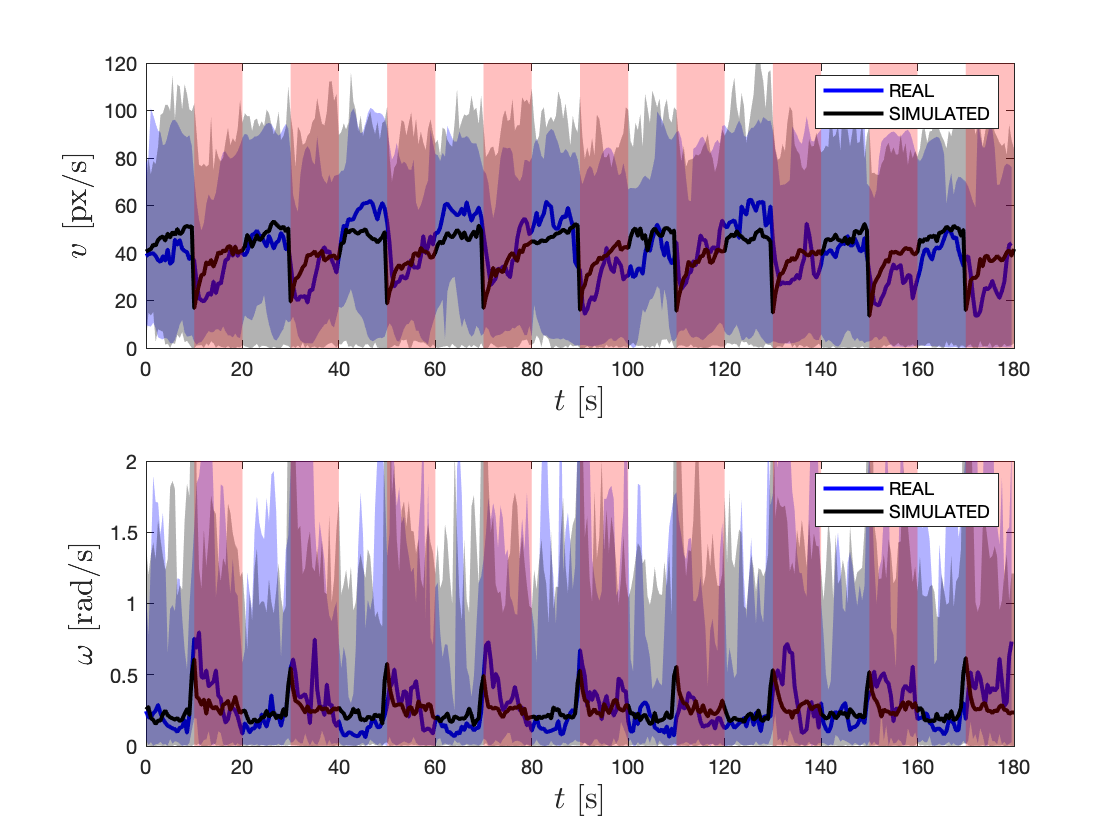} 
        \caption{}
        \label{subfig:identification_validation_input_time}
    \end{subfigure} 
    
    \begin{subfigure}[t]{0.45\textwidth}
        \centering
        \includegraphics[trim={0 20 0 0},clip,width=1\textwidth]{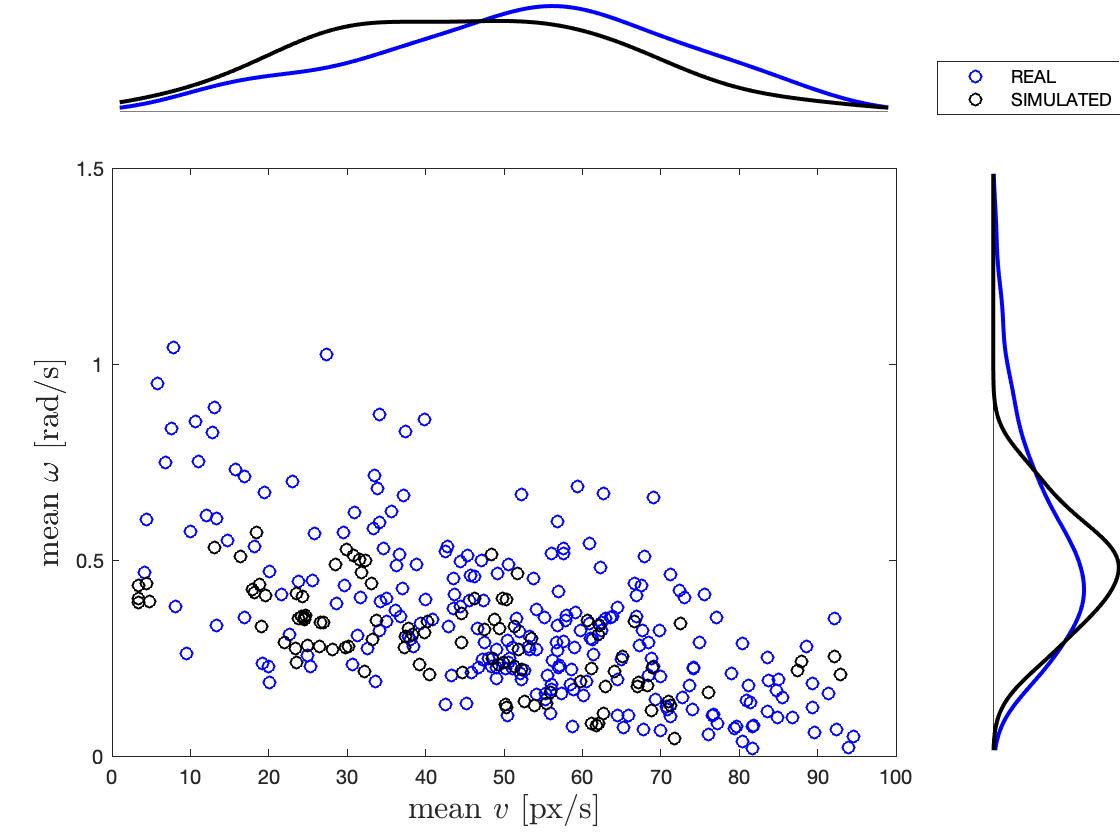} 
        \caption{}
        \label{subfig:identification_validation_dark_scatter}
    \end{subfigure} 
    \begin{subfigure}[t]{0.45\textwidth}
        \centering
        \includegraphics[trim={0 20 0 0},clip,width=1\textwidth]{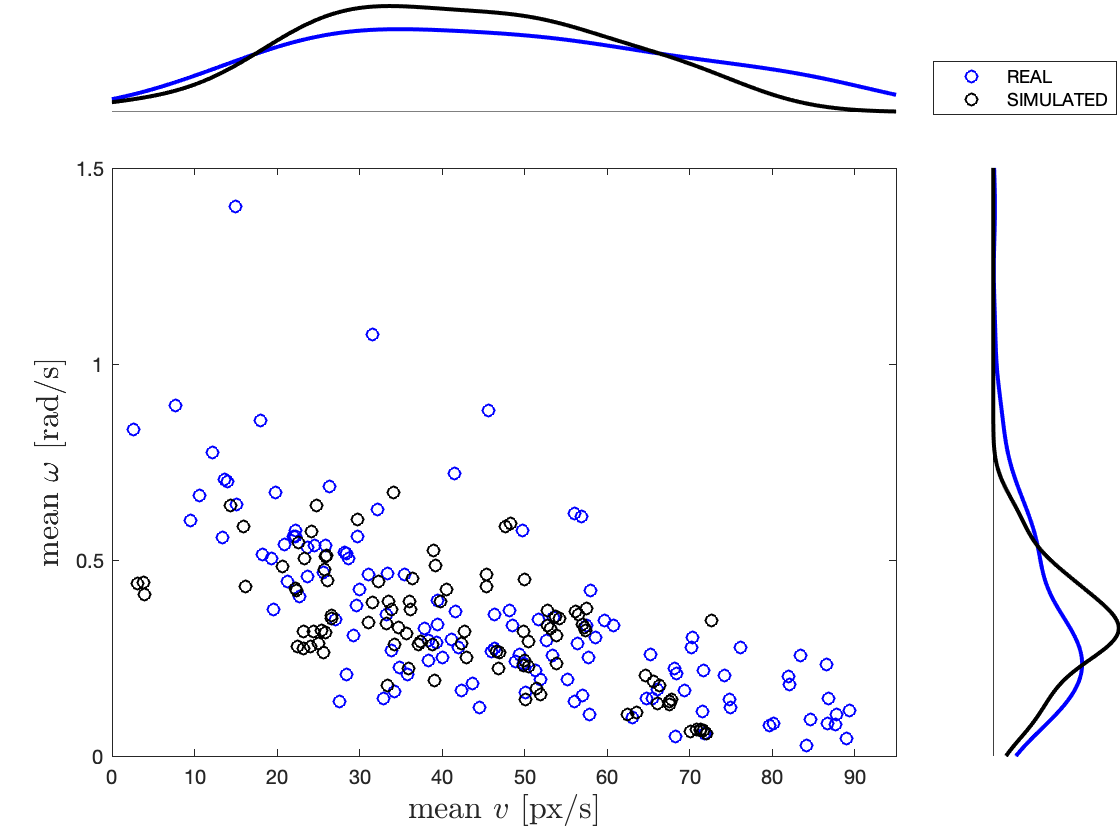} 
        \caption{}
        \label{subfig:identification_validation_input_scatter}
    \end{subfigure} 
    
    \begin{subfigure}[t]{0.3\textwidth}
        \centering
        \includegraphics[trim={0 55 0 72},clip,width=1\textwidth]{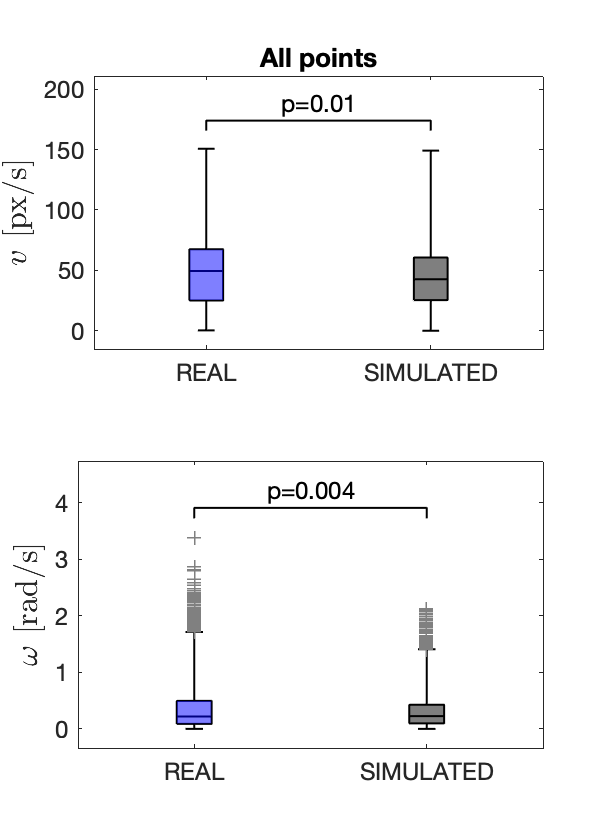}      \caption{}
        \label{subfig:identification_validation_dark_box}
    \end{subfigure}  
    \qquad
    \begin{subfigure}[t]{0.3\textwidth}
        \centering
        \includegraphics[trim={0 55 0 72},clip,width=1\textwidth]{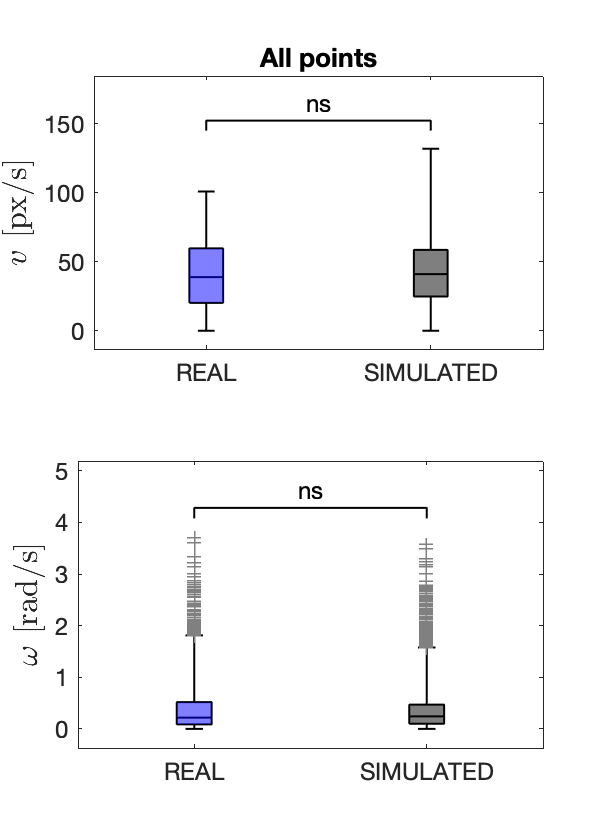}      \caption{}
        \label{subfig:identification_validation_input_box}
    \end{subfigure} 
    \vspace{-0.25cm}
    \caption{Comparison between experimental and synthetic data, for two experiments without inputs (a, c, e) and with switching input (b, d, f). 
    (a, b) Time evolution of speed and absolute angular velocity. Red shaded area indicate light input is on.
    (c, d) Each dot represents an agent, either real or simulated, and the position is given by its mean speed and absolute angular velocity.
    Marginal distributions are obtained by Kernel Density Estimation.
    (e, f) Data points, for either real or simulated agents.
    Statistical significance is computed by Mann-Whitney U test. 
    All real and simulated experiments lasted 180\,s, with a sampling time of 0.5\,s. 
    Both simulations use the virtual agents identified from the experimental data with switching inputs.
    }
    \label{fig:identification_validation}
\end{figure}

\clearpage
\section{Towards spatial control}
The final goal of our project is using the \ac{DOME} to implement spatial control of the studied microorganisms, according to the scheme in Figure \ref{fig:dome_control_scheme}.
Specifically, we will use the mathematical model developed in this Chapter to design a feedback control strategy able to steer the movement of the microorganisms, so to obtain the desired  distribution and achieving controlled density regulation.
Such control strategy may leverage the flexibility offered by the \ac{DLP} to combine spatial inputs with singe-cell ones, and exploit the knowledge provided by the model to achieve unseen performance.
Moreover, it will be easily applied to any species of microorganisms whose movement can be captured by our model. 

\begin{figure}[t]
    \centering   
    \includegraphics[trim={0 0 0 0},clip,width=1\textwidth]{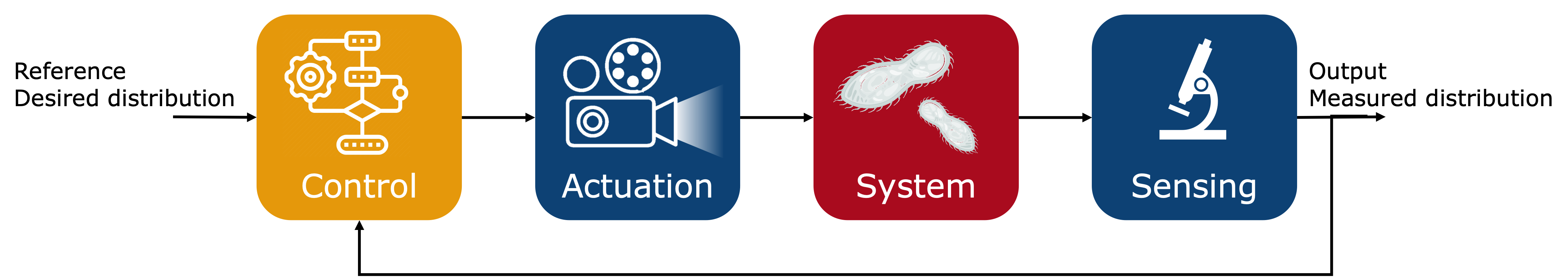} 
    \caption{Control scheme to implement spatial control in the \ac{DOME}.
    }
    \label{fig:dome_control_scheme}
\end{figure}

To this aim we run preliminary experiments, complementary to those described in Section \ref{sec:experiments}, to study the response of the microorganisms to spatial inputs.
These experiments are described in Table \ref{tab:spatial_experiments} and depicted in Figure \ref{fig:experiments_spatial_inputs}, while some representative results are shown in Figure \ref{fig:spatial_inputs_example}.
It can be easily seen that, as known, \textit{Euglena} shows a strong photodispersion behaviour, and something similar is observed for \textit{Volvox}.
Instead, the response of \textit{Paramecia} is less clear.
These observations allowed to test on \textit{Euglena} a simple open-loop approach, based on the static projection of a dark patter on a lit background.
This naive approach, while occasionally providing satisfactory results (see Figure \ref{fig:BCL_Euglena}), is limited to a single specie and proved very sensible to the experimental conditions.

\begin{figure}[t]
    \centering
    \includegraphics[trim={0, 0, 0, 320},clip,width=0.9\linewidth]{figs/experiments.pdf}
    \caption{Graphical representation of the experiments executed to characterize the spatial response of the microorganisms.}
    \label{fig:experiments_spatial_inputs}

    \vspace{0.5cm}
    \begin{tabular}{p{0.5cm}p{2.2cm}p{8cm}}
    \hline
     & Experiment & Description \\
    \hline
    \multirow{9}{*}{\rotatebox[origin=c]{90}{Spatial inputs}}
    \arrayrulecolor{lightgray}
         & Half-Half        & 10s OFF, half sample OFF, half ON.\\ \cline{2-3}
         & Gradient\newline lateral & 10s OFF, linear light gradient from one side of the sample to the other.\\ \cline{2-3}
         & Gradient center light & 10s OFF, linear light gradients from the sides of the sample to the center.\\ \cline{2-3}
         & Gradient center dark & 10s OFF, linear light gradients from the center of the sample to the sides.\\ \cline{2-3}
         & Circle light & 10s OFF, lit circle on dark background.\\ \cline{2-3}
         & Circle dark & 10s OFF, dark circle on lit background.\\
    \arrayrulecolor{black} \hline
    \end{tabular}
    \captionof{table}{Experiments executed to characterize the spatial response of the microorganisms.}
    \label{tab:spatial_experiments}
\end{figure}

\begin{figure}[t]
    \centering
    \begin{subfigure}[t]{0.48\textwidth}
        \includegraphics[width=\linewidth]{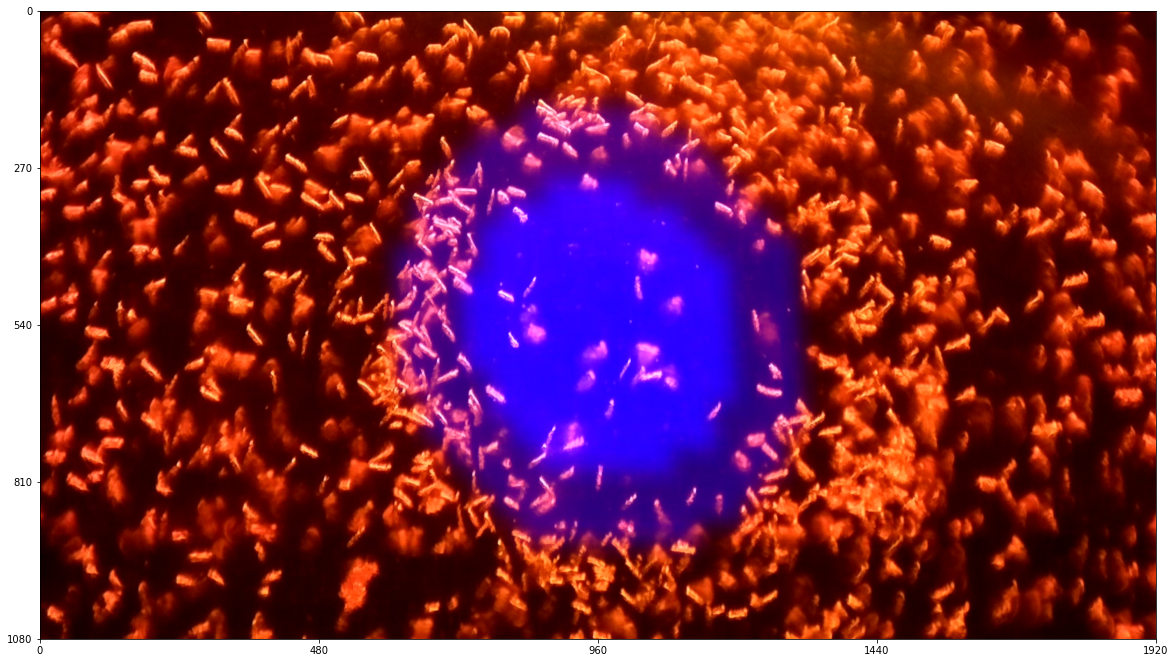}
        \caption{\textit{Euglena}}
    \end{subfigure}
    \begin{subfigure}[t]{0.48\textwidth}
        \includegraphics[width=\linewidth]{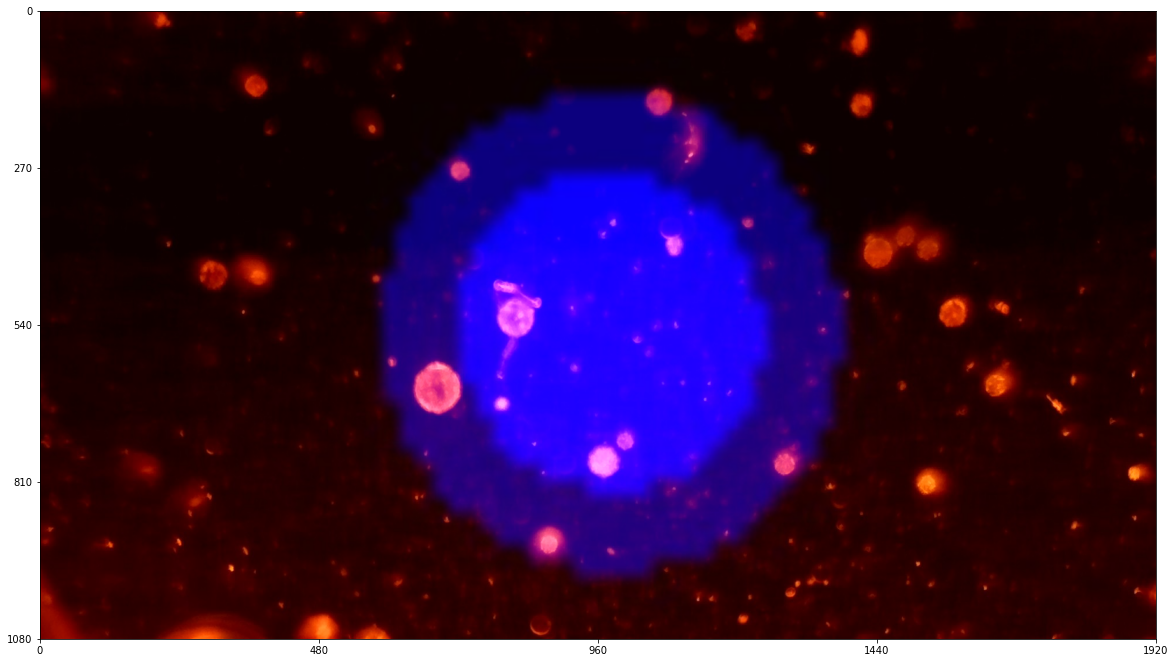}
        \caption{\textit{Volvox}}
    \end{subfigure}

    \begin{subfigure}[t]{0.48\textwidth}
        \includegraphics[width=\linewidth]{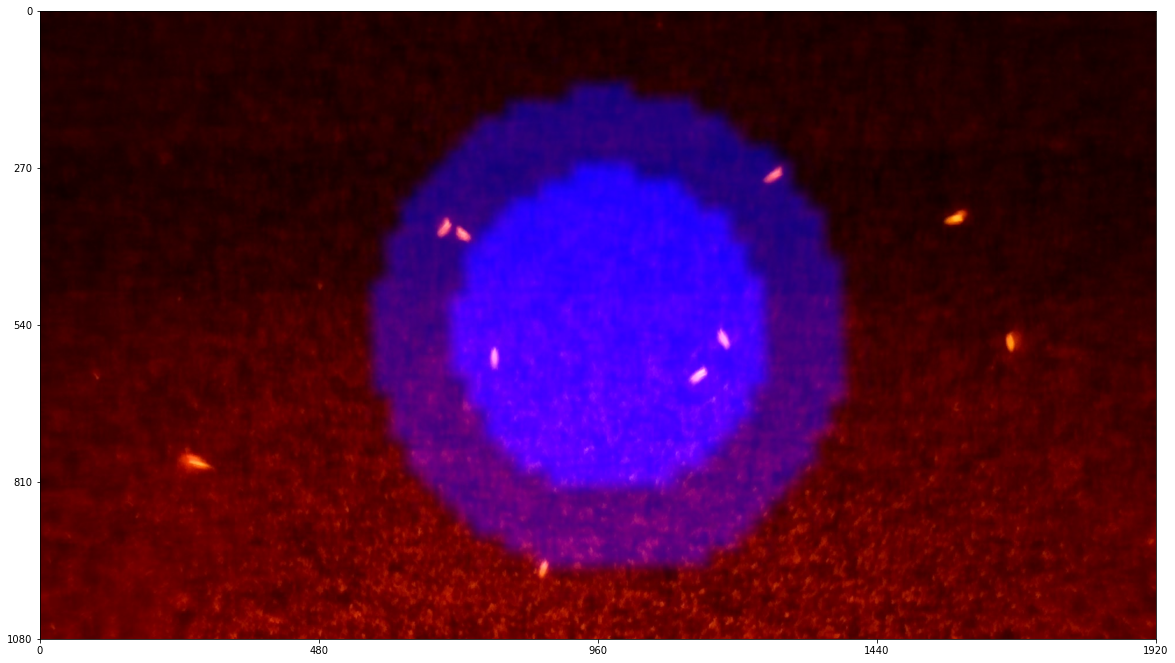}
        \caption{\textit{P. bursaria}}
    \end{subfigure}
    \begin{subfigure}[t]{0.48\textwidth}
        \includegraphics[width=\linewidth]{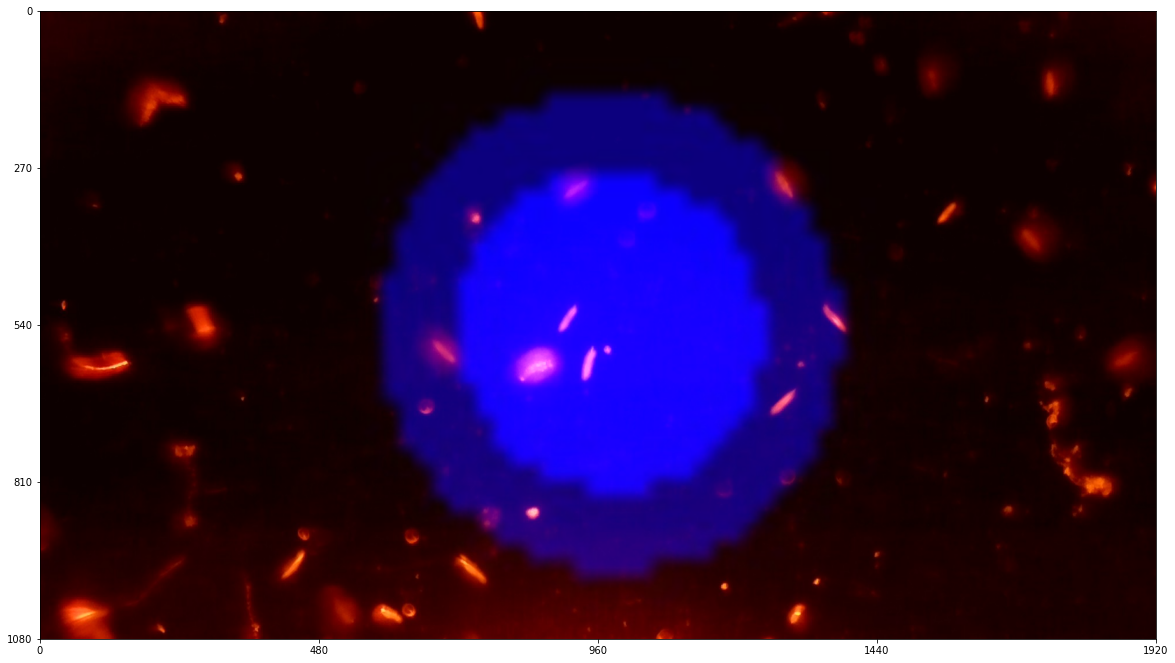}
        \caption{\textit{P. caudatum}}
    \end{subfigure}
    \caption{Response of microorganisms to a lit circle on dark background. 
    The experimental conditions are those reported in Table \ref{tab:exp_param}.
    The projected blue light pattern is not present in the original images, thanks to the red light filter, and was added digitally.
    }
    \label{fig:spatial_inputs_example}
\end{figure}

\begin{figure}[t]
    \centering
    \includegraphics[trim={0, 0, 0, 0},clip,width=0.9\linewidth]{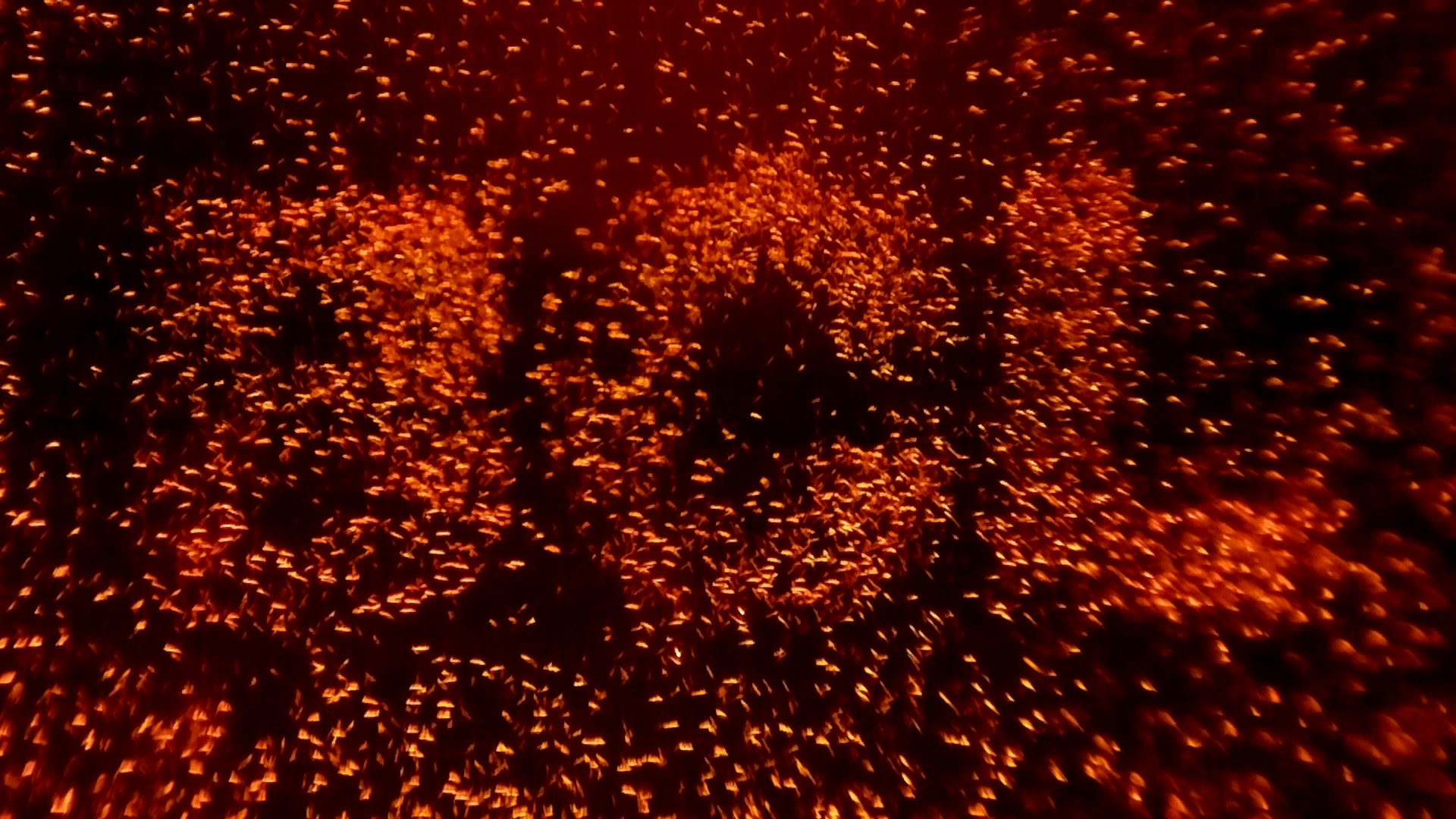}
    \caption{Spatial control with static input. A colony of \textit{Euglena} forms the letters "BCL", for BioCompute Lab.}
    \label{fig:BCL_Euglena}
\end{figure}

A thorough investigation of the control aspects will be carried out next, but is beyond the scope of this thesis and will be presented in future publications.

\section{Discussion}
In this Chapter we introduced a new model for stochastic mobile agents, extending the classic \acl{PTW} model with a second \acl{SDE} modelling the longitudinal speed, and inputs to take into account the effect of light.
This allowed to capture both the movement of microorganisms in a dark environment, and their complex response to light inputs, showing a combination of step-up photophobic, negative photokinetic and positive photoklinokinetic responses, further enriched by adaptation.
The model has a generic structure and can be easily parameterized from the experimental data, meaning that it may suite a variety of mobile micro-agents and effectively describe their motion.
Finally, we discussed how to proceed in the direction of developing new strategies to control the distribution of microorganisms in both space and time.
Crucially, relatively simple but descriptive models, which allow both simulations and model based design, can be instrumental for the development of  more effective feedback strategies.

To date, the model has been tested and validated only on \textit{Euglena}, but we are currently working to extending the application of this methodology to obtain  models also for the other species in our data-set, i.e. \textit{Volvox}, \textit{Paramecium bursaria} and \textit{Paramecium caudatum}.
Moreover, we aim to study how macroscopic density regulation (i.e. photoaccumulation or photodispersion) emerges from the individual responses of the organisms.

\chapter{Conclusions}
\thispagestyle{empty} 
\label{ch:conclusions}

This thesis explored the emerging behaviours displayed by \acf{LS-MAS}, particularly focusing on their spatial organization.
The ubiquity of such systems, made of several interacting agents, spans from natural phenomena to artificial systems.
These reflected into the two application domains discussed in this work, spatial control of biological populations and swarm robotics. 
The motivation stemmed from the many open challenges and the untapped possibilities, generated by the intricate dynamics of such systems.

\paragraph{Contributions.}
Our study began in Chapter \ref{ch:background} by surveying the existing Literature to establish a robust foundation and better understand the peculiarities of \ac{LS-MAS}. 
This preliminary step helped identify critical gaps and challenges in modeling, analyzing, and controlling their emergent behaviors.

Then, in Part \ref{part:geometric_pattern}, focusing on geometric pattern formation, we devised a novel distributed control algorithm to achieve self-organization of swarms into specific geometric patterns, like triangular and square lattices, representing a valuable advancement in the field of spatial organization of swarming robots.
After introducing and validating the novel algorithm, we developed a formal proof of the convergence of swarming systems towards rigid lattices, when these are driven by suited virtual forces.
This validates the stability of many existing  solutions, thereby enhancing their applicability and reliability.
In Part \ref{part:spatial_microorganisms}, parting from geometric patterns, we ventured into the study of spatial behaviours of biological systems, particularly the movement of microorganisms and their response to light stimuli.
Starting from experimental data, this led to the quantitative characterization of the key features of the kinematics of the microscopic algae \textit{Euglena gracilis}, and how its motion is influenced by light.
The resulting insights represented the starting point for the development of a novel data-driven mathematical model that clearly captures the behavior of the studied microorganisms, both in terms of motion and light response.
This model lays the foundation for future research on the spatial control of light sensitive cellular populations.

\paragraph{Methodological Insights.}
Our approach merged formal and numerical tools, together with experiments.
A formal approach, based on graph and Lyapunov theories, allowed the development of formal guarantees for the stability of rigid lattices.
The experiments were carried out leveraging innovative platforms, particularly the Robotarium was used to validate our distributed control strategy on real swarming robots, while the \ac{DOME} allowed the study and the collection of a rich data-set on a range of light sensitive microorganisms. 
Lastly, numerical tools were used throughout the thesis to analyze data and carry out simulations.
This resulted in the development of a new software platform, SwarmSim, to simulate mobile multi-agent systems, and the software to use the \ac{DOME}, which also includes a computer vision package to automatically detect and track moving objects.
These tools were used to validate our results, while also paving the way for future advancements in the study of \ac{LS-MAS} and their spatial behaviours.

\paragraph{Implications and Future Directions.}
The advancements made in this thesis have significantly contributed to understanding and addressing the challenges posed by \ac{LS-MAS}. However, certain critical aspects remain as open questions, presenting avenues for future research and exploration.
As regards geometric pattern formation, we contributed to the development of both novel control algorithms and formal proofs; nevertheless much remains to be addressed, aiming for a general control approach and the corresponding proof of convergence, so as to guarantee the formation of any required lattices. 
Looking at the behaviour of biological agents, we presented our first results, that point in the right direction, towards the development of integrated methodologies to model and control the movement and the emerging spatial behaviours of light sensitive microorganisms.
To this aim, it is needed to extend and improve our model, and the connected parametrization procedure, so as to be able to describe a wider range of micro-agents, and finally to control their distribution.
In a more general prospective, despite the numerous advancements achieved by researchers, obtaining a full understanding of the multiple scales in which \ac{LS-MAS} operate, and their interplay, remains a crucial goal.

\paragraph{Closing Remarks.}
In conclusion, this thesis, by addressing critical challenges in modeling, analyzing, and controlling  emergent behaviors of \acl{LS-MAS}, represents a step forward in unraveling their inherent complexity.
The work presented therein stemmed from a three years long journey that took place between the Sincro research group, at the University of Naples Federico II, and the BioCompute Lab at the University of Bristol.
Specifically, the work reported in Part \ref{part:geometric_pattern} was developed in collaboration with Giancarlo Maffettone, Marco Coraggio and Davide Fiore, while the experimental work presented in Part \ref{part:spatial_microorganisms} was carried out during a nine months period at the BioCompute Lab (University of Bristol), under the supervision of Dr. Thomas Gorochowski.

Now, as we look towards the horizon, the journey continues, brimming with opportunities to explore, innovate, and refine our understanding of \acl{LS-MAS}, propelling us towards a future where their potential is fully realized across diverse domains.

\appendix 

\chapter{SwarmSim: a framework for Multi-Agent Systems simulation}
\thispagestyle{empty} 
\label{ch:swarmsim}

\begin{figure}[t]
    \centering
    \includegraphics[width=0.6\linewidth]{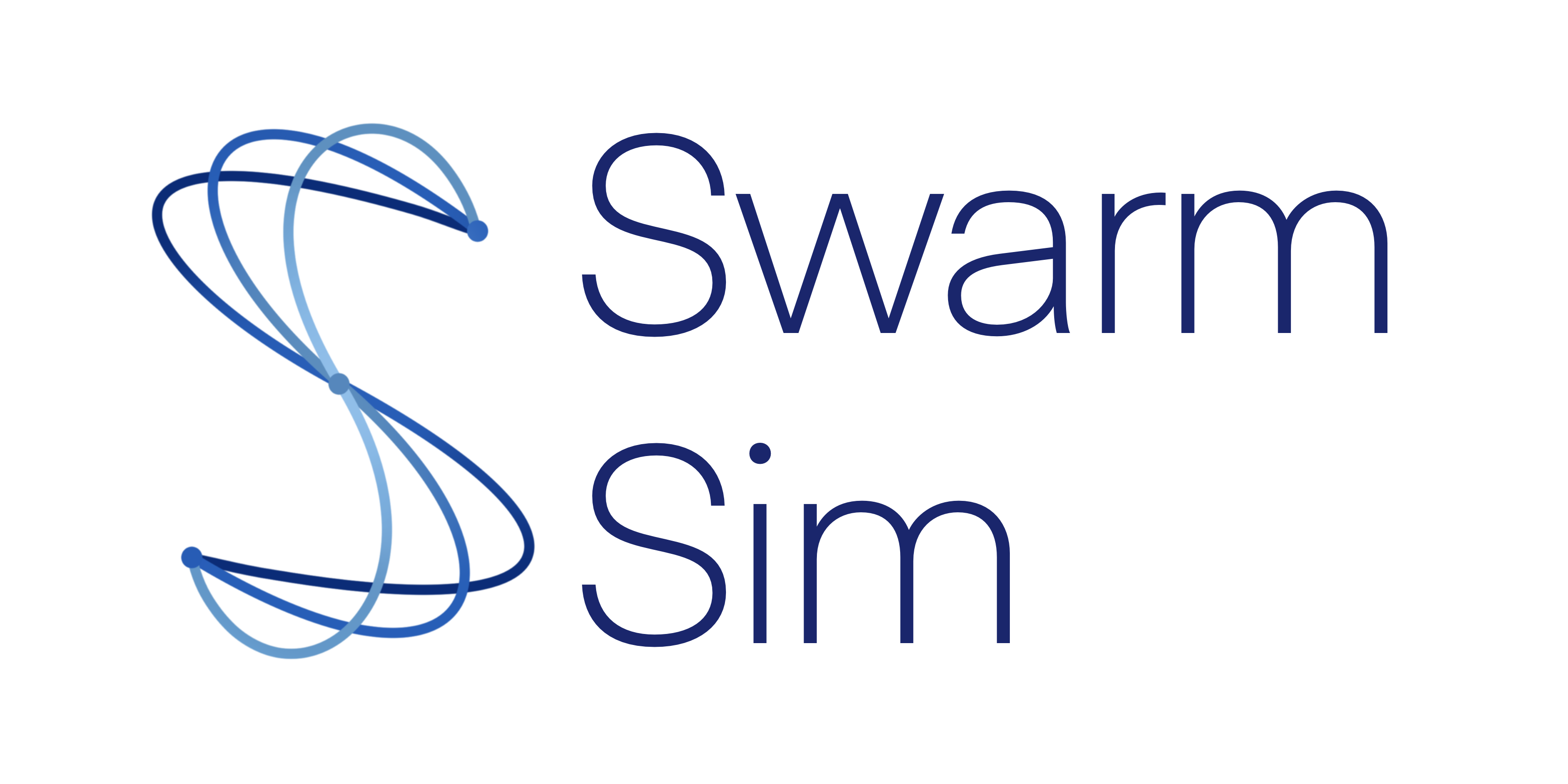}
    \caption{Logo of SwarmSim}
    \label{fig:swarmsim_logo}
\end{figure}

SwarmSim is a {\sc{Matlab}} software package crafted for the simulation of swarms of mobile agents, in both 2D and 3D environments. We developed it to provide researchers with an intuitive and flexible platform to analyze the intricate dynamics of multi-agent systems and test new control algorithms. This chapter summarize the utility, functionalities, and significance of SwarmSim. 
The up to date software and the complete documentation are available online at \url{www.github.com/diBernardoGroup/SwarmSimPublic}.
Figure \ref{fig:swarmsim_logo} depicts the logo of SwarmSim.

\section{Core functionalities and practical utility}
The software operates on a modular architecture that provides flexibility and adaptability. Central to its functionality is the ability to manipulate both the dynamics and the interactions governing agents' behavior. Users can either implement new dynamical models, tailored to their specific needs, or use pre-existing models embedded within the system, such as first and second order integrators, Lévy walkers \cite{Zaburdaev2015}, \acl{PTW} agents \cite{Gautrais2009} (see Figure \ref{subfig:swarmsim_examples_PTW}), etc. 
A central aspect of SwarmSim is its support for distributed control laws. 
Also in this case, users can either leverage embedded control algorithms based on virtual forces, such as those for geometric pattern formation discussed in Chapters \ref{ch:dist_cont} and \ref{ch:convergence} (see Figures \ref{subfig:swarmsim_examples_2D} and \ref{subfig:swarmsim_examples_3D}) or in \cite{Spears2004,Torquato2009}, or implement new ones.
This enables the observation and analysis of emergent phenomena stemming from the interactions among agents.

SwarmSim also facilitates the evaluation of performance metrics, providing quantitative measures to assess and compare different  scenarios (see Figure \ref{subfig:swarmsim_examples_metrics}). This capability is instrumental in gauging the efficacy of control strategies, evaluating their robustness, scalability and flexibility.
Moreover, the software supports extensive parallel simulation, allowing to easily scrutinize stochastic effects, explore diverse initial conditions, and investigate the influence of varying parameters on the resulting system's behavior (see Figure \ref{subfig:swarmsim_examples_parameters}).
The integration of local stability analysis through linearization further enriches the software, by enabling the exploration the system's stability around specific configurations.
Finally, SwarmSim includes functions to visualize all relevant information, such as the position of the agents, their trajectories and interactions, the time evolution of relevant quantities or the relation between selected parameters and output metrics.
This software was used for all the simulations showed in this thesis, to validate control algorithms for geometric pattern formation in Chapters \ref{ch:dist_cont} and \ref{ch:convergence}, and to simulate the movement of microorganisms in Chapter \ref{ch:modelling}.

Compared to other simulators, such as SwarmLab \cite{Soria2020}, tailored on specific applications (i.e. flocking) and pre-defined dynamics, SwarmSim offers improved flexibility, in terms of both agents dynamics and emerging behaviour.

\begin{figure}[t]
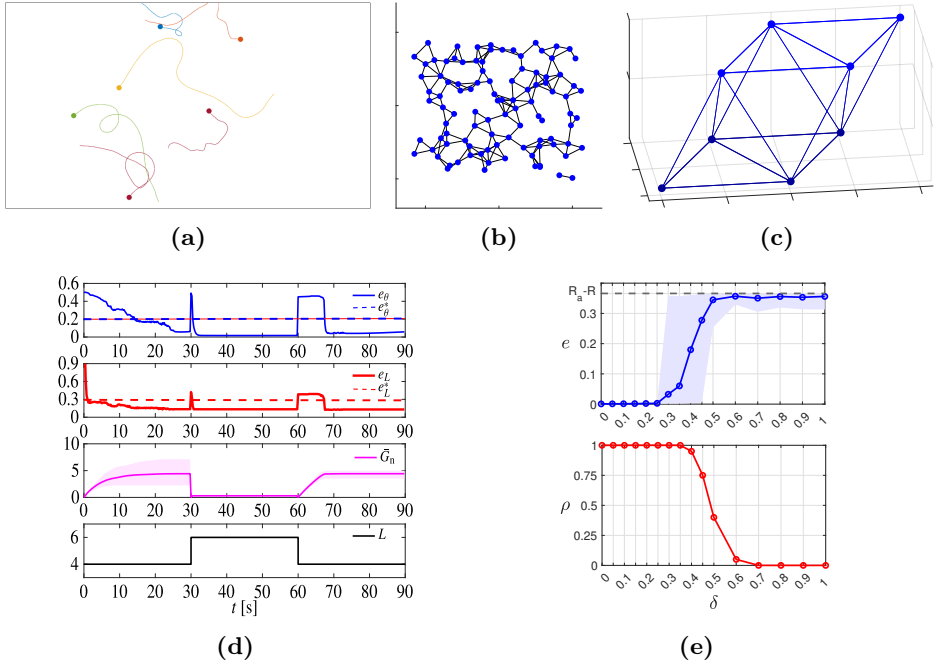

    \centering
    \begin{subfigure}[t]{0.4\textwidth}
        \centering
        \includegraphics[height=2.7cm]{figs/Euglena_trajectories_simulated.png}
        \caption{}
        \label{subfig:swarmsim_examples_PTW}
    \end{subfigure}
    \begin{subfigure}[t]{0.24\textwidth}
        \centering
        \includegraphics[height=2.7cm]{figs/d=06_x_0.pdf}
        \caption{}
        \label{subfig:swarmsim_examples_2D}
    \end{subfigure}
    \begin{subfigure}[t]{0.33\textwidth}
        \centering
        \includegraphics[height=2.7cm]{figs/lattice_3D_N=8.pdf}
        \caption{}
        \label{subfig:swarmsim_examples_3D}
    \end{subfigure}
    
    \begin{subfigure}[t]{0.48\textwidth}
        \centering
        \includegraphics[height=0.76\linewidth]{figs/adaptive_tuning_dynamic_lattice.pdf}
        \caption{}
        \label{subfig:swarmsim_examples_metrics}
    \end{subfigure}
    \begin{subfigure}[t]{0.48\textwidth}
        \centering
        \includegraphics[height=0.8\linewidth]{figs/e_rho_2D_small.pdf}
        \caption{}
        \label{subfig:swarmsim_examples_parameters}
    \end{subfigure}
    
    \caption{Example of figures generated by SwarmSim. 
    (a) Trajectories of \aclp{PTW}.
    (b) Agents in 2D space with interaction network.
    (c) Agents in 3D space with interaction network.
    (d) Time evolution of selected quantities of interest.
    (e) Dependency of two selected quantities of interest with respect to a parameter.
    }
    \label{fig:swarmsim_examples}
\end{figure}

\section{Discussion}

In conclusion, SwarmSim supports researchers by enabling exploration and analysis of emergent behaviors in multi-agent systems. Furthermore, the software serves as a valuable tool to test and validate novel distributed control algorithms.
SwarmSim, while being developed for swarm robotics, can result useful for different disciplines, such as biology and complex systems analysis. Additionally, its intuitive interface makes it a suitable tool for introducing students to the simulation of multi-agent systems.

\chapter{Tracking software}
\thispagestyle{empty} 
\label{ch:tracker}



Being able to automatically extract relevant information from images and video is a critical feature in many recent applications, including our research.
Therefore, together with the \ac{DOME} software, presented in Section \ref{sec:dome_software}, we developed an automatic tracking software. It allows to analyze the images acquired from a static camera to automatically detect the moving objects and track their movement over time. 
It combines computer vision techniques to detect the moving objects, with estimation and optimization algorithms, to assign them consistent ids over time. 
This software is written in Python using the OpenCV library.

The algorithm is made of three main steps, \emph{background modelling}, \emph{objects detection} and \emph{objects tracking}, as depicted in Figure \ref{fig:tracker}.
In the following we discuss the function and the structure of each of these steps.

\begin{figure}
    \centering
    \includegraphics[width=1\linewidth]{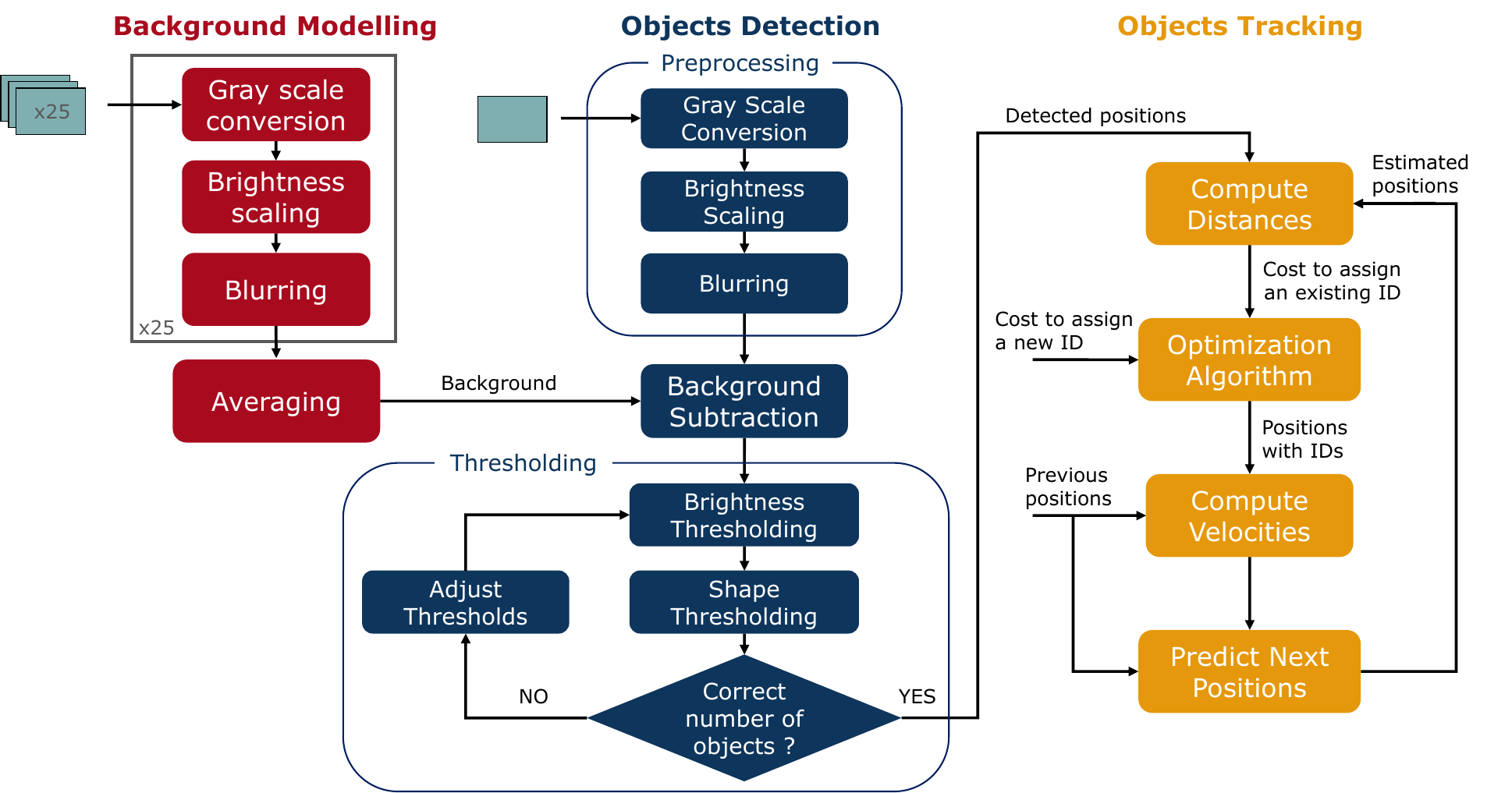}
    \caption{Schematic of the tracking algorithm.}
    \label{fig:tracker}
\end{figure}

\section{Background modelling}
Given a collection of images acquired during an experiment the software first builds a model of the background, that is an image representing all and only the elements that are static.

To do this a subset of images is selected, specifically we used $N_{\text{BG}}$ images uniformly distributed during the experiment.
Each of these images is converted to grey scale, according to the \texttt{channel} parameter, which allows to select either one of the three color channels (i.e. red, green and blue) or the combined brightness.
Then the brightness is adjusted so that the darkest pixel becomes the new zero, while the brightest one becomes the new one.
Finally, a median blurring is applied to remove noise and smooth the contours of the objects.

These image are then collapsed into a single one, by a pixel-wise median over the time dimension.
The resulting grey scale image represents the model of the background, that will be used for the background subtraction.
Figure \ref{fig:background_modelling} shows an example of this process.

\begin{figure}[t]
    \centering
    \includegraphics[width=1.05\linewidth]{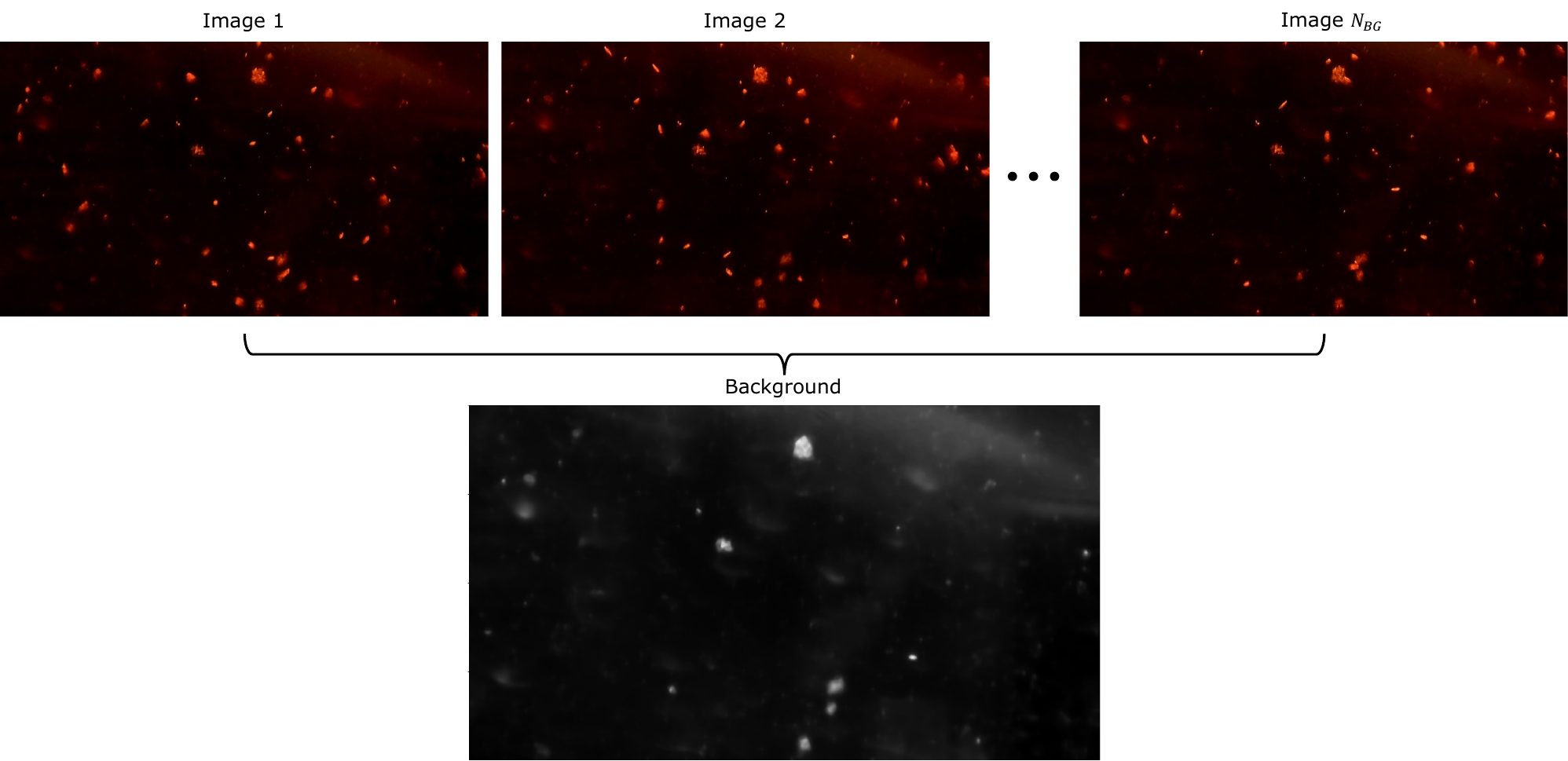}
    \caption{Example of background modelling.}
    \label{fig:background_modelling}
\end{figure}

\section{Objects detection}
After the background modelling, the software analyses the images one by one, in temporal order, to get the positions of the moving objects. 
Each frame undergoes the same pre-processing steps applied during the background modelling, namely grey scale conversion, brightness adjustment and blurring. This ensures the images stay consistent with the background. 
Then, the foreground, that represents only the moving objects, is computed by pixel-wise \emph{background subtraction} as $foreground=\max(0, frame - background)$.

The actual image segmentation is then performed by pixel-wise brightness thresholding, that is $mask = foreground\geq b_{\min}$. The \texttt{findContours} function is then applied to the resulting black and white mask, to detect the contours of the objects.
For each of the detected objects we compute the area $a$, the perimeter $p$ and the Polsby–Popper compactness measure  \cite{Polsby1991}, as $c=(4 \pi a) / (p^2) \in [0;1]$.
A second selection happens by retaining only the objects whose area and compactness fall in the given ranges, $[a_{\min}; a_{\max}]$ and $[c_{\min}; c_{\max}]$ respectively.
This allows to discern the objects we are looking for from debris or elements of the background not completely removed.
Figure \ref{fig:objects_detection} shows the objects detection process step by step.

\begin{figure}[t]
    \centering
    \begin{subfigure}[t]{0.48\linewidth}
        \centering
        \includegraphics[trim={40, 20, 10, 20},clip,width=\linewidth]{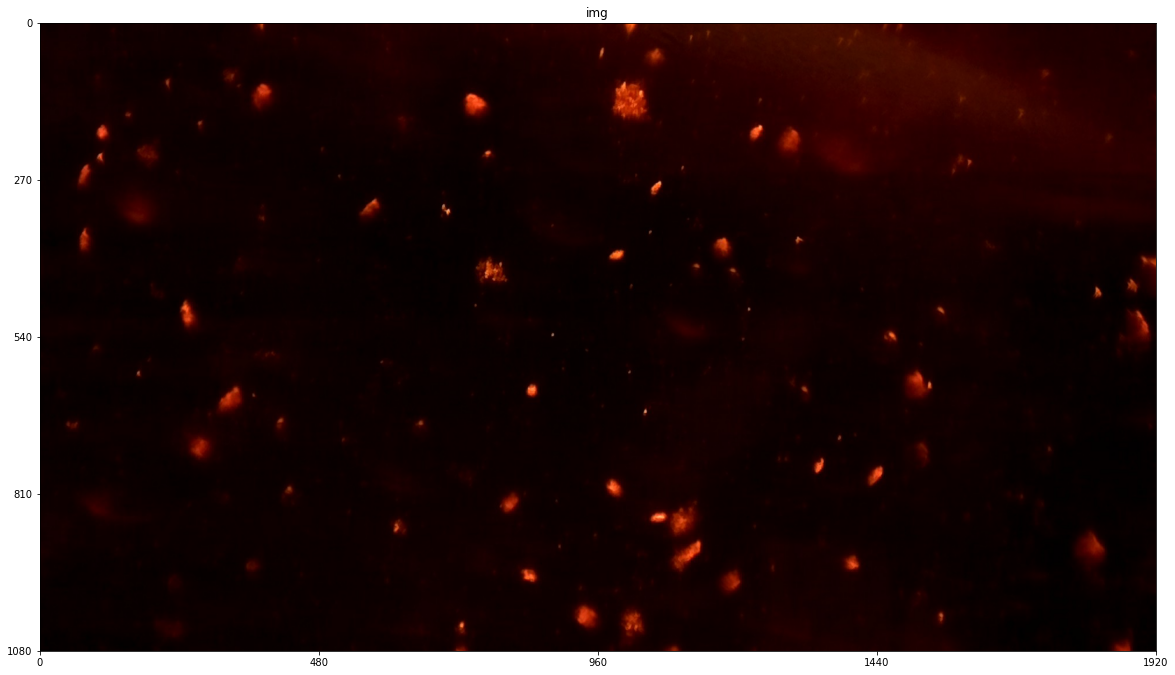}
        \caption{Image}
    \end{subfigure}
    \begin{subfigure}[t]{0.48\textwidth}
        \centering
        \includegraphics[trim={40, 20, 10, 20},clip,width=\linewidth]{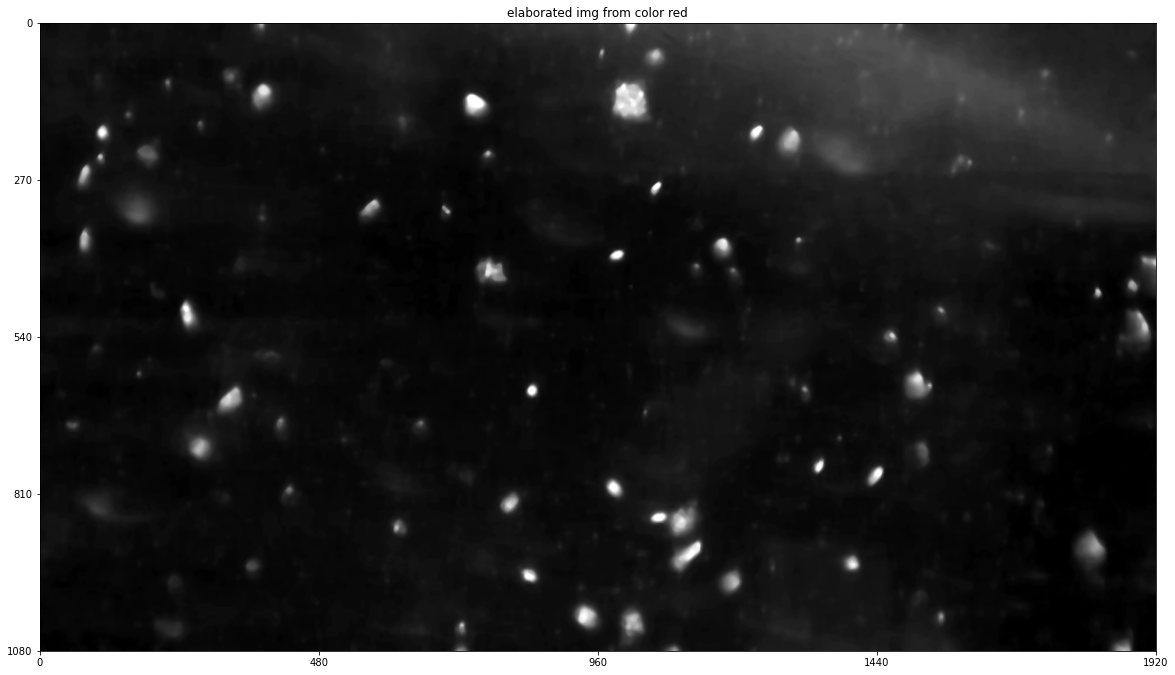}
        \caption{Processed image}
    \end{subfigure}
    
    \begin{subfigure}[t]{0.48\textwidth}
        \centering
        \includegraphics[trim={40, 20, 10, 20},clip,width=\linewidth]{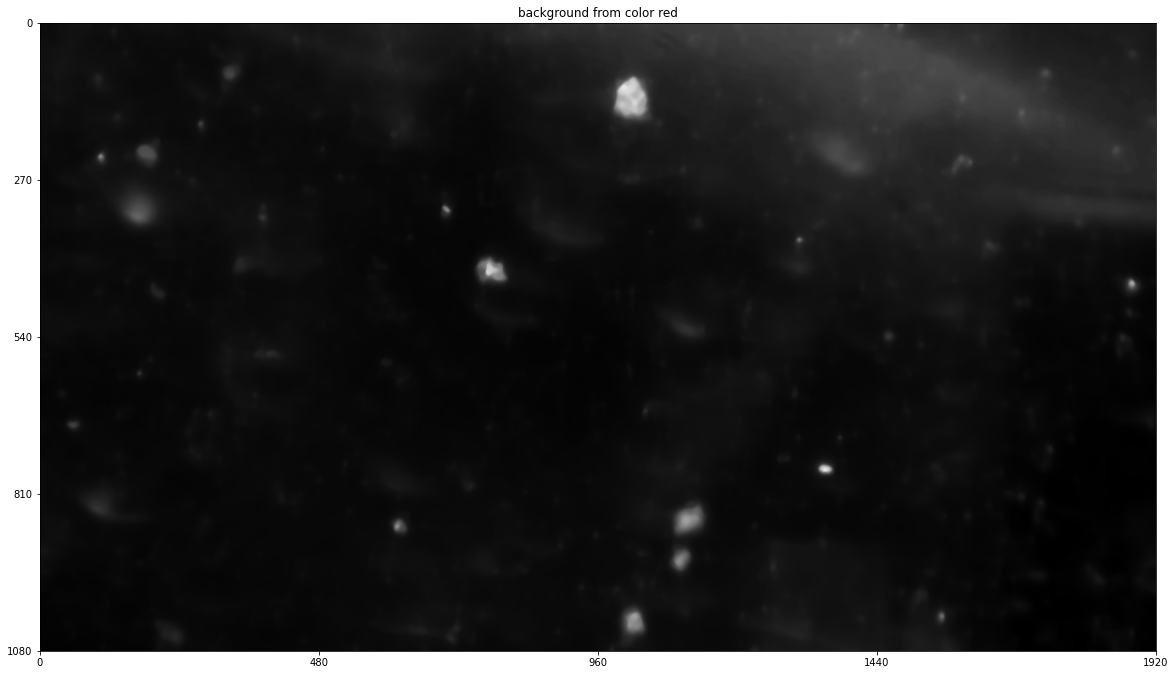}
        \caption{Background}
    \end{subfigure}
    \begin{subfigure}[t]{0.48\textwidth}
        \centering
        \includegraphics[trim={40, 20, 10, 20},clip,width=\linewidth]{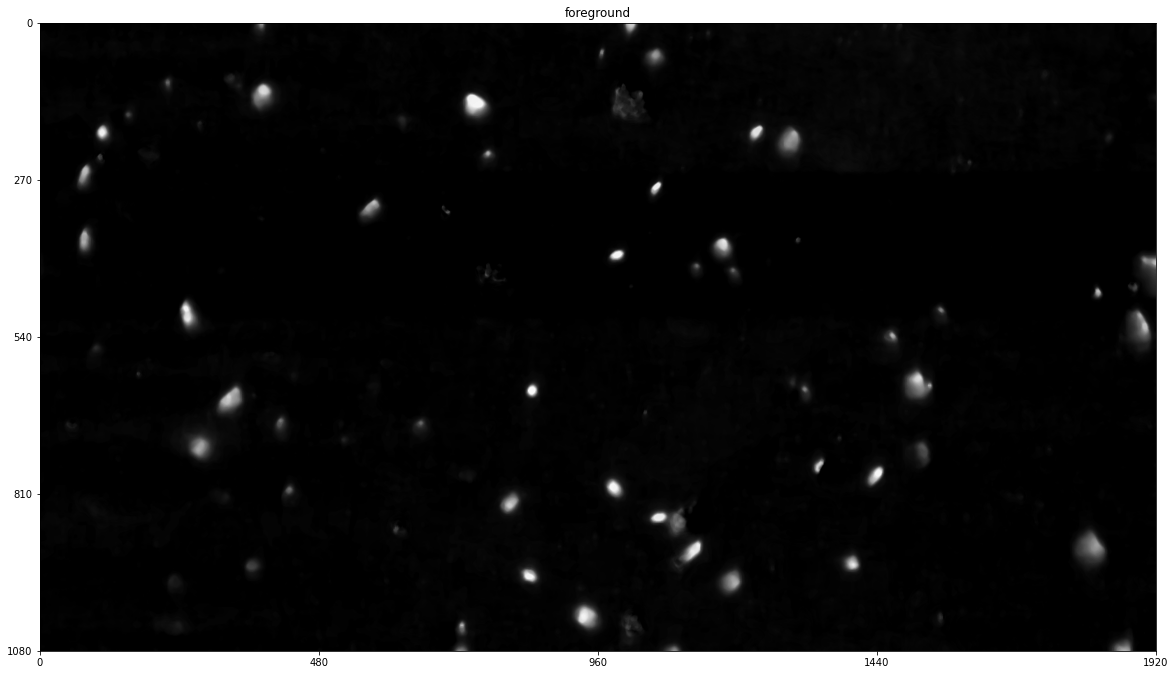}
        \caption{Foreground}
    \end{subfigure}
    
    \begin{subfigure}[t]{0.48\textwidth}
        \centering
        \includegraphics[trim={40, 20, 10, 20},clip,width=\linewidth]{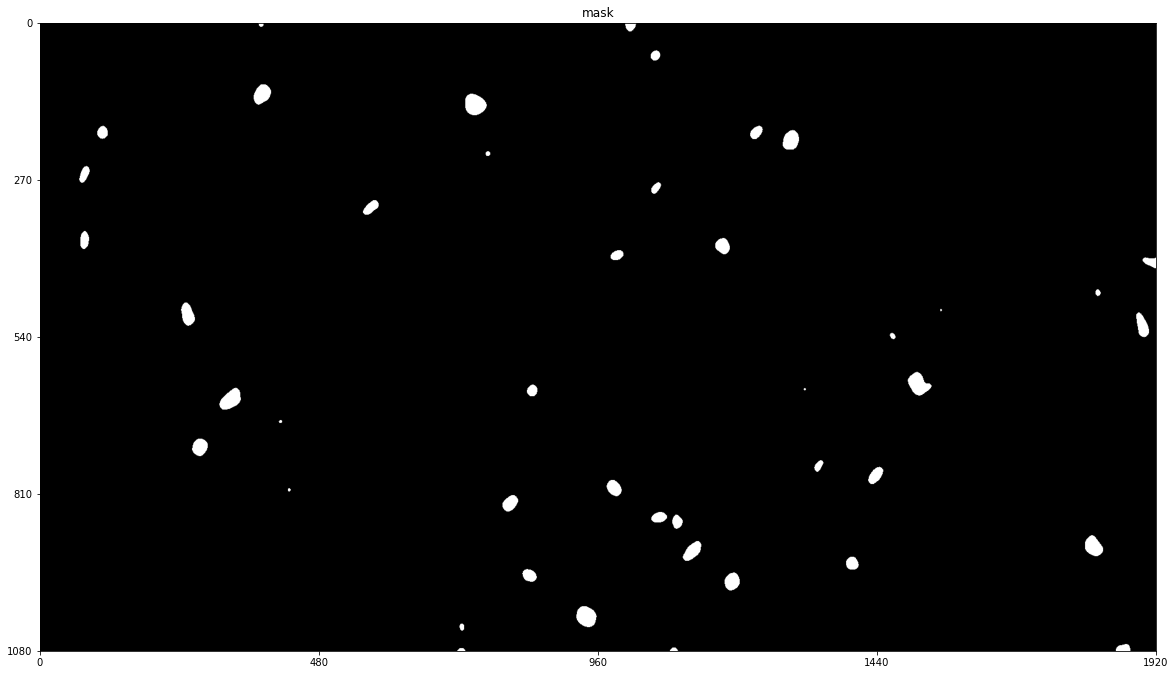}
        \caption{Mask}
    \end{subfigure}
    \begin{subfigure}[t]{0.48\textwidth}
        \centering
        \includegraphics[trim={40, 20, 10, 20},clip,width=\linewidth]{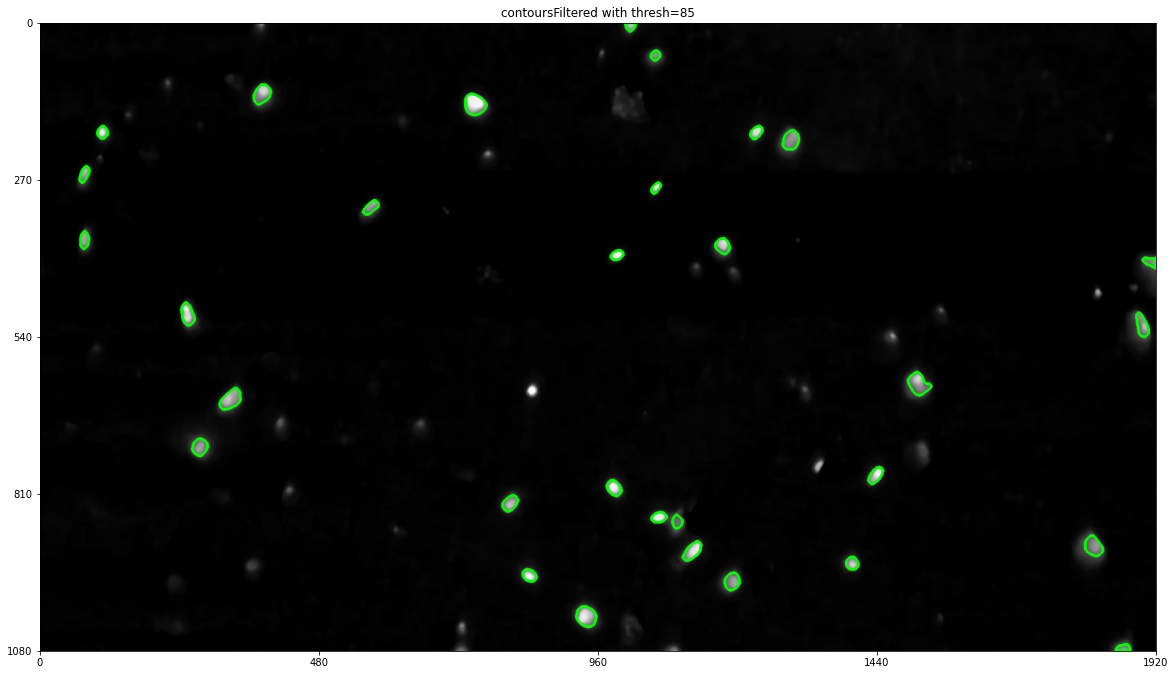}
        \caption{Detected objects}
    \end{subfigure}
    
    \caption{Step by step image elaboration for object detection. 
    (a) Original image.
    (b) Processed image by selecting the red channel, adjusting the brightness and applying the blur.
    (c) Model of the background.
    (d) Foreground obtained by subtracting the background from the processed image.
    (e) Black and white mask obtained by thresholding the brightness of the pixels.
    (f) Contours of the detected objects, after the selection based on area and compactness.
    These images were acquired by the \ac{DOME}, during an experiment with \textit{Euglena}.
    }
    \label{fig:objects_detection}
\end{figure}

The number of detected objects is then compared with the expected one, if it is not farther than $\delta$ from the expected one the detection is considered successful, and the positions of the centers of the detected objects are returned.
Otherwise, all the thresholds $a_{\min}, a_{\max}, c_{\min}, c_{\max}$ and $b_{\min}$ are relaxed or tightened by a factor $\lambda$, depending on whether the number of detected objects is smaller or larger than the expected one.
The detection is then repeated with the new thresholds.
This iterative procedure, should the condition on the number of detected objects not be satisfied, is forcibly stopped after a maximum number of iterations.
For the first frame the expected number of objects is set equal to the number of detected ones, so that the condition is automatically satisfied, while for subsequent frames the expected number of objects is set to the number of objects detected in the previous frame.
This automatic adaptation of the thresholds helps in keeping a more consistent detection, also in the presence of changes in the illumination.

\section{Objects tracking}
The last step consists in assigning ids to the detected objects and keeping these consistent over time.
Given the current detected positions and the estimated ones the distance between the two, for all the possible couples, is computed.
The matching cost is then defined as the square of the distance.
Moreover, a cost for the allocation of a new id is computed for each of the detected objects, depending on the distance from the edges of the camera, so that objects close to an edge have a lower cost for the allocation of a new id.

The Jonker-Volgenant optimization algorithm \cite{Crouse2016} is then used to solve the assignment problem, while minimizing the total cost.
Specifically, each detected object will be given an id, either an existing or a new one.

Finally, for each object the position at the next time step is estimated as $\vec{p}_{k+1} = \vec{p}_{k} + \gamma (\vec{p}_{k}-\vec{p}_{k-1})$, where $\gamma \in \BB{R}_{\geq0}$ is the inertia parameter.
The estimated position will be used at the next time step to perform the id matching.

When an existing id is not assigned to any of the detected objects it is not discarded and its estimated position is updated as usual, but its  \emph{inactivity counter} is incremented by one.
The inactivity counter is then used when computing the matching costs, so that longly inactive objects have a larger matching cost. Specifically this additional cost increases with the square of the inactivity counter.
This allows to keep a consistent tracking of the object even if it "disappears" for one or more consecutive frames, that might happen when it overlaps with another object or something in the background.
Figure \ref{fig:tracking_example} shows an example application of this algorithm.

\begin{figure}[t]
    \centering
    \includegraphics[width=\linewidth]{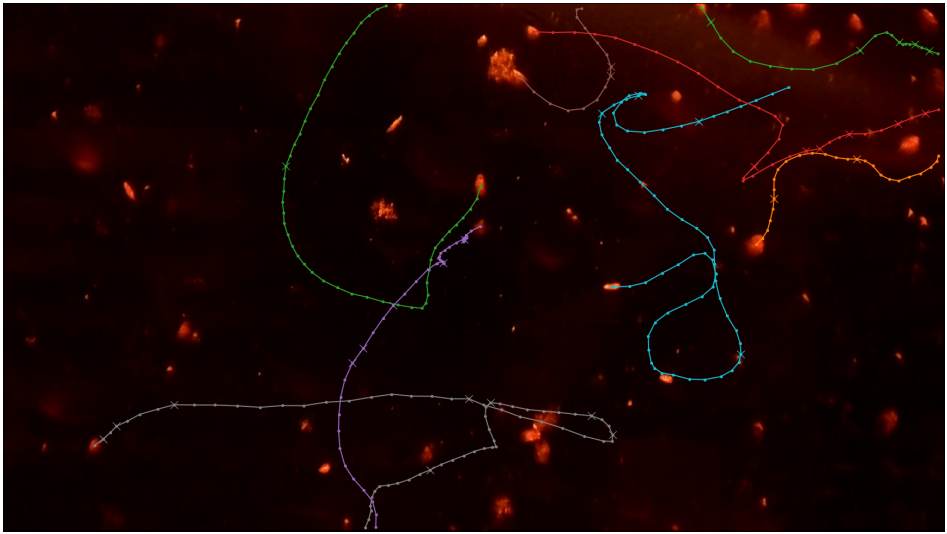}
    \caption{Trajectories generated by the tracking software during an experiment with \textit{Euglena}.}
    \label{fig:tracking_example}
\end{figure}

\section{Discussion}
This software allows to automatically detect and track moving objects, using images acquired from a static camera.
Despite not leveraging the most modern techniques, based on machine learning, it provides satisfactory results, with an open door for further improvements.
It was originally developed to analyze data from the \ac{DOME} experiments, but the generality of the algorithm makes it suitable for different applications.
In particular it was designed to detect bright objects moving on a dark background, but could be easily adapted to work in the opposite situation.
Table \ref{tab:tracking_parameters} shows the values of the parameters used to track the \textit{Euglena}, using the images from the \ac{DOME}.
Together with the rest of the software developed to use the \ac{DOME}, this is currently being refined for future publication.

\begin{table}
    \centering
    \begin{tabular}{l l l}
        \hline
        Parameter   & Description & Value\\
        \hline
        \texttt{channel}& Color channel used for gray-scale conversion & Red\\
        $N_{\text{BG}}$ & Number of images used for background modelling & 25\\
        $b_{\min}$  & Brightness threshold & 85/255\\
        $a_{\min}$  & Minimum area & 175\,px\\
        $a_{\max}$  & Maximum area & 1500\,px\\
        $c_{\min}$  & Minimum compactness & 0.55\\
        $c_{\max}$  & Maximum compactness  & 0.90\\
        $\delta$    & Tolerance on the number of detected objects & 25\%\\
        $\lambda$   & Thresholds adjustment gain & 2\%\\
        $\gamma$    & Inertia parameter for position estimation& 0.9\\
        \hline
    \end{tabular}
    \caption{Parameters used to track \textit{Euglena}.}
    \label{tab:tracking_parameters}
\end{table}

\clearpage

\thispagestyle{empty} 
\vspace*{.3\textheight}
\begin{center}
{\fontfamily{calligra}\selectfont
{\LARGE
{\color{main_color}
The End
}
}
}
\end{center}
\thispagestyle{empty} 

\cleardoublepage

\small
\bibliographystyle{plain}
\singlespacing
\setlength{\parskip}{0.2em}
\thispagestyle{empty} 
\paragraph{}
\addcontentsline{toc}{chapter}{Bibliography}
\bibliography{support/LSMAS, support/Synt_Bio}
\cleardoublepage

\chapter*{Author's publications}
\addcontentsline{toc}{chapter}{Author's publications}
\thispagestyle{empty} 
\label{ch:author_publications}

\noindent The results presented in Chapters \ref{ch:dist_cont} and \ref{ch:convergence} of this thesis were published in:

\begin{enumerate}
\item A. Giusti, M. Coraggio and M. di Bernardo, 2023, "Local Convergence of Multi-Agent Systems Toward Rigid Lattices" in \textit{IEEE Control Systems Letters}, vol. 7, pp. 2869-2874, doi: \href{https://ieeexplore.ieee.org/abstract/document/10160116}{10.1109/LCSYS.2023.3289060}.

\item A. Giusti, GC. Maffettone, D. Fiore, M. Coraggio and M. di Bernardo, 2023, "Distributed control for geometric pattern formation of large-scale multirobot systems" in \textit{Frontiers in Robotics and AI}, doi: \href{https://www.frontiersin.org/articles/10.3389/frobt.2023.1219931}{10.3389/frobt. 2023.1219931}.
\end{enumerate}

Moreover, the content of Chapters \ref{ch:dome} and \ref{ch:modelling} is currently being refined for publication.
\cleardoublepage

\chapter*{External resources}
\addcontentsline{toc}{chapter}{External resources}
\thispagestyle{empty} 
\label{ch:external_res}

\noindent For up to date information about the following items, visit the corresponding web page.

\begin{itemize}
\item The DOME \url{www.theopendome.org}.

\item SwarmSim \url{www.github.com/diBernardoGroup/SwarmSimPublic}.

\item SINCRO research group at the University of Naples Federico II \\ \url{www.sites.google.com/site/dibernardogroup}.

\item BioCompute Lab at the University of Bristol \\ \url{https://biocomputelab.github.io}.

\end{itemize}
\cleardoublepage

\end{document}